\documentclass[12pt,oneside,american,intoc,index=totoc,BCOR=10mm,captions=tableheading,titlepage,breaklinks=true,headsepline=0.5pt]{scrbook}
\usepackage{lmodern}

\usepackage[T1]{fontenc}
\usepackage[latin9]{inputenc}
\usepackage[a4paper]{geometry}
\usepackage{scrlayer-scrpage}
\addtokomafont{pagehead}{\small}
\usepackage{xcolor}
\usepackage{scrhack}
\usepackage{babel}
\usepackage{array}
\usepackage{verbatim}
\usepackage{float}
\usepackage{units}
\usepackage{textcomp}
\usepackage{multirow}
\usepackage{amsmath}
\usepackage{amsthm}
\usepackage{amssymb}
\usepackage{graphicx}
\usepackage{setspace}
\PassOptionsToPackage{normalem}{ulem}
\usepackage{ulem}
\usepackage{nomencl}
\providecommand{\printnomenclature}{\printglossary}
\providecommand{\makenomenclature}{\makeglossary}
\makenomenclature
\usepackage{hyperref}
\hypersetup{
 pdfpagelayout=OneColumn, pdfnewwindow=true, pdfstartview=XYZ, plainpages=false}

\makeatletter

\newcommand{\lyxmathsym}[1]{\ifmmode\begingroup\def\b@ld{bold}
  \text{\ifx\math@version\b@ld\bfseries\fi#1}\endgroup\else#1\fi}

\providecommand{\tabularnewline}{\\}
\providecolor{lyxadded}{rgb}{0,0,1}
\providecolor{lyxdeleted}{rgb}{1,0,0}

\DeclareRobustCommand{\lyxsout}[1]{\ifx\\#1\else\sout{#1}\fi}

\usepackage[figure]{hypcap}

\let\myTOC\tableofcontents
\renewcommand\tableofcontents{%
  \pdfbookmark[1]{\contentsname}{}
  \myTOC
  }

\addto\captionsamerican{ 

}
\setkomafont{captionlabel}{\bfseries}
\setcapindent{1em}

\usepackage{calc}

\usepackage{indentfirst}

\usepackage{jabbrv}

\let\mySection\section\renewcommand{\section}{\suppressfloats[t]\mySection}

\usepackage[Lenny]{fncychap}
\ChTitleVar{\huge\scshape\centering}

\usepackage{mcite}
\usepackage{cite}
\usepackage{graphicx}
\usepackage{makecell}
\usepackage{tensor}
\usepackage{gensymb}
\usepackage{bbold}
\usepackage{bm}
\usepackage{color}
\usepackage{array}
\usepackage{rotating}
\usepackage{multirow}
\usepackage{mathtools}
\usepackage{amsfonts,amstext,amsmath,amsthm,amssymb}
\usepackage{mathbbol,tensor,diagbox,multirow}
\usepackage{xcolor,mathrsfs,enumerate}
\usepackage{xcolor}
\usepackage{url}
\usepackage{times}
\usepackage{bbm}
\usepackage{soul}
\usepackage{chngcntr}

\newcommand{\bra}[1]{\langle #1|}
\newcommand{\ket}[1]{|#1\rangle}
\newcommand{\braket}[2]{\langle #1|#2\rangle}
\newcommand{\ketbra}[2]{| #1 \rangle \langle #2 |}

\DeclareMathOperator{\Tr}{Tr}

\newtheorem{thm}{Theorem}

\@ifundefined{showcaptionsetup}{}{%
 \PassOptionsToPackage{caption=false}{subfig}}
\usepackage{subfig}
\makeatother

\begin{document}
\thispagestyle{empty}

\begin{center}
\includegraphics[width=0.1\textwidth]{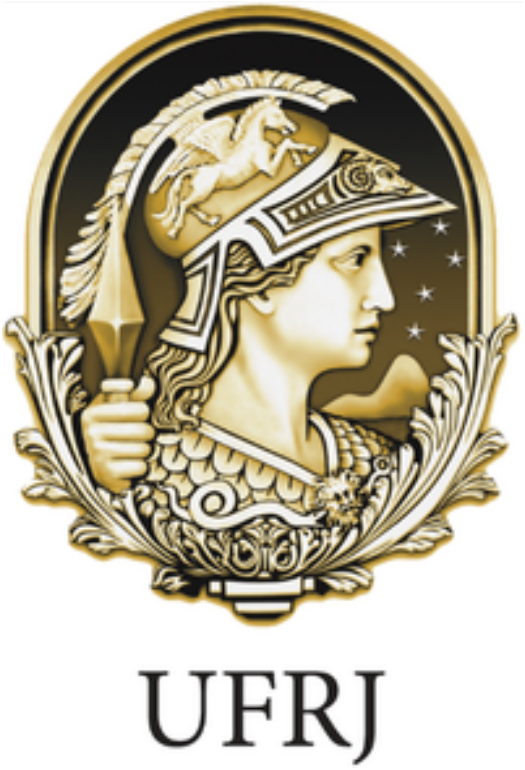}
\end{center}

\begin{center}
\begin{minipage}[t]{0.6\columnwidth}%
\begin{center}
UNIVERSIDADE FEDERAL DO RIO DE JANEIRO
\par\end{center}
\begin{center}
INSTITUTO DE FÍSICA
\par\end{center}%
\end{minipage}
\par\end{center}

\vspace{3cm}

\begin{center}
\textbf{\Large{}Classical and Quantum Light: }{\Large\par}
\par\end{center}

\begin{center}
\textbf{\Large{}Versatile tools for quantum foundations and quantum
information} 
\par\end{center}

\vspace{4cm}

\begin{center}
{\large{}Thais de Lima Silva}{\large\par}
\par\end{center}

\vspace{3cm}

\vfill{}

\begin{center}
Rio de Janeiro
\par\end{center}

\vspace{-1cm}

\begin{center}
Abril, 2020
\par\end{center}

\vspace{-2cm}

\cleardoublepage{}

\thispagestyle{empty}

\begin{center}
UNIVERSIDADE FEDERAL DO RIO DE JANEIRO
\par\end{center}

\begin{center}
INSTITUTO DE FÍSICA
\par\end{center}

\vspace{3cm}

\begin{center}
\textbf{\Large{}Classical and Quantum Light: }{\Large\par}
\par\end{center}

\begin{center}
\textbf{\Large{}Versatile tools for quantum foundations and quantum
information} 
\par\end{center}

\vspace{2cm}

\begin{center}
{\large{}Thais de Lima Silva}{\large\par}
\par\end{center}

\begin{center}
ORIENTADOR: Stephen Patrick Walborn
\par\end{center}

\begin{center}
CO-ORIENTADOR: Grabriel Horacio Aguilar
\par\end{center}

\vspace{3cm}

\begin{flushright}
\emph{}%
\begin{minipage}[c][1\totalheight][t]{0.6\columnwidth}%
\emph{Tese apresentada como parte dos requisitos para obtenção do
título de doutora em Física pelo programa de pós-graduação do Instituto
de Física da Universidade Federal do Rio de Janeiro.}%
\end{minipage}
\par\end{flushright}

\vfill{}

\begin{center}
Rio de Janeiro
\par\end{center}

\vspace{-1cm}

\begin{center}
April, 2020
\par\end{center}

\vspace{-2cm}

\cleardoublepage{}

\frontmatter

\thispagestyle{empty}%
\noindent\begin{minipage}[t]{1\columnwidth}%
\end{minipage}\vspace{5cm}

\noindent\begin{minipage}[t]{1\columnwidth}%
\begin{center}
À memória de meu tio Oswaldo Vicente de Lima.
\par\end{center}%
\end{minipage}

\vfill{}

\begin{flushright}
\begin{minipage}[t]{0.5\columnwidth}%
\begin{singlespace}
\noindent \begin{flushright}
\emph{Dedico-me sobretudo aos gnomos, anões, sílfides e ninfas que
me habitam a vida. Dedico-me à saudade de minha antiga pobreza, quando
tudo era mais sóbrio e digno e eu nunca havia comido lagosta. O que
me atrapalha a vida é escrever. E não esquecer que a estrutura do
átomo não é vista mas sabe-se dela. Sei de muita coisa que não vi.}
\par\end{flushright}
\end{singlespace}
\vspace{-0.7cm}

\begin{flushright}
\textbf{Clarice Lispector}, ``A hora da estrela''
\par\end{flushright}%
\end{minipage}
\par\end{flushright}

\vspace{2cm}

\cleardoublepage{}

\chapter*{Agradecimentos}

Uma tese nasceu, com dores de parto e em meio a uma pandemia sem precedentes.
Isolada no último mês de escrita, porém jamais sem o apoio distante
de muitos que estiveram presentes durante esses quatro anos de doutorado
e dez anos de física. Tantos são os que merecem dedicatórias e agradecimentos
que incorro no risco de ser injusta e esquecer-me de alguns nomes
ou até desprivilegiar alguém pela ordem em que os cito. Não sei se
isso é de qualquer importância àqueles que me cercam, assim como não
soube se a própria tese era de qualquer importância, por vezes ela
perdeu o sentido: não salvaria vidas e parece ser só o que importa
durante esses dias sombrios. Por outro lado, por que salvar vidas?
Pelo número de vidas salvas ou pela individualidade de cada uma? Sendo
assim, a minha vida também deve importar e as particularidades ligadas
a ela também. Essa tese importa, nem que seja somente para mim. O
conteúdo dela certamente não revoluciona a ciência, mas fornece alguns
tijolinhos para sustentar essa enorme construção científica. E assim
é feita a ciência: somente vez ou outra uma revolução, mas sempre
se sustentando nos tijolinhos.

E eu sempre me sustentando na minha família. Meus pais, dona Ini e
seu Magno, e minha irmã, Tina, sempre acreditaram em mim e nas minhas
escolhas, mesmo quando elas envolveram me mudar para longe deles.
Me mudando para longe, não pude participar ativamente do crescimento
da minha sobrinha, Thayná, e de minha prima, Isa, espero que ao menos
eu sirva de inspiração para essas criaturinhas. Devo me sentir orgulhosa
e, família, vocês também devem se orgulhar do trabalho que fizemos.
Meus pais, que sequer tiveram oportunidade de concluir o ensino fundamental,
conseguiram me impulsionar para que eu agora tenha a oportunidade
de concluir um doutorado! Muito obrigada pela dedicação de vocês,
espero conseguir retribuir de alguma forma. Agradeço também a toda
a família, tios e primos, a união dessa família sempre será uma motivação,
em especial à Tinti e ao Uncle que sempre serviram de inspiração e
também à Grangran, a avó mais teimosa e mais divertida que eu poderia
ter.

Além da minha família original, não posso deixar de agradecer a uma
família que me adotou e incorporou como se eu fosse um deles desde
quando os conheci. Muito obrigada à família Zanco, especialmente ao
Jônatas que por tanto tempo esteve ao meu lado, alguém que me apoiou,
me deu suporte e me ensinou tanta coisa para tornar a vida mais leve.
Não poderia me esquecer também da minha mãe carioca, Druzila, que
me recebeu em sua casa no primeiro ano de doutorado como a uma filha
e me deu não somente abrigo, mas sua amizade e seus cuidados. 

Muitas foram as amizades que os tantos anos de Física me trouxeram,
algumas passageiras, outras que ficam pra vida. Quantas foram as vezes
em que quase fomos roubados por macacos enquanto comíamos biscoito
frito às 16h, Leandro? Companheiro de natação, de escrita de dissertação
e agora de tese, de trabalhos de EaD, de finais de semana na universidade,
de insolação... Obrigada pela amizade incondicional de quase dez anos
e pela disposição em sempre ajudar e ouvir. 

Assim que cheguei à UFRJ fui levada ao que seria meu escritório e
que seria dividido com dois malucos. Como foram divertidos os primeiros
dias, Kainã e Renato fizeram eu me sentir importante e acolhida, ganhava
até paçoca e bolacha piraquê! Através deles e de sua capacidade incrível
de socialização, conheci muitos colegas de instituto. Foi através
deles que conheci uma pessoa incrível, uma mulher sonhadora, destemida
e que domina as palavras como poucas vezes vi. Obrigada, Carol, pela
amizade, pelas festas loucas, pelas trocas de segredinhos e por me
fazer ver como feminilidade, força e conquista de respeito devem andar
juntos. Não acho que já tenha te dito isso: te admiro muito, admiro
sua determinação em não seguir a corrente, mas ir contra, ir bailando
sobre pernas de pau ao encontro dos seus ideais. Falando em mulher
forte e admiração, não posso deixar de mencionar e agradecer à Murielvis,
a amiga mais surpreendente que jamais tive, de infância no canavial
até paraquedismo, ela abarca tudo que a vida tem a oferecer com uma
coragem e uma força que não parecem possíveis ao julgar erroneamente
a aparência.

Obrigada a todos os colegas de laboratório e de grupo, sem vocês esses
anos teriam sido pobres, até mesmo de inspiração para trabalhar. Foram
muitas as festinhas de aniversário, confraternização aleatória ou
celebrações de defesa, muitos imagem e ação, chocolates e outros doces
compartilhados... Sem momentos de procrastinação pós almoço com Rodrigo
ou com Márcio ou com Ranieri, ou com todos juntos, o doutorado não
teria a leveza e a graça que teve. E claro que, sem a ajuda do Rani,
nem mesmo o doutorado talvez fosse possível, após várias horas de
discussões, misturadas com procrastinação baseada nos mais diversos
temas e jogos, vários ``Você pode ler isso aqui que eu escrevi?''
seguidos de um trocadilho e um ``sim'', após alguns aniversários
de Thainery, diversas conversas sobre os questionamentos profundos
da vida, só me resta agradecer e torcer para que nossa amizade e colaboração
continuem por muitos anos mais. Agradeço ao Márcio também por me mostrar
ridiculamente como é possível ter vida social, dominar várias línguas,
traduzir vídeos do YouTube, dar ótimos churrascos, editar vídeos de
memes e ainda ser super produtivo no trabalho. Esse agradecimento
é o mais próximo que eu consigo chegar de um vídeo de aniversário
com bolo de imagem e ação, considere retribuído.

Não posso também não agradecer ao Victor, mas honestamente não sei
o que dizer. Sua importância durante esse período foi inegável, transformadora
e por que não dizer nutritiva, dado que tudo se iniciou como uma troca
de marmitas. E mesmo o seu afastamento me trouxe coisas maravilhosas,
como esses dois amigos que não posso deixar de mencionar. Pedro e
Matheus, vocês foram e são fundamentais na minha vida e mesmo no meu
trabalho. A companhia de vocês, todos os jantares compartilhados e
as besteiras ditas aliviaram minha carga e me ajudaram a me valorizar
e confiar em mim. Obrigada por me integrarem em tão pouco tempo, a
dedicação e o bom humor de vocês me inspira e não aceito que nossa
amizade não seja pra sempre.

E como não citar os elementos fundamentais nessa trajetória: todos
os professores que dela participaram. Muito obrigada a todos, tantos
nomes fundamentais que não há espaço para citar todos. Desde antes
de ingressar na universidade, um professor do instituto de Física
da UFG já me influenciou na minha escolha. Ter invadido a sala do
Caparica em uma visita à universidade foi fundamental, não só para
a escolha, como também para não me frustrar com ela, já que ele me
disse algo como ``A vida, o trabalho e as conquistas de um cientista
são muito diferentes da visão idealizada da ficção''. E meus primeiros
passos como cientista foram guiados por meu primeiro orientador que
me acompanhou na graduação e no mestrado, meu obrigada ao Ardiley.
Obrigada também a todos os professores do grupo de Informação quântica
da UFRJ pelos ensinamentos e pela amizade, em especial ao Fabricio
que, à sua maneira às vezes rude, sempre se preocupou comigo, e ao
Leandro que agora me acolhe para uma nova fase, o pós-doutorado. Finalmente,
declaro minha gratidão ao meu orientador, Steve, e ao meu coorientador,
Gabo, vocês acreditaram em mim, me acompanharam e me incentivaram,
recebi não só ensinamentos e ajuda, mas confiança e amizade, e tive
o melhor ambiente de trabalho possível.

Por último, gostaria de agradecer ao CNPq, sem a bolsa de doutorado
não teria sido possível.\selectlanguage{american}

\cleardoublepage{}

\chapter*{Resumo}

Feixes ópticos oferecem muitos graus de liberdade a serem explorados.
Há graus de liberdade discretos como polarização, momento angular
orbital e caminhos discretos. Existem também graus de liberdade contínuos,
como frequência, momento e posição transversal. Além da possibilidade
de emaranhar photons nesses muitos graus de liberdade, isso faz da
luz uma ferramenta extremamente útil e versátil para investigações
em fundamentos de mecânica quântica e em informação quântica. Desde
o início da informação e da computação quântica, experimentos fotônicos
têm tido um papel crucial que vai desde testes fundamentais da teoria
até a implementação de protocolos de informação quântica. Nesta tese,
essa importância e versatilidade é endossada apresentando novas contribuições
que exploram tanto graus de liberdade discretos como contínuos. A
tese inicia-se com dois experimentos que utilizam luz clássica e exploram
a analogia entre a função de onda de sistemas quânticos e a amplitude
da onda eletromagnética. O primeiro é uma simulação da dinâmica de
uma partícula quântica relativística na qual utiliza-se a analogia
entre campo próximo/distante e a função de onda em posição/momento,
bem como a analogia entre spin e polarização. Esta simulação permite
observar claramente o chamado \textit{zitterbewegung, }movimento trêmulo
de partículas livres, com boa visibilidade para valores ajustáveis
de massa da partícula. O segundo trabalho é relacionado à teoria de
medidas mutuamente imparciais que são efetivamente discretas, porém
construídas a partir de variáveis contínuas, implementadas novamente
no perfil transversal de um feixe luminoso clássico. Demonstra-se
teoricamente que tais medidas não são compatíveis com variáveis discretas
nem contínuas, uma vez que o número máximo de medidas mutuamente imparciais
possível não se comporta como nenhuma das duas possibilidades. Na
segunda parte desta tese, são apresentados três trabalhos utilizando
graus de liberdade discretos de polarização e caminho de fótons. O
primeiro trata da redefinição da correlação quântica não local chamada
\textit{steering} no cenário multipartido, baseada em uma inconsistência
na definição anterior, a saber, a criação desta correlação através
de operações que supostamente não seriam capazes de criá-la. Este
fenômeno é chamado exposição de \textit{steering} quântico. Neste
trabalho, é construído um protocolo para geração de qualquer assemblage
bipartido com \textit{steering} a partir de um assemblage tripartido
sem a correlação. Em geral, tal protocolo não é realizável quanticamente,
no entanto, apresentamos um exemplo obtido a partir de um estado quântico
tripartido em que a exposição de \textit{steering} ocorre e a observamos
experimentalmente para fótons emaranhados. Os demais trabalhos são
relacionados à implementação experimental de canais quântico de um
qbit, um deles é um canal particular para o qual testamos não-Markovianidade
usando uma medida operacional chamada correlação condicional de passado-futuro
(CPF). É mostrado que, mesmo com erros e estatística finita inerente
à implementação experimental, a correlação CPF é capaz de detectar
efeitos de memória que vão além da capacidade de outros quantificadores.
A tese é finalizada com uma proposta para a realização de qualquer
canal quântico de um qbit, em que o qbit é representado pela polarização
de fótons únicos. Diferentemente de outros trabalhos correlatos, nossa
proposta não depende da implementação clássica de combinações convexas
de canais, além de não requerer sistemas auxiliares adicionais, já
que estes são providos por graus de liberdade de caminho do próprio
fóton.

\bigskip{}

\noindent \textbf{Palavras-chave:} Ótica clássica e quântica; ótica
paraxial; fótons emaranhados; simulação quântica e clássica; equação
de Dirac; medidas mutuamente imparciais; \textit{steering} quântico
multipartido; não-markovianidade quântica; canais quânticos. \selectlanguage{american}

\cleardoublepage{}

\chapter*{Abstract}

Light beams offer many degrees of freedom to be explored. There are
discrete ones as polarization, angular orbital momentum and discrete
paths. There are also continuous ones, like frequency, momentum and
transverse position. In addition to the possibility of entangling
photons in these many degrees of freedom, it makes light a very useful
and versatile tool for quantum information and quantum foundation
purposes. Since the very beginning of quantum information and quantum
computation, photonic experiments have played a crucial role that
ranges from testing the foundations of quantum theory to implementing
quantum information protocols. In this thesis, we endorse its importance
and versatility by presenting novel contributions that further explore
both discrete and continuous degrees of freedom. It begins with two
experiments that use classical light and explore its analogous behavior
to quantum systems. The first one is a simulation of the dynamics
of a relativistic quantum particle in which we use the analogy between
near/far transverse fields and position/momentum wavefunctions as
well as the analogy between spin and polarization. Our simulation
enables us to clearly observe the so called zitterbewegung, the trembling
motion of free particles, with good visibility and with a tunable
value of particle mass. The second work is related to the theory of
mutually unbiased measurements that are effectively discrete but constructed
from continuous variables systems, which is again experimentally implemented
on the transverse field profile of a classical light beam. We theoretically
prove that these measurements are actually neither continuous or discrete,
since the maximum number of mutually unbiased measurements possible
does not behave like any of those. In the second part of the thesis,
three works are presented that use the polarization and path discrete
degrees of freedom. The first one is a redefinition of the quantum
nonlocal correlation called steering in the multipartite scenario,
based on an inconsistency in the previous definition, namely the creation
of this correlation from scratch using operations that supposedly
would not be able to do so. We call this exposure of quantum steering.
In this work we build a protocol to generate any steerable bipartite
assemblage from a tripartite unsteerable one, although this protocol
is not realizable with quantum states, we come out with a quantum
example for which this exposure phenomenon is observed with entangled
photons. The other two works are related to the experimental implementation
of quantum channels of qubits, one of them is a particular channel
for which we test for non-Markovianity using a operational measure
called conditional past-future (CPF) correlation. We show that, even
with finite statistics inherent to an experiment and with experimental
errors, this CPF correlation is able to detect memory effects beyond
other non-Markovianity quantifiers. The thesis finishes with a proposal
for an experimental realization of any quantum channel of a single
qubit, where the qubit is realized by the polarization of single photons.
Differently from other works, our proposal does not rely on classical
implementation of convex superposition, also it does not need any
extra ancillary systems, since the ancillas are provided by path degrees
of freedom of the photon itself.

\bigskip{}

\noindent \textbf{Keywords:} Quantum and classical optics; paraxial
optics; entangled photons; quantum and classical simulation; Dirac
equation; mutually unbiased measurements; multipartite quantum steering;
quantum non-Markovianity; quantum channels.

\cleardoublepage{}

\ihead{}

\ohead{}
\chead{}

\cfoot{}

\ofoot{\thepage}

\cleardoublepage{}

\tableofcontents

\cleardoublepage{}

\mainmatter

\ihead{}

\ohead{\leftmark}

\ifoot{}

\cfoot{}

\ofoot[
]{\thepage}

\addchap{Introduction}

Optical experiments were at the heart of the two big revolutions in
Physics that occurred in the beginning of the twentieth century. On
one side, one can cite the interference experiment of Michelson and
Morley which supported the Theory of Relativity \cite{michelson1887}.
On the other side, the discovery \cite{lenard1902} and subsequent
explanation \cite{einstein1905} of photoelectric effect was one of
the motivations for the development of quantum physics. 

Ever since the first formulations of quantum theory, optics has played
a central role in its development. Not only the theory itself has
been initiated by the photoeletric effect, but also it was frequently
an optics experiment that served as the most suitable platform for
testing some non-intuitive features of the theory. For the latter,
we can cite, for example, the incredible amount of experiments to
test quantum nonlocality, with the first unambiguous experiments on
violation of Bell inequalities \cite{Bell1964} by quantum correlations
realized by A. Aspect \textit{et al.} \cite{aspect1981,aspect1982b,Aspect1982}
using photons produced by cascade emission, and subsequently many
experiments using photons produced by parametric down conversion \cite{ou1988,shih1988,kwiat1995},
only to cite a few. Also related to quantum correlations, we can mention
the experiment proving the weirdness of entanglement for multipartite
systems without the necessity of inequalities using Greenberger-Horne-Zeilinger
states \cite{pan2000}. Furthermore, the first quantum teleportation
implementation was performed with photonic states \cite{bouwmeester1997};
wave-particle complementarity has been tested many times with delayed
choice experiments, and quantum erasers, for instance (see Ref. \cite{ma2016}
and references therein); and the two-particles interference exhibiting
a Hong-Ou-Mandel dip because of the symmetry in the quantum bosonic
state of photons was also verified \cite{hong1987}. 

Moreover, with the advent of quantum information theory, photons have
become a natural physical system for \textit{quantum information transmission}
because, first and foremost, they are the fastest carriers of information
available and, due to their lack of charge and mass, they have a reduced
interaction with the environment, making them able to transmit signals
through large distances outside the protected environment of a laboratory
\cite{zeilinger2005}. Thus, many quantum key distribution protocols
have been realized using photons \cite{jennewein2000,bennett1992}.
More recently, these protocols are becoming closer to practical application
with many realizations of quantum information transmission through
long distance fibers \cite{Hiskett_2006,korzh2015,pan2016} and using
Earth satellites \cite{Yin2017,satellite2018} or drones \cite{drone2020}.
Photons have their limitations as a platform for \textit{quantum information
processing} because of the difficulty to store them and also to build
multipartite entangled states, due exactly to the feature that favors
the transmission of information: the difficulty of producing interactions.
However, even in this direction photonic experiments have contributed
with the first attempts of proving the supremacy of quantum computers,
by the use of many-photons interference in boson sampling \cite{pan2019,Brod2019}.

The enormous number of interesting results is ascribable to the versatility
of light. It offers many degrees of freedom which can be explored
independently or jointly, with the possibility of producing states
with entanglement between different degrees of freedom. There are
discrete ones such as polarization, angular orbital momentum, discrete
spatial modes and number of photons \cite{knill2001}. There are also
continuous ones, like frequency, momentum and transverse position
\cite{Braunstein05}.

After all these successful demonstrations of the power of optics as
a tool for the study of quantum theory, there are still many challenges
and room for new and interesting research. This thesis intends to
give some additional contribution to the research field of quantum
theory using optics experiments. It is divided in two parts according
to the experimental platform used. Each part begins with an overview
of the main experimental techniques and devices employed (chapters
1 and 4). Part I contains two experiments using continuous variables
provided by the transverse degrees of freedom of a classical light
beam. Part II contains three experiments that use discrete degrees
of freedom of pairs of photons produced via spontaneous parametric
down conversion. The experiments may be classified in three general
topics: quantum kinematics (Chapter 3), quantum dynamics of a single
system (chapters 2, 6, and 7) and quantum correlations (Chapter 5).
Although sharing the same experimental platforms, the works presented
here are substantially different in their theoretical support, thus
each of them has its own technical introduction to the topic explored
as self-contained as it is possible, such that the chapters can be
read in any order. The five experiments presented approach different
aspects of quantum theory in a fundamental or applied feature: 
\begin{itemize}
\item In a slightly more applied direction, we present the optical simulation
of a relativistic spin $\nicefrac{1}{2}$ particle in Chapter 2. This
experiment, although realized using a classical optical beam, goes
along with the idea of quantum simulation in which a controllable
system is used to mimic the dynamics of another system that one usually
cannot access \cite{blatt2012}. We observe the so-called \textit{zitterbewegung,
}the trembling motion of free relativistic particles predicted by
Dirac equation.
\item In Chapter 3 the description of a single continuous variable quantum
system is explored by means of mutually unbiasedness (MU). Mutually
unbiased bases refer to pairs of Hilbert space bases for which the
projective measurement of any element of one basis on the other basis
gives equiprobable outcomes \cite{Durt10}. This concept can be defined
for discrete or continuous variables systems and is directly related
to complementarity of observables. In practice, MU is not directly
observed for continuous variables systems, but can be recovered by
discretizing the measurements through periodic coarse graining \cite{Tasca18a}.
In this chapter, we investigate the construction of an arbitrary number
of such periodic coarse grained measurements satisfying MU and how
this construction resembles the discrete and continuous cases.
\item If the system is composed of subsystems, then it may present correlations
that are stronger than the ones allowed by classical physics. One
such correlation is quantum steering, which appears in a scenario
were some subsystems are not accessible at the quantum state level
\cite{Cavalcanti2017}. In Chapter 5, we explore multipartite quantum
steering. It is shown that the current definition of this correlation
can lead to inconsistencies, allowing for the creation of the correlation
from scratch by applying some operations that admittedly should not
be able to do so. We experimentally demonstrate this phenomenon, showing
that it can be strong enough to be detected even under unavoidable
experimental imperfections. 
\item The content of Chapter 6 is related to the dynamics of a single quantum
system. Although the textbooks usually present the evolution of a
quantum system as a unitary transformation, the general transformations
are not unitary and given as the solution of master equations. Decoherence
and dissipation are phenomena induced by the unavoidable coupling
of an open quantum system with its environment. When describing this
kind of system dynamics, some important approximation are usually
considered. A paradigmatic example is the Born-Markovian approximation
(BMA), which considers that the reservoir is not altered significantly
due to the presence of the system. Nevertheless, even when a quantum
master equation is obtained beyond the BMA, most of the identifiers
of quantum memory may indicate the absence of any non-Markovian (memory)
effect. For example, dynamics characterized by positive time-dependent
rates are usually classified as Markovian ones~\cite{vega,BreuerReview,plenioReview}.
In this context, conditional past-future correlations (CPF) were shown
to be good memory indicators, predicting non-Markovianity in dynamics
usually considered as Markovian \cite{budini,budiniPRA}. In this
chapter, we provide theoretical and experimental evidence on the feasibility
of measuring and detecting departures from the BMA by using CPF correlations
for the decay of a two level system (polarization of a photon) in
a bosonic bath (spatial modes of the photon). 
\item Closing the thesis, in Chapter 7, a proposal for simulating the aforementioned
general evolution for a qubit is presented, there the qubit is implemented
in the polarization degree of freedom of single photons. 
\end{itemize}

\part{Transverse spatial degree of freedom}

~\ihead{}

\ohead{\textbf{Chapter~\thechapter}~\leftmark}

\ifoot{}

\cfoot{}

\ofoot[
]{\thepage}

\chapter{Experimental techniques \label{chap:Paraxial}}

In this chapter, the experimental techniques and devices common to
the following two chapters are presented. In these two works the experiments
are realized using classical light beams, and we will de interested
in the continuous variables coming from the beam transverse position
and discrete variables coming from the polarization degree of freedom.
The first section presents the functioning of the waveplates used
to manipulate the polarization. We use collimated light beams satisfying
the paraxial approximation, which is presented in Section \ref{sec:Paraxial-approximation}.
It is possible to change to which phase space representation we are
looking to by performing optical Fourier transforms (to go from position
to momentum representation and vice versa) or optical fractional Fourier
transforms (to change between two arbitrary phase space direction
representations), this is the subject of Sec. \ref{sec:Fractional-Fourier-transform}.
Finally, in the last section we present the device used to manipulate
the spatial profile of the light beam, it is known as a spatial light
modulator and enables one to imprint programmable position dependent
phases to the wave fronts.

\section{Polarization transformations: wave plates\label{sec:Placas-de-onda}}

Birefringent crystals (App. \ref{chap:Anexo2}) have different properties
for different electric field directions. Because of this anisotropy,
they are ideal to manipulate the polarization of electromagnetic waves,
by transforming the polarization of a beam or even to separate different
polarization components. Wave plates are slabs of a birefringent crystal
cut to have their optical axis perpendicular to the incidence direction.
This way it is guaranteed that no walk-off between two orthogonal
polarization directions occurs. The effect of the wave plate is to
set a phase difference between ordinary and extraordinary polarization
directions. For fixed refractive indexes $n_{e}$ ($x$ polarization)
and $n_{\vartheta}$ ($y$ polarization), the phase difference and
accordingly the polarization transformation is manipulated by changing
the crystal width $d$ and rotation angle $\theta$ around the $z$
axis.

\begin{figure}[H]
\noindent \centering{}\includegraphics[width=0.4\textwidth]{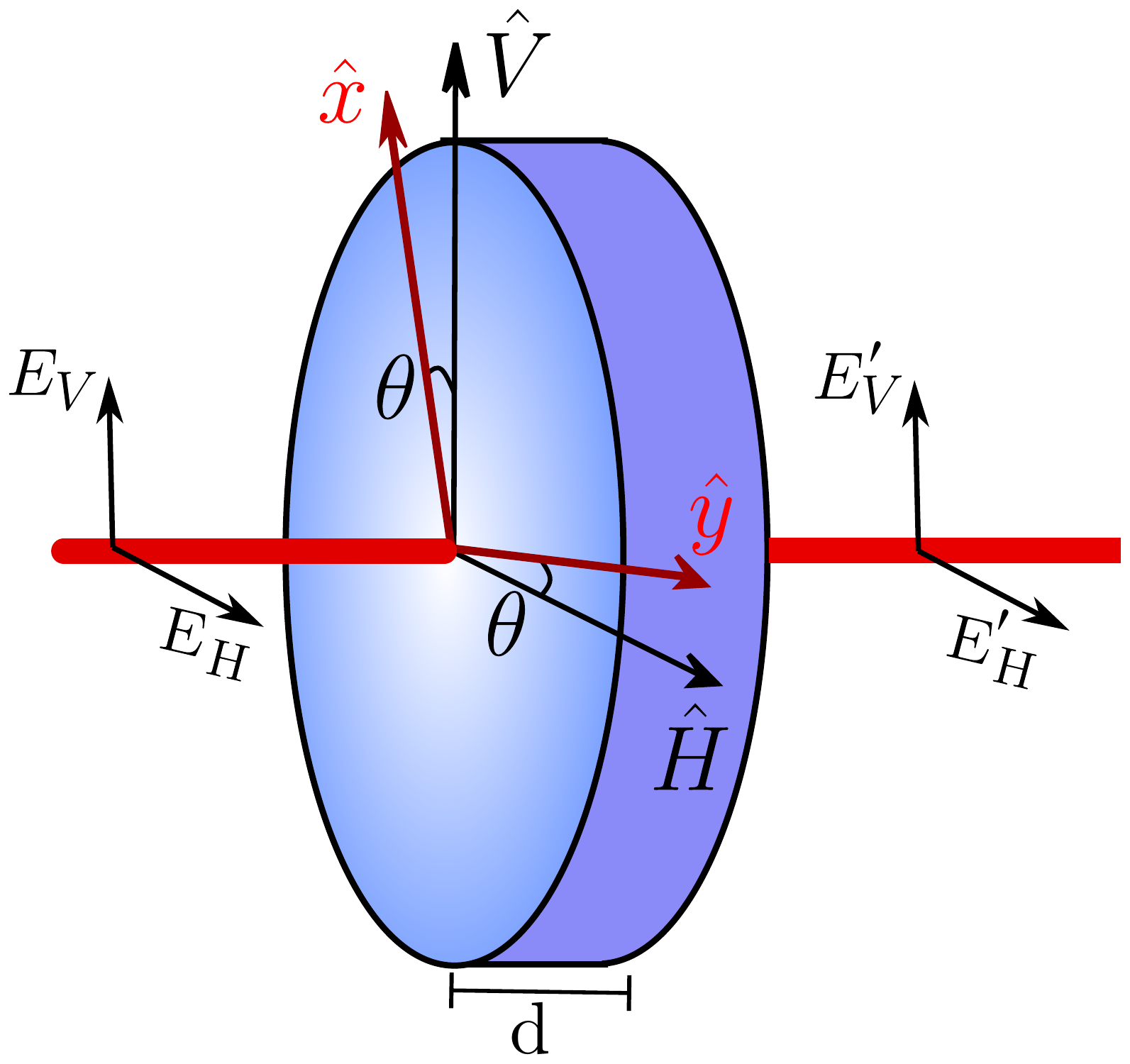}\caption{\foreignlanguage{english}{\textit{A wave plate}: a piece of a uniaxial crystal of width $d$
cut such that its optical axis $\hat{x}$ lie on the incidence plane
face. By changing the rotation angle $\theta$ and the width of the
plate, the transformation relative to the convenient polarization
basis $\{\hat{H},\hat{V}\}$ can be selected. }
}
\end{figure}

\selectlanguage{english}%
Let us consider the plane wave solutions for a given frequency $\omega_{0}$
propagating in the $\hat{z}$ direction. Because the optical axis
is perpendicular to the propagation direction, so is the electric
field inside the material. The most general plane wave with this features
reads
\begin{equation}
\boldsymbol{E}(z,t)=\left(E_{x}e^{i\omega_{0}t-i2\pi\frac{n_{e}}{\lambda_{0}}z}\hat{\boldsymbol{x}}+E_{y}e^{i\omega_{0}t-i2\pi\frac{n_{\vartheta}}{\lambda_{0}}z}\hat{\boldsymbol{y}}\right)=e^{i\omega_{0}t-i2\pi\frac{n_{e}}{\lambda_{0}}z}\left(E_{x}\hat{\boldsymbol{x}}+E_{y}e^{-i2\pi\frac{n_{\vartheta}-n_{e}}{\lambda_{0}}z}\hat{\boldsymbol{y}}\right),\label{eq:camp}
\end{equation}
with $E_{x}$ and $E_{y}$ the field amplitudes in the extraordinary
and ordinary directions, respectively, and $\lambda_{0}=\frac{2\pi c}{\omega_{0}}$
is the wavelength of the beam in vacuum.

It is convenient to define the laboratory coordinate system in the
horizontal and vertical directions. The crystal and laboratory references
are related by 
\begin{equation}
\hat{\boldsymbol{x}}=\cos\theta\,\hat{\boldsymbol{H}}+\sin\theta\,\hat{\boldsymbol{V}}\qquad\hat{\boldsymbol{y}}=-\sin\theta\,\hat{\boldsymbol{H}}+\cos\theta\,\hat{\boldsymbol{V}}.
\end{equation}
A plane wave in the imminence of entering the wave plate $(z=0)$
has field amplide $\boldsymbol{E}_{0}=E_{H}\hat{\boldsymbol{H}}+E_{V}\hat{\boldsymbol{V}}.$
After crossing the wave plate $z=d$, according to (\ref{eq:camp}),
it is transformed to 

\begin{multline}
\boldsymbol{E}(d)=e^{-i2\pi\frac{n_{e}}{\lambda_{0}}d}\Big\{\left[\left(E_{H}\cos\theta+E_{V}\sin\theta\right)\cos\theta-\left(-E_{H}\sin\theta+E_{V}\cos\theta\right)\sin\theta\:e^{-i2\pi\frac{n_{\vartheta}-n_{e}}{\lambda_{0}}d}\right]\hat{\boldsymbol{H}}\\
\left[\left(E_{H}\cos\theta+E_{V}\sin\theta\right)\sin\theta+\left(-E_{H}\sin\theta+E_{V}\cos\theta\right)\cos\theta\:e^{-i2\pi\frac{n_{\vartheta}-n_{e}}{\lambda_{0}}d}\right]\hat{\boldsymbol{V}}\Big\}.\label{eq:fase}
\end{multline}
If the wave plate width is such that $d=(2m+1)\lambda_{0}/2(n_{1}-n_{2})$
,with $m$ an integer number, then it is called a \textit{half wave
plate} (HWP). The phase difference caused by a HWP is $\pi$ and $e^{-i2\pi\frac{n_{\vartheta}-n_{e}}{\lambda_{0}}d}=-1$.
It can be seen from (\ref{eq:fase}) that the resultant transformation
$(E_{H},E_{V})\rightarrow(E_{H}^{\prime},E_{V}^{\prime})$ resulting
from a HWP set at angle $\theta$ in the $\{\hat{\boldsymbol{H}},\hat{\boldsymbol{V}}\}$
basis is given by the transformation matrix
\begin{equation}
\hat{HWP}(\theta)=\left[\begin{array}{cc}
\cos(2\theta) & \sin(2\theta)\\
\sin(2\theta) & -\cos(2\theta)
\end{array}\right].\label{eq:HWP}
\end{equation}
Thus, a half wave plate transforms linear polarization into linear
polarization because it does not introduce any complex phase between
the field amplitude components. Particularly, if $\theta=45^{\circ}$,
the HWP transforms horizontal into vertical polarization and vice
versa. Another interesting particular configuration is $\theta=22.5^{\circ}$.
In this case, horizontally (vertically) polarized light is transformed
into a beam with diagonal (antidiagonal) polarization $\hat{\boldsymbol{D}}=\frac{\hat{\boldsymbol{H}}+\hat{\boldsymbol{V}}}{\sqrt{2}}$
$\left(\hat{\boldsymbol{A}}=\frac{\hat{\boldsymbol{H}}-\hat{\boldsymbol{V}}}{\sqrt{2}}\right)$.

If the plate width is related to the light wavelength by $d=(4m+1)\lambda_{0}/4(n_{1}-n_{2})$
, with $m$ an integer number, it is called \textit{quarter wave plate}
(QWP). The phase difference between ordinary and extraordinary waves
after crossing the plate is then $e^{-i2\pi\frac{n_{\vartheta}-n_{e}}{\lambda_{0}}d}=i$,
considering that $n_{e}>n_{\vartheta}$ as is the case of quartz,
the material composing the wave plates we use. If the angle between
the crystal axis and the vertical direction is $\theta$ the resulting
transformation in the $\{\hat{\boldsymbol{H}},\hat{\boldsymbol{V}}\}$
basis is given by the operator
\begin{equation}
\hat{QWP}(\theta)=\frac{1}{\sqrt{2}}e^{i\frac{\pi}{4}}\left[\begin{array}{cc}
1-i\cos(2\theta) & i\sin(2\theta)\\
i\sin(2\theta) & 1+i\cos(2\theta)
\end{array}\right],\label{eq:QWP}
\end{equation}
as can be obtained from Eq. (\ref{eq:fase}). In particular, when
$\theta=45^{\circ}$, the QWP transforms linear to circular polarization
and vice versa
\begin{equation}
\hat{\boldsymbol{H}}\stackrel{{\scriptscriptstyle QWP}}{\longrightarrow}\hat{\boldsymbol{R}}=e^{i\frac{\pi}{4}}\frac{\left(\hat{\boldsymbol{H}}-i\hat{\boldsymbol{V}}\right)}{\sqrt{2}}\qquad\hat{\boldsymbol{V}}\stackrel{{\scriptscriptstyle QWP}}{\longrightarrow}\hat{\boldsymbol{L}}=e^{-i\frac{\pi}{4}}\frac{\left(\hat{\boldsymbol{H}}+i\hat{\boldsymbol{V}}\right)}{\sqrt{2}},
\end{equation}
$\hat{\boldsymbol{R}}$ and $\hat{\boldsymbol{L}}$ are complex unity
vectors for right and left polarization, respectively.

For both types, the integer $m$ defines the order of the wave plate.
For applications with classical laser light with a big coherence length
such as the one discussed in the next chapter , the order of the wave
plates used does not degrade the interference between ordinary and
extraordinary light. In the second part of this thesis, however, the
experiments are made with down converted photons which have a short
coherence length comparable with the possible optical path difference
inside the wave plate. In that case, it is desirable to work with
zero order plates $(m=0)$.

For this part of the thesis, we only consider the special angles mentioned
above. In Sec. \ref{sec:Projective-measurements}, we show that any
polarization projection can be performed by a pair QWP-HWP. In Sec.
\ref{sec:Unitary-transformations} we also show that any unitary transformation
in polarization can be realized by a set HWP-QWP-HWP.

\section{Paraxial approximation\label{sec:Paraxial-approximation}}

The electric field of an electromagnetic wave propagating in vacuum
satisfies the wave equation
\begin{equation}
\nabla^{2}\boldsymbol{E}(t,\boldsymbol{x})-\frac{1}{c^{2}}\frac{\partial^{2}}{\partial t^{2}}\boldsymbol{E}(t,\boldsymbol{x})=0,
\end{equation}
where $c$ is the light velocity. The simplest solution to this equation
is a plane wave $E_{0}\exp[i\omega_{\mathbf{k}}t-i\mathbf{k}\cdot\boldsymbol{x}]$
with frequency $\omega$ and wave vector $\mathbf{k}$ satisfying
$\omega_{\mathbf{k}}=c|\mathbf{k}|=c\,k$. A plane wave itself does
not represent a physical field since it is spread all over space and
time, but the set of plane wave functions is a complete set of solutions
of the wave equation such that
\[
\boldsymbol{E}(t,\boldsymbol{x})=\int d^{3}k\:\tilde{\boldsymbol{A}}(\boldsymbol{k})e^{i\omega_{\boldsymbol{k}}t-i\boldsymbol{k\cdot x}}
\]
is the most general solution possible, $\tilde{\boldsymbol{A}}(\boldsymbol{k})$
is a complex vector perpendicular to the wave vector $\mathbf{k}$.
We are interested in describing the monocromatic collimated light
beam emitted by a laser. This solution has the property of being well
localized in space, it has a well defined propagation direction and
it does not diverges much during the propagation. This features allows
for the so called paraxial approximation which leads to the Helmholtz
equation, as we describe next, whose solutions are well known. 

A monocromatic paraxial wave propagating in the $z$ direction is
a solution composed by plane waves with frequency $\omega$ whose
transverse wave vector components are much smaller than the component
in the direction of propagation, i.e. $k_{x},k_{y}<<k$ and $k_{z}\sim k$
or writing $k_{z}=k-\Delta k$
\begin{equation}
\Delta k=k-\left(k^{2}-k_{y}^{2}-k_{z}^{2}\right)^{\frac{1}{2}}<<k=\frac{2\pi}{\lambda}.\label{eq:deltaz-1}
\end{equation}
The electric field can be rewritten as 
\begin{equation}
\boldsymbol{E}(t,\boldsymbol{x})=\left[\int d^{2}k\:\tilde{\boldsymbol{A}}(\boldsymbol{k})e^{-i\left(k_{x}x+k_{y}y+\Delta kz\right)}\right]e^{i\omega t-ikz}=\boldsymbol{A}(\boldsymbol{x})e^{i\omega t-ikz},\label{eq:solpar}
\end{equation}
that is, a plane wave propagating in the $z$ direction with envelope
$\boldsymbol{A}(\boldsymbol{x})$. Because of the paraxial condition
(\ref{eq:deltaz-1}), the envelope varies slowly with the propagation
distance. The substitution of solution (\ref{eq:solpar}) into the
wave equation shows that the envelope satisfies
\begin{equation}
\nabla_{t}^{2}\boldsymbol{A}(\boldsymbol{x})-2ik\frac{\partial\boldsymbol{A}(\boldsymbol{x})}{\partial z}+\frac{\partial^{2}\boldsymbol{A}(\boldsymbol{x})}{\partial z^{2}}=0,\label{eq:kkk}
\end{equation}
$\nabla_{t}^{2}$ stands for the Laplacian operator in the transverse
coordinates $x$ and $y$. Considering the slow variation with the
propagation distance, the Taylor series of the envelope
\[
\boldsymbol{A}(x,y,z+\Delta z)=\boldsymbol{A}(\boldsymbol{x})+\frac{\partial\boldsymbol{A}(\boldsymbol{x})}{\partial z}\Delta z+\frac{1}{2}\frac{\partial^{2}\boldsymbol{A}(\boldsymbol{x})}{\partial z^{2}}\Delta z^{2}+...
\]
can be approximated by the first order expansion for $\Delta z\sim\lambda$,
regarding that 
\[
\frac{1}{2}\left|\frac{\partial^{2}\boldsymbol{A}(\boldsymbol{x})}{\partial z^{2}}\right|\lambda^{2}<<\left|\frac{\partial\boldsymbol{A}(\boldsymbol{x})}{\partial z}\right|\lambda,
\]
 or equivalently
\[
\left|\frac{\partial^{2}\boldsymbol{A}(\boldsymbol{x})}{\partial z^{2}}\right|<<2k\left|\frac{\partial\boldsymbol{A}(\boldsymbol{x})}{\partial z}\right|.
\]
 Applying this approximation to Eq. (\ref{eq:kkk}) give us the paraxial
Helmholtz equation 
\begin{equation}
\nabla_{t}^{2}\boldsymbol{A}(\boldsymbol{x})+2ik\frac{\partial\boldsymbol{A}(\boldsymbol{x})}{\partial z}=0
\end{equation}
satisfied by the envelope field in the paraxial condition. This equation
is completely analogous to the free Schrödinger equation for a unit-mass
particle if we make the associations $z\rightarrow t$ and $k\rightarrow\nicefrac{1}{\hbar}$.
In this analogy, the complex amplitude of the field is associated
to the particle wavefunction, and the beam intensity distribution
to the probability density of detecting the particle. This analogy
has been widely used in many experiments emulating quantum systems
using light \cite{Tasca11,lemos2012}. 

The Hermite-Gauss functions form a particular complete set of solutions
to the Helmholtz equation for each component of the vector $\boldsymbol{A}(\boldsymbol{x})$
{\cite{saleh1991}}. The output of the monomode
laser we use in our experiments is a Hermite-Gauss function of zero
order expressed as 
\[
A_{0,0}(x,y,z)=A_{0}\frac{w_{0}}{w(z)}\exp\left[\frac{-\rho^{2}}{w^{2}(z)}\right]\exp\left[-ik\frac{\rho^{2}}{2R(z)}+i\zeta(z)\right],
\]
where:
\begin{itemize}
\item $A_{0}$ is a complex constant , $\rho^{2}=x^{2}+y^{2}$ is the distance
to the center of the beam;
\item $w(z)=w_{0}\left[1+\left(\frac{z}{z_{0}}\right)^{2}\right]^{1/2}$
is the beam width in position $z,$ were $z_{0}=\frac{w_{0}^{2}\pi}{\lambda}$
is known as the Rayleigh range;
\item The origin of the $z$ axis is defined such that the beam waist, that
is the position is which the beam width has the smallest value possible
$w_{0}$, is located at $z=0$;
\item $R(z)=z\left[1+\left(\frac{z_{0}}{z}\right)^{2}\right]$ is the radius
of curvature of the wave fronts in position $z$;
\item $\zeta(z)=\mathrm{tg}^{-1}\left(\frac{z}{z_{0}}\right)$ is called
Gouy phase.
\end{itemize}
\textbf{Free space propagation:} Let us consider only one component
of the electric field of a paraxial wave propagating in the $z$ direction
and that this component distribution on the plane $z=z_{1}$ is equal
to a function $f_{1}(x,y)=E(x,y,z_{1})$. We want to relate this initial
transverse profile to the field distribution $f_{2}(x,y)=E(x,y,z_{2})$
in a posterior position $z=z_{2}$ if the electromagnetic wave is
propagating in free space. The paraxial approximation applied to Eq.
(\ref{eq:deltaz-1}) gives
\begin{equation}
\Delta k\approx\frac{k_{x}^{2}+k_{y}^{2}}{2k}.
\end{equation}
Thus, defining the Fourier transform of the initial distribution 
\[
f_{1}(x,y)=\frac{1}{(2\pi)^{2}}\int dk_{x}\int dk_{y}\:F_{1}(k_{x},k_{y})e^{-ik_{x}x-ik_{y}y}
\]
 and defining the propagation distance $d=z_{2}-z_{1}$, the field
distribution in position $z=z_{2}$ is obtained with Eq. (\ref{eq:solpar})
as 
\begin{equation}
f_{2}(x,y)=\frac{e^{-ikd}}{(2\pi)^{2}}\int dk_{x}\int dk_{y}\:F_{1}(k_{x},k_{y})e^{-ik_{x}x-ik_{y}y}e^{-i\frac{k_{x}^{2}+k_{y}^{2}}{2k}d},\label{eq:freeprop}
\end{equation}
but this is the inverse Fourier transform of the product $F_{1}(k_{x},k_{y})e^{-i\frac{k_{x}^{2}+k_{y}^{2}}{2k}d},$
and by the convolution theorem \cite{goodman1996} it can be written
as 
\begin{equation}
f_{2}(x,y)=e^{-ikz}\int dx^{\prime}\int dy^{\prime}\:f_{1}(x^{\prime},y^{\prime})\,h(x-x^{\prime},y-y^{\prime}),\label{eq:freeprop2-1}
\end{equation}
where we defined $h(x,y)=\frac{i}{\lambda d}\exp\left[-ik\frac{x^{2}+y^{2}}{2d}\right]$
as the Fourier transform of the free space transfer function $e^{-i\frac{k_{x}^{2}+k_{y}^{2}}{2k}d}$.

\textbf{Propagation through a thin lens: }Consider a monocromatic
paraxial wave crossing a lens with central width $\Delta z$ made
of a homogeneous isotropic material whose refractive index is $n$.
While inside the lens, the wave vector changes its modulus due to
the change in the refractive index. The lens width and thus the optical
path differ\label{eq:freeprop2}ence acquired by the beam depend on
the transverse position relative to the center of the lens. If $f_{1}(x,y)=E(x,y,z_{1})$
is the electric field amplitude in the plane immediately before the
lens, the net effect in the electric field amplitude in the plane
immediately after the lens $f_{2}(x,y)=E(x,y,z_{1}+\Delta z)$ is
a quadratic phase 
\begin{equation}
f_{2}(x,y)=e^{-ink\Delta z}\exp\left[ik\frac{x^{2}+y^{2}}{2f}\right],\label{eq:lens}
\end{equation}
where $f$ is the focal distance of the lens which relates to the
refractive index of the lens material and to the curvature of the
lens depending on its exact shape. This expression is obtained under
some assumptions: that the incident beam is narrow compared to the
lens curvature, and that the lens is thin and the incidence is almost
normal, such that direction changes in the wave vector are neglected.
The last condition is valid for all wave vectors composing the beam
and is equivalent to requiring the beam to be paraxial. 

\section{Fractional Fourier transform and phase space variables\label{sec:Fractional-Fourier-transform}}

The fractional Fourier transform is a generalization of the ordinary
Fourier transform and can be defined in the context of rotations in
the phase space of a continuous variables system. As a mathematical
tool it has applications in classical signal processing, quantum physics
and in the solution of differential equations \cite{Ozaktas01}.

Let us consider the adimensional position and momentum operators of
a continuous variables system satisfying the canonical commutation
relation $[\hat{x},\hat{p}]=i$. Let $\ket{x}$ be the eigenstate
of the position operator $\hat{x}$ with eigenvalue $x$. The eigenstates
of the position operator satisfy the completeness relation $\int dx\,\ket{x}\bra{x}=\mathbb{1}$,
such that any pure state $\ket{\psi}$ of the system can be uniquely
represented by its position representation or wavefunction $\psi(x)=\langle x|\psi\rangle$
as 
\begin{equation}
\ket{\psi}=\int dx\:\psi(x)\ket{x}.
\end{equation}
The same statements are valid for the momentum eigenstates $\{\ket{p}\}$,
being $\tilde{\psi}(p)=\braket{p}{\psi}$. Because of the commutation
relation between position and momentum, the scalar product of eigenstates
of the two operators is given by $\braket{p}{x}=\frac{e^{ipx}}{\sqrt{2\pi}}$
and thus 
\begin{equation}
\tilde{\psi}(p)=\int dx\,\braket{p}{x}\braket{x}{\psi}=\frac{1}{\sqrt{2\pi}}\int dx\,e^{-ipx}\psi(x),
\end{equation}
that is, momentum and position representations are connected via a
Fourier transform (FT). Accordingly, the position representation is
obtained as the inverse Fourier transform of the momentum wavefunction.

Analogously, we can consider the eigenstates $\{\ket{q_{\theta}}\}$
of any operator defined as 
\begin{equation}
\hat{q}_{\theta}=\cos\theta\,\hat{x}+\sin\theta\,\hat{p}.\label{eq:qtheta}
\end{equation}
The operator $\hat{q}_{\theta}$ can be seen as a rotated position
in phase space, a two dimensional space which orthogonal axis are
defined by the position and momentum values. The eigenstates of $\hat{q}_{\theta}$
satisfy \cite{Ozaktas01}
\begin{equation}
\braket{q_{\theta}}{x}=\mathcal{F}^{\theta}=\sqrt{\frac{1-i\cot\theta}{2\pi}}e^{i\frac{\cot\theta}{2}\left(x^{2}+q_{\theta}^{2}\right)-ixq_{\theta}\csc\theta},\label{eq:kernel}
\end{equation}
they form a basis for the space of states and satisfy $\int dq_{\theta}\ket{q_{\theta}}\bra{q_{\theta}}=\mathbb{1}$
such that any state can be written as 
\begin{equation}
\ket{\psi}=\int dq_{\theta}\:\bar{\psi}(q_{\theta})\ket{q_{\theta}}.
\end{equation}
The $q_{\theta}$ representation and the position representation are
related through the integral transformation
\begin{equation}
\bar{\psi}(q_{\theta})=\braket{q_{\theta}}{\psi}=\int dx\,\psi(x)\bra{q_{\theta}}\ket{x}\label{eq:FrFTrep}
\end{equation}
with kernel given by (\ref{eq:kernel}). This transformation is called
fractional Fourier transform (FrFT). The FrFT is linear and additive,
meaning that realizing two subsequent transformations with angles
$\theta_{1}$ and $\theta_{2}$ is equivalent to realizing only one
transformation with angle $\theta_{1}+\theta_{2}$. In the particular
case of $\theta=\pi/2$, it is possible to see that the FrFT consistently
reduces to the ordinary FT. Analogously, any rotation in phase space
from $q_{\theta}$ to $q_{\theta^{\prime}}$ representations is obtained
through a FrFT with kernel $\mathcal{F}^{\Delta\theta}$, with $\Delta\theta=\theta^{\prime}-\theta$.

The FrFT has many optical realizations and even one of its first treatments
came in the context of classical optics \cite{Ozaktas93,Lohmann93}.
We show next the methods we use in our experiments to realize FTs
and FrFT. Although the FT is contained in the FrFT, we present the
two transformations separately as we use a simpler method to realize
FTs when position and momentum are the only phase space variables
needed.

\subsection{Optical Fourier transform\label{subsec:Optical-Fourier-transform}}

As was mentioned before, the transverse electric field distribution
of a paraxial wave is analogous to the wavefunction of a quantum particle
in position representation. In the same way, the Fourier transform
of the field distribution can be regarded as the analogous to the
momentum representation of the state. The Fourier transform of the
transverse field distribution can be obtained experimentally by the
apparatus shown in Fig. \ref{fig:Configura=0000E7=0000E3o-utilizando-uma}
as is described in what follows.

Let us consider again a monocromatic paraxial beam with wavelength
$\lambda$, propagating in the $z$ direction, which is perpendicular
to a lens of focal distance $f$ and width $\Delta$, located one
focal distance apart from the plane $z=0$. In a plane $A$ located
in the origin of the $z$ axis, a component of the electric field
is described by a function $f(x,y)=E(x,y,0)$ with Fourier transform
$F(k_{x},k_{y})$. We are interested in relating $f(x,y)$ and the
field distribution $g(x,y)=E(x,y,2f+\Delta)$ after the beam has propagated
in free space by a distance $f$, traversed the lens and then again
freely propagated by another focal distance.

\begin{figure}[h]
\noindent \begin{centering}
\includegraphics[width=0.6\textwidth]{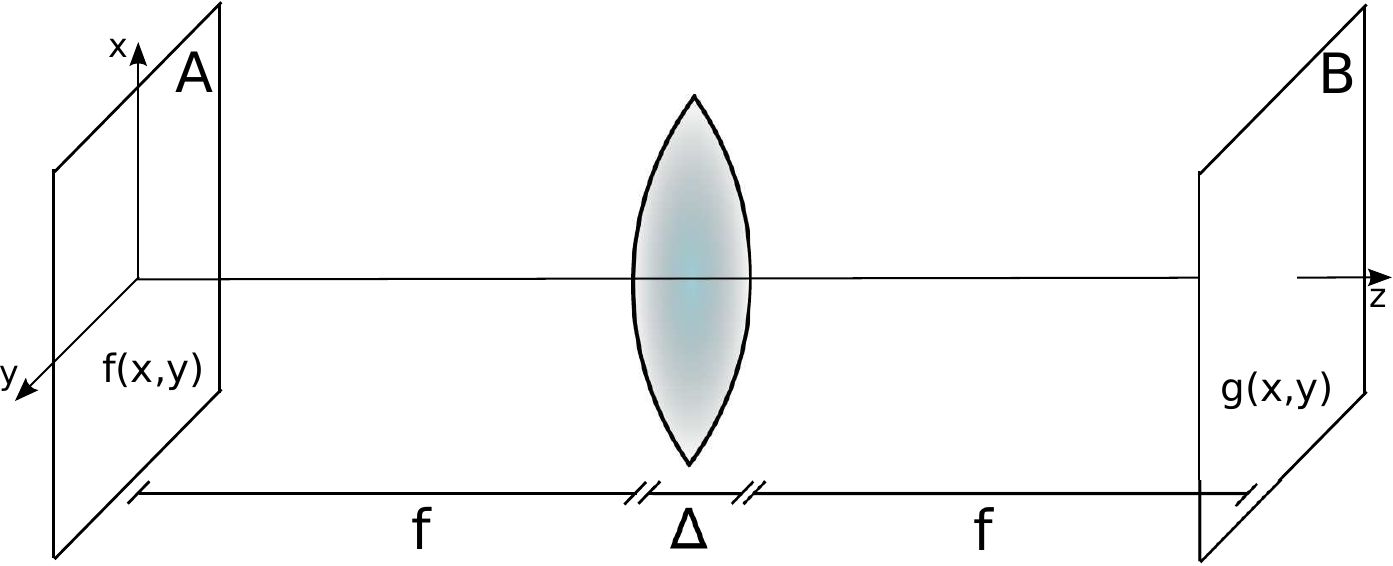}
\par\end{centering}
\caption{Configuration using a lens to implement the optical Fourier transform
between the two focal planes.\label{fig:Configura=0000E7=0000E3o-utilizando-uma}}
\end{figure}

After the first free space propagation, according to Eq. (\ref{eq:freeprop}),
the field distribution becomes
\begin{equation}
E(x,y,\mathsf{f})=\frac{e^{-ik_{0}\mathsf{f}}}{(2\pi)^{2}}\int dk_{x}\int dk_{y}\:F(k_{x},k_{y})e^{i\mathsf{f}\frac{k_{x}^{2}+k_{y}^{2}}{2k_{0}}}e^{-ik_{x}x-ik_{y}y}.
\end{equation}
The wave then crosses the lens acquiring a quadratic phase (Eq. (\ref{eq:lens}))
and becoming

\begin{eqnarray}
E(x,y,\mathsf{f}+\Delta) & = & e^{-ik_{0}\Delta}e^{ik_{0}\frac{x^{2}+y^{2}}{2\mathsf{f}}}E(x,y,\mathsf{f})\label{eq:U_poslente}\\
 & = & \frac{e^{-ik_{0}(\mathsf{f}+\Delta)}}{(2\pi)^{2}}e^{ik_{0}\frac{x^{2}+y^{2}}{2\mathsf{f}}}\int dk_{x}\int dk_{y}\:F(k_{x},k_{y})e^{i\mathsf{f}\frac{k_{x}^{2}+k_{y}^{2}}{2k_{0}}}e^{-ik_{x}x-ik_{y}y}.\nonumber 
\end{eqnarray}
It is now easier to directly use Eq. (\ref{eq:freeprop2-1}) rather
than the Fourier propagation to obtain the field distribution in the
focal plane $B$ after the last free space propagation as\footnote{\noindent Plugging Eq. (\ref{eq:U_poslente}) into Eq. (\ref{eq:freeprop2-1})
gives\vspace{-0.5cm}
 
\begin{multline*}
g(x,y)=\frac{i}{\lambda_{0}\mathsf{f}}\frac{e^{-ik_{0}(2\mathsf{f}+\Delta)}}{(2\pi)^{2}}\int dx'\int dy'\int dk_{x}\int dk_{y}\:e^{\frac{-ik_{0}}{2\mathsf{f}}\left((x-x')^{2}+(y-y')^{2}\right)}e^{\frac{ik_{0}}{2\mathsf{f}}\left(x'^{2}+y'{}^{2}\right)}\\
e^{\frac{i\mathsf{f}}{2k_{0}}\left(k_{x}^{2}+k_{y}^{2}\right)}e^{-i(k_{x}x'+k_{y}y')}F(k_{x},k_{y})
\end{multline*}
\vspace{-1cm}
\begin{multline*}
g(x,y)=\frac{i}{\lambda_{0}\mathsf{f}}\frac{e^{-ik(2\mathsf{f}+\Delta)}}{(2\pi)^{2}}e^{-\frac{ik_{0}}{2\mathsf{f}}\left(x^{2}+y{}^{2}\right)}\int dk_{x}\int dk_{y}\:\left[\int dx'e^{i\left(\frac{k_{0}x}{\mathsf{f}}-k_{x}\right)x'}\right]\left[\int dy'\,e^{i\left(\frac{k_{0}y}{\mathsf{f}}-k_{y}\right)y'}\right]\\
e^{\frac{i\mathsf{f}}{2k_{0}}\left(k_{x}^{2}+k_{y}^{2}\right)}F(k_{x},k_{y}).
\end{multline*}
The lens width can be neglected in comparison with the focal distance
and the position integrals are identified as Dirac delta functions
since $\delta(k)=\frac{1}{2\pi}\int dx\,e^{ikx}$, giving\vspace{-0.5cm}
 
\begin{multline*}
g(x,y)=\frac{i}{\lambda_{0}\mathsf{f}}e^{-ik_{0}2\mathsf{f}}e^{-\frac{ik_{0}}{2\mathsf{f}}\left(x^{2}+y{}^{2}\right)}\int dk_{x}\int dk_{y}\:\delta\left(\frac{k_{0}x}{\mathsf{f}}-k_{x}\right)\,\delta\left(\frac{k_{0}y}{\mathsf{f}}-k_{y}\right)e^{\frac{i\mathsf{f}}{2k_{0}}\left(k_{x}^{2}+k_{y}^{2}\right)}F(k_{x},k_{y}),
\end{multline*}
which leads to Eq. (\ref{eq:resultTF}) after integration.}
\begin{equation}
g(x,y)=\frac{ie^{-2ik_{0}\mathsf{f}}}{\lambda_{0}\mathsf{f}}F\left(\frac{k_{0}x}{\mathsf{f}},\frac{k_{0}y}{\mathsf{f}}\right),\label{eq:resultTF}
\end{equation}
that is , the wave amplitude in the position $(x,y)$ on plane $B$
is proportional to the Fourier transform of the initial amplitude,
evaluated for the transverse wave vector $(k_{x}=\frac{k_{0}x}{\mathsf{f}},k_{y}=\frac{k_{0}y}{\mathsf{f}})$.
Thus, we can regard the first and second focal planes as the position
and momentum spaces, respectively.

\begin{figure}[h]
\noindent \begin{centering}
\includegraphics[width=0.9\textwidth]{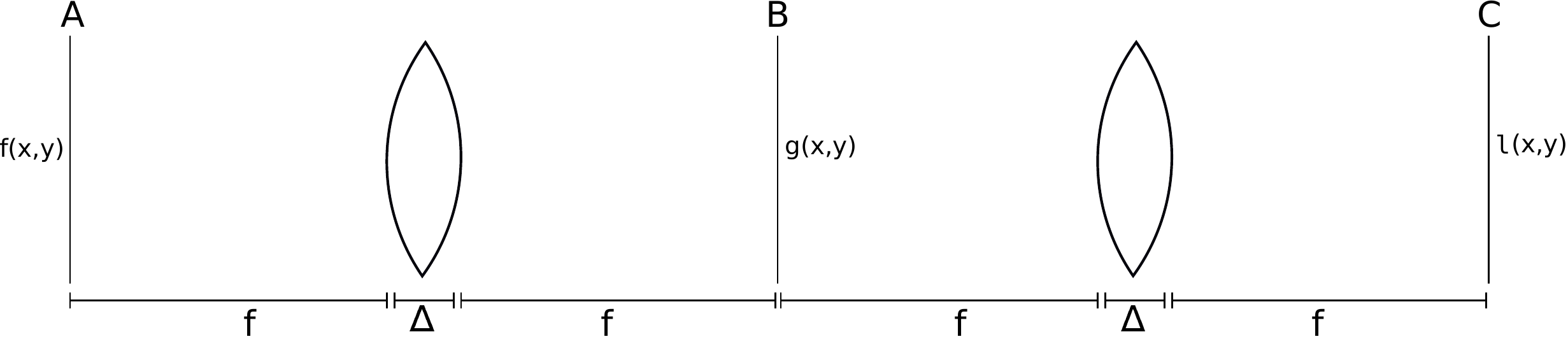}
\par\end{centering}
\caption{The optical Fourier transform scheme duplicated.\label{fig:TFinv}}
\end{figure}

The ``inverse'' Fourier transform is obtained replicating the lens
system as shown in Fig. \ref{fig:TFinv}. As the Fourier transform
of $g(x,y)$ is given by 
\begin{eqnarray}
G(k_{x},k_{y}) & = & \int dx\int dy\:g(x,y)e^{ik_{x}x+ik_{y}y}\nonumber \\
 & = & ie^{-2ik_{0}\mathsf{f}}\mathsf{f}\lambda_{0}f\left(-\frac{\mathsf{f}}{k_{0}}k_{x},-\frac{\mathsf{f}}{k_{0}}k_{y}\right),
\end{eqnarray}
then the field amplitude in plane $C$ is the inverted image of plane
$A$ 
\begin{equation}
l(x,y)=e^{-2ik_{0}\mathsf{f}}f(-x,-y),
\end{equation}
what allow us to say that in plane $C$ we have again the position
space although it is inverted.

\subsection{Optical fractional Fourier transform}

\begin{figure}[h]
\noindent \begin{centering}
\includegraphics[width=0.6\columnwidth]{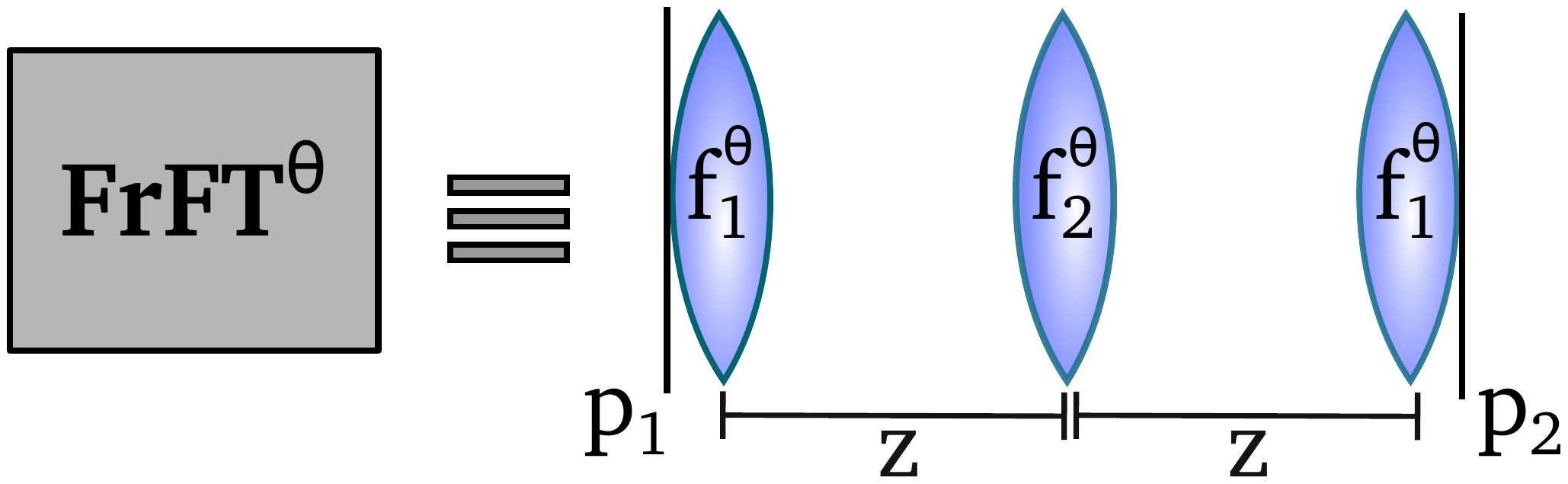}
\par\end{centering}
\caption{Three lenses scheme for optical FrFT: the field distribution in the
transverse plane $p_{2}$ is equals to the FrFT of the field distribution
in the input plane $p_{1}$.\label{fig:Three-lenses-scheme}}

\selectlanguage{american}%
\end{figure}

In what follows it is shown that the three lenses configuration illustrated
in Fig. \ref{fig:Three-lenses-scheme} realizes a FrFT of the transverse
field distribution between the input ($p_{1}$) and output ($p_{2}$)
planes, with any phase-space rotation-angle $\theta$ if the lenses
focal lengths are suitably chosen as 
\begin{eqnarray}
\frac{1}{f_{1}^{\theta}} & = & \frac{1}{z}\left(1-\frac{\cot(\theta/2)}{2}\right)\label{eq:f1}\\
1-\frac{z}{2f_{2}^{\theta}} & = & \sin\theta.\label{eq:f2}
\end{eqnarray}

Let us consider a monochromatic light beam with wavelength $\lambda$
propagating in the $z$ direction perpendicular to the lenses coming
from the left to the right. The spatial electric field distribution
is given in any point of space by $E(x,y,z)$. In the plane $p_{1}$
immediately before the first lens the field is a function $E(x,y,0)=f(x,y)$
of the coordinates on the plane. For simplicity, the lenses are considered
to be cylindrical, such that the curvature in $y$ direction is infinity
and the lens acts only in the $x$ direction. We can forget about
the dependence in coordinate $y$ since the beams we use are separable
in the transverse variables and the $y$-dependent part is only affected
by the free propagation divergence. Also for simplicity, the propagation
phases $e^{ikz}$ are ignored throughout the calculation. So, after
the first lens the field amplitude becomes (Eq. (\ref{eq:lens}))
\begin{equation}
E(x,\Delta_{1})=f(x)\exp\left(\frac{i\pi x^{2}}{\lambda f_{1}^{\theta}}\right),
\end{equation}
where $\Delta_{1}$ and $f_{1}$ are the width and the focal length
of this lens. Then the wave freely propagates by a distance $z$ after
which the field distribution is given by the convolution (Eq. (\ref{eq:freeprop2-1}))
\begin{equation}
E(x,\Delta_{1}+z)=\frac{1}{\lambda z}\int dx'\:f(x')\exp\left(\frac{i\pi x'^{2}}{\lambda f_{1}^{\theta}}\right)\exp\left[-i\pi\frac{(x-x')^{2}}{\lambda z}\right].
\end{equation}

Another quadratic phase is imprinted by the second lens whose width
is $\Delta_{2}$and whose focal length is $f_{2}$. The field becomes
\begin{align*}
E(x,\Delta_{1}+z+\Delta_{2})= & \frac{1}{\lambda z}\int dx'\:f(x')\exp\left(\frac{i\pi x'^{2}}{\lambda f_{1}^{\theta}}\right)\\
 & \exp\left(\frac{i\pi x{}^{2}}{\lambda f_{2}^{\theta}}\right)\exp\left[-i\pi\frac{(x-x')^{2}}{\lambda z}\right]
\end{align*}

After the second free space propagation, using again the convolution
(\ref{eq:freeprop2-1}), the field amplitude transforms to
\begin{align*}
E(x,\Delta_{1}+2z+\Delta_{2})= & \frac{1}{(\lambda z)^{2}}\int dx''\int dx'\:f(x')\exp\left(\frac{i\pi x'^{2}}{\lambda f_{1}^{\theta}}\right)\\
 & \exp\left(\frac{i\pi x''^{2}}{\lambda f_{2}^{\theta}}\right)\exp\left[-i\pi\frac{(x''-x')^{2}}{\lambda z}\right]\exp\left[-i\pi\frac{(x-x'')^{2}}{\lambda z}\right]
\end{align*}

Let us take a look at the integral in the variable $x''$
\begin{equation}
\int dx''\exp\left[\frac{i\pi x''^{2}}{\lambda f_{2}^{\theta}}-\frac{2i\pi x''^{2}}{\lambda z}+\frac{2i\pi x''(x+x')}{\lambda z}\right].
\end{equation}
Substituting (\ref{eq:f2}) and solving the Gaussian integral it becomes
\begin{equation}
\int dx''\exp\left[-\frac{2i\pi\sin\theta x''^{2}}{\lambda z}+\frac{2i\pi x''(x+x')}{\lambda z}\right]=\sqrt{\frac{\lambda z}{2i\sin\theta}}\exp\left[\frac{i\pi(x'+x)^{2}}{2\lambda z\sin\theta}\right].\label{eq:interm}
\end{equation}

The field amplitude is then given by
\begin{equation}
E(x,\Delta_{1}+2z+\Delta_{2})=\frac{1}{(\lambda z)^{2}}\sqrt{\frac{\lambda z}{2i\sin\theta}}\int dx'\:f(x',y')\exp\left(\frac{i\pi x'^{2}}{\lambda f_{1}^{\theta}}+\frac{i\pi(x'+x)^{2}}{2\lambda z\sin\theta}-\frac{i\pi x'^{2}}{\lambda z}-\frac{i\pi x{}^{2}}{\lambda z}\right),
\end{equation}
Finally, using (\ref{eq:f1}) and the trigonometric relation $\cot(\theta/2)=\frac{1}{\sin\theta}+\cot\theta$
all the undesired phases in $x'$ are canceled out. The remaining
spurious phases in $x$ are removed by the last lens. The field at
the output plane is given by
\begin{equation}
E(x,2\Delta_{1}+2z+\Delta_{2})=\frac{1}{(\lambda z)^{2}}\sqrt{\frac{\lambda z}{2i\sin\theta}}\int dx'\:f(x',y')\exp\left\{ \frac{i\pi\,}{2\lambda z}\left[\cot\theta\,\left(x'^{2}+x^{2}\right)-2xx'\csc\theta\right]\right\} ,
\end{equation}
that is precisely the FrFT of angle $\theta$ of the field amplitude
in the input plane $f(x)$. The scaling factor $\sqrt{2\lambda z}$
is common to both phase space variables and das not depend on the
FrFT order. If there are no further transformations after the second
free propagation, then the last lens is not required, since it just
corrects the phase.

This scheme was proposed in Ref. \cite{Rodrigo09} and is not the
only way to perform a optical FrFT. For example, a simpler scheme
using only one lens could be used \cite{Ozaktas93}, but then the
free propagation distances must be changed for each phase space direction
one wants to access. The three-lenses scheme allows one to keep the
free space propagation distances fixed. That is, the lenses can have
fixed positions, if the focal lengths are changed accordingly. As
we show in the next section, since the a lens effect is to imprint
a position dependent phase to the wave front, a spatial light modulator
can be used to mimic a lens of any focal length. Thus the three-lenses
scheme provides a method to realize optical FrFT in a programmable
way.

\selectlanguage{american}%

\selectlanguage{english}%

\section{Spatial Light Modulator}

\textit{Spatial light modulator} (SLM) is a common term used to identify
devices that modulate the phase, the amplitude, or the polarization
of a light beam by means of diverse physical phenomena like acusto-optic
and electro-optic effects, and liquid crystal anisotropy \cite{saleh1991}.
In this section the basic operation of a liquid crystal SLM is explained.

Basically, liquid crystals (LC) are materials which are in a fluid
phase as liquids, being able to adapt their shapes to the recipient,
at the same time that they present anisotropic features like crystals
\cite{grabmaier1975}. Typically, materials which present a LC phase
are composed by elongated molecules. The anisotropy then comes from
the alignment of the long molecule axis in a preferred direction.
When the molecules are in average all aligned in the same direction,
but their centers are randomly distributed, the LC has only one anisitropic
axis. In this case, the LC is called nematic and it behaves like a
uniaxial crystal with optic axis in the same direction as the molecule
orientation.

\begin{figure}[h]
\noindent \begin{centering}
\subfloat[\foreignlanguage{american}{}]{\includegraphics[width=0.4\textwidth]{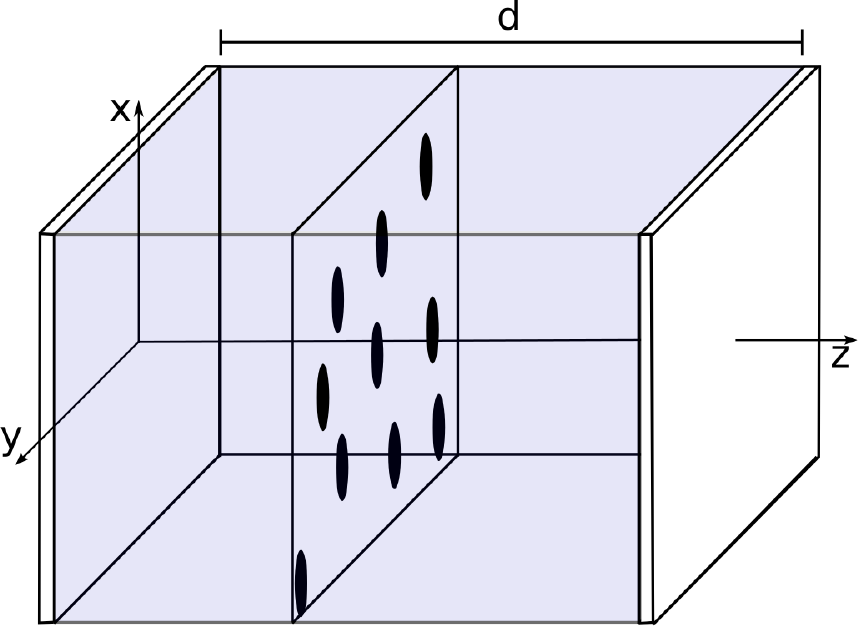}

}\qquad{}\subfloat[\foreignlanguage{american}{}]{\includegraphics[width=0.4\textwidth]{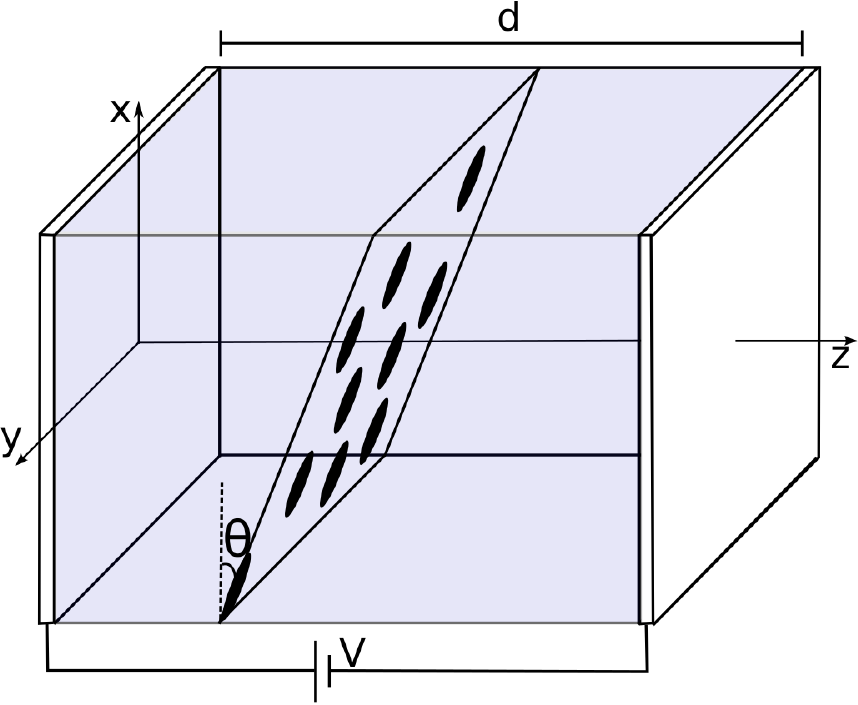}

}
\par\end{centering}
\caption{(a) Nematic liquid crystal cell with width $d$, enclosed by two glass
plates which keep the LC molecules aligned with the $x$ direction.
(b) If a potential difference $V$ is applied through the cell in
the $z$ direction, it generates a dipole force that tilts the molecules
by an angle $\theta$ in the stead position. \label{fig:(a)-C=0000E9lula-de}}
\end{figure}

Let us consider a cell of a nematic LC with thickness $d$ and a square
face as shown in Fig. \ref{fig:(a)-C=0000E9lula-de}-a). Because of
the interaction with the two glass plates involving the LC material,
the molecules tend to align with the $x$ direction. A transverse
field can be applied with a electric potential difference $V$ across
the cell, generating a dipole force in the $z$ direction which cause
the molecules to rotate until they reach a steady position at angle
$\theta$ relative to the $x$ axis (Fig. \ref{fig:(a)-C=0000E9lula-de}-b)).
The inclination of the molecules is given as a function of the potential
difference by\cite{saleh1991}
\begin{equation}
\theta=\begin{cases}
0, & V<V_{c}\\
\frac{\pi}{2}-2\textrm{tg}^{-1}\left[\textrm{exp}\left(-\frac{V-V_{c}}{V_{0}}\right)\right], & V>V_{c}
\end{cases},
\end{equation}
where $V_{0}$ is a constant characterizing the material, and $V_{c}$
is the critical potential above which the molecules start tilting.

The direction of the molecules long axis is the extraordinary direction
of the crystal with index of refraction $n_{e}$, while the two perpendicular
directions have refractive index $n_{\vartheta}$. When an electromagnetic
wave propagates inside the LC cell in the $z$ direction, it feels
the ordinary refractive index $n_{\vartheta}$ if it is $y$-polarized,
while the refractive index will be

\begin{equation}
\frac{1}{n(\theta)^{2}}=\frac{\cos^{2}\theta}{n_{\vartheta}^{2}}+\frac{\textrm{sen}^{2}\theta}{n_{e}^{2}}
\end{equation}
for $x$-polarized waves (App. \ref{chap:Anexo2}). Thus, the optical
path of the $x$-polarized component inside the LC cell can be manipulated
by changing the applied voltage. Analogously to Eq. (\ref{eq:camp}),
this optical path difference results in a phase difference between
the two orthogonal polarizations given by
\begin{equation}
\phi(\theta)=2\pi\frac{n_{2}-n(\theta)}{\lambda_{0}}d,
\end{equation}
that is, an electrically-controlled phase modulation occurs.

\begin{figure}[h]
\noindent \centering{}\includegraphics[width=0.6\textwidth]{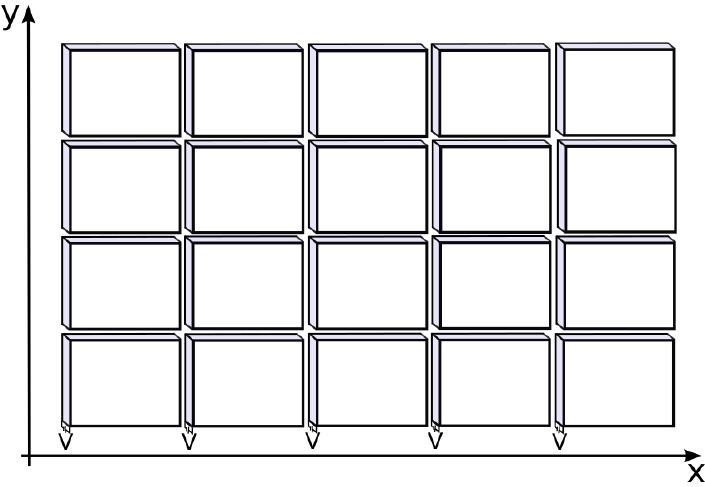}\caption{Many LC identical cells disposed side by side composing a SLM screen.
Each cell's voltage is individually controlled such that a position
dependent phase is imprinted to the wave front of the $x$-polarized
component of a electromagnetic wave.\label{fig:modulador}}
\end{figure}

A SLM is a screen composed of thousands of microscopic LC cells arranged
side by side as shown in Fig. \ref{fig:modulador}. The voltage across
each cell is individually electronic controlled. So it is possible
to program the SLM to apply any position depend phase $f(x,y)$ by
changing the voltages as to make $2\pi\frac{n_{2}-n(\theta(x,y))}{\lambda_{0}}d=f(x,y)$,
where $\theta(x,y)$ is the inclination of the molecules inside the
cell located at position $(x,y)$ on the screen plane. The refractive
index in the $y$ direction is fixed and the phase aquired by the
$x$-polarized component is relative between the two components of
polarization. The action of a SLM in the $\{\hat{x},\hat{y}\}$ basis
is given by the position dependent operator 
\begin{equation}
\hat{P}\left[f(x,y)\right]=\left[\begin{array}{cc}
e^{if(x,y)} & 0\\
0 & 1
\end{array}\right],\label{eq:slm}
\end{equation}
disregarding the global phase. Of course the applied phase is not
the spatially continuous function $f(x,y)$ but a discretized version
of it, whose resolution is given by the size of the LC cells.

In our experiments, we use Holoeye SLMs, which have the long molecule
axis is in the horizontal and thus the horizontal polarization is
modulated. The resolution of these SLMs is $1920\times1080$ pixels,
with $8\,\mu$m pixel pitch. Therefore, if a phase function varies
considerably in a interval smaller than $8\,\mu$m, then the SLM does
not capture its detailed features. These devices are build on a silicon
matrix in which lies the electronic control parts, thus the light
is not transmitted through the cells but it is reflected passing twice
through the LC{ \cite{johnson1993}}. These
SLMs have a fill factor of $92\%$ meaning that $8\%$ of the total
active area is empty space between neighbor cells. This causes part
of the incident light to scatter by diffraction and approximately
$40\%$ of the light is lost in each SLM use. This kind of SLM works
coupled to a computer as an additional screen. The phase function
is converted to a gray scale image which is projected in the SLM.
Each value of the gray scale number of 8 bits (natural numbers from
0 to 255) is associated to a value in the phase interval usually with
a roughly linear correspondence, which can be adjusted manually for
each SLM. Ideally, if the phase interval is $[0,2\pi]$, then the
color $0$ corresponds to applied phase $0$ and color 255 corresponds
to apply a phase of $2\pi$. For each of the 256 gray scale values,
a voltage is assigned. This operation mode discretize also the values
of phase possible, since only the values of the form $(n-1)\frac{2\pi}{255}$
can be set, with $n=1,2,...,256$. 

\subsection{SLM calibration}

A SLM typically works for a large range of wavelengths imprinting
phases in an adjustable phase interval. To assure that, for a fixed
wavelength ($632.8\,$nm produced by a He-Ne laser in our case), the
right phase interval is reached, a calibration process is needed.
To find out what is the phase imprinted for each gray tone we interfere
a modulated wave with one that has not been applied a phase difference.
As the SLM modulates only the horizontal polarization, we interfere
the incoming horizontal and vertical polarizations in a HWP as is
illustrated in Fig. \ref{fig:Set-up-for}. First of all a horizontally
polarized beam is prepared from a continuous wave laser using a polarizing
beam splitter (PBS), a device that transmits horizontal polarization
while reflecting the vertical polarized component. The beam then passes
through a QWP set to $-45^{\circ}$ which transforms the polarization
to circular
\begin{equation}
E_{0}\hat{\boldsymbol{H}}\stackrel{\small{QWP_{@45^{\circ}}}}{\longrightarrow}E_{0}\frac{\hat{\boldsymbol{H}}+i\hat{\boldsymbol{V}}}{\sqrt{2}},
\end{equation}
$E_{0}$ is the field amplitude after the PBS. The light is then reflected
by the SLM whose screen uniformly projects the same grayscale which
corresponds to some unknown phase $\theta$. The correspondent transformation
is given by
\begin{equation}
E_{0}\frac{\hat{\boldsymbol{H}}+i\hat{\boldsymbol{V}}}{\sqrt{2}}\stackrel{\small{SLM_{@\theta}}}{\longrightarrow}E_{0}\frac{\hat{\boldsymbol{H}}e^{i\theta}+i\hat{\boldsymbol{V}}}{\sqrt{2}}.\label{eq:calibreslm}
\end{equation}
After that , the beam goes back through the same path passing again
through the QWP which mixes modulated and unmodulated components
\begin{equation}
E_{0}\frac{\hat{\boldsymbol{H}}e^{i\theta}+i\hat{\boldsymbol{V}}}{\sqrt{2}}\stackrel{\small{QWP_{@45}}}{\longrightarrow}E_{0}\frac{\hat{\boldsymbol{H}}\left(e^{i\theta}-1\right)+i\hat{\boldsymbol{V}}\left(e^{i\theta}+1\right)}{2}.\label{eq:calibreslm2}
\end{equation}
Finally, the polarization components of the beam are separated by
the PBS and the vertical component intensity is measured by a power
meter. For a given $\theta$ the measured intensity $I$ is proportional
to the modulus square of the vertical electric field and thus 
\begin{equation}
I\propto\cos^{2}\frac{\theta}{2}\label{eq:power}
\end{equation}
what allows us to determine the phase $\theta$ corresponding to a
given grayscale level. 

\begin{figure}[h]
\noindent \begin{centering}
\includegraphics[width=0.5\textwidth]{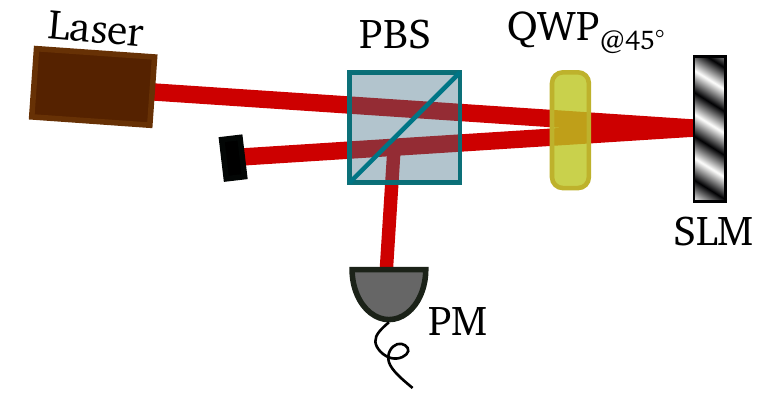}\caption{Set up for an SLM calibration: a light beam at the frequency the SLM
is going to be calibrated for passes through a polarizing beam spliter
(PBS) which lets only the horizontally polarized component be transmitted.
A HWP set to $22.5^{\circ}$ transforms the light to diagonal polarization.
Only the horizontal component is modulated by the SLM. Horizontal
and vertical polarization interfere when passing again through the
HWP and the intensity of vertical polarization component reflected
by the PBS is measured by a power meter (PM).\label{fig:Set-up-for}}
\par\end{centering}
\selectlanguage{american}%
\end{figure}

If the SLM is set to work in the range $[0,2\pi]$ with linear correspondence
between gray scale values and angles, then the measured power as a
function of the color should behave like Eq. (\ref{eq:power}), with
maximums at $0$ and $255$. To calibrate the SLM we vary the grayscale
along all 256 possible values, measuring the power of the vertically
polarized component for each one. When using the factory setting,
we obtain the red dots shown in Fig. \ref{fig:Measured-power-as}-a),
a behavior similar to the expected (blue solid line in the figure)
but with a slightly larger frequency. We use this measurements to
determine a new correspondence function between color and voltage
across the LC cells. After the reconfiguration, the power measurement
returns the correct behavior (Fig. \ref{fig:Measured-power-as}-b)).

\begin{figure}[h]
\noindent \begin{centering}
\includegraphics[width=0.9\textwidth]{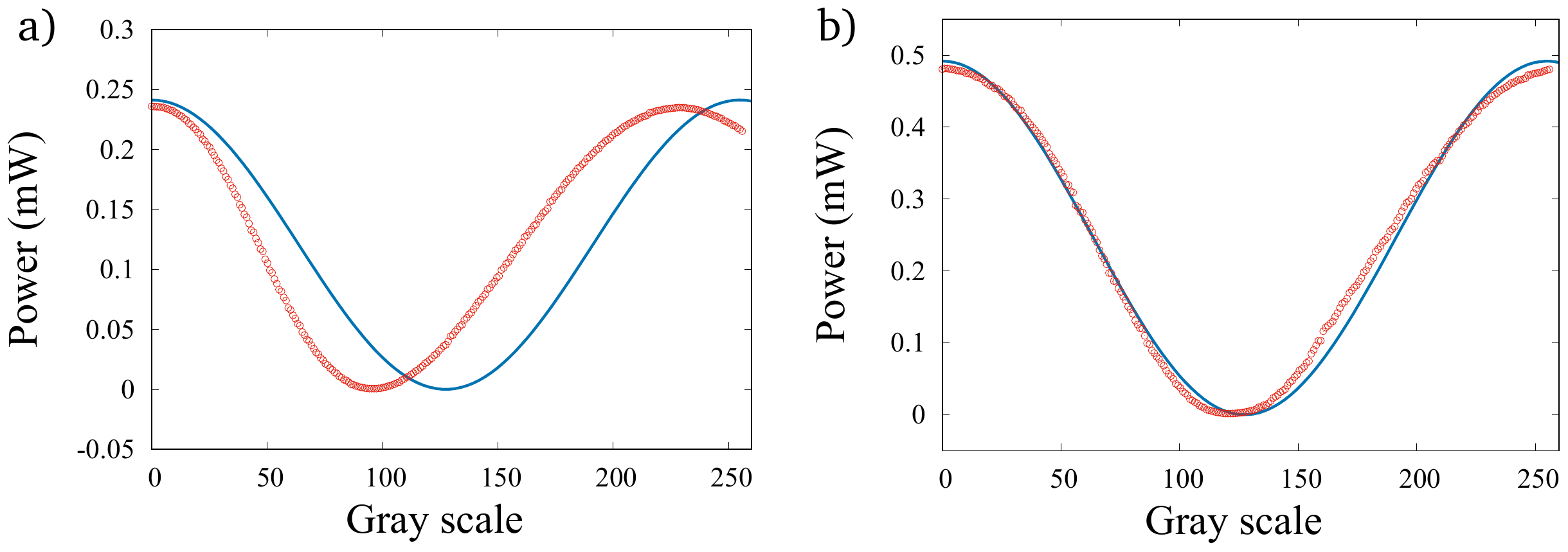}
\par\end{centering}
\caption{Measured power as a function of the gray color projected on the SLM
screen in a calibration process: a) with the factory setting, b) after
adjusting the voltage configuration to wavelength $632.8\,$nm and
phase range $[0,2\pi]$. In both plots the blue solid line is the
expected cosine function and the red dots are measured power values.
\label{fig:Measured-power-as}}
\end{figure}

\subsection{Beam positioning in the SLM plane}

Besides of assuring that the correct phase value is imprinted, a proper
coordinate system on the SLM plane must be defined, what means to
find a physically meaningful origin point since the direction of the
axes is defined by the rectangular array of LC cells. The origin is
defined relatively to the center of the light beam without any previous
spatial alteration\footnote{In the case when the beam hits other SLM screens before the one that
is being calibrated, all the previous ones must be kept with uniform
phase such that the position of the beam is not altered.}. 

\begin{figure}[h]
\noindent \begin{centering}
\includegraphics[width=0.4\columnwidth]{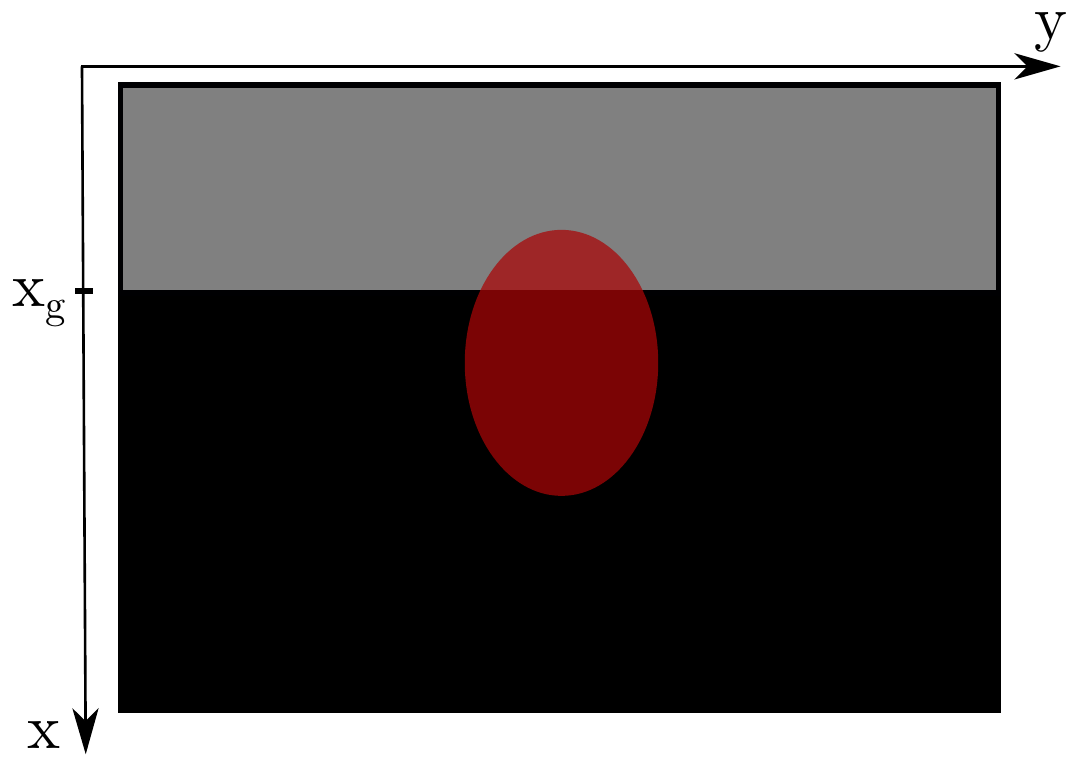}
\par\end{centering}
\caption{SLM screen projecting the image for positioning the beam. Part of
the red light beam gets the phase given by black color and part gets
the phase associated to the gray tone. \label{fig:SLM-screen-projecting}}
\end{figure}

We consider only one coordinate on the plane, let us say the vertical
position, and ignore the dependence on the second one, as it is the
case in our experiments. The laser produces a Gaussian mode with amplitude
proportional to $f(x,y)=\exp\left[-\frac{(x-x_{0})}{2\Delta_{x}^{2}}\right]\exp\left[-\frac{(y-y_{0})}{2\Delta_{y}^{2}}\right]$.
Let us consider a setup as the one used for phase calibration (Fig.
\ref{fig:Set-up-for}) but instead of a uniform grayscale in the entire
screen, the SLM projects an image as shown in Fig. \ref{fig:SLM-screen-projecting}
with the screen divided into two colors. When we wrote the field amplitude
in Eq. (\ref{eq:calibreslm}), we were not concerned with the spatial
distribution since the whole beam was receiving the same phase and
the spatial distribution would be integrated to give the total power.
Now, after the SLM action and after the QWP, instead of Eq. (\ref{eq:calibreslm2}),
we have two phase regions and we can write for the field amplitude
\begin{equation}
\boldsymbol{E}\propto\begin{cases}
f(x,y)\frac{\hat{\boldsymbol{H}}\left(e^{i\theta_{1}}-1\right)+i\hat{\boldsymbol{V}}\left(e^{i\theta_{1}}+1\right)}{2} & x<x_{g}\\
f(x,y)\frac{\hat{\boldsymbol{H}}\left(e^{i\theta_{2}}-1\right)+i\hat{\boldsymbol{V}}\left(e^{i\theta_{2}}+1\right)}{2} & x\geq x_{g}
\end{cases},
\end{equation}
$\theta_{1}$ and $\theta_{2}$ are the phases of each color region.
Consequently, the total power measured when selecting only vertical
polarization is proportional to 
\begin{equation}
\int_{-\infty}^{\infty}\,dy\left[\int_{-\infty}^{x_{g}}\,dx\:f(x,y)^{2}\,\cos^{2}\frac{\theta_{1}}{2}+\int_{x_{g}}^{\infty}\,dx\:f(x,y)^{2}\,\cos^{2}\frac{\theta_{2}}{2}\right].
\end{equation}
If the colors are chosen such that $\theta_{1}=\pi$ and $\theta_{2}=0$,
because the beam is Gaussian, the power will be a displaced complementary
error function of $x_{g}$ $\textrm{erfc}(x_{g}-x_{c})$ which center
$x_{c}$ coincides with the center of the Gaussian beam $x_{0}$.
So, by varying the position of the gray $\pi$-phase stripe from pixel
0 to the $1080$th pixel we can reconstruct the complementary error
function and determine the central position of the beam on the SLM
as shown in Fig. \ref{fig:Measured-power-in}. In general, we cannot
use exactly the setup of Fig. \ref{fig:Set-up-for} because the beam
positioning must be done with the whole experimental setup for the
actual experiment mounted, but any configuration able to interfere
horizontal and vertical polarization after the SLM action will behave
as a error-like function when the screen is scanned by the $\pi$-phase
region.

\begin{figure}[h]
\noindent \begin{centering}
\includegraphics[width=0.5\columnwidth]{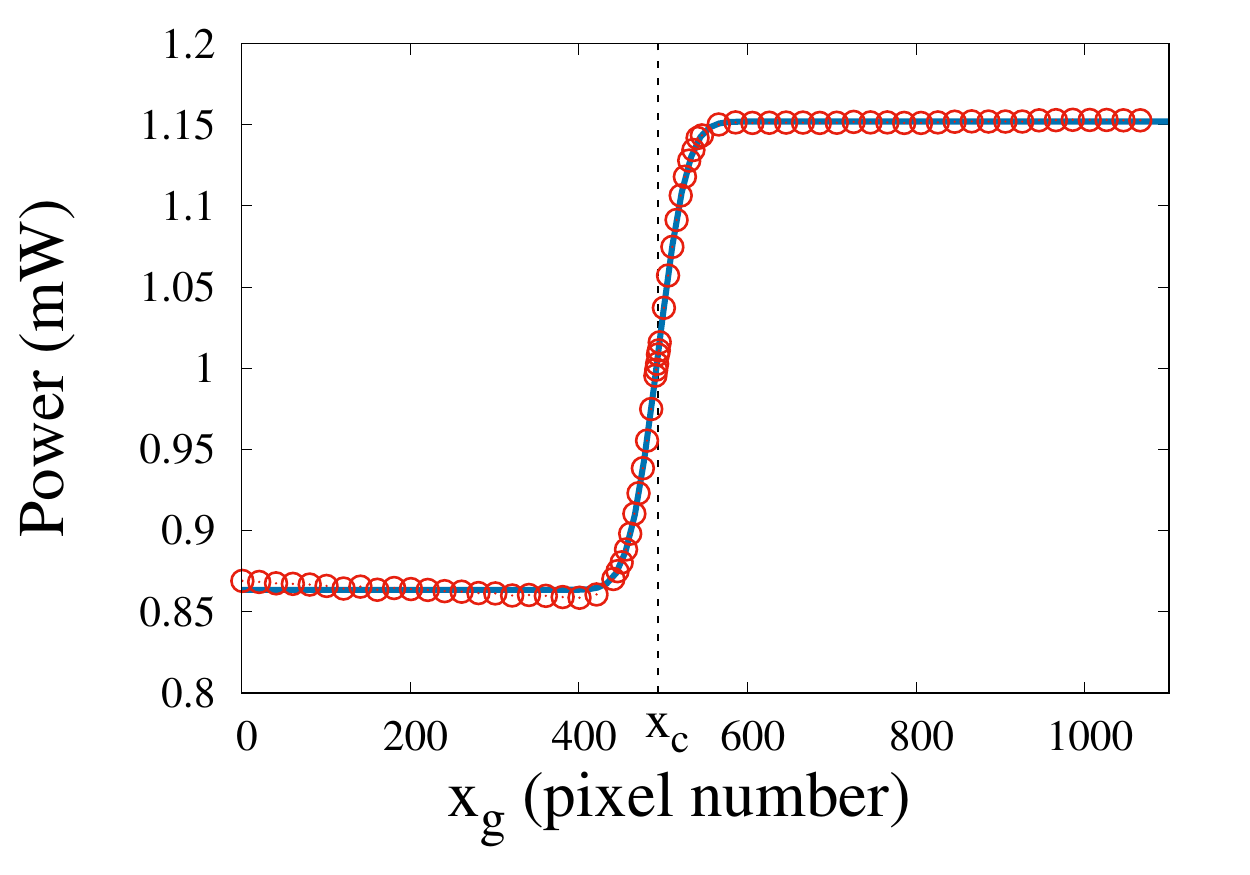}
\par\end{centering}
\caption{Power of vertical polarization component as a function of the end
position of the gray stripe. The dots are the measured values and
the blue line is the fitted complementary error function. \label{fig:Measured-power-in}}

\selectlanguage{american}%
\end{figure}

\subsection{Amplitude masks from phase modulation \label{subsec:Amplitude-masks-from}}

It is possible to produce amplitude masks from phase modulation by
using a phase-only SLM as a diffraction grating. A diffraction grating
is an optical element that periodically modulates the phase of a incident
beam. Thus, setting a SLM to imprint a periodic phase $\exp[i\,f(y)]$
makes it work as a diffraction grating in the $y$ direction. Consider
a paraxial incident beam making a small angle $\theta_{i}$ with the
plane $xz$, where $z$ is the normal direction to the SLM plane.
If the imprinted phase has a period $\Lambda$ which is much larger
than the beam wavelength $\lambda$, then the reflection of the beam
by the SLM generates several beams at angles \cite{saleh1991}
\[
\theta_{m}=\theta_{i}+m\frac{\lambda}{\Lambda},\quad m\in\mathbb{Z}.
\]
Thus, the small the period of the phase function, the larger is the
separation between two consecutive diffracted beams. However, the
period, and consequently the beam separation, is limited by the SLM
pixel size. The relative power in each diffraction order is determined
by the shape of the phase function. 

Now, let us consider that the $y$-direction periodic phase is multiplied
by a mask 
\[
M(x)=\begin{cases}
1 & x\in R\\
0 & \textrm{otherwise,}
\end{cases}
\]
that is, the mask vanishes outside the region $R$ of the $x$ axis,
which could be the union of several disconnected intervals. An example
of the resultant phase pattern is shown in Fig. \ref{fig:SLM-used-as}.
The effect is to select the portion of the beam incident in the desired
region, that is, the beam falling outside $R$ is reflected to the
zero order of diffraction and the higher orders of diffraction contain
only the beam incident inside $R$. Different mask functions may also
be used, allowing for a more general amplitude modulation. The diffraction
grating can also be summed to any other phase function varying with
the $x$ coordinate, such that this coordinate can be independently
manipulated. The advantage of using only the higher order diffracted
beams is to ensure that the whole beam has been modulated, otherwise
it would not be diffracted. However, it is worth noting that this
application of an SLM only works for beams linearly polarized in the
direction of the LC molecules, and thus cannot be used in experiments
which tangle polarization and spatial degrees of freedom.

\begin{figure}[h]
\noindent \begin{centering}
\includegraphics[width=0.5\columnwidth]{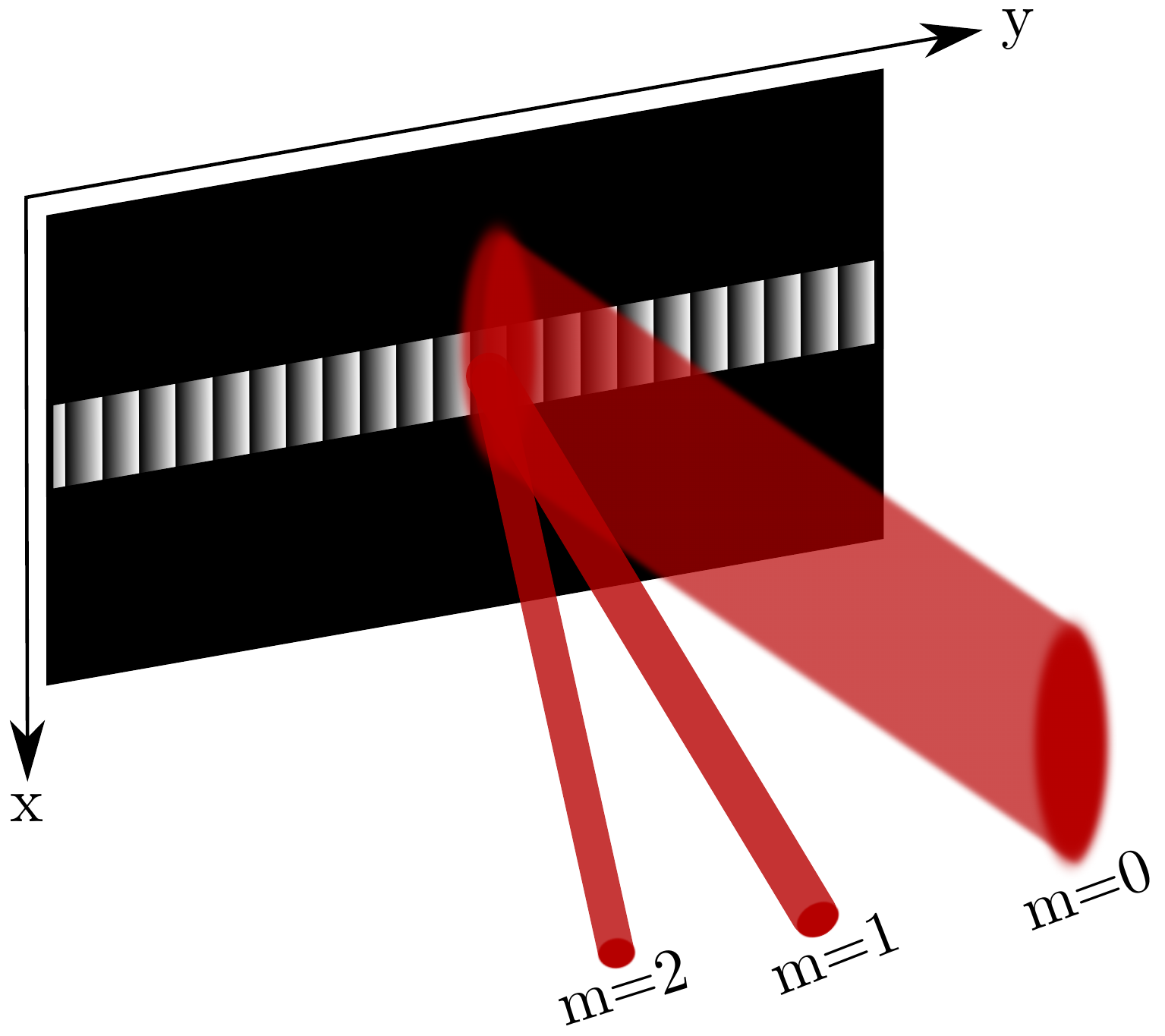}
\par\end{centering}
\caption{SLM used as an amplitude mask: only the portion of the incident beam
(not shown in the picture) that hits the SLM in the graded region
will be present in order $m\protect\geq1$ of difraction. The beam
with the order $m=0$ of difraction also contains the portion reflected
in the non-graded region.\label{fig:SLM-used-as}}

\selectlanguage{american}%
\end{figure}
\selectlanguage{american}

~\ihead{}

\ohead{\textbf{Chapter~\thechapter}~\leftmark}

\ifoot{}

\cfoot{}

\ofoot[
]{\thepage}

\chapter{Optical simulation of the free Dirac Equation\label{chap:Dispositivos-=0000F3pticos}}

In this chapter, I present a classical optics simulation of the one-dimensional
Dirac equation for a free particle. For this purpose we use both discrete
and continuous degrees of freedom. Positive and negative energy components
are represented by orthogonal polarizations of a free propagating
beam, while the spatial profile represents the spatial wave function
of the particle. Using a completely tunable time parameter, we observed
the oscillation of the average value of the Dirac position operator--known
as \emph{Zitterbewegung} (ZB). We are also able to measure the so
called mean-position operator which is a single-particle observable
and presents no oscillations. Our work opens the way for optical simulation
of interesting phenomenon of relativistic systems, as well as condensed-matter
physics, without any requirement for specially engineered medium. 

This work was done in collaboration with Ardiley T. Avelar, Rafael
M. Gomes and Emile R.F. Taillebois from Federal University of Goiás,
and Stephen P. Walborn. My contribution to this work was devising
the simulation protocol, designing the experiment, setting the experiment
up, and also to do the data analysis and help writing the paper which
is now published in Phys. Rev. A \cite{silva2019}. 

\section{Introduction}

\noindent \qquad{}Although the Dirac equation \cite{dirac1928} represents
a historical landmark in the quantum description of relativistic systems
-- satisfactorily explaining the origin of spin and predicting the
existence of antimatter \cite{anderson1933} -- it still provokes
a lot of discussion about its interpretation, even when applied to
describe the simplest physical system, that is, a free particle. In
this case, the Dirac equation predicts intriguing phenomena, for instance
the ZB \cite{schrodinger1930} and Klein's paradox \cite{klein1929},
which impede the single-particle (SP) interpretation of the Dirac
equation. As fundamental effects in the understanding of relativistic
influence over quantum theory, they have contributed to the transition
to the many-body quantum field approach \cite{grobe2004}.

The important technical difficulties involved in the direct observation
of several relativistic quantum predictions have led to an increased
interest in their simulations in trapped ions \cite{gerritsma2010,Gerritsma2011},
photonic crystals \cite{zhang2008}, confined light \cite{otterbach2009},
graphene \cite{katsnelson2006}, optical superlattices \cite{Dreisow2010},
Bose-Einstein condensates \cite{salger2011,leblanc2013} and ultracold
atoms \cite{vaishnav2008}. Among the unexpected effects of Dirac
equation, the ZB -- the flickering motion of a free relativistic
quantum particle described by a Dirac wavefunction with positive and
negative energy components -- is one of most investigated due to
its interesting counterintuitive nature.

From a SP perspective, Dirac's equation must be interpreted as the
simultaneous solving of two independent problems, for instance, the
single free evolution of both electron (positive energy sub-space)
and positron (negative energy sub-space). Therefore, there is no meaning
in assigning a physical interpretation to an operator that is not
SP, that is, an operator that mixes positive and negative subspaces.
In the SP approach, physical results must be obtained by projecting
SP observables over the subspace corresponding to the problem of interest.
Despite the previous studies \cite{gerritsma2010,Dreisow2010,leblanc2013},
an important feature was not explored: the Dirac position operator
related to the flickering motion is not a SP observable, i.e. it cannot
be written as the direct sum of its positive and negative energy projections. 

For Dirac's theory, a SP position observable exists and is obtained
using the so-called Foldy--Wouthuysen transformation (FWT) \cite{foldy1950}
-- a momentum dependent unitary transformation that diagonalizes
Dirac's Hamiltonian and is at the kernel of important algorithms used
to obtain quantum relativistic corrections \cite{obukhov2001,quach2015}.
This observable is often called mean-position operator and does not
exhibit the oscillatory behavior characteristic of the ZB, a result
that generates doubts concerning its actual existence.

Here, the simulation of the Dirac one-dimensional free evolution and
the ZB is performed using the transverse degrees of freedom of a paraxial
light beam, where different components of the spinor are represented
by different polarization components of the beam. This physical setup
is well suited for the purpose of quantum simulation, as it allows
for implementation of the dynamical phases with easy tuning of the
important physical parameters \cite{lemos2012}. Besides being a proof-of-concept
for the optical simulation of Dirac particles, the present approach
differs from others in the theoretical procedure adopted to perform
the simulation. Previous works perform a direct simulation of the
1+1D Dirac Hamiltonian, while the present approach performs the evolution
in the diagonalized Foldy--Wouthuysen representation (FWR) and permits
one to switch back and forth between this and Dirac's representation,
allowing us to investigate the behavior of both the Dirac position
and the mean-position operators.

\section{Dirac Equation and position operator}

Consider the 1D Dirac equation 
\begin{equation}
i\hbar\frac{\partial}{\partial t}\psi_{D}=\hat{\mathcal{H}}_{D}\psi_{D}=(c\hat{p}\sigma_{1}+mc^{2}\sigma_{3})\psi_{D},\label{dirac}
\end{equation}
where $c$ is the speed of light, $\hat{p}$ is the momentum operator,
$m$ is the mass of the particle, and $\sigma_{i}$ are the usual
Pauli matrices. The information of this system is encoded in the spinor
$\psi_{D}$ that has only two components which are related to positive
and negative energy states in the particle's rest frame, i.e. spin
degrees of freedom are eliminated by the dimensional constrain \cite{schuabl}.
In Dirac's coordinate representation, the momentum operator $\hat{p}$
assumes the usual form $-i\hbar\frac{\partial}{\partial x}$, where
$x$ is the so-called Dirac coordinate associated to the multiplication
operator $\hat{x}_{D}\psi_{D}(x)=x\psi_{D}(x)$. Since the Hamiltonian
operator is not diagonal in this representation, the positive and
negative energy eigenstates are non-trivial and assume, respectively,
the forms $\psi_{p}^{+}(x,t)=u(p)e^{-i\varepsilon(p)t/\hbar}e^{ipx/\hbar}$
and $\psi_{p}^{-}(x,t)=v(p)e^{i\varepsilon(p)t/\hbar}e^{-ipx/\hbar}$,
with $\varepsilon(p)\equiv\sqrt{(mc^{2})^{2}+(pc)^{2}}$, $u(p)=[2mc^{2}(\varepsilon(p)+mc^{2})]^{-1/2}\begin{pmatrix}\varepsilon(p)+mc^{2} & cp\end{pmatrix}^{T}$
and $v(p)=[2mc^{2}(\varepsilon(p)+mc^{2})]^{-1/2}\begin{pmatrix}cp & \varepsilon(p)+mc^{2}\end{pmatrix}^{T}$.

The non-diagonal form of $\hat{\mathcal{H}}_{D}$ in Dirac's representation
is evinced by the commutator $[\hat{\mathcal{H}}_{D},\hat{x}_{D}]=-ic\hbar\sigma_{1}$
and leads to the Heisenberg picture evolution given by \cite{thaller}
\begin{equation}
\begin{aligned}\hat{x}_{D}(t)= & \hat{x}_{D}(0)+c^{2}\hat{p}\hat{\mathcal{H}}_{D}^{-1}t-\frac{c\hbar\hat{\mathcal{H}}_{D}^{-1}}{2i}\left(e^{2i\hat{\mathcal{H}}_{D}t/\hbar}-1\right)\left(c\hat{p}\hat{\mathcal{H}}_{D}^{-1}-\sigma_{1}\right).\end{aligned}
\label{eqEvol}
\end{equation}
The first two terms on the right represent the expected linear time
evolution of a free particle, the last term being associated to the
ZB. This flickering motion is accompanied by other particularities
of the $\hat{x}_{D}$ operator. Indeed, the evolution given in (\ref{eqEvol})
is derived from the equation of motion $\dot{\hat{x}}_{D}=\frac{i}{\hbar}[\hat{\mathcal{H}}_{D},\hat{x}_{D}]=c\sigma_{1},$
which implies that, although $\langle\dot{\hat{x}}_{D}\rangle=\langle c^{2}\hat{p}\hat{\mathcal{H}}_{D}^{-1}\rangle$,
the eigenvalues associated to the velocity ${\dot{\hat{x}}_{D}}$
are restricted to $\pm c$, a remarkable result which contributes
to raise doubts as to the correct interpretation of $\hat{x}_{D}$
as definition of position. These peculiarities of the operator $\hat{x}_{D}$
arise from the fact that this is not a SP observable, i.e. $\hat{x}_{D}\neq\hat{P}_{+}\hat{x}_{D}\hat{P}_{+}^{\dagger}+\hat{P}_{-}\hat{x}_{D}\hat{P}_{-}^{\dagger}$,
where $\hat{P}_{\epsilon}=\frac{1}{2mc^{2}}\left(\begin{smallmatrix}mc^{2}+\epsilon\varepsilon(p) & -\epsilon cp\\
\epsilon cp & mc^{2}-\epsilon\varepsilon(p)
\end{smallmatrix}\right)$ is the projection operator over the subspace of states with energy
sign $\epsilon$.

To obtain a SP position, the FWT must be applied to diagonalize the
Dirac Hamiltonian. For the 1D Dirac free particle, this canonical
transformation is given by the momentum dependent unitary operator
\begin{equation}
\hat{U}(\hat{p})=e^{i\hat{S}(\hat{p})},\quad\textrm{with}\hat{\quad S}(\hat{p})\equiv\frac{\sigma_{2}}{2}\mathrm{tg}^{-1}\left(\frac{\hat{p}}{mc}\right).\label{eq:FWoperator}
\end{equation}
In the resulting FWR, the original Dirac Hamiltonian is given by $\hat{\mathcal{H}}_{D}^{\prime}=\sigma_{3}\varepsilon(\hat{p})$,
and the former $\hat{x}_{D}$ operator by $\hat{x}_{D}^{\prime}=\hat{x}_{FW}+\frac{\hbar mc^{3}}{2\varepsilon(p)^{2}}\sigma_{2}$,
where $\hat{x}_{FW}$ is the new multiplication operator in the FWR.
The operator $\hat{x}_{FW}$ is the so called mean-position operator
and, unlike the operator $\hat{x}_{D}$, it is a SP observable since
$\hat{x}_{FW}=\hat{P}_{+}^{\prime}\hat{x}_{FW}\hat{P}_{+}^{\prime}+\hat{P}_{+}^{\prime}\hat{x}_{FW}\hat{P}_{+}^{\prime}$,
where $\hat{P}_{\epsilon}^{\prime}=\left(\begin{smallmatrix}\delta_{\epsilon+} & 0\\
0 & \delta_{\epsilon-}
\end{smallmatrix}\right)$ are the energy projectors in the new representation.

Aside from being SP, the operator $\hat{x}_{FW}$ also satisfies the
equation $\dot{\hat{x}}_{FW}=c^{2}\hat{p}\hat{\mathcal{H}}_{D}^{\prime}$,
resulting in the Heisenberg picture evolution \cite{thaller} 
\begin{equation}
\hat{x}_{FW}(t)=\hat{x}_{FW}(0)+c^{2}\hat{p}\hat{\mathcal{H}}_{D}^{\prime-1}t
\end{equation}
that is linear in time, as expected for a free particle. Thus, as
stated before, the ZB does not occur for this operator.

Here, as a proof-of-concept for the simulation of relativistic systems
using free propagating light beams, the simulation of both the Dirac
and FWRs will be performed in a single setup. This difference with
other simulation procedures open the possibility for future investigations
on more complex FWTs associated to relativistic scenarios involving
interactions.

\section{Simulation Protocol and Experiment}

One way to simulate the dynamics associated to Eq. (\ref{dirac})
is to directly implement the evolution operator $\exp\left(-\frac{i\hat{\mathcal{H}}_{D}t}{\hbar}\right)$,
which is usually a tough task due to the non-diagonal character of
$\hat{\mathcal{H}}_{D}$. This difficulty can be overcome by using
the FWT, since this transformation allows the time evolution operator
to be written as $\exp\left(-\frac{i\hat{\mathcal{H}}_{D}t}{\hbar}\right)=\hat{U}^{-1}(\hat{p})\exp\left(-\frac{i\hat{\mathcal{H}}_{D}^{\prime}t}{\hbar}\right)\hat{U}(\hat{p})$,
and $\hat{\mathcal{H}}_{D}^{\prime}$ is a diagonal operator. This
operator can be implemented in an optical beam by considering the
vertical coordinate on the transverse plane as the particle's position
and the horizontal (vertical) polarization as the superior (inferior)
component of the spinor. The horizontal spatial degrees of freedom
on the transverse plane play no relevant role in the experiment. Although
a spinor is a mathematical object which transforms very specifically
under a reference frame change, it is not a concern for this simulation
since the reference frame is assumed to be fixed.

The optical transformations required for the simulation are polarization
transformations (acting as nondiagonal operators) and phase shifts
(used to introduce momentum dependent phases). The former are obtained
with the suitable application of wave plates and the last are realized
by SLMs, which are able to imprint programmable position dependent
phases $\textrm{exp}[i\:\textrm{f}(x,y)]$ in the horizontal polarization.
The momentum-dependent phases are applied in the momentum space defined
as the optical Fourier transform of the position space where the initial
state is prepared. The position plane is shown as a dashed line in
Fig. \ref{fig:Experimental-setup}, the SLMs are placed such that
the optical Fourier transform connects position and momentum planes
as presented in Sec. \ref{subsec:Optical-Fourier-transform}. The
action of a quarter wave plate (QWP) set to $45^{\circ}$ is given
by the operator $\hat{Q}:=\hat{QWP}(45^{\circ})=\nicefrac{e^{i\tfrac{\pi}{4}}}{\sqrt{2}}\left(\mathbb{1}-i\sigma_{1}\right)$
{[}see Eq. (\ref{eq:QWP}){]}, while $\hat{H}:=\hat{HWP}(45^{\circ})=\sigma_{1}$
describes a half wave plate at $45^{\circ}$ {[}see Eq. (\ref{eq:HWP}){]}.
The action of a SLM is equivalent to applying $\hat{P}[f(x,y)]=\textrm{exp}[i\:\textrm{f}(x,y)]\sigma_{+}\sigma_{-}+\sigma_{-}\sigma_{+}$
over the transverse profile spinor., as can be seen from Eq. (\ref{eq:slm}).
Using this operator representation for the optical devices, it follows
that the operator (\ref{eq:FWoperator}) can be written in momentum
representation as
\begin{equation}
\hat{U}(p)=\hat{Q}\hat{P}\left[-\theta(p)\right]\hat{H}\hat{P}\left[\theta(p)\right]\hat{Q},
\end{equation}
with $2\theta(p)=\mathrm{tan}^{-1}\left(\frac{p}{mc}\right)$. We
express the inverse FWT in an analogous fashion. As the Hamiltonian
is diagonal in the FWR, the transformed time evolution operator is
achieved via the application of the dynamical phase $\exp\left[\pm it\varepsilon(p)/\hbar\right]$
in each spinor component using waveplates and the SLM, which concludes
the simulation. A summary of the analogy between the optical simulator
and the simulated system is given in Table I.

\begin{table}[h]
\noindent \centering{}\caption{Summary of the optical analogy}
\begin{tabular}{c|c}
\hline 
\textbf{Optical System}  & \textbf{Simulated System}\tabularnewline
\hline 
\hline 
Vertical transverse position  & $x$\tabularnewline
\hline 
Transverse profile of  & \multirow{2}{*}{$\psi_{D_{1}}(x)$}\tabularnewline
horizontal polarization & \tabularnewline
\hline 
Transverse profile of  & \multirow{2}{*}{$\psi_{D_{2}}(x)$}\tabularnewline
vertical polarization & \tabularnewline
\hline 
QWP@45°  & $\ensuremath{\hat{Q}=\nicefrac{e^{i\tfrac{\pi}{4}}}{\sqrt{2}}\left(\mathbb{1}-i\sigma_{1}\right)}$\tabularnewline
\hline 
HWP@45°  & $\hat{H}=\sigma_{1}$\tabularnewline
\hline 
\multirow{2}{*}{SLM printing phase $f(x,y)$ } & $\hat{P}[f(x,y)]=$ \tabularnewline
 & $\textrm{exp}[i\:\textrm{f}(x,y)]\sigma_{+}\sigma_{-}+\sigma_{-}\sigma_{+}$\tabularnewline
\hline 
Normalized horizontal  & \multirow{2}{*}{$\left|\psi_{D_{1}}(x)\right|^{2}$}\tabularnewline
polarization intensity at $x$ & \tabularnewline
\end{tabular}
\end{table}

The experimental scheme is shown in Fig.\ref{fig:Experimental-setup}.
A He-Ne laser with wavelength 632.8 nm and two Holoeye reflective
SLMs, each of which divided into halves to operate twice, are used.
As in this experiment we use both vertical and horizontal polarizations,
we need to use the zero order diffraction of the SLM although not
all the light is modulated in this order. The Fourier transforms are
made by plano-convex cylindrical lenses with 150 mm focal distance
such that the position space (mirrors and camera) is in one focal
plane and the momentum space is in the opposite focal plane where
the SLM is located. The reason to use cylindrical lenses is that only
the vertical transverse coordinate is used and thus only this direction
must be transformed. The momentum $p$ and the position on the SLM,
$X$, are connected by $X=\frac{\lambda f}{h}p$, where $f$ is the
focal distance and $\lambda$ is the laser wavelength \cite{saleh1991}.
In terms of $X$, the applied phases become $2\theta(X)=\mathrm{tan}^{-1}\left(\frac{h}{mc}\frac{X}{\lambda f}\right)$
and $t\varepsilon(p)/\hbar=2\pi t\sqrt{\left(\frac{X}{\lambda f}\right)^{2}+\left(\frac{mc}{h}\right)^{2}}$,
so the parameters we need to set are the speed of light $c$ and the
Compton wavelength $\lambda_{C}=h/mc$, which are easily tunable since
they enter as programmable parameters in the imprinted phases. Notice
that contrary to what is usual in optical simulations \cite{lemos2012},
the time in our simulation does not correspond to the propagation
distance of the beam as it would come in a direct analogy between
Schrödinger and paraxial Helmholtz equations \cite{dragoman2013}.
Since the time coordinate also comes up as a programmable parameter,
we could in principle take measurements for as many time values as
we wish inside a time interval. This also implies that the unit of
measurement for time is an arbitrary $\tau$. In this realization
we chose $\Delta t/\tau=1$.

\begin{figure}[h]
\noindent \centering{}\includegraphics[width=0.65\columnwidth]{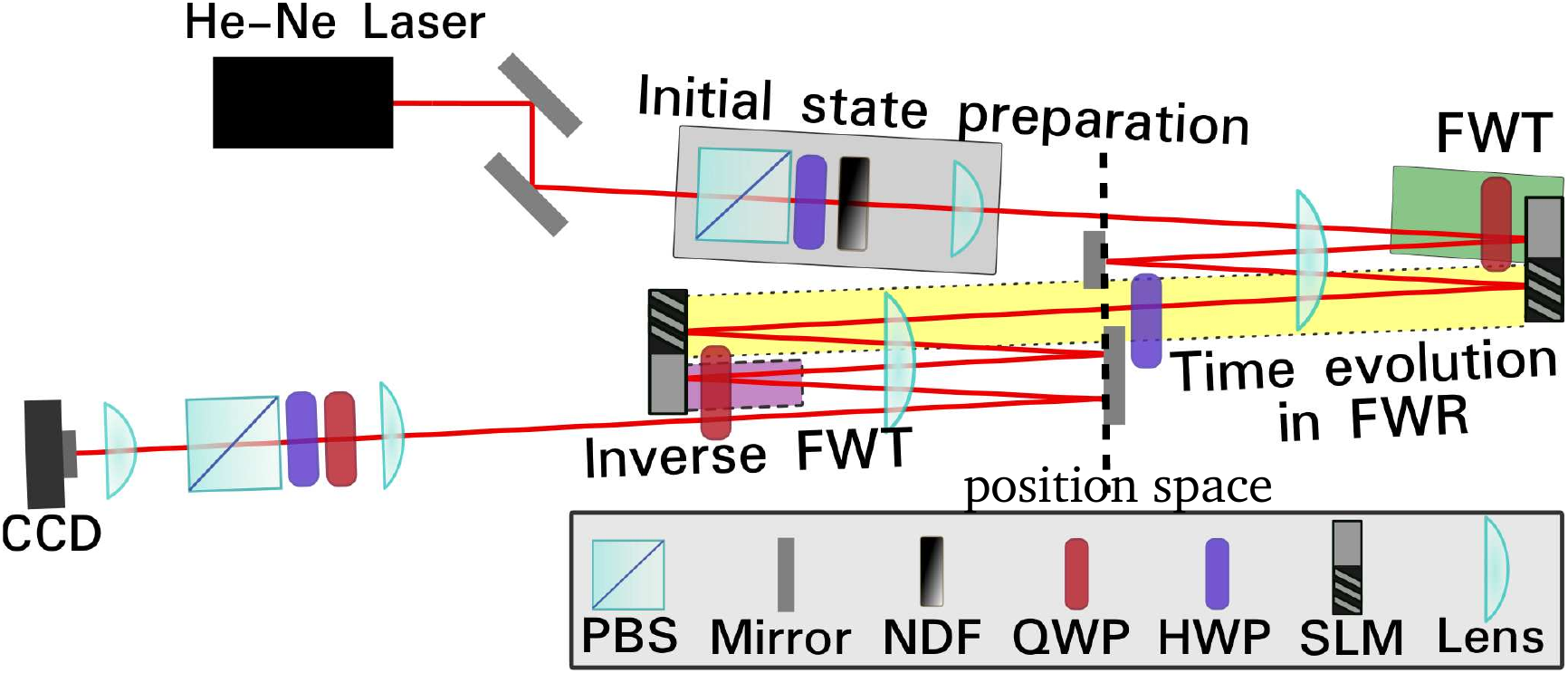}\caption{\label{fig:Experimental-setup}A He-Ne laser, prepared with an initial
gaussian profile and anti-diagonal polarization state, is sent through
an optical system designed to implement the Dirac Hamiltonian. The
grey shaded regions of the SLMs implement the FWTs, while the striped
regions implement the free-evolution. Lenses are used to map the field
profile among the different planes of the SLMs, and wave plates to
control the polarization state. A CCD camera is used to register the
intensity profile of the output field. Additional details are provided
in the text. }
\end{figure}


The laser produces a Gaussian spatial profile separable in the $x$
and $y$ coordinates, so the initial spinor is 
\begin{equation}
\psi_{D}(x,t=0)=\frac{e^{-i\pi x^{2}/(\lambda R)}e^{-x^{2}/(4\Delta^{2})}}{(\sqrt{2\pi}2\Delta)^{1/2}}\begin{pmatrix}a\\
b
\end{pmatrix},\label{eq:initial_state}
\end{equation}
where $a$ and $b$ are the normalized horizontal and vertical polarization
coefficients ($|a|^{2}+|b|^{2}=1$), $\Delta$ is the beam width in
the vertical direction, and $R$ is the vertical radius of curvature
of the beam in the initial position plane. The propagation and Gouy
terms of the Gaussian beam only introduce global phases which do not
affect the dynamical evolution \cite{saleh1991}. We start with $a=-b=1$,
but changing $a$ and $b$ would enable us to prepare different positive
and negative energy superpositions. Two cylindrical lenses are placed
before the first position space in order to manipulate the initial
momentum distribution which depends on $R$ and therefore on $\Delta$.
Using a beam profiler, we determined the initial state parameters
to be $\Delta=48,6\,\mu m$ and $\lambda R/\pi=(2.2\Delta)^{2}$.

The average position of the simulated particle is calculated as 
\[
\langle\hat{x}_{D}\rangle(t)=\sum_{i=1,2}\int dx\,x|\psi_{Di}(x,t)|^{2},
\]
where $\left|\psi_{Di}(x,t)\right|^{2}$ is proportional to the light
intensity of polarization component $i$ at position $x$ on the transverse
plane measured by a CCD camera placed at the output position space.
Each instant of time corresponds to one programmable-phases configuration
and one intensity-profile measurement. It is worth noting that the
evolved state is accessible for any time value.

\section{Results}

\begin{figure}[h]
\noindent \centering{}\includegraphics[width=0.7\columnwidth]{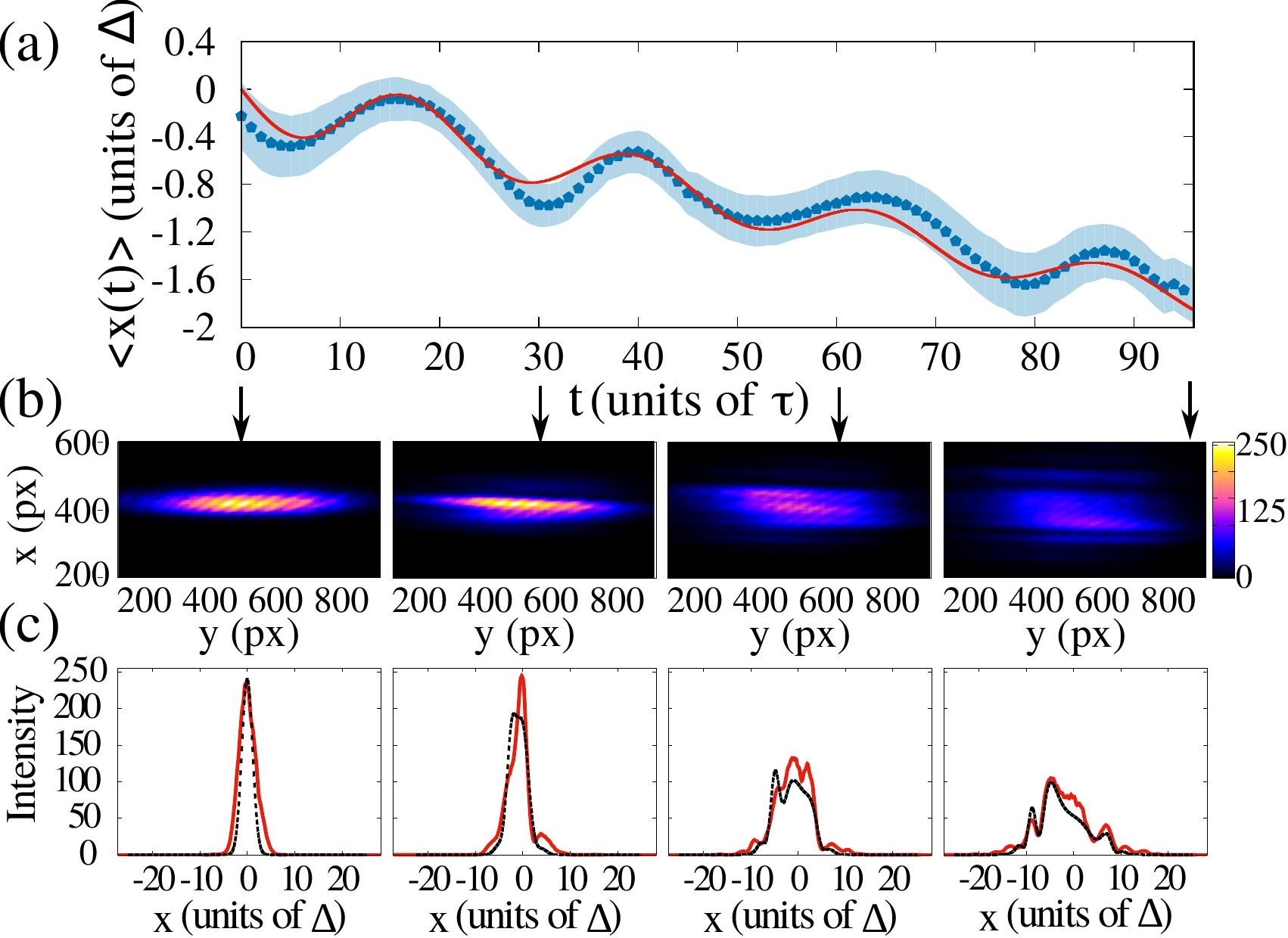}\caption{\label{fig:L5Results-1} (a) Mean position $\langle\hat{x}_{D}(t)\rangle$
as a function of the time parameter $t$ for Compton length $\lambda_{C}/\Delta=5$.
Points are experimental data, the dashed line is the numerical solution
and filled band is the experimental error of one $\sigma$. (b) The
camera image and (c) the one dimensional state got from marginalizing
the intensity measured by the camera over $y$ for four values of
the parameter $t$, namely $t/\tau=0,\:30,\:60\textrm{ and }95$,
are shown. In (c) red solid lines are experimental data and black
dashed curves are numerical solutions.}
\end{figure}


A summary of our experimental procedure for the particular case $\lambda_{C}=5\Delta$
is depicted in Fig.\ref{fig:L5Results-1}. In Fig.\ref{fig:L5Results-1}-a)
we present the mean position $\langle\hat{x}_{D}\rangle(t)$ as a
function of $t$, the ZB being evidenced by the oscillatory behavior.
The solid red line is the theoretical prediction, while points are
experimental results obtained from the images shown in Fig.\ref{fig:L5Results-1}-b).
The shaded region represents uncertainty of one $\sigma$. Fig.\ref{fig:L5Results-1}-b)
shows samples of the data collected by the CCD camera for some instants
of time, the $x$ distribution being obtained by considering only
a fixed $y$ coordinate at the center of the beam. The $x$ distributions
used to calculate $\langle\hat{x}_{D}\rangle(t)$ are shown in Fig.\ref{fig:L5Results-1}-c).
This procedure assumes that $x$ and $y$ intensity distributions
remain separable throughout all the apparatus. This is true in the
ideal case, however the cylindrical lenses can introduce some non-separability
as one can see in the slightly tilted elliptical intensity pattern
shown in Fig. \ref{fig:L5Results-1}-b). The non-separability causes
the initial state to be not entirely pure. Since our experimental
results agree well with theory, we conclude that these effects are
negligible for the present experiment.

For a fixed initial state and speed of light $c=0.1\,\nicefrac{\Delta}{\tau}$,
we measured the average position in Dirac's representation for different
values of the Compton wavelength, as shown in Fig.\ref{fig:AllMeanPosition}.
We fitted the average position with the function $vt+A\sin(\omega t+\delta)$
for each $\lambda_{C}$ to estimate the mean velocity, amplitude and
frequency of the oscillation. These quantities are shown in Fig.\ref{fig:amp_freq}.
Our experimental results are in agreement with the expected linear
dependence of amplitude and inverse dependence of frequency on $\lambda_{C}$
for small $\lambda_{C}$ \footnote{From Eq.(\ref{eqEvol}) we have that the amplitude of oscillations
for a given momentum eigenstate is proportional to $ch/\epsilon(p)$,
which for small $\lambda_{C}$, i.e., large mass, can be approximated
by $\lambda_{C}$. On the other hand, the frequency is proportional
to $\epsilon(p)/h$, that approaches $c/\lambda_{C}$ for large masses.}, as can be seen in Fig.\ref{fig:amp_freq}-b). This is consistent
with the fact that the ZB visibility in Fig.\ref{fig:AllMeanPosition}
increases for smaller values of $\lambda_{C}$.

\begin{figure}[h]
\begin{centering}
\includegraphics[width=0.7\columnwidth]{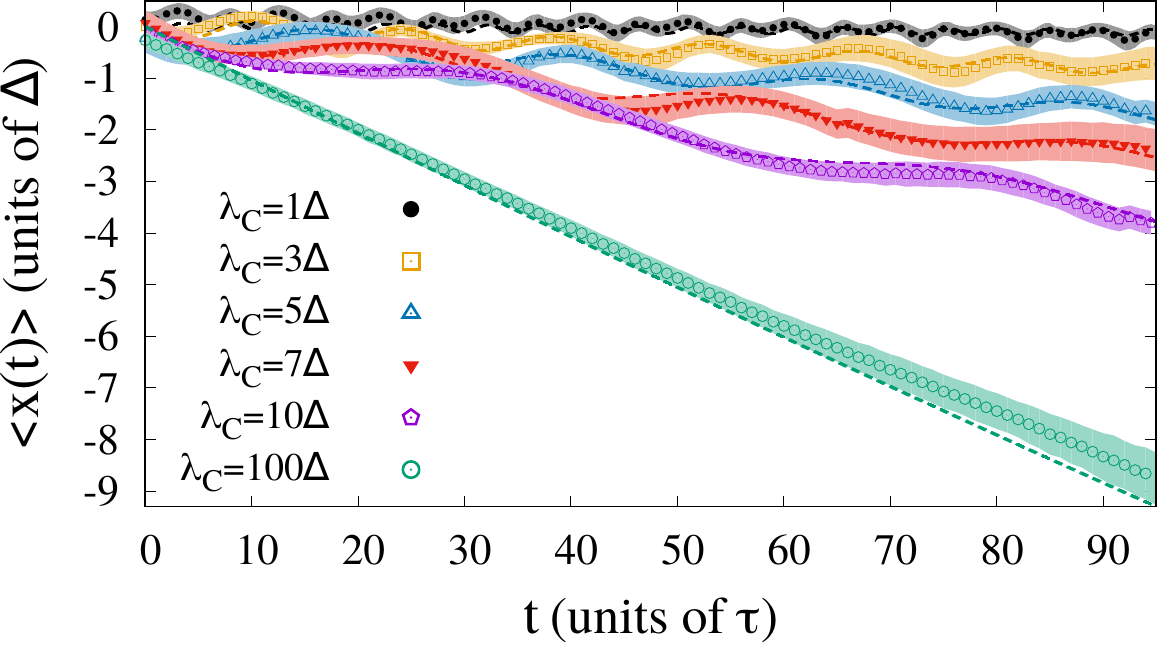} 
\par\end{centering}
 \caption{\label{fig:AllMeanPosition} Average position $\langle\hat{x}_{D}(t)\rangle$
as a function of the time parameter $t$ for Compton lengths $\lambda_{c}/\Delta=1,3,5,7,10,100$.
Dashed lines are numerical predictions of the theory using Eq. (\ref{dirac}),
while points are experimental results. The filled bands are experimental
errors of one $\sigma$.}
\end{figure}


The different inclinations exhibited in Fig.\ref{fig:AllMeanPosition}
are due to the fact that each mass, i.e. Compton wavelength, is associated
to a different velocity distribution, even the momentum distribution
being the same for all values of $\lambda_{C}$. Although the initial
state (\ref{eq:initial_state}) has zero average momentum, this is
not true for the mean velocity in Fig.\ref{fig:amp_freq}-a). As is
expected, the mean velocity falls quadratically with $\lambda_{C}$
for large masses (small $\lambda_{C}$), while it is close to the
speed of light for very small masses ($\lambda_{C}=100$).

The agreement between the experimental ZB data and the theoretical
predictions confirm that our optical setup is well suited for the
study of 1+1D relativistic dynamical systems, the theoretical extension
to larger dimensions being discussed in the Sec. \ref{sec:More-spatial-dimensions}.
Beside serving as a proof-of-concept, the proposed setup permits to
investigate the system in the FWR, an interesting possibility since
it allows to describe the dynamics of the system according to the
single-particle perspective, i.e. assigning physical sense only to
the projections of single-particle operators over the subspaces of
definite sign of energy.

From the SP perspective, operators that are not block-diagonal in
the FWR, as is the case for $\hat{x}_{D}$, have no physical meaning,
since they mix components of positive and negative energy that are
associated to two distinct problems. On the other hand, operators
that are block-diagonal in the FWR, as $\hat{x}_{FW}$ or the Dirac
Hamiltonian, may have a physical sense assigned to their positive
and negative projections. In this sense, the correct description of
the dynamics of a single electron (positron), for example, should
be given by $\hat{x}_{FW,+}=P_{+}^{\prime}\hat{x}_{FW}P_{+}^{\prime}$
($\hat{x}_{FW,-}=P_{-}^{\prime}\hat{x}_{FW}P_{-}^{\prime}$). In our
setup, the average $\langle\hat{x}_{FW}\rangle$ measured using both
positive and negative components of the spinor can be obtained by
measuring the transverse profile of the beam before the inverse FWT.
However, since the positive and negative components of the spinor
are encoded in the horizontal and vertical polarizations of the beam
in the FWR, the single-particle position dynamics described by $\hat{x}_{FW,+}$
(particle) and $\hat{x}_{FW,-}$ (anti-particle) is also accessible
by simply selecting one of the polarizations prior to the CCD measurement
in the FWR.

\begin{figure}[h]
\noindent \centering{}\includegraphics[width=0.7\columnwidth]{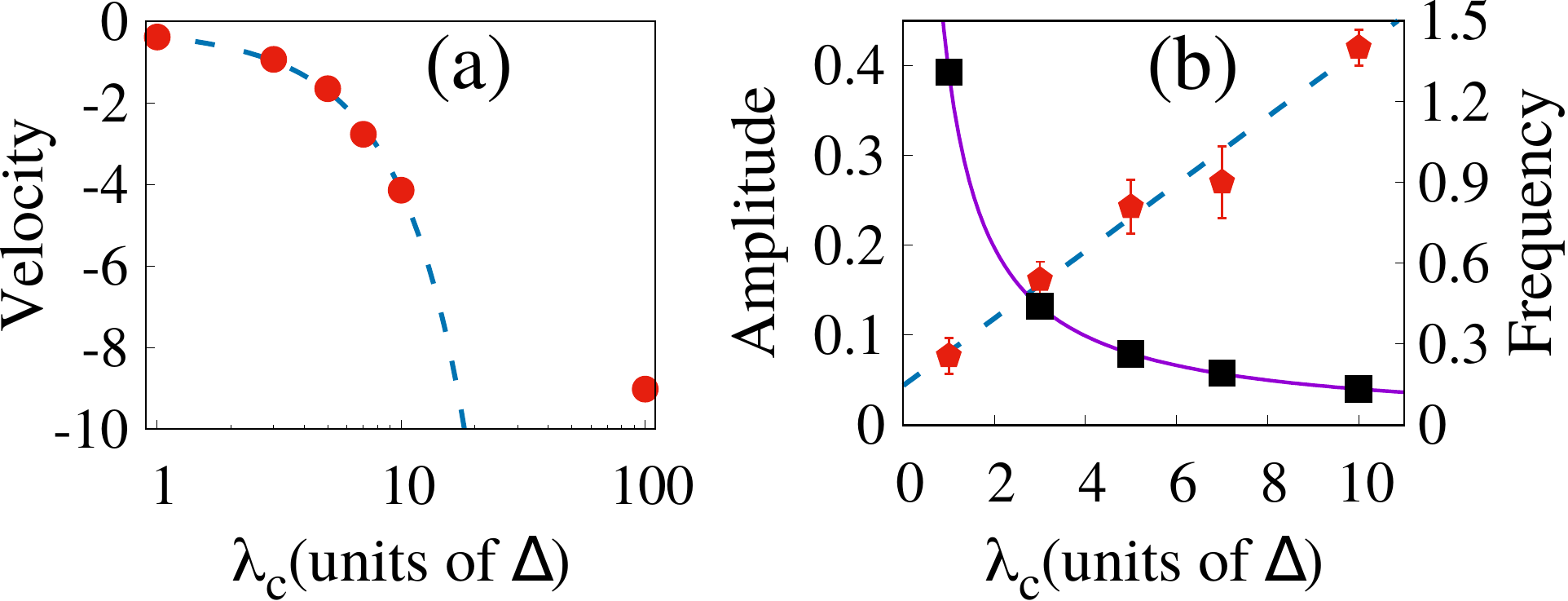}
\caption{\label{fig:amp_freq} (a) Mean velocity, (b) amplitude (red circles)
and frequency (black squares) of ZB obtained by fitting a sinusoidal
function to the average position $\langle\hat{x}_{D}(t)\rangle$ for
each Compton wavelength. For small values of Compton wavelength (large
masses) the mean velocity falls quadratically with $\lambda_{C}$
(dashed blue line) and for a large value of Compton wavelength (vanishing
mass) it appoximates the speed of light set on the experiment. The
amplitude dependence with $\lambda_{C}$ is in well agreement with
linear behavior (dashed blue line) while frequency is proportional
to $1/\lambda_{c}$ (solid purple line). }
\end{figure}


Experimental results for the mean-position operator are shown in Fig.
\ref{fig:L5FWMeanPosition} for $\lambda_{C}=5\Delta$. The experimental
data concerning the ZB effect for the $\hat{x}_{D}$ operator is also
plotted for comparison (blue points). Measurements of $\langle\hat{x}_{FW,+}(t)\rangle$
and $\langle\hat{x}_{FW,-}(t)\rangle$ are shown as black dots. As
is expected from the independence of the two problems in the SP description,
we have two independent mean trajectories corresponding to the free
evolution of the particle and the corresponding anti-particle. The
ZB is not present for these mean trajectories and a linear behavior
in time is observed, as it was expected. As mentioned earlier, we
were also able to measure the mean value of the FW mean-position operator
$\hat{x}_{FW}$. The results are plotted as the red circles and represent
an average of the positive and negative projections cases. The small
deviation from a perfect linear behavior can be explained assuming
that the SLMs do not modulate all the incident light but a fraction
of it, as is shown in the inset picture which shows the same mean
values as the experimental plot but obtained from a numerical simulation
of the experiment for modulation efficiency of $95\%$ in each SLM.

\begin{figure}[!h]
\noindent \centering{}\includegraphics[width=0.65\columnwidth]{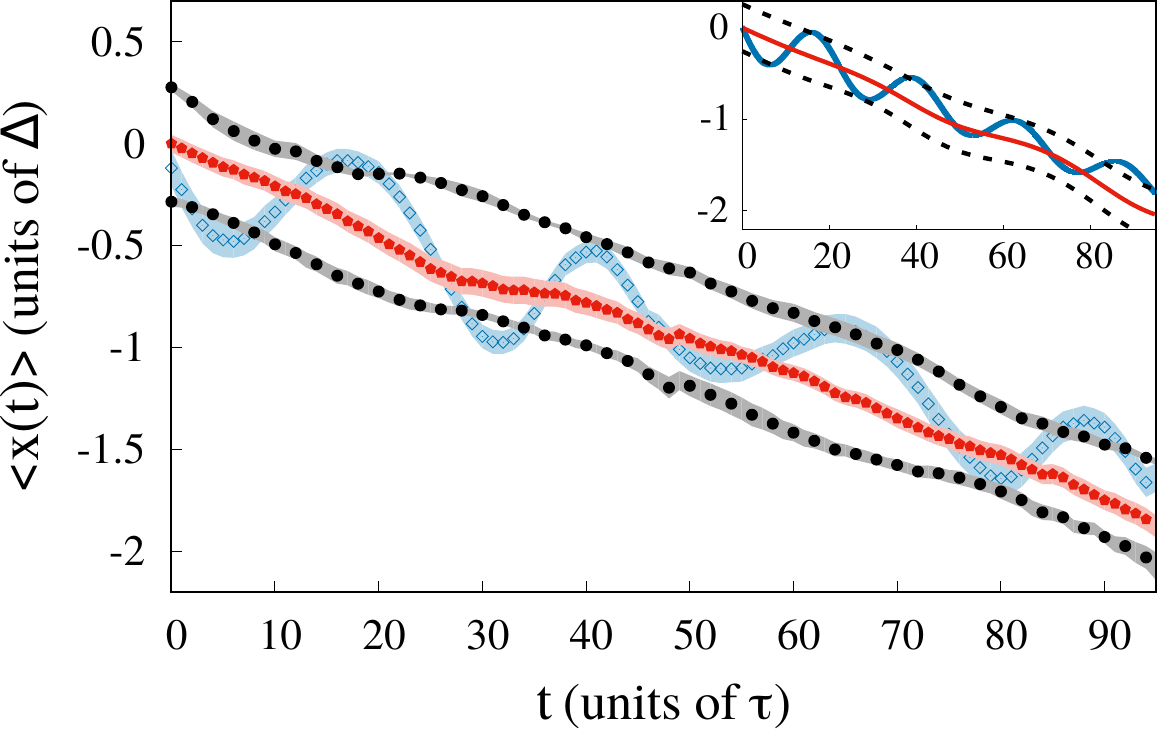}\caption{\label{fig:L5FWMeanPosition} Experimental mean positions $\langle\hat{x}_{D}(t)\rangle$
(blue dots) and $\langle\hat{x}_{FW}(t)\rangle$ (red circles) for
$\lambda_{C}/\Delta=5$. Measurements of $\hat{x}_{FW}(t)$ projected
over positive energy components (horizontal polarization in FWR) as
well as negative energy components (vertical polarization in FWR)
are shown as black dots. The shaded regions are the mean value uncertainties.
The $\langle\hat{x}_{FW}(t)\rangle$ data presents no ZB and fits
a linear dependence with $R^{2}=99.6\%$. The inset picture shows
the result of a numeric simulation of the experimental setup assuming
ideal devices except that the SLMs have efficiency of $95\%$.}
\end{figure}

\section{More spatial dimensions and potentials\label{sec:More-spatial-dimensions}}

The goal of this section is to show the simulation protocols for particles
in two and three spatial dimensions. Because we use the transverse
spatial degrees of freedom of a light beam, we only have at disposal
two coordinates to play the role of particle position. Also, the polarizations
used as spinor components, allows us to simulate only a two-components
spinor. Thus, this restricts our simulation protocol to one and two
spatial dimensions, the last one is presented in Section \ref{subsec:Simulation-of-2+1}.
Despite the apparent impossibility of simulating a particle existing
in a three dimensional space using our scheme, in Section \ref{subsec:Simulation-of-3+1}
we show that the simulation is possible for some particular cases. 

Finally, it would be desirable to include interactions in the particle
Hamiltonian. The drawback of using a simulation based on the FWT is
that this transformation is not exact for most of the potential functions
\cite{foldy1950}. But for a particular class of potentials it is
possible to include interactions in our simulation as we show in Section
\ref{subsec:Simulation-of-1+1}.

\subsection{Simulation of 2+1 dimensional Dirac Equation\label{subsec:Simulation-of-2+1}}

Consider the Dirac Hamiltonian for a free particle existing in a 2D
space 
\begin{equation}
H=c(\sigma_{1}p_{x}+\sigma_{2}\mathrm{p}_{y})+mc^{2}\sigma_{3}.
\end{equation}
The unitary transformation $e^{iS_{2}}=e^{-i\frac{\sigma_{1}p_{x}-\sigma_{2}p_{y}}{|\mathbf{p}|}\theta(\mathbf{p})}$
is the FWT which diagonalizes the Hamiltonian in this case, with $\theta(\mathbf{p})=\frac{1}{2}\textrm{tg}^{-1}\frac{|\mathbf{p}|}{mc}$.
For the same reasons as in the 1D situation, there is no spin if the
space is restricted to two dimensions and the particle state is a
spinor with two components. Once we manage to construct FWT from optical
device operators, the simulation protocol is made possible identifying
again the transverse profile of a laser beam in the two orthogonal
polarizations with the spinor components and identifying the two transverse
coordinates with particle position. To show that in fact there is
such a decomposition let us define a new momentum dependent phase
\begin{equation}
\theta'(\mathbf{p})=\begin{cases}
\mathrm{tg}^{-1}(p_{y}/p_{x}), & p_{y}>0\textrm{ or }p_{y}=0,p_{x}>0\\
\pi+\mathrm{tg}^{-1}(p_{y}/p_{x}), & p_{y}<0\textrm{ or }p_{y}=0,p_{x}<0,
\end{cases}\label{eq:thetalinha-1}
\end{equation}
motivated by the polar expression $p_{x}-ip_{y}=|\mathbf{p}|e^{-i\theta'}$.
In terms of the two phase functions the FWT reads 
\begin{equation}
e^{iS_{2}}=\left(\begin{array}{cc}
\cos\theta(\mathbf{p}) & e^{-i\theta'}\sin\theta(\mathbf{p})\\
-e^{i\theta'}\sin\theta(\mathbf{p}) & \cos\theta(\mathbf{p})
\end{array}\right).\label{eq:TFW2D-1}
\end{equation}
Using the definitions given in the main text it is easy to see that
the operator sequence 
\[
\hat{P}\left[-\theta^{'}(\mathbf{p})\right]\hat{Q}\hat{P}\left[-\theta(\mathbf{p})\right]\hat{H}\hat{P}\left[\theta(\mathbf{p})\right]\hat{Q}\hat{P}\left[\theta^{'}(\mathbf{p})\right]
\]
is equal to the FW unitary. The transformed diagonalized time evolution
is then given by the product $\hat{H}\hat{P}\left[i\varepsilon(\mathbf{p})t/\hbar\right]\hat{H}\hat{P}\left[-i\varepsilon(\mathbf{p})t/\hbar\right]$.
Just to conclude the protocol, the inverse FWT is given analogously
by the same set of devices as the FWT with different imprinted phases
and different angles for the wave plates.

\subsection{Simulation of 3+1 dimensional Dirac Equation for a particular class
of initial states\label{subsec:Simulation-of-3+1}}

Let us consider the 3+1 Dirac equation 
\begin{equation}
i\hbar\frac{\partial\psi(\mathbf{x},t)}{\partial t}=\left[c\boldsymbol{\alpha}\cdot\boldsymbol{p}+mc^{2}\beta\right]\psi(\mathbf{x},t)
\end{equation}
with the standard choice of Dirac matrices $\beta=\begin{pmatrix}\mathbb{1}_{2\times2} & 0\\
0 & -\mathbb{1}_{2\times2}
\end{pmatrix}\quad,\qquad\alpha_{k}=\begin{pmatrix}0 & \sigma_{k}\\
\sigma_{k} & 0
\end{pmatrix}$. The four components of the spinor $\psi(\mathbf{x},t)$ accounts
for the two signs of the energy and for the two spin projections along
a fixed direction. The FWT reads $e^{iS_{3}}=\cos{\theta(\mathbf{p})}+\beta\frac{\boldsymbol{\alpha\cdot}\mathbf{p}}{|\mathbf{p}|}\mathrm{sen}\,\theta(\mathbf{p})$
with the same definition for $\theta(\mathbf{p})$ as before.

A general simulation of the above equation, besides of requiring a
four dimensional object to emulate the four spinor components, it
would also require three spatial degrees of freedom, while the presented
setup allows for just two. Instead of proposing a complete simulation,
let us consider only the particular family of initial states given
by 
\begin{equation}
\psi(\mathbf{x},t=0)=\begin{cases}
\frac{1}{\sqrt{L}}\left(\begin{array}{c}
\tilde{\phi}_{1}(x,y)\\
\vdots\\
\tilde{\phi}_{4}(x,y)
\end{array}\right) & -\frac{L}{2}<z<\frac{L}{2}\\
0 & \mathrm{elsewhere}
\end{cases}.
\end{equation}
In momentum space this state becomes 
\begin{equation}
\tilde{\psi}(t=0,\boldsymbol{p})=\left(\begin{array}{c}
\tilde{\phi}_{1}(p_{x},p_{y})\\
\vdots\\
\tilde{\phi}_{4}(p_{x},p_{y})
\end{array}\right)\sqrt{\frac{L}{2\pi\hbar}}\left(\frac{\mathrm{sen}\,\frac{p_{z}L}{2\hbar}}{\frac{p_{z}L}{2\hbar}}\right),
\end{equation}
which depedence on $p_{z}$ behaves like a $\delta(p_{z})$ for large
values of $L$.

Thus we can approximate the FW transformed state by 
\begin{equation}
\tilde{\psi}'(\boldsymbol{p},t=0)\approx\left(\cos{\left(\theta(p_{x},p_{y},0)\right)}+\beta\frac{\alpha_{x}p_{x}+\alpha_{y}p_{y}}{|(p_{x},p_{y},0)|}\mathrm{sen}\,\theta(p_{x},p_{y},0)\right)\tilde{\psi}(\boldsymbol{p},t=0),
\end{equation}
 and all the momentum dependent phases only depend on two coordinates
and can be applied with SLMs.

The transformed state is explicitly written as 
\begin{equation}
\psi'(\boldsymbol{p},t=0)=\left(\begin{array}{cccc}
\cos{\theta} & 0 & 0 & \frac{p_{x}-ip_{y}}{|\mathbf{p}|}\mathrm{sen}\,\theta\\
0 & \cos{\theta} & \frac{p_{x}+ip_{y}}{|\mathbf{p}|}\mathrm{sen}\,\theta & 0\\
0 & \frac{-p_{x}+ip_{y}}{|\mathbf{p}|}\mathrm{sen}\,\theta & \cos{\theta} & 0\\
\frac{-p_{x}-ip_{y}}{|\mathbf{p}|}\mathrm{sen}\,\theta & 0 & 0 & \cos{\theta}
\end{array}\right)\left(\begin{array}{c}
\phi_{1}\\
\phi_{2}\\
\phi_{3}\\
\phi_{4}
\end{array}\right)\sqrt{\frac{L}{2\pi\hbar}}\left(\frac{\mathrm{sen}\,\frac{p_{z}L}{2\hbar}}{\frac{p_{z}L}{2\hbar}}\right)
\end{equation}
 we notice that the FWR only mixes the components two by two what
makes possible to simulate it using two beams without any joint transformation
between them. Then we can address the transverse profiles of horizontal
and vertical polarizations of the first beam to $\phi_{1}$ and $\phi_{4}$
and the FWR as well as the subsequent diagonal time evolution do not
mix this components with the two remaining. Moreover, each pair of
mixed components transforms like the two spatial dimensions case (Eq.
(\ref{eq:TFW2D-1})) with the suitable phase signs.

The interesting thing about three spatial dimensions simulation is
that it would enable us to investigate also spin effects like the
spin analogous to \emph{Zitterbewegung}.

\subsection{Simulation of 1+1 dimensional Dirac Equation for a particular class
of potentials\label{subsec:Simulation-of-1+1}}

Employing the strategy of \cite{sabin2012}, we show in this subsection
that, if the initial state is conveniently prepared, our approach
is able to simulate the Dirac equation for a particular class of potentials.
To this end, consider the 1D Dirac equation 
\begin{equation}
i\hbar\frac{\partial}{\partial t}\psi_{D}=(c\hat{p}\sigma_{1}+mc^{2}\sigma_{3}+V(x))\psi_{D},\label{dirac2}
\end{equation}
where $V(x)$ is a spinorial potential of the form 
\begin{equation}
V(x)=V_{1}(x)\sigma_{1}.
\end{equation}

For the above particular case, we define $\phi_{D}$ such that 
\begin{equation}
\psi_{D}=e^{-i\frac{1}{\hbar c}\int V_{1}(x^{\prime})\mathbb{1}dx^{\prime}}\phi_{D}.
\end{equation}
The substitution of this state on Eq.(\ref{dirac2}) shows that the
spinor $\phi_{D}$ evolves according to the free Dirac equation (Eq.(\ref{dirac})).
Thus, if we want to simulate the time evolution of the initial state
$\psi_{D}(x,t=0)$, we need to prepare the state $\phi_{D}(x,t=0)=\exp\left[i\frac{1}{\hbar c}\int V_{1}(x^{\prime})\mathbb{1}dx^{\prime}\right]\psi_{D}(x,t=0)$
and the dynamics of the free evolution. From the experimental point
of view, this corresponds to applying a local phase in position space
to all the components of the spinor before performing the free evolution
in the way as it is shown in the main text.

\section{Discussion and Conclusions}

Our experiment demonstrates how relativistic dynamics can be studied
using classical optics, and opens the way to more sophisticated investigations.
For this purpose it would also be desirable to produce more general
initial states. This can be accomplished using intensity and phase
masks in the initial state preparation. The state produced in this
experiment had zero average momentum, but simply shifting the momentum
in all SLMs phases by the same $\Delta p$ can be interpreted as if
the state has non-vanishing average momentum.

In principle the method implemented in this simulation using the FWT
could be applied for other simulation schemes of the Dirac equation,
however this transformation requires applying a phase shift that is
proportional to the inverse tangent of momentum. In our approach the
application of this phase is fairly easy, thanks to the spatial light
modulator (SLM). However, in other systems, this is quite challenging.
Typically in continuous variable quantum simulators one can implement
Gaussian Hamiltonians, but non-Gaussian operations (third order and
above) are quite difficult \cite{Tasca11}. Thus, we believe that
our approach is quite interesting in this regard, as it allows one
to employ the FWT and investigate relevant aspects of it.

Albeit here we focused 1+1D case, the extension for 2+1D and for some
initial states in 3+1D is straightforward as shown in the Sec. \ref{sec:More-spatial-dimensions}.
The first is a direct extension considering the second transverse
coordinate of the beam as the second spatial degree of freedom of
the simulated particle. A 2+1D simulation also allows for investigation
of electronic behavior in bidimensional condensed matter systems such
as graphene \cite{katsnelson2006}, but still do not present any spin
effect. For a general 3+1D simulation it would be necessary a third
beam coordinate what is not available in this scheme. In spite of
this limitation, we showed a class of initial states which dependence
on the third coordinate does not alter the time evolution. In this
case, the two extra spinor components are provided by the polarization
components of a second beam.

It is well known that there is no exact FWT for the non-free Dirac
equation, i.e., if we add a potential to the free Dirac equation (\ref{dirac})
the Dirac Hamiltonian becomes no longer diagonalizable with one single
unitary transformation \cite{schuabl}. This seems to be a very limiting
factor of our simulation technique and indeed it is if we try to implement
the actual FWT for a potential problem. Instead of doing so, we can
try to find other kinds of unitary transformations which reproduce
the time evolution and are experimentaly feasible with the available
optical elements. Up to now we know that at least for a particular
class of potentials it is possible to break the time evolution operator
into a position dependent phase which carries all the information
about the potential followed by the free evolution presented in this
work. This particular case was also discussed in Sec. \ref{sec:More-spatial-dimensions}.

In conclusion, we have presented an all-optical simulation of the
dynamics of a one-dimensional relativistic free point particle, where
the beam's spatial profile plays the role of the particle's wavefunction,
and its orthogonal polarization components are associated to spinor
components. Our experiment is based on the diagonalization of the
Dirac Hamiltonian using the FWT, which allowed for the decomposition
of the unitary evolution into operations that are realizable with
off-the-shelf optical components. Adjusting the tunable time parameter
we observed the oscillatory ZB phenomenon for Dirac's position operator.
Using our experimental FWT, we were also able to address this phenomenon
from a single-particle perspective, where the position description
is given by the positive and negative energy projections of the single-particle
mean-position operator. This approach allowed us to observe the absence
of ZB oscillations for the particle and anti-particle single-particle
dynamical evolutions.\selectlanguage{american}

\ihead{}

\ohead{\resizebox{\columnwidth}{!}{\textbf{Chapter~\thechapter}~\leftmark}}

\ifoot{}

\cfoot{}

\ofoot[
]{\thepage}

\chapter{Mutual unbiasedness of coarse-grained measurements for an arbitrary
number of phase space observables\label{chap:MUM}}

Observables of continuous quantum variables can be made discrete by
binning them together, resulting in an observable with a finite number
$d$ of outcomes. These operators allow one to reproduce some properties
of measurements on discrete quantum systems. One example is mutual
unbiasedness, which continuous variable operators satisfy only in
limits that are unphysical, but physical discretized operators can
satisfy perfectly, as in the discrete case. In this chapter, it is
shown that binning of continuous observables can lead to operators
that are in a sense neither continuous nor discrete. In particular,
it is shown that the maximum number of mutually unbiased measurements
is three for even $d$, which is analogous to the continuous case.
However, for prime $d$ we can find $d+1$ mutually unbiased observables,
surpassing the continuous case and in partial analogy to the discrete
case. To illustrate this, an optical experiment is presented showing
four mutually unbiased measurements with $d=3$ outcomes. For odd
non-prime $d$, it is shown theoretically that the maximum number
of unbiased measurements follows neither the discrete nor the continuous
regimes. 

This work was done in collaboration with \L ukasz Rudnicki from the
Center for Theoretical Physics in Poland, Daniel Tasca from Fluminense
Federal University, and Stephen Walborn. My contribution to this work
is both theoretical and experimental. In the theoretical part, I showed
the solution for the possible angle between phase space directions,
showed that the pair dimensions are forbidden for 4 or more phase
space directions, and also demonstrated the maximum number of directions
in the odd dimension case. In the experimental part, I designed and
built the experiment and made the data analysis. This work is being
prepared for publication.

\section{Introduction}

Quantum physics separates itself from classical physics in a number
of ways. One of these is the incompatibility of measurements, which
lies at the heart of the complementarity principle \cite{bohr1928b},
uncertainty relations \cite{toscano18}, quantum contextuality \cite{amaral18},
the violation of Bell's inequalities \cite{brunner14}, quantum random
number generation \cite{herrero17}, among other topics. Mutual unbiasedness
(MU) plays a fundamental role in incompatibility. Two observables
are mutually unbiased if measurements of one observable on the eigenstates
of the other observable produce a set of equiprobable outcomes. The
bases of the underlying Hilbert space associated to two mutually unbiased
observables are said to be mutually unbiased basis (MUB).

For practical purposes, it is of fundamental importance to know what
is the maximum number of simultaneously MUBs for a given Hilbert space
dimension and also how to build a set of MUBs with the maximum number
of elements. By a set of simultaneously MUBs we mean that any pair
of bases taken from the set satisfies MU conditions. Besides the fundamental
mathematical interest in MU, many quantum information protocols rely
on the use of more than two simultaneously MUBs. It is known, for
example, that measuring a quantum system in the maximal set of MUBs
is the minimal and optimal set of measurements to completely determine
the quantum state of the system \cite{wootters1989}. For a discrete
variables system with dimension $d$, the maximum number of MUBs possible
is $d+1$. Although this upper bound is valid for any $d$, only if
$d=q^{m}$, with $q$ a prime number and $m$ a positive integer,
the maximal set with $d+1$ elements is known to exist \cite{wootters1989}.
If $d$ is not the power of a prime number, few things are known about
the existence or construction of a MUB set even for the smallest dimension
possible $d=6$, and numerical \cite{butterley2007,brierley2008,Durt10,Brierley2010,raynal2011}
as well as analytical \cite{brierley2009,paterek2009} evidences point
to the existence of only three MUBs in this case . 

In contrast to finite dimensional systems, continuous variables (CV)
systems also have MUBs, the position and momentum operators bases
being standard examples. However, instead of allowing for the construction
of infinitely many simultaneously MUBs as would be if the limit of
discrete case were valid, it allows for only three simultaneously
MUBs \cite{Weigert08}. On the other hand, because of the finite resolution
of CV detectors and the impossibility of producing eigenstates of
CV operators, the unbiasedness is not observed in practice. Recently,
it was proposed that a periodic coarse graining (PCG) of the CV measurements
can recover the unbiasedness relations if the period of the measurements
in different phase space directions is adequately chosen \cite{Tasca18a}.
In this scheme, the CV system is mapped to a effective DV system with
dimension $d$ equals to the number of possible outcomes of the measurement.
In this case, it is referred to as mutually unbiased measurements
(MUM) instead of bases. The natural questions are: \textit{how many
MUMs one can have in this PCG scheme? And knowing that a certain number
of MUMs is possible, how to construct them? }\emph{Is there a ``scaling
rule\textquotedbl{} as a function of $d$? Do these measurements resemble
more their continuous or discrete counterparts?} The previous works
showed the existence of pairs \cite{Tasca18a} and triples \cite{paul18}
of such PCG MUMs. Here we answer these questions, first providing
a general recipe to construct mutually unbiased measurements (MUMs)
together with their experimental realization in a optical setup. We
then show that PCG observables display a behavior that is reminiscent
of both continuous and discrete variables systems. For even dimensionality
$d$, it is shown that there are at most three MUMs, as in the continuous
case. On the other hand, for odd $d$ there is some agreement with
the discrete case. We show that for prime $d$ there are at most $M=d+1$
MUMs, like the discrete case. However, when $d$ odd and not prime
then is no correspondence with neither the continuous nor the discrete
case. 

The chapter is organized as follows. In Sec. \ref{sec:Mutually-unbiased-basis}
the concept of MUBs in DV and CV variables is presented. In Sec. \ref{sec:Mutually-unbiased-periodic}
the construction of PCG MUM of Ref. \cite{Tasca18a} is revised for
only one pair of measurements. Our contribution is contained in Secs.
\ref{sec:Construction-of-several} and \ref{sec:Maximum-number-of}
with a method to construct several PCG MUMs for any dimension parameter,
including its experimental realization, and a proof of the maximum
number of PCG MUMs depending on the number of outcomes, respectively.
Sec. \ref{sec:Concluding-remarks} concludes this chapter.

\section{Mutually unbiased basis and measurements\label{sec:Mutually-unbiased-basis}}

Consider first a system with a finite dimensional Hilbert space which
dimension is $d$. Consider also two orthonormal bases $\{\ket{a_{j}}\}_{j=0,...,d-1}$
and $\{\ket{b_{j}}\}_{j=0,...,d-1}$ that may be regarded as the set
of eigenstates of the two observables $A=\sum_{j}a_{j}\ket{a_{j}}\bra{a_{j}}$
and $B=\sum_{j}b_{j}\ket{b_{j}}\bra{b_{j}}$, respectively. The bases
are said to be mutually unbiased if the absolute value of the inner
product of any pair of states, one from each basis, is a fixed number,
that is
\begin{equation}
\left|\bra{a_{j}}b_{k}\rangle\right|=d^{-\nicefrac{1}{2}}\quad\forall\;j,k=0,\cdots,d-1,
\end{equation}
where the value $d^{-\nicefrac{1}{2}}$ is due to the normalization
of the states \cite{Durt10}. In other words, if the system is initially
prepared in an eigenstate of $A$, any outcome of a subsequent measurement
of $B$ is equally probable, and vice-versa. $A$ and $B$ are said
to be extreme complementary observables: if one of them is known,
the other is completely unknown. The statement in terms of measurements
is useful to extend the concept of unbiasedness to more general measurement
processes. Consider two positive-operator valued measurements $\{\hat{\Omega}_{j}^{(a)}\}$
and $\{\hat{\Omega}_{j}^{(b)}\}$, they are said to be mutually unbiased
measurements if, for any state $\rho$ satisfying $p_{j^{\prime}}^{(a)}=\Tr\left[\hat{\Omega}_{j^{\prime}}^{(a)}\rho\left(\hat{\Omega}_{j^{\prime}}^{(a)}\right)^{\dagger}\right]=\delta_{jj^{\prime}}$
for some $j$, we have 
\begin{equation}
p_{k}^{(b)}=\Tr\left[\hat{\Omega}_{k}^{(b)}\rho\left(\hat{\Omega}_{k}^{(b)}\right)^{\dagger}\right]=\frac{1}{d}\quad\forall\:k=0,...,d-1,
\end{equation}
where $d$ in this case is the number of outcomes, $p_{j^{\prime}}^{(a)}$
($p_{k}^{(b)}$) is the probability of outcome $j^{\prime}$ ($k$)
when performing measurement $a$ ($b$). In other words, if the measurement
$a$ is deterministic in the sense that only outcome $j$ is detected
with unity probability, then the outcomes of measurement $b$ on the
same state are equiprobable. This statement is also valid if $a$
and $b$ are interchanged.

Maybe the most celebrated pair of complementary observables is position
and momentum, two continuous variable operators. Indeed, most quantum
mechanics textbooks highlight the fact that MU between position ($\hat{x}$)
and momentum ($\hat{p}$) operators can be demonstrated by $|\langle x|p\rangle|=1/\sqrt{2\pi}$
(we set $\hbar$ = 1 throughout). What is somewhat less well-known
is the fact that \emph{any} two non-parallel phase space operators
$\hat{q}_{\theta}$ and $\hat{q}_{\theta^{\prime}}$ (Eq. \eqref{eq:qtheta})
are mutually unbiased, which can be demonstrated via 
\begin{equation}
|\langle q_{\theta^{\prime}}|q_{\theta}\rangle|=\left({2\pi|\sin\Delta\theta}|\right)^{-1/2},\label{eq:CVs}
\end{equation}
where we assume that $q_{\theta}$ and $q_{\theta^{\prime}}$ are
characterized by angles $\theta$ and $\theta^{\prime}$ in phase
space, and $\Delta\theta\equiv\theta^{\prime}-\theta$ is the angle
between them \footnote{We note that in the limit $\theta_{j}k\rightarrow0$, the limit must
be taken before the absolute value to recover the normalization to
the usual Dirac delta function}, as illustrated in Fig. \ref{fig:Phase-space-directions}.

\begin{figure}[h]
\begin{centering}
\includegraphics[width=0.3\columnwidth]{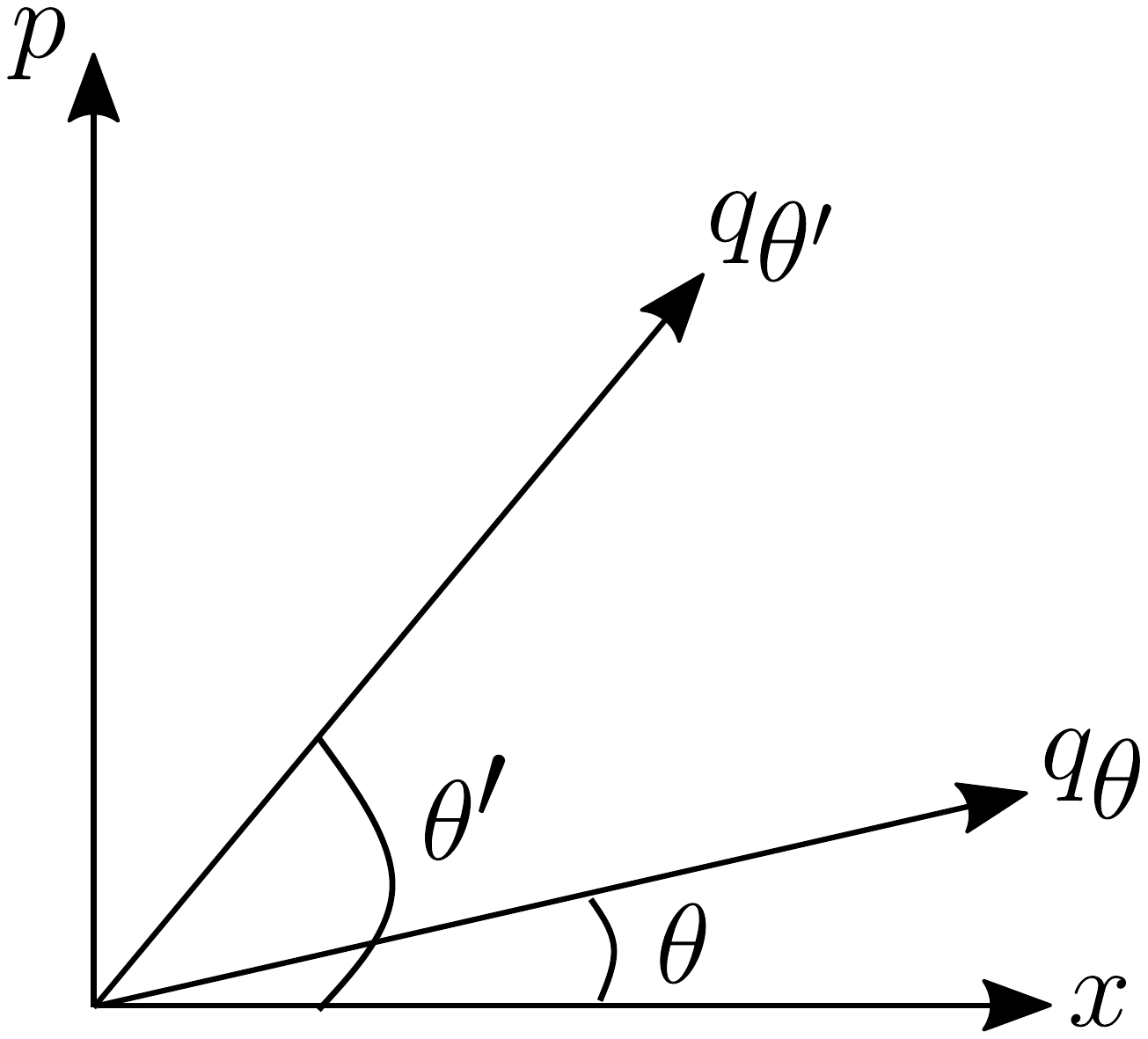}
\par\end{centering}
\caption{Phase space directions and angles.\label{fig:Phase-space-directions}}

\selectlanguage{american}%
\end{figure}

Mutually unbiasedness (MU) brings up one major difference between
continuous variable and discrete variable quantum systems. Discrete
and finite $d$-dimensional quantum systems admit at most $d+1$ mutually
unbiased bases \cite{wootters1989}. This means that the maximal set
for which any pair of bases are mutually unbiased has at most $d+1$
elements. When $d$ is the power of a prime number the existence of
such a maximal set is guaranteed \cite{Ivonovic1981,wootters1989,Bandyopadhyay2002,klappenecker2003}.
On the other hand, it was shown in Ref. \cite{Weigert08} that there
are at most three MU bases for a CV system (also known as a ``qumode\textquotedbl ).
MU for three phase space operators can be achieved by defining the
relative angle $\Delta\theta=2\pi/3$, such that all pairs of operators
satisfy Eq. \eqref{eq:CVs} with the same right-hand side (RHS) \cite{Weigert08,Paul16}.

\begin{figure}[h]
\begin{centering}
\includegraphics[width=0.95\columnwidth]{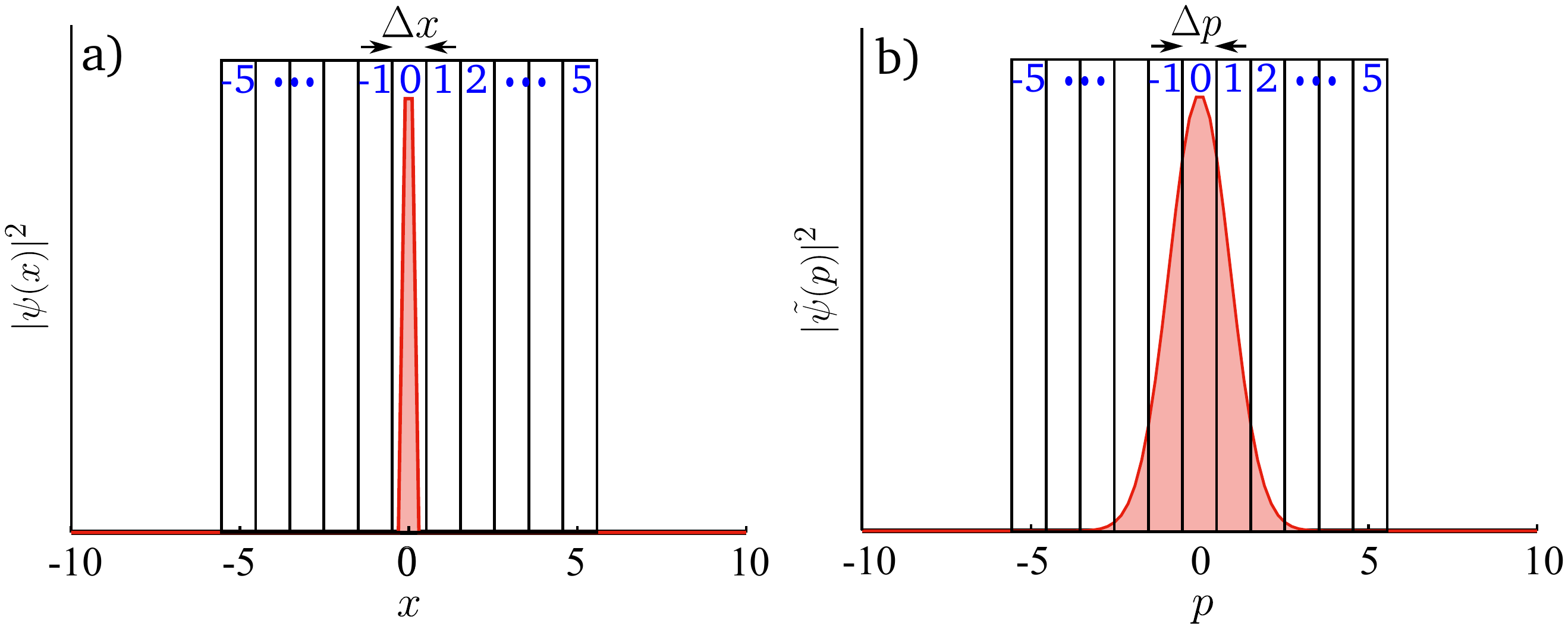}
\par\end{centering}
\caption{Coarse grained measurements in a) position and b) momentum over a
state $\ket{\psi}$. \label{fig:Course-grained-measurements}}

\selectlanguage{american}%
\end{figure}

Observables of continuous quantum variables can be discretized by
a ``binning\textquotedbl{} procedure, dividing the Hilbert space
into a finite number of discrete parts, as exemplified in Fig. \ref{fig:Course-grained-measurements}.
One common example is the parity operator, with eigenstates that are
symmetric or anti-symmetric with respect to the origin. There are
a number of reasons to pursue this type of discretization. For example,
it is well known that it is difficult to employ phase space operators
to demonstrate quantum non-locality \cite{revzen05}. This has led
to a number of binning schemes of measurements or states \cite{gilchrist98,banaszek98,banaszek99,wenger03,cavalcanti11}.
More fundamentally, the continuous variable eigenstates in \eqref{eq:CVs}
are not physical \cite{Braunstein05}. In real-world experiments,
they are approximated by states that are localized around some mean
value, which renders these physical eigenstates no longer mutually
unbiased. In addition, measurements in any quantum system suffer from
some amount of coarse graining, which follows from the fact that any
measurement device has some finite resolution. In this way, the measurements
are not projections over an eigenstate of the CV observable, but a
projection on the region of each detector
\begin{equation}
\hat{\Omega}_{k}^{(q_{\theta})}=\int_{\left(k-\frac{1}{2}\right)\Delta q_{\theta}}^{\left(k+\frac{1}{2}\right)\Delta q_{\theta}}\,dq_{\theta}\:\ket{q_{\theta}}\bra{q_{\theta}},\label{eq:binning}
\end{equation}
where $k$ is an integer and $\Delta q_{\theta}$ is the detector
aperture. For example, consider a very localized Gaussian state in
position representation, such that the probability of detecting the
particle out of detector $0$ is negligible, as the one represented
in Fig. \ref{fig:Course-grained-measurements}-a). The more localized
the state is, it still has a finite width in position, and consequently
a finite width in momentum as well {[}Fig. \ref{fig:Course-grained-measurements}-b){]}.
Therefore, although we may have unity probability of detecting the
system in the detector $0$ when measuring position, we do not have
equal probability of detecting the system in any detector when measuring
momentum, and the MU for this real measurement process is lost. In
other words, the measurements of Eq. \eqref{eq:binning} for position
and momentum are not MUM. The coarse-grained observables describing
these measurements though satisfy uncertainty relations \cite{toscano18},
but they are not complementary in the sense o MU. This inherent coarse
graining of real-world states, as well as real-world measurements,
motivates the search for coarse-grained mutually unbiased observables.

\section{Mutually unbiased periodic coarse grained measurements\label{sec:Mutually-unbiased-periodic}}

It has been shown recently that one path to MU is through the definition
of periodic coarse grained (PCG) observables, where physical mutually
unbiased measurement pairs \cite{Tasca18a} and mutually unbiased
measurement triples \cite{paul18} were demonstrated theoretically
and experimentally. An schematic representation of the PCG measurement
is shown in Fig. \ref{fig:The-periodic-coarse}. The continuum of
values for the phase space observable $q_{\theta}$ is binned uniformily,
according to the size of the detectors. Instead of associating an
outcome to each detector, the bins are labeled periodically with natural
numbers from $0$ to $d-1$. Anytime the quantum system is detected
in a box labeled by $j$, the outcome $j$ is attributed to that measurement.
In this formulation there naturally appears a ``dimensionality\textquotedbl{}
parameter $d$, given by the number of possible measurement outcomes.
The period of the coarse grained measurement, $T_{\theta}$, is equal
to the size of the detectors times the dimension.

\begin{figure}[h]
\begin{centering}
\includegraphics[width=0.6\columnwidth]{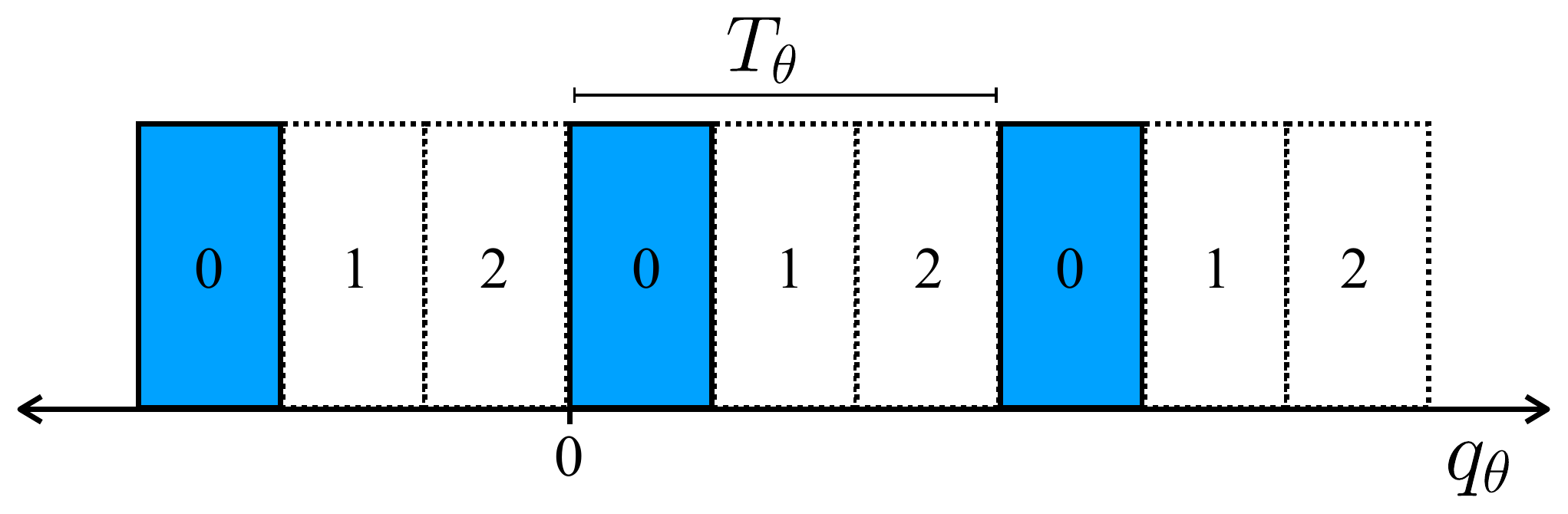}
\par\end{centering}
\caption{The periodic coarse graining scheme for $d=3$ and period $T_{\theta}$.
Anytime a particle is detected in the blue regions, the measurement
is attributed the outcome value $0$. The solid line represents the
mask function $M_{0}(q_{\theta},T_{\theta})$.\label{fig:The-periodic-coarse}}

\selectlanguage{american}%
\end{figure}

This PCG measurement can be described by the measurement operators
\begin{equation}
\hat{\Omega}_{k}^{(\theta)}=\int_{-\infty}^{\infty}\,dq_{\theta}\,M_{k}(q_{\theta}-q_{\theta}^{\text{cen}},T_{\theta})\,\ket{q_{\theta}}\bra{q_{\theta}},\label{eq:measurop}
\end{equation}
where the mask function describes the detectors and is regarded as
\textcolor{black}{the periodic square waves }
\begin{equation}
M_{k}(q_{\theta},T_{\theta})=\begin{cases}
1 & k\frac{T_{\theta}}{d}\leq q_{\theta}(\text{mod}\,d)<(k+1)\frac{T_{\theta}}{d}\\
0 & \textrm{otherwise},
\end{cases}
\end{equation}
such that $\hat{\Omega}_{k}^{(\theta)}$ projects a state on the subspace
of $\hat{q}_{\theta}$-eigenstates corresponding to all bins periodically
labeled by $k$.\textcolor{black}{{} The displacement parameter $q_{\theta}^{\textrm{cen}}$
is included to allow for freedom to define the origin} and will be
taken to be zero for simplicity. The Fourier series of the mask function
is written as 
\begin{equation}
M_{k}(q_{\theta},T_{\theta})=\frac{1}{d}+\sum_{n\in\mathbb{Z}/\{0\}}\frac{1-e^{-i\frac{2\pi n}{d}}}{i2\pi n}e^{-i\frac{2\pi nk}{d}}e^{-i\frac{2\pi n}{T_{\theta}}q_{\theta}}.\label{Fouriermask}
\end{equation}

Let us now consider two PCG measurements $\left\{ \hat{\Omega}_{k}^{(\theta)}\right\} {}_{k}$
and $\left\{ \hat{\Omega}_{k}^{(\theta^{\prime})}\right\} _{k}$as
the one described. They are defined in two nonparallel directions
characterized by angles $\theta$ and $\theta^{\prime}$ in phase
space, with periods $T_{\theta}$ and $T_{\theta^{\prime}}$, respectively.
Without any loss of generality we can consider a pure state satisfying
\begin{equation}
p_{k}^{(\theta)}=\bra{\psi}\hat{\Omega}_{k}^{(\theta)}\ket{\psi}=\delta_{kl}\quad\forall\,k=0,...,d-1.\label{eq:probtheta}
\end{equation}
The $q_{\theta}$ and $q_{\theta^{\prime}}$ representations of this
state are denoted as $\braket{q_{\theta}|\psi}=\psi(q_{\theta})$
and $\braket{q_{\theta^{\prime}}|\psi}=\bar{\psi}(q_{\theta^{\prime}})$.
In order to determine the condition for MU of the two PCG measurements,
the probability of detection in direction $\theta^{\prime}$ must
be calculated. By using Eq. \eqref{eq:measurop} with the Fourier
series \eqref{Fouriermask} and the FrFT connection between the two
phase space representations {[}Eq. \eqref{eq:FrFTrep}{]}, one obtains
\begin{equation}
p_{j}^{(\theta^{\prime})}=\bra{\psi}\hat{\Omega}_{j}^{(\theta^{\prime})}\ket{\psi}=\frac{1}{d}+\sum_{n\in\mathbb{Z}/\{0\}}\frac{1-e^{-i\frac{2\pi n}{d}}}{i2\pi n}\int dq_{\theta}\,e^{i\phi_{j}^{(n)}(q_{\theta})}\,\psi^{*}\left(q_{\theta}\right)\,\psi\left(q_{\theta}-n\frac{2\pi\sin\Delta\theta}{T_{\theta^{\prime}}}\right),\label{eq:probthetalin}
\end{equation}
where $\Delta\theta=\theta-\theta^{\prime}$ and $\phi_{j}^{(n)}(q_{\theta})=n\frac{2\pi\cos\Delta\theta}{T_{\theta^{\prime}}}q_{\theta}-\left(n\frac{2\pi}{T_{\theta^{\prime}}}\right)^{2}\frac{\sin(2\Delta\theta)}{4}-\frac{2\pi j}{d}$.
For the two measurements to be unbiased, $p_{j}^{(\theta^{\prime})}$
must be equal to $\nicefrac{1}{d}$ because we started from a localized
state relative to the measurement in $q_{\theta}$. Thus, the sum
in the right hand side shall be forced to vanish. The terms in the
sum with $n$ multiple of $d$ are already null because in this cases
$e^{-i\frac{2\pi n}{d}}=1$. Because of Eq. \eqref{eq:probtheta},
$\psi(q_{\theta})$ has nonnull values only inside the mask $M_{l}(q_{\theta},T_{\theta})$.
Hence, if $q_{\theta}$ is any value for which $\psi(q_{\theta})\neq0$,
then $\psi(q_{\theta}+\Delta q_{\theta})=0$, and so is $\psi^{*}\left(q_{\theta}\right)\,\psi\left(q_{\theta}+\Delta q_{\theta}\right)$,
provided that $\frac{T_{\theta}}{d}\leq\Delta q_{\theta}(\text{mod }T_{\theta})\leq(d-1)\frac{T_{\theta}}{d}$.
Particularly, if the absolute value of the increment $\frac{2\pi\sin\Delta\theta}{T_{\theta^{\prime}}}$
for $n=1$ in Eq. \eqref{eq:probthetalin} is equal to $m\frac{T_{\theta}}{d}$,
with $m$ an integer number non-multiple of $d$, the corresponding
term in the summation vanishes for any $\ket{\psi}$ satisfying \eqref{eq:probtheta}.
For the other values of $n$ non-multiple of $d$, we then have that
$n\frac{2\pi\sin\Delta\theta}{T_{\theta^{\prime}}}=nm\frac{T_{\theta}}{d}$
must be non-multiple of the period in order to all the integrals to
vanish. This arguments yield the MUM condition for the PCG measurements
\begin{equation}
T_{\theta}T_{\theta^{\prime}}=\frac{2\pi d\left|\sin\Delta\theta\right|}{m},\qquad\frac{nm}{d}\notin\mathbb{N}\quad\forall\,n=1,...,d-1,\label{eq:MUMcond}
\end{equation}
which is a relation between periods dependent on the dimensionality
parameter chosen as well as on the angle between the two phase space
directions. This relation was first shown in Ref. \cite{Tasca18a}
for position and momentum only $(\Delta\theta=\frac{\pi}{2})$ and
then extended to any two phase space observables in Ref. \cite{paul18}.
In the last reference, the authors also show that is it possible to
build a triple (and no more than three) of PCG measurements pairwise
MU when all the $m$ numbers present in the period relations are equal
to one. The cases above \cite{Tasca18a,paul18} are interesting, but
even in the ideal continuous case, we can have MU measurements for
two or three phase space directions. A natural question is, using
this periodic coarse graining, can we go beyond the CV case, obtaining
more than three MU observables? To answer the question about the possibility
of having a set with more MUMs and how many MUMs compose the maximal
set for a given dimension, one might try to find a set of periods
and angles for which the relation \eqref{eq:MUMcond} is satisfied
for all pairs of measurements for some $m$. In the next section,
it is presented a way to construct such set of measurements in which
the angles are fixed and the equations are solved for the periods.
The dimensions possible for a given number of measurement directions
come as part of the solution.

\section{Construction of several MUMs \label{sec:Construction-of-several}}

Following Refs. \cite{Tasca18a,paul18}, let us consider $R$ phase
space operators $\hat{q}_{j}$ as in Eq. \eqref{eq:CVs}, related
to each other via phase space rotations, and each characterized by
an angle $\theta_{j}$, for $j=0,...,R-1$, illustrated in Fig. \ref{fig:vecs}
a). We can then define define $R$ coarse-grained projective measurement
operators:

\begin{equation}
\hat{\Omega}_{j}^{(m)}=\int dq_{j}\,M_{r}^{(j)}\left(q_{j}-q_{j}^{\textrm{cen}};T_{j}\right)\left|q_{j}\right\rangle \left\langle q_{j}\right|,\label{Projectors}
\end{equation}
with detector apertures encoded in ``mask functions'' $M_{r}^{(j)}$,
such that \textcolor{black}{$\sum_{m=0}^{d-1}M_{r}^{(j)}=1$. The
parameter $T_{j}$ is the period of the mask function.} The outcome
probabilities produced by the set of projectors \eqref{Projectors}
then define the PCG of the probability distribution associated with
the phase-space variable $q_{j}$. Since we work with dimensionless
variables, the mask parameter $T_{j}$ is also dimensionless. 

\begin{figure}[h]
\centering{}\includegraphics[width=0.5\columnwidth]{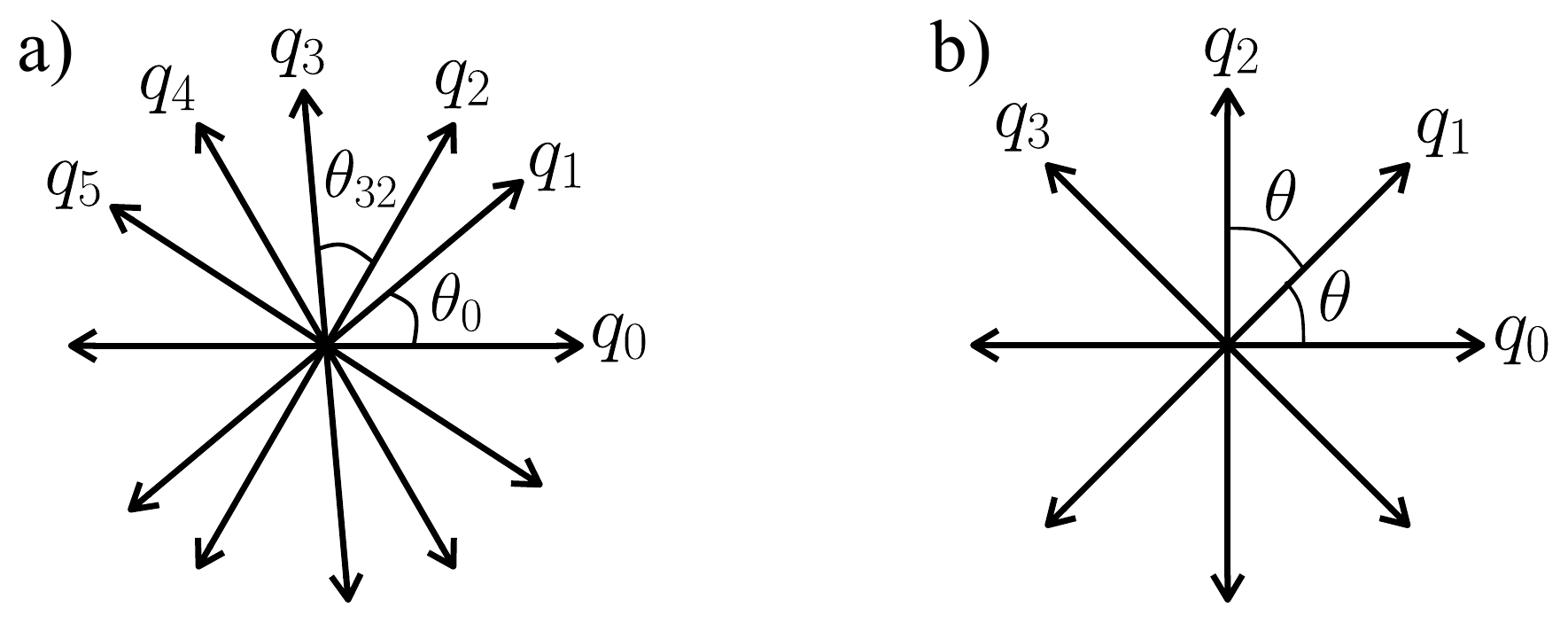}
\caption{Examples of phase space variables for a) R = 6 and b) the R = 4 case
implemented experimentally. }
\label{fig:vecs} 
\end{figure}


We assume that $\theta_{0}=0$, without loss of generality, and $\theta_{j}>\theta_{k}$
if $j>k$. The condition \eqref{eq:MUMcond} for mutual unbiasedness
of PCG operators of the sort \eqref{Projectors} is written as 
\begin{equation}
T_{j}T_{k}m_{jk}=2\pi d|\sin\theta_{jk}|,\label{eq:MUMcond1}
\end{equation}
where $\theta_{jk}\equiv\theta_{j}-\theta_{k}$ and $m_{jk}$ is a
positive integer. The MU condition \eqref{eq:MUMcond1} also requires
that 
\begin{equation}
\frac{m_{jk}n}{d}\notin\mathbb{N}\,\,\forall\,n=1,\dots d-1.\label{eq:MUMcond1-b}
\end{equation}

To construct a general recipe to obtain an arbitrary number of PCG
MUMs, let us consider that $\sin\theta_{j}\geq0$, which defines variables
$q_{j}$ in the upper semi-plane of phase space. This is not a restriction,
as variables in the lower half-plane can be taken to the upper half-plane
by a reflection through the origin: $q_{j}\rightarrow-q_{j}$. Using
the fact that $\theta_{0}=0$, we then have $R-1$ conditions of the
form: 
\begin{equation}
T_{j}=\frac{2\pi d\sin\theta_{j}}{m_{j0}T_{0}}.\label{eq:periods}
\end{equation}
Plugging the above equations for the periods into condition \eqref{eq:MUMcond1}
with $j,k\neq0$, we have 
\begin{equation}
T_{j}T_{k}m_{jk}=2\pi d|\sin\theta_{jk}|=\frac{(2\pi d)^{2}\sin\theta_{j}\sin\theta_{k}m_{jk}}{T_{0}^{2}m_{j0}m_{k0}}.\label{eq:MUMcond2}
\end{equation}
Considering $j>k$, this can be rewritten as
\begin{equation}
\frac{m_{jk}}{m_{j0}m_{k0}}=\frac{T_{0}^{2}}{2\pi d}(\cot\theta_{k}-\cot\theta_{j}).\label{eq:MUMcond4}
\end{equation}
Since the left-hand side is composed of all natural numbers, it is
a rational number, which results in a general restraint for the right-hand
side (RHS). Condition \eqref{eq:MUMcond4} is enough to prove several
important results concerning periodic discretization.

There is quite a bit of freedom in constraint \eqref{eq:MUMcond4}
concerning the period $T_{0}$, as well as the angles $\theta_{j},\theta_{k}$
(for $j,k\neq0$). With some specification, we can construct a useful
recipe for finding a general mutually unbiased set. As a step in this
direction, let us choose angles that are distributed at integer multiples
of an angle $\theta$, such that $\theta\leq\pi/2$ and $\theta_{j}=j\theta$.
Moreover, we will choose $T_{0}^{2}=\pi d\tan\theta$, so that we
can write 
\begin{equation}
\frac{m_{jk}}{m_{j0}m_{k0}}=\frac{\tan\theta}{2}(\cot k\theta-\cot j\theta).\label{eq:MUMcond5}
\end{equation}
Choosing then 
\begin{equation}
\tan\theta=\sqrt{N/N^{\prime}},\label{eq:MUMcond6}
\end{equation}
where $N,N^{\prime}\in\mathbb{N}$, we can prove that the RHS of Eq.
\eqref{eq:MUMcond5} is always a rational number for all $j,k=1,\dots,M-1$.
The proof is provided in the next subsection. This allows one to find
suitable numbers $m_{j0}$, $m_{k0}$ and $m_{jk}$ that satisfy Eq.
\eqref{eq:MUMcond5}. Eqs. \eqref{eq:MUMcond1-b}, \eqref{eq:MUMcond5}
and \eqref{eq:MUMcond6} then define the conditions for a set of mutually
unbiased PCG observables. 

\subsection{Solution to equation (\ref{eq:MUMcond4})}

The RHS of Eq. \eqref{eq:MUMcond4} must be a rational number for
all allowed values of $j$ and $k$. Thus, all elements in the sequence
$\{a_{k}=\frac{T_{0}^{2}}{2\pi d}\cot\theta_{k}\}_{k}$ should be
rational. There are two possibilities to assure this: requiring the
product of $\frac{T_{0}^{2}}{2\pi d}$ and $\cot\theta_{k}$ to be
rational, what implies a relation between the period and the angle,
or requiring each term in the product to be independently rational.
In both cases, mathematical induction can be used to figure out the
solutions. According to this method, we need to ensure the validity
of the statement for the first element in the sequence ($k=1$ in
our case). The statement is valid for all $k$ if it can be demonstrated
that the assumption of validity for any value $k-1$ implies its validity
for $k$.

In what follows, we use the recurrence relation 
\begin{equation}
\cot(k\theta)=\frac{\cot\theta\cot(k-1)\theta-1}{\cot\theta+\cot(k-1)\theta},\label{recurrence-1}
\end{equation}
and all $n_{i}$ are natural numbers.

\subsubsection*{First solution}

We want to show what are the conditions that make all elements of
the sequence $a_{k}$ to be rational numbers. A relation between the
period $T_{0}$ and the angle $\theta$ comes from the assertion that
the first element of this sequence is a rational: 
\begin{equation}
a_{1}=\frac{T_{0}^{2}}{2\pi d}\frac{1}{\tan\theta}=\frac{n_{1}}{n_{2}}\quad\Rightarrow\quad T_{0}=\sqrt{2\pi d\tan\theta\frac{n_{1}}{n_{2}}}.\label{k1-1}
\end{equation}

Now suppose that for an arbitrary $k$ the $(k-1)$-th element is
rational, that is 
\begin{equation}
a_{k-1}=\frac{T_{0}^{2}}{2\pi d}\cot(k-1)\theta=\frac{n_{3}}{n_{4}},\label{kminus1-1}
\end{equation}
this should imply the $k$-th element to also be a rational number.
Using the relations \eqref{recurrence-1}, \eqref{k1-1}, and \eqref{kminus1-1}
we have 
\begin{equation}
a_{k}=\frac{n_{1}}{n_{2}}\frac{n_{2}n_{3}-n_{1}n_{4}\tan^{2}\theta}{n_{2}n_{3}+n_{1}n_{4}},
\end{equation}
which is not rational unless 
\begin{equation}
\tan\theta=\sqrt{\frac{n_{5}}{n_{6}}}.\label{eq:tan-1}
\end{equation}
Therefore, if conditions \eqref{k1-1} and \eqref{eq:tan-1} are satisfied,
the inductive proof is concluded. The solution presented before for
the period and angle are particular cases with $n_{1}=1$, $n_{2}=2$.

\subsubsection*{Second solution}

We define a new sequence $\{b_{k}=\cot(k\theta)\}_{k}$ and $a_{k}=\frac{T_{0}^{2}}{2\pi d}b_{k}$.
One possibility for the elements of $\{a_{k}\}$ to be rational is
that the elements of $\{b_{k}\}$ are rational and 
\begin{equation}
\frac{T_{0}^{2}}{2\pi d}=\frac{n_{1}}{n_{2}}\quad\Rightarrow\quad T_{0}=\sqrt{2\pi d\frac{n_{1}}{n_{2}}}.\label{period-1}
\end{equation}

A condition for the angle comes from the requirement that $b_{1}$
is a rational number: 
\begin{equation}
b_{1}=\frac{n_{3}}{n_{4}}\quad\Rightarrow\quad\tan\theta=\frac{n_{4}}{n_{3}}.
\end{equation}

It comes directly from relation \eqref{recurrence-1} that, if $b_{k-1}$
is rational, then it follows that $b_{k}$ is also a rational number.
This solution allows for choosing $T_{0}$ and $\theta$ independently. 

\subsection{Even dimensionality parameter}

There are some interesting conditions that can be derived about the
dimension parameter $d$. As a first result, we show that, \emph{for
$d$ even, there are at most $R=3$ mutually unbiased PCG operators}.
To prove this, let us analyze once more the conditions \eqref{eq:MUMcond5}
for $k=1$ and $j=2$, which gives 
\begin{equation}
\frac{m_{20}m_{10}}{m_{21}}=4\cos^{2}{\theta}.\label{eq:deven1}
\end{equation}
For $R=3$ this is the only condition that must be satisfied, and
it alone is not prohibitive for any $d\geq2$. For any number of bases
$R>3$, an additional condition for $k=1$ and $j=3$ is also present.
The two conditions are related, since, using \eqref{eq:MUMcond5}
and \eqref{eq:deven1} we can write 
\begin{equation}
\frac{m_{30}m_{10}}{m_{31}}=4\cos^{2}{\theta}-1=\frac{m_{20}m_{10}-m_{21}}{m_{21}}.\label{eq:m30}
\end{equation}
Consider $d$ even. Then, using $n=d/2$ in \eqref{eq:MUMcond1-b}
determines that all $m_{jk}$ must be odd. From \eqref{eq:m30} we
can write $m_{30}=\left(m_{31}m_{20}m_{10}-m_{31}m_{21}\right)/(m_{21}m_{10})$,
and assuming that all $m_{jk}$ appearing in the RHS of this expression
are odd, then the numerator turns out to be even and thus $m_{30}$
must be even, which violates condition \eqref{eq:MUMcond1-b}. Thus,
there is no valid solution for even dimension if the number of phase
space operators $R>3$. 

\subsection{Examples for $d$ odd}

\begin{table}[h]
\centering{}%
\begin{tabular}{|c|c|c|>{\raggedright}m{6.8cm}|c|}
\hline 
{\small{}$\theta$} & {\small{}$\tan\theta$ } & {\small{}$R$ } & \multirow{1}{6.8cm}{\centering{}{\small{}$m$ values }} & {\small{}allowable $d$ }\tabularnewline
\hline 
{\small{}$\frac{\pi}{4}$ } & {\small{}1 } & {\small{}4 } & \centering{}{\small{}$m_{20}=2$, all other $m_{jk}=1$ } & {\small{}$d\geq3$, odd}\tabularnewline
\hline 
\noalign{\vskip\doublerulesep}
\multirow{2}{*}{{\small{}$\frac{\pi}{6}$ }} & \multirow{2}{*}{{\small{}$\frac{1}{\sqrt{3}}$ }} & \multirow{2}{*}{{\small{}6 }} & \multirow{2}{6.8cm}[0.1cm]{\centering{}{\small{}$m_{20}=m_{40}=m_{42}=3$, $m_{30}=m_{41}=m_{52}=2$,
all other $m_{jk}=1$ }} & \multirow{2}{*}{{\small{}$d=5,7,11,13,17,19,23,25,29$ }}\tabularnewline[\doublerulesep]
 &  &  &  & \tabularnewline
\hline 
{\small{}$0.35$ rad} & {\small{}$\sqrt{\frac{2}{15}}$} & {\small{}8} & \centering{}{\small{}many values} & {\small{}$d=7,11,19,23,29$}\tabularnewline
\hline 
{\small{}$\approx0.3$ rad } & {\small{}$\frac{1}{\sqrt{9}}$} & {\small{}9 } & \centering{}{\small{}many values } & {\small{}$d=11,17,19,23,29$ }\tabularnewline
\hline 
\end{tabular}\caption{\label{tab:1} Some results for $R=4,6,8,9$.}
\end{table}

Let us now consider some particular cases. One can see that previous
results \cite{paul18} for $R=3$ are retrieved when $\tan\theta=\sqrt{3}$
(here in present notation we have $q_{0}=x$, $q_{1}=-s$ and $q_{2}=r$,
in terms of previous variables \cite{paul18}) which from Eq. \eqref{eq:MUMcond1}
gives all $m_{jk}=1$. 

In table \ref{tab:1} we show results for $R=4,6,8$ and $9$. The
allowable dimensions $d$ were tested up to $d=30$. $R=3,4,6$ are
the only values that allow division of the first half-plane into equal
``slices\textquotedbl , while still maintaining $\tan^{2}\theta$
rational. For other values, this is not possible. For example, to
divide the semi-plane into 7 equal slices, we need $\theta=\pi/7$,
but this does not result in $\tan\theta=\sqrt{N/N^{\prime}}$. Thus,
we must choose an $N$ and $N^{\prime}$ which results in a $\theta<\pi/7$.
To take advantage of the entire phase space, in some sense, it seems
logical to choose the angles so that $\theta$ is as large as possible,
though this is not necessary. The results were obtained choosing the
value of $\theta$ and fixing $m_{10}=1$ from what is possible to
obtain the values of the other $m_{jk}$ as numerators and denominators
of the RHS of Eq. \ref{eq:MUMcond5}. With the values of all $m_{jk}$
we could check which are the allowed values of $d$ satisfying \ref{eq:MUMcond1}. 

Looking at our results for $R=4,6$ and $8$ one might be tempted
to assume that there are at most $R=d+1$ MUMs for odd $d$, as is
known to be true for some cases in discrete quantum mechanics. In
fact, this is the case if $d$ is a prime number, as we shall show
in Sec. \ref{sec:Maximum-number-of}. 

\subsection{Experimental realization}

To confirm and explore our results we performed a classical optics
experiment, similar to those of Refs. \cite{Tasca18a,paul18}. Optical
fractional Fourier transforms (FrFT) and amplitude masks were used
to prepare and measure the transverse spatial profile of a laser beam,
as shown schematically in Fig. \ref{fig:exp}. Both the optical FrFTs
as well as the amplitude masks were implemented using spatial light
modulators (SLMs), as shown in Fig. \ref{fig:exp} and described in
detail in Chapter \ref{chap:Paraxial}. The vertical transverse coordinate
is used as the system CV while the horizontal coordinate is used for
the amplitude masks diffraction. In each reflection by an SLM, only
the first order diffraction in the horizontal plane is taken for the
following operations. It is well known that when proper scaling is
chosen, the FrFT of order $\alpha$ is equivalent to a rotation in
phase space by the angle $\alpha$. Using the three-lens FrFT scheme
introduces a scaling factor such that the adimensional ($T_{j}$)
and physical ($T'_{j}$) periods are related by $T'_{j}=\sqrt{\frac{\lambda z}{\pi}}T_{j}$,
where $z=0.29\,$m is the distance between the lenses (quadratic phases
implemented by the SLMs) and $\lambda=632.9$nm is the laser wavelength
produced by the HeNe laser used. Moreover, for practical reasons the
physical periods are given in number of SLM pixels. The pixel size
of the Holoeye SLMs used here is $8\mu\text{m}$. The initial state
is fixed and is prepared as a colimated Gaussian beam with width $(2.54\pm0.06)\,$mm
at the plane of the first SLM. It is considered as the state in the
position-eigenstates representation. In the preparation stage, a FrFT
of order $j\theta$ was implemented on the transverse profile, followed
by the application of an amplitude mask $M_{r}^{(j)}$ of period $T_{j}^{\prime}$.
This maps the position representation in the first SLM plane to the
$q_{j\theta}$-representation in the third SLM plane, where the amplitude
mask is also applied to prepare a localized state in respect to PCG
$q_{j\theta}$ measurement. The measurement stage consisted of an
FrFT of order $(k-j)\theta$, mapping the $q_{j\theta}$-representation
in the third SLM plane to the $q_{k\theta}$-representation in the
third SLM plane, and an amplitude mask $M_{s}^{(k)}$. The full field
of the resulting output beam was then attenuated and detected with
a single photon detector. 

\begin{figure}[h]
\centering{}\includegraphics[width=0.7\columnwidth]{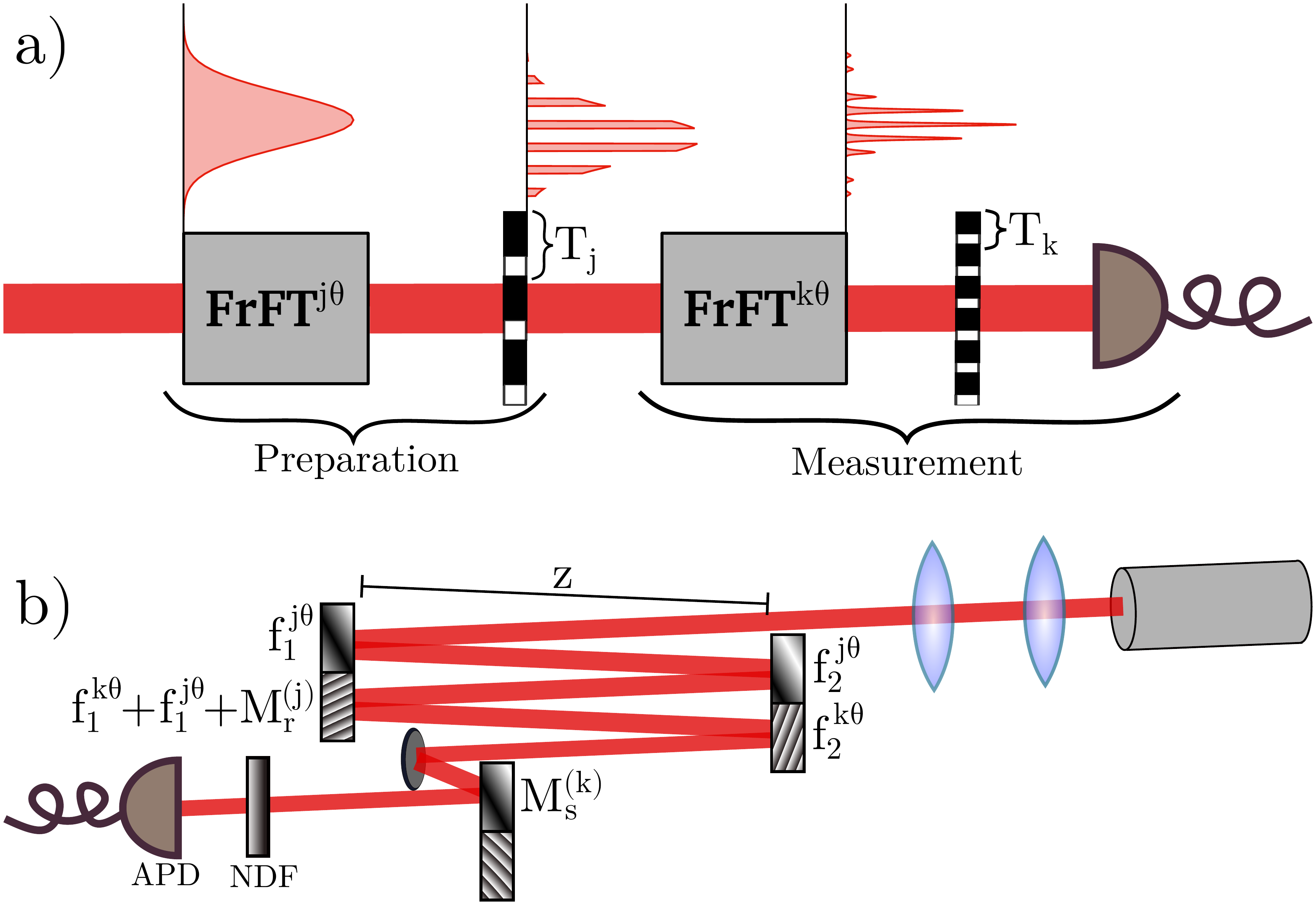} \caption{a) Schematic of the experiment. Fractional Fourier Transforms (FrFT)
and periodic amplitude masks (M) are used to prepare and measure the
transverse spatial profile of the laser beam. All output light is
incident on full-field single-photon detector. The scheme used to
implement FrFT is composed by three lenses with separation distance
$z$. If the focal distances satisfy $f_{1}^{\theta}=z\left(1-\frac{\cot(\theta/2)}{2}\right)^{-1}$
and $f_{2}^{\theta}=z(1-\sin\theta)^{-1}$ then the transverse profile
of the optical field in the plane $p_{2}$ is the FrFT of order $\theta$
of the field in the plane $p_{1}$ up to a scaling factor that is
independent of the FrFT order. b) Experimental setup using spatial
light modulators (SLM) for preparation and measurement. The output
of a $632.8$nm He-Ne laser is enlarged and collimated using two lenses.
The FrFTs are performed modulating the phase accordingly to the quadradic
phase of lenses with the right focal distances. The use of SLMs allows
us to make any order FrFT, which would be challenging with actual
lenses. The amplitude masks are also implemented using phase-only
modulators applying diffraction gratings and collecting only the first
diffraction order. Finally, the beam is attenuated with a neutral
density filter (NDF) and detected by a single photon detector. More
details provided in main text.}
\label{fig:exp} 
\end{figure}



\begin{center}
\begin{table} \centering \begin{tabular}{c|c|c|c|c|c|} \multicolumn{1}{c}{\multirow{6}{*}{\begin{turn}{90} \hspace{-0.5cm} Preparation \end{turn}}} & \multicolumn{5}{c}{Measurement}\tabularnewline \cline{2-6}   &  & 0 & 1 & 2 & 3\tabularnewline \cline{2-6}   & 0 & $0.161(3)$ & $1.5846(2)$ & $1.579(1)$ & $1.5847(2)$\tabularnewline \cline{2-6}   & 1 & $1.5841(8)$ & $0.143(3)$ & $1.583(3)$ & $1.584(4)$\tabularnewline \cline{2-6}   & 2 & $1.5846(1)$ & $1.5848(1)$ & $0.140(5)$ & $1.5838(1)$\tabularnewline \cline{2-6}   & 3 & $1.5844(2)$ & $1.5847(1)$ & $1.5844(5)$ & $0.162(3)$\tabularnewline \cline{2-6}  \end{tabular} \caption{Entropy value for different preparation and measurement direction in the case $R=4$ and $d=3$. } \label{tab:2} \end{table} 
\par\end{center}

We tested the case of $R=4$ MUMs with dimension parameter $d=3$
and $\theta=\pi/4$ for all nine combinations of preparation and measurement
phase space directions. We chose the period of mask 0 to be $T_{0}^{\prime}=93$
pixels, since this value is the closest integer number to the exact
value ($92.7476$ pixels) satisfying condition \eqref{eq:MUMcond6}
with $N=1$ and $N'=4$. Using \eqref{eq:periods}, and choosing $m_{10}=m_{30}=1$
we have $T_{1}^{\prime}=T_{3}^{\prime}=131.165$ pixels that was approximated
by $132$ pixels so the bin width is an integer. Using these values,
we tested MUM conditions between preparations $j=0,1,3$ and measurement
$k=2$, giving results shown in Fig. \ref{fig:results}. For each
measurement, the detection mask was scanned in all three positions
($d=3$), and the number of photon counts registered. We then calculated
the detection probabilities $p_{i}$, where $i=0,1,2$ refers to the
three amplitude masks, as well as the Shannon entropy $E=-\sum_{i}p_{i}\log_{2}p_{i}$
of the probability distributions, plotted in Fig. \ref{fig:results}
as a function of the period $T_{2}^{\prime}$ of the mask used in
the measurement stage. Vertical lines show values at which the period
$T_{2}^{\prime}$ corresponds to allowable $m_{2j}$ values. We can
see that at several places the entropy reaches its maximum value of
$\log_{2}3\approx1.5849$, which indicates the probability distribution
is uniform, corresponding to a MUM result. In order to have a set
of MUMs, the entropy must have its maximum value for all preparation
directions with the same value of $T_{2}^{\prime}$, which only happens
for the periods corresponding to some set of $m_{2j}$ that satisfies
\eqref{eq:MUMcond1} and \eqref{eq:MUMcond1-b} for all $j$ simultaneously\footnote{This is visually observed for periods $T_{2}^{\prime}>25$px approximately.
For periods smaller than that value, the entropy is too close to its
maximum value and it is not possible to tell by the plot which periods
give exactly the maximum entropy.}. The period parameter $T_{2}^{\prime}=93$ pixels satisfies the MUM
condition in all plots. Moreover, it corresponds to $m_{21}=m_{23}=1$
and $m_{20}=2$, as predicted by our theoretical results in Table
\ref{tab:1}. The entropy values obtained for these mask periods are
given in table \ref{tab:2}. We obtained results very close to the
maximum value of $\log_{2}3$ in all cases, indicating MUM results
for $R=4$ PCG measurements. To test the operation of our setup, we
also made measurements with equal preparation and measurement phase
space direction. The entropy ideally would vanish in this case, but
it has non-zero experimental values as can be seen in Table \ref{tab:2}.
This is due to the existence of some background noise that makes the
probability of preparing and measuring the system in the same mask
to be slightly smaller than 1 ($\approx0.98$), which is enhanced
by the structure of Shannon entropy. Similar results were obtained
for all combinations of preparation and measurement and are shown
in App. \ref{chap:Anexo0}. Furthermore, in all plots we can observe
that the entropy decreases greatly when $m_{2j}=3$, which is not
allowed by Eq. \eqref{eq:MUMcond1-b} when $d=3$. 

\begin{figure}[H]
\centering{}\includegraphics[width=8cm]{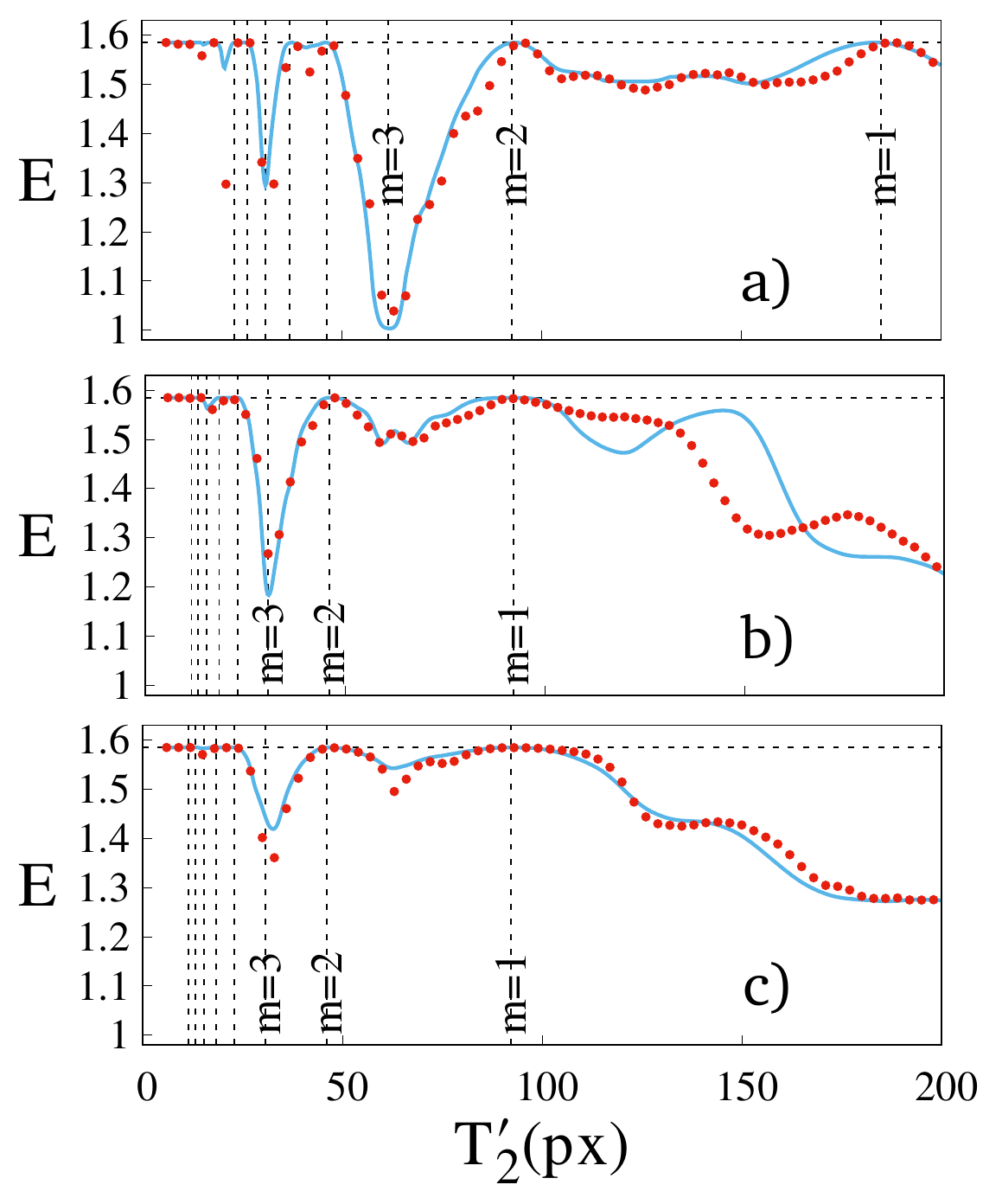} \caption{Example of results for $R=4$ MUMs, using angle $\theta=\pi/4$ and
preparation directions with a) $j=0$, b) $j=1$ and c) $j=3$. In
all graphs the measurement used is $k=2$. Entropy is plotted as a
function of the period $T_{2}^{\prime}$ of the measurement mask,
with preparation mask $M_{0}^{(j)}$ fixed with the period that satisfies
the MUM condition. The blue curves correspond to theoretical predictions
with the same initial state as the measured one.}
\label{fig:results} 
\end{figure}

\section{Maximum number of MUMs: general case\label{sec:Maximum-number-of}}

Our main results concern the maximum number of MUMS: First, let us
consider $d$ even. Using $n=d/2$ in \eqref{eq:MUMcond1-b} determines
that all $m_{jk}$ must be odd. Now, for $R\geq4$ measurements we
have conditions \eqref{eq:MUMcond4} for $jk=12,13,23$. Combining
these three constraints leads to $m_{23}m_{10}+m_{12}m_{30}=m_{13}m_{20}$,
which requires one $m$ to be even and violates Eq. \eqref{eq:MUMcond1-b}.
Thus, there are at most $R=3$ MUMs for even $d$ in the general case,
showing that it is not a particularity of our construction.

Now, for odd $d$ the PCG given in \eqref{Projectors} and satisfying
constraints \eqref{eq:MUMcond1} and \eqref{eq:MUMcond1-b}, we can
show that the maximum number of MUMs is given by $R_{max}=\varphi(d)+2$,
where $\varphi(d)=d\,\Pi_{p|d}\left(1-\frac{1}{p}\right)$ is Euler's
totient function counting the positive integers less than $d$ that
are relatively prime to $d$, and $p|d$ denotes all prime divisors
of $d$. Since $\varphi(d)=d-1$ when $d$ is prime we have $R_{max}=d+1$
for prime $d$. The proof is provided below. 

\subsubsection*{Proof for $d$ odd}

Consider again the MUM condition in the form of Eq. \eqref{eq:MUMcond4}.
This condition is particularly valid for $k=1$ and $j\geq2$ giving
\begin{equation}
\frac{m_{j1}}{m_{j0}m_{20}}=\frac{T_{0}^{2}}{2\pi d}(\cot\theta_{1}-\cot\theta_{j}).
\end{equation}
The previous equation can be substituted into the more general relation
for $j>k$ and $j,k\geq2$ yielding a relation involving only the
$m_{jk}$'s
\begin{equation}
m_{j1}m_{k0}-m_{j0}m_{k1}=m_{jk}m_{10}.\label{eq:ms}
\end{equation}

Let us introduce an expansion of all $m_{jk}$ as 
\begin{equation}
m_{jk}=l_{jk}d+n_{jk},
\end{equation}
where $l_{ij}$ the integer quotient of the division of $m_{jk}$
by $d$ and $n_{jk}=m_{jk}\:\text{mod}\,d$ is the integer remainder.
By plugging this parametrization into \eqref{eq:ms}, one can see
that
\begin{equation}
n_{10}m_{jk}=Z_{jk}d+n_{j1}n_{k0}-n_{j0}n_{k1},
\end{equation}
with 
\[
Z_{jk}=-l_{10}m_{jk}+n_{j1}l_{k0}+n_{k0}l_{j1}+l_{j1}l_{k0}d-\left(n_{k1}l_{j0}+n_{j0}l_{k1}+l_{k1}l_{j0}d\right)
\]
an integer. Notice that $n_{10}=1,2,\,...\,,d-1$, thus one instance
of condition \eqref{eq:MUMcond1-b} reads
\begin{equation}
\frac{n_{10}m_{jk}}{d}\notin\mathbb{N}\quad\Leftrightarrow Z_{jk}+\frac{n_{j1}n_{k0}-n_{j0}n_{k1}}{d}\notin\mathbb{N}.
\end{equation}
Hence, as $Z_{jk}$ is integer, a necessary condition for the set
of MUMs with $R$ measurements to exist is that 
\begin{equation}
(n_{j1}n_{k0}-n_{j0}n_{k1})\text{mod }d\neq0\quad\forall\:j,k\in\{2,3,...,R-1\}.\label{eq:conditions2}
\end{equation}

It follows from the fact that $n_{j1}n_{k0}\neq n_{k1}n_{j0}$ must
be satisfied for all pairs $i,j$ that \begin{subequations}\label{restrictions}
\begin{equation}
n_{j1}=n_{j'1}\Rightarrow n_{j0}\neq n_{j'0}
\end{equation}
\begin{equation}
n_{j0}=n_{j'0}\Rightarrow n_{j1}\neq n_{j'1}
\end{equation}
\begin{equation}
n_{j1}=n_{j0}\Rightarrow n_{j'0}\neq n_{j'1}\quad\forall\,j'\neq j.
\end{equation}
\end{subequations}

Let us define the $2\times(R-2)$-dimensional matrix 
\begin{equation}
\Omega=\left(\begin{matrix}n_{20}\quad n_{30}\quad...\quad n_{R-1,0}\\
n_{21}\quad n_{31}\quad...\quad n_{R-1,1}
\end{matrix}\right).\label{eq:N}
\end{equation}
The conditions \eqref{eq:conditions2} are thus equivalent to say
that the determinant of all submatrices $2\times2$ of $\Omega$ must
be congruent to zero modulo $d$, that is
\begin{equation}
\left|\begin{matrix}n_{j0} & n_{k0}\\
n_{j1} & n_{k1}
\end{matrix}\right|\mathrm{mod}\,d\neq0.\label{eq:subdet}
\end{equation}
Besides being more elegant, it is also helpful in having a more clear
view of the problem.

The problem of finding the maximum number of MUMs is now equivalent
to the problem of finding the maximum number of columns of matrix
\eqref{eq:N} such that condition \eqref{eq:subdet} is obeyed. Now,
let us look for the cases when condition \eqref{eq:subdet} is not
satisfied, i.e. the cases for which 
\begin{equation}
(n_{j0}n_{k1}-n_{k0}n_{j1})\equiv0\quad(\text{mod}\,d).\label{eq:subdet2}
\end{equation}
From Euclides Algorithm we know that for any $k\in\mathbb{Z},\;k<d$
$\exists\,\alpha,\beta\in\mathbb{Z}\quad\text{s.t.}\quad\alpha k+\beta d=\text{gdc}(k,d)$.
If $d$ is prime, then $\text{gdc}(k,d)=1$ since $k$ is smaller
than $d$. In this case, we can conclude that $k$ has a modular inverse
\begin{equation}
[\alpha k+\beta d](\text{mod}\,d)=1\quad\Rightarrow\quad[\alpha k](\text{mod}\,d)=1,
\end{equation}
what leads to $[\alpha ki](\text{mod}\,d)=i$ for any $i\in\mathbb{Z}$.
Thus, all elements of matrix \eqref{eq:N} have a modular inverse
with respect to $d$ for $d$ prime, which we denote by $\bar{n}_{ij}^{-1}$,
such that $n_{ij}\bar{n}_{ij}^{-1}\equiv1(\text{mod}\,d)$. Multiplying
\eqref{eq:subdet2} by $\bar{n}_{1j}^{-1}\bar{n}_{1k}^{-1}$ we are
left with 
\begin{equation}
(\bar{n}_{k0}^{-1}n_{k1}-\bar{n}_{j0}^{-1}n_{j1})\equiv0\quad(\text{mod}\,d)\label{eq:subdet3}
\end{equation}
or 
\begin{equation}
\bar{n}_{k0}^{-1}n_{k1}\equiv\bar{n}_{j0}^{-1}n_{j1}\quad(\text{mod}\,d).\label{eq:subdet4}
\end{equation}
Thus each column of $\Omega$ is characterized by a number $\chi_{k}=\bar{n}_{k0}^{-1}n_{k1}\,(\text{mod}\,d)$.
Accordingly, to satisfy \eqref{eq:subdet}, each column must have
a different value of $\chi_{k}$. As $1\leq\chi_{k}\leq d-1$, only
$d-1$ columns are allowed, therefore it is not possible to have more
than $d+1$ MUMs. 

If $d$ is not a prime number, then the general MUM condition \eqref{eq:MUMcond1-b}
implies that all elements of $\Omega$ are still coprime with $d$.
The same argument can be used to show that only $\varphi(d)$ columns
are allowed, $\varphi(d)$ being Euler's totient function of $d$,
i.e. the number of all coprimes with $d$ smaller than $d$ calculated
by $\varphi(d)=d\,\Pi_{p|d}\left(1-\frac{1}{p}\right)$, where $p|d$
denotes all prime divisors of $d$. Therefore $\varphi(d)+2$ is the
maximum number of MUMs in this case.

A particular form of the matrix $\Omega$ that satisfies \eqref{eq:subdet}
is 
\begin{equation}
N=\left(\begin{matrix}1 & \dots & 1 & \dots & 1\\
n_{21} & \dots & n_{j1} & \dots & n_{R-1,1}
\end{matrix}\right),\label{eq:N2}
\end{equation}
where $R=\varphi(d)+2$. Note that for $\Omega$ of this form, the
maximum value of $R$ is easily determined. We can see that no two
$n_{j1}$'s can be equal, since in this case condition \eqref{eq:subdet}
would not be satisfied. This, together with the constraints on $n_{ij}$,
determine that the maximum number $K_{max}$ of $n_{2j}$'s is then
given by the number of non-factors of $d$, as shown above.

We have shown the maximum number of MUMs allowed, but not necessarily
one can built a set containing the maximum number of measurements
since condition \eqref{eq:conditions2} are only necessary but not
sufficient. The possibility of having the maximum number of MUMs will
depend on the particular choice of phase space directions. For example,
as shown in Table \ref{tab:1}, if $\theta=\tan^{-1}\left(\sqrt{\frac{2}{15}}\right)$
with eight MUMs, dimension $d=7$ is allowed, but if $\theta=\tan^{-1}\left(\sqrt{\frac{1}{15}}\right)\approx0.25$rad
the dimensions allowed up to $30$ are \{13,17,19,23,29\}, which does
not include $d=R-1$. 

\section{Concluding remarks \label{sec:Concluding-remarks}}

MU is an essential concept in quantum mechanics and quantum information,
about which there are still some basic open questions, such as the
existence of a maximal set of such basis for a general dimension.
For the current known cases, continuous and discrete systems differ
in the number of bases or measurements contained in the maximal MU
set. In this chapter we extended the recently proposed PCG measurements
for CV systems that satisfy MU conditions \cite{Tasca18a}. This kind
of measurement, although performed on CV variables systems, has a
finite discrete number of outcomes resembling a discrete variable
system in some aspects. Here we showed how to construct an arbitrary
number of such PCG MUMs satisfying the MU conditions pairwise. In
our construction, the measurements are determined by the choice of
one angle and one period, for which we found MU conditions. For a
given number of measurements, we could find the allowed dimensions.
For the inverse question, given a dimension $d$, we showed that,
if $d$ in even, the maximum number of measurements is equal to three,
as is the case of the original CV system. On the other hand, if $d$
is odd then the maximum number of MUMs is determined by the number
of prime factors of $d$ and reproduces the discrete case for $d$
prime. It was not expected that the PCG MUMs would follow the discrete
or continuous behavior. Actually it is surprising that it does resemble
both for particular dimensions. These results seems merely mathematical
and reconnect them with physics, we also showed here an experimental
realization of the constructed measurements in which the continuous
variable is taken to be the transverse position of an attenuated light
beam. In our scheme, the FrFT connecting different phase space representations
is performed in a programmable way, allowing for preparation and measurement
in any phase space directions without changing the setup. It is shown
for the case of $d=3$ and $R=4$ that the measurements only present
MU if all the periods are chosen in accordance to the solutions we
found. 

An interesting future direction for this work is to identify a utility
for these results in quantum information. For example, can these MU
observables assist in a task such as tomography, cryptography, or
random number generation? \selectlanguage{american}

\part{Discrete degrees of freedom}

~\ihead{}

\ohead{\textbf{Chapter~\thechapter}~\leftmark}

\ifoot{}

\cfoot{}

\ofoot[
]{\thepage}

\chapter{Experimental techniques \label{chap:SPDC}}

In this Chapter we provide some experimental techniques used in the
discrete degrees of freedom experiments. At the single photon level
the methods used in our experiments are often equivalent to those
used in classical optics. Thus they are partly presented in a classical
manner. This Chapter does not intend neither to be complete nor rigorous,
but its intention is to present the main concepts necessary to understand
the experiments and the experimental issues.

\section{Generating polarization entangled pairs of photons}

Entangled pairs of photons can be generated through the process called
spontaneous parametric down conversion (SPDC). In this process, the
passage through a transparent second-order nonlinear medium can sometimes
cause a photon of a pump beam of frequency $\omega_{p}$ to split
in two photons of lower frequencies $\omega_{s}$ and $\omega_{i}$,
historically called signal and idler photons. Because they are generated
together, conservation laws are responsible for the two photons to
be correlated in many degrees of freedom such as frequency, momentum,
orbital angular momentum and polarization \cite{walborn2010}\footnote{The term parametric refer to the fact that no energy and momentum
is transferred to the medium, so the conservation laws apply only
to the three photons system.}. 

To conserve energy, the frequencies of signal and idler must sum up
to that of their parent photon (Fig. \eqref{fig:SPDC}-c)), that is
\begin{equation}
\omega_{i}+\omega_{s}=\omega_{p}.\label{eq:frequency}
\end{equation}
Momentum conservation implies 
\begin{equation}
\mathbf{k}_{s}+\mathbf{k}_{i}=\mathbf{k}_{p},\label{eq:phase}
\end{equation}
where the $\mathbf{\mathbf{k}}$s are the wave vectors of each mode
(Fig. \eqref{fig:SPDC}-b)). Relations \eqref{eq:frequency} and \eqref{eq:phase}
are called frequency and phase matching conditions, respectively. 

As the pump and the generated fields have quite different frequencies,
the dispersion of the medium causes them to travel at different velocities
because they experience different indices of refraction. Moreover,
isotropic media have null second order susceptibility (the electric
polarization must be an odd function of the electric field such that
the inversion of the last causes an inversion of the polarization
without any alteration). Thus this conversion process requires anisotropy
to take place. In anisotropic media the index of refraction depends
not only on the frequency, but also on the direction of propagation
and the polarization of the propagating field (See Appendix \ref{chap:Anexo2}).
Actually, this dependence is beneficial, since, together with other
properties like crystal thickness, the manipulation of the direction
of the optical axis of the crystal permits one to control the phase
matching and choose the type of SPDC allowed, as is described in the
sequence. 

\begin{figure}[h]
\centering{}\includegraphics[width=0.8\textwidth]{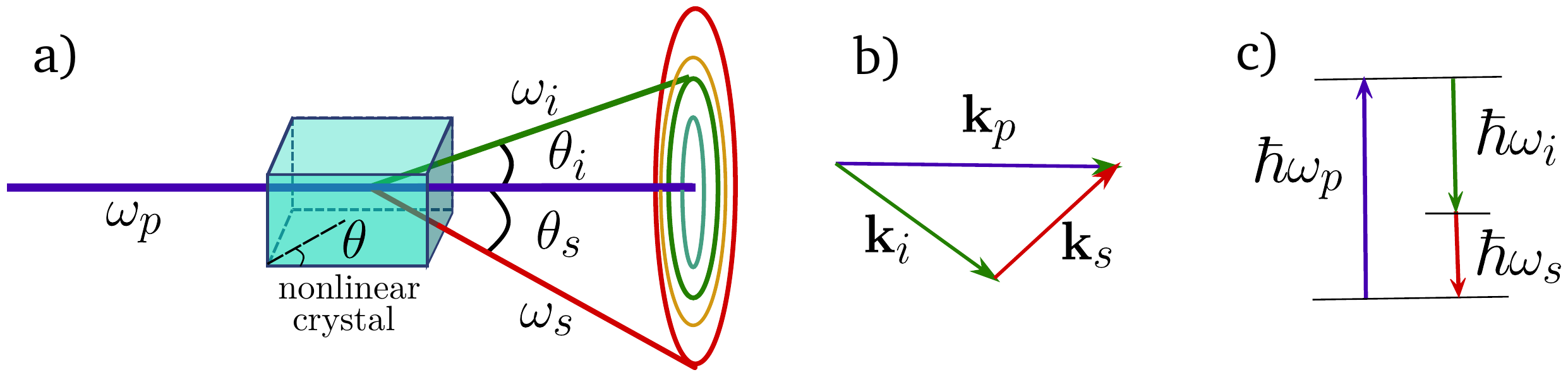}\caption{\textbf{a)} A collimated pump beam with frequency
$\omega_{p}$ crosses a slab of second-order nonlinear crystal generating
a continuum of lower frequency pairs of modes spatially distributed
according to the phase matching conditions. Here the example of the
so-called type-I phase matching: the generated beams have the same
polarization and form a cone centered around the incident beam. The
exit angles of the converted pairs is controlled by the optic axis
angle $\theta$. Two beams generated together, non-collinear with
the pump beam propagation, are represented by the red and green rays.
Each pair of beams conserve \textbf{b)} momentum and \textbf{c)} energy.}
\label{fig:SPDC}
\end{figure}

In a uniaxial crystal, as is the case of the $\beta$-barium-borate
(BBO) crystal used in our experiments, for each propagation direction
there are two orthogonaly polarized modes with different indices of
refraction. The ordinary wave ($\vartheta$) with index of refraction
independent of the propagation direction, and the extraordinary wave
($e$), whose refractive index depends on the angle $\theta$ that
the propagation direction makes with the optic axis of the crystal.
The phase-matching condition \eqref{eq:phase} can be separated into
components giving 
\begin{equation}
\begin{array}{c}
\omega_{s}n_{\vartheta,e}(\omega_{s},\theta)\sin\theta_{s}=\omega_{i}n_{\vartheta,e}(\omega_{i},\theta)\sin\theta_{i}\\
\omega_{s}n_{\vartheta,e}(\omega_{s},\theta)\cos\theta_{s}+\omega_{i}n_{\vartheta,e}(\omega_{i},\theta)\cos\theta_{i}=\omega_{p}n_{\vartheta,e}(\omega_{p},\theta)\cos\theta_{p},
\end{array}
\end{equation}
where the first equation is for the component perpendicular to the
pump beam direction and the second is the parallel component. The
norm of the wave vector is $\mathbf{|k}|=n\omega/c$. The index of
refraction of each mode must take into account if its polarization
is ordinary or extraordinary. These equations can be solved, together
with the frequency matching, for the angle of the optic axis $\theta$
of the crystal fixing any combinations of ordinary and extraordinary
polarization for the three waves. If the polarization of the two converted
beams is the same , the phase-matching is said to be of type-I and
if they are orthogonal it is said to be of type-II. In type-I SPDC,
the generated modes form coaxial cones and two corresponding modes
are diametrically opposite because of momentum conservation, as shown
in Fig. \eqref{fig:SPDC}-a). In type-II SPDC, two cones for the two
different polarizations are generated. If the crystal orientation
satisfies the phase-matching for a input polarization of the pump,
the orthogonal polarization will not be able the give rise to SPDC.

\begin{figure}[h]
\begin{centering}
\includegraphics[width=0.5\textwidth]{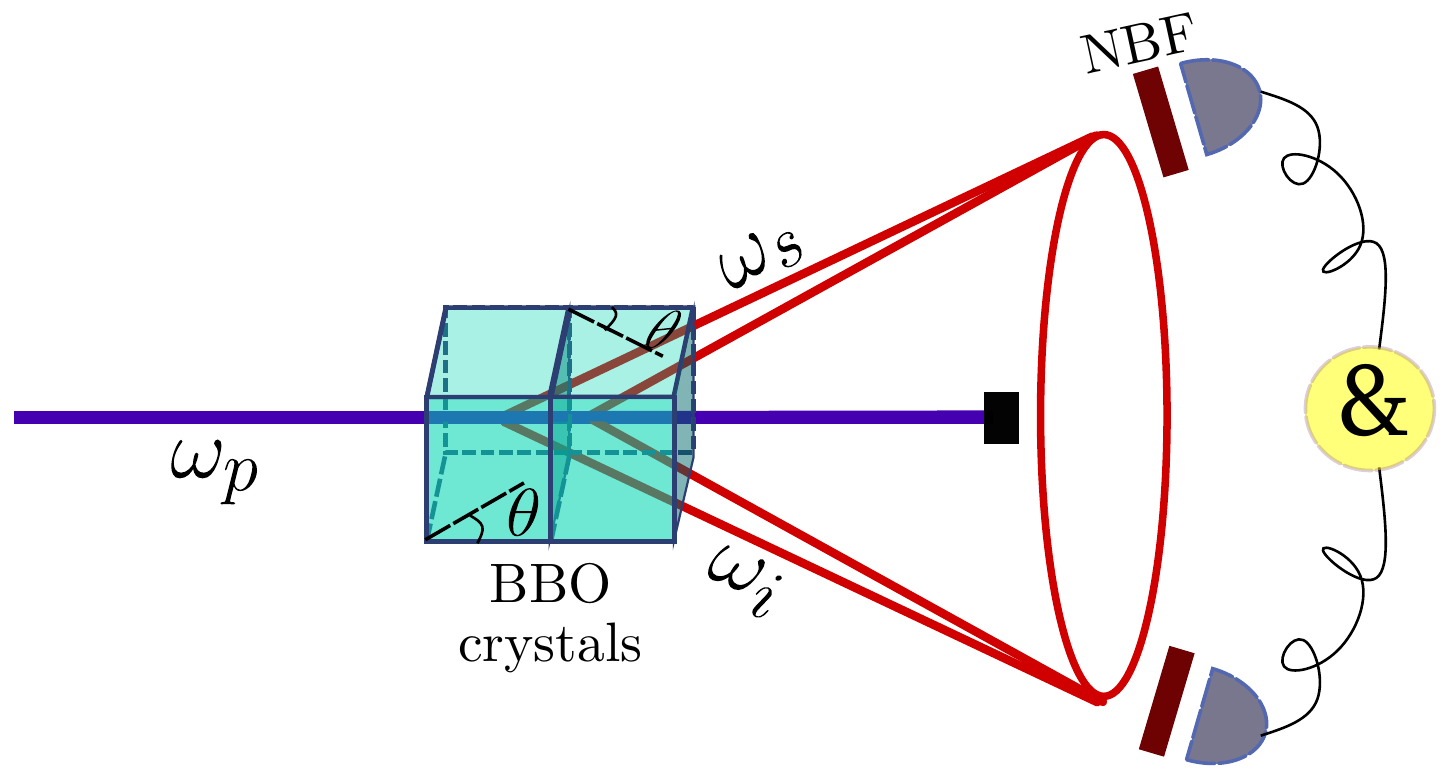}
\par\end{centering}
\caption{SPDC scheme used in our experiments: two adjacent type-I BBO crystals
cut as to attain the same phase-matching, one is rotated of $\pi/2$
relative to the other, such that the phase-matching is obeyed in each
crystal by orthogonal polarizations. The crystals are thin enough
to generate indistinguishable down-converted photons. The signal and
idler detectors are placed symmetrically relative to the pump beam
direction and collect the photon pairs with $\omega_{s}=\omega_{i}=\omega_{p}/2$.
Only detections in coincidence are considered. Narrow band spectral
filters (NBF) are placed in front of each detector.}

\end{figure}

In our experiments, we use a two-type-I-crystals source of entangled
photons \cite{Kwiat99}. In this source, two thin type-I crystals
are put adjacent, such that the pump beam passes through both. The
two crystals are identically cut and form the same angle $\theta$
with the normal incidence direction, but they are rotated by an angle
$\pi/2$ relative to each other such that, for the same propagation
direction, the ordinary polarization of the first crystal is in the
same direction as the extraordinary polarization of the second one.
The phase matching is such that a photon with horizontal polarization
(state $\ket{H}_{p}$\} arriving at the first crystal produces two
vertically polarized ones $\ket{V}_{s}\ket{V}_{i}$, while a vertically
polarized pump photon (state $\ket{V}_{p}$) generates a pair of horizontally
polarized photons $\ket{H}_{s}\ket{H}_{i}$ in the second crystal.
If the two crystals are thin enough, then their cones of down-converted
light coincide and the photons generated in one or in the other crystal
are indistinguishable. The result is that each pair of photons is
generated in a coherent superposition of being created in the first
or in the second crystal. If the pump photons are in the superposition
$\left(\ket{H}_{p}+e^{i\varphi_{p}}\ket{V}_{p}\right)/\sqrt{2}$,
thus for a given pair of corresponding down-converted modes, the polarization
state is (up to an irrelevant global phase)
\[
\frac{\ket{H}_{s}\ket{H}_{i}+e^{i\varphi}\ket{V}_{s}\ket{V}_{i}}{\sqrt{2}}.
\]
The phase difference $\varphi$ is due to the thickness of the crystals
and can be changed by controlling the phase difference $\varphi_{p}$
of the input laser. The desired entangled state can be obtained by
placing a QWP at $0^{\circ}$ in the pump beam, tilting this wave
plate around the vertical axis slightly changes the optical path length
inside it, allowing for tuning of the phase $\varphi$. 

The above described source of pairs of entangled photons can also
be used as a single photon source as one of the photons may be used
only to herald the presence of the other.

\subsection{Coherence length }

The phase-matching conditions are satisfied by a continuum of pairs
of modes. In fact, the SPDC for a single mode pump beam can be effectively
described by the interaction Hamiltonian
\begin{equation}
H_{I}=\sum_{\sigma_{s},\sigma_{i}}\int d\mathbf{k}_{s}\int d\mathbf{k}_{i}\,g_{\mathbf{k}_{s},\sigma_{s}}g_{\mathbf{k}_{i},\sigma_{i}}\delta(\mathbf{k}_{s}+\mathbf{k}_{i}-\mathbf{k}_{p})\delta(\omega_{s}(\mathbf{k}_{s})+\omega_{i}(\mathbf{k}_{i})-\omega_{p}(\mathbf{k}_{p}))a_{\mathbf{k}_{s},\sigma_{s}}^{\dagger}a_{\mathbf{k}_{i},\sigma_{i}}^{\dagger}a_{p}+H.c.,
\end{equation}
where $g_{\mathbf{k},\sigma}$ depends on the second-order nonlinear
coefficient of the media, on the volume of the crystal and on the
index of refraction for the mode with wave vector $\mathbf{k}$ and
polarization $\sigma$. The summation is over the two polarization
directions of the converted photons. Operator $a_{\mathbf{k},\sigma}^{\dagger}$
creates a photon with wave vector $\mathbf{k}$ and polarization $\sigma$,
while operator $a_{p}$ annihilates a photon in the pump mode. $H.c.$
stands for Hermitian conjugate. This effective description of the
SPDC process is valid under several assumptions, among which are the
assumptions of weak power of the pump laser such that the time between
two down conversions is relatively large, and that the crystal is
large as compared to the wavelength of the three beams \cite{walborn2010}.
Since the coupling is weak, such that a pump photon passes through
the crystal without being absorbed with high probability, there is
only a small probability that it can generate a pair of photons and
a negligible probability of generating a higher number of photons.
Thus, the time evolution operator in the interaction picture can then
be approximated by its first order Taylor expansion 
\begin{equation}
U(t)=e^{-itH_{I}}\approx\mathbb{1}-itH_{I}.
\end{equation}
As a spontaneous process the initial state can be regarded as vacuum
in the down converted modes and a strong classical field with amplitude
$E_{p}$ in the pump mode. The non-linearity of the material then
produces the two photon state
\begin{multline}
\ket{\psi}=C\Big(\ket{vac}-itE_{p}\sum_{\sigma_{s},\sigma_{i}}\int d\mathbf{k}_{s}\int d\mathbf{k}_{i}\,g_{\mathbf{k}_{s},\sigma_{s}}g_{\mathbf{k}_{i},\sigma_{i}}\delta(\mathbf{k}_{s}+\mathbf{k}_{i}-\mathbf{k}_{p})\delta(\omega_{s}(\mathbf{k}_{s})+\omega_{i}(\mathbf{k}_{i})-\omega_{p}(\mathbf{k}_{p}))\quad\\
\ket{1_{\mathbf{k}_{s},\sigma_{s}}}\ket{1_{\mathbf{k}_{i},\sigma_{i}}}\Big)\label{eq:estado}
\end{multline}
where $\ket{vac}$ is the vacuum state in the down converted modes
and $\ket{1_{\mathbf{k},\sigma}}$ is the state with one photon in
mode $\mathbf{k},\sigma$, $C$ is a normalization constant. The interaction
time $t$ is the time for crossing the crystal. As was mentioned before,
the photons produced by SPDC are non-monocromatic, rather it is quite
the opposite: they are broadband photons. In the experiments, however,
the photons are postselected by the position of the detectors whose
narrow aperture selects only a narrow range of momenta, thus reducing
the spatial and spectral bandwidths considered. Moreover, narrow bandwidth
filters are placed in front of the detectors. Also, the vacuum contribution
is neglected, since we consider only detections in coincidence between
the signal and idler detectors. In particular, we use a pump laser
centered at 325 nm and collect the degenerated pairs of photons at
650 nm with a single mode fiber. 

The large bandwidth causes the photons to have a quite small coherence
length. Roughly speaking, the coherence length is the maximum path
difference between two parts of a split beam such that their recombination
still gives rise to interference. It is expressed as $l_{c}\equiv\lambda^{2}/\Delta\lambda$
and is $\approx40\,\mu$m for a $10$ nm filter centered around $650$
nm. This means that, for example , a diagonal state $\ket{D}=(\ket{H}+\ket{V})/\sqrt{2}$
separated into horizontal and vertical polarization components and
recombined with a path difference much larger than the coherent length
will become the mixed state $(\ket{H}\bra{H}+\ket{V}\bra{V})/2$. 

The finite bandwidth of the converted photons also leads to a time
duration of the order of femtoseconds for the converted photons. This
also means that they are considered to be detected in coincidence
within a time interval of the order of femtoseconds. In the experiment,
though, we use a time window of $4$ ns. It does not cause photons
of different pairs to be considered coincident because of the big
interval between down-conversions and it allows for path differences
of the order of one meter between photons belonging to the same down-conversion.

\subsection{Obtaining error bars: Poisson distribution}

There are many methods to estimate the error of a measurement. For
example, the same measurement can be repeated ideally an infinite
number of times at the same conditions and the standard deviation
can be regarded as the error associated to the measurement. Another
common approach is to theoretically estimate the error from the previous
knowledge of the probability distribution behind stochastic results.
This last method, called Monte Carlo estimation, is quite useful when
the experimentalist does not have access to many repetitions of the
measurement. In this case, the mean value over the few measurements
realized is used as the mean value of the distribution and a large
number of random results is generated artificially from this distribution.
The quantities of interest and their standard deviation are then calculated
from this artificial experimental data.

SPDC is a very inefficient process and the probability of production
of one pair of photons from the coherent state of the pump laser is
very small, meaning that $|C|^{2}$ in Eq. \eqref{eq:estado} is close
to one. The probability of detection in coincidence of a pair of twin
photons is further diminished because of the narrow aperture of the
detectors and also their efficiency, as well as because of the frequency
filters. In this way, the time between two consecutive down-conversions
is large enough so that they can be regarded as independent events,
as well as two consecutive coincidence detections. Let us consider
that in a certain time interval $\delta t$ there is a probability
$q$ of one pair detection and that the probability of detecting two
or more pairs in this time interval is negligible. Then, the probability
of having $n$ detections in a time interval $T=N\delta t$, with
$N>n$, is given by the binomial distribution
\begin{equation}
p(n;T)=\left(\begin{array}{c}
N\\
n
\end{array}\right)q^{n}(1-q)^{N-n},\label{eq:binomial}
\end{equation}
where $q^{n}(1-q)^{N-n}$ is the probability that the $n$ first time
intervals $\delta t$ are going to register a coincidence count while
in the last $N-n$ subintervals the detectors will not click, $\left(\begin{array}{c}
N\\
n
\end{array}\right)=\frac{N!}{n!(N-n)!}$ accounts for all the possible sequences of intervals with clicks
and without clicks. Provided that $\delta t$ is small enough to have
at most one detection in this interval, it can be chosen arbitrarily.
Thus, one can make it as small as desired while keeping the mean number
of detections in the finite time interval $T$ , $\langle n\rangle=Nq$
, constant. When this limit is applied to Eq. \eqref{eq:binomial},
the binomial distribution becomes the Poisson distribution \cite{mandel1995}
\begin{equation}
p(n;T)=\frac{e^{-\langle n\rangle}\langle n\rangle^{n}}{n!},\label{eq:poisson}
\end{equation}
which is a one-parameter distribution that only depends on the mean
number of occurrences in the interval $T$. Knowing that the down-conversion
events behave according to a Poisson distribution, Monte Carlo can
be applied to produce artificial experimental results. This is done
by randomly picking values of $n$ according to the distribution \eqref{eq:poisson},
with $\langle n\rangle$ being the number of coincidence counts averaged
over the few real measurements realized. To get rid of unlikely events
coming from the tail of the distribution, which could spoil the mean
value, we also eliminate measurements far from the median of the measured
values before calculating their mean.

\section{Generating path entanglement \label{sec:Generating-path-entanglement}}

In the works presented in the following chapters, we use discrete
path degrees of freedom generated by beam displacers. A beam displacer
(BD) is a birefringent crystal cut as to separate an input beam with
arbitrary polarization into two orthogonally-polarized parallel beams
as shown in Fig. \ref{fig:BDoperating}. 

The device is a parallelepipedic piece of a uniaxial crystal with
optical axis ($AB$ direction in the Figure) orthogonal to one of
the input-face edges and making a $\pi/2-\theta$ angle with the other
edge. Let us define horizontal the polarization in the $\hat{y}$
direction and vertical the one in the $\hat{x}$ direction. A beam
propagating in air in the $\hat{z}$ direction reaches the BD perpendicularly
to the input face. Snell's law guarantees that the two refracted beams
will have wave vectors also in the $\hat{z}$ direction as the incidence
angle is $0\text{°}$. If the BD is oriented in the way shown in Fig.
\ref{fig:BDoperating}-a), the horizontal polarization
is perpendicular to the optical axis and will be the ordinary wave,
which has energy propagation in the same direction of the wave vector
(see Appendix \eqref{chap:Anexo2}). The vertical polarization, on
the other hand, is the extraordinary wave (only the electric displacement
remains vertical inside the material) and the energy propagates with
angle $\alpha$ with respect to the wave vector given by Eq. \eqref{eq:alpha}
(Appendix \eqref{chap:Anexo2}). After leaving the crystal the two
beams keep having wave vectors in the $\hat{z}$ direction, therefore
propagating parallel to each other, but now their centers are vertically
separated, with separation $d$ determined by the angle $\theta$
and the length $\ell$ of the crystal. For example, the BDs used in
our experiments are made of calcite ($n_{\vartheta}=1.658$, $n_{e}=1.486$
\cite{yariv1984}) cut at $\theta=45^{\circ}$ and length of $2.5$
cm, what gives separation angle $\alpha=6.25^{\circ}$ and final separation
of $2.7$ mm. If the crystal is oriented as in Fig. \ref{fig:BDoperating}-b),
then the horizontal polarization turns out to be the extraordinary
wave and the two beams come out the BD horizontally separated. If
the BD is tilted around $\hat{z}$ by an angle $\phi$, then combinations
of horizontal and vertical polarizations will come out displaced along
a line also rotated by $\phi$. 

\begin{figure}[h]
\begin{centering}
\includegraphics[width=0.9\textwidth]{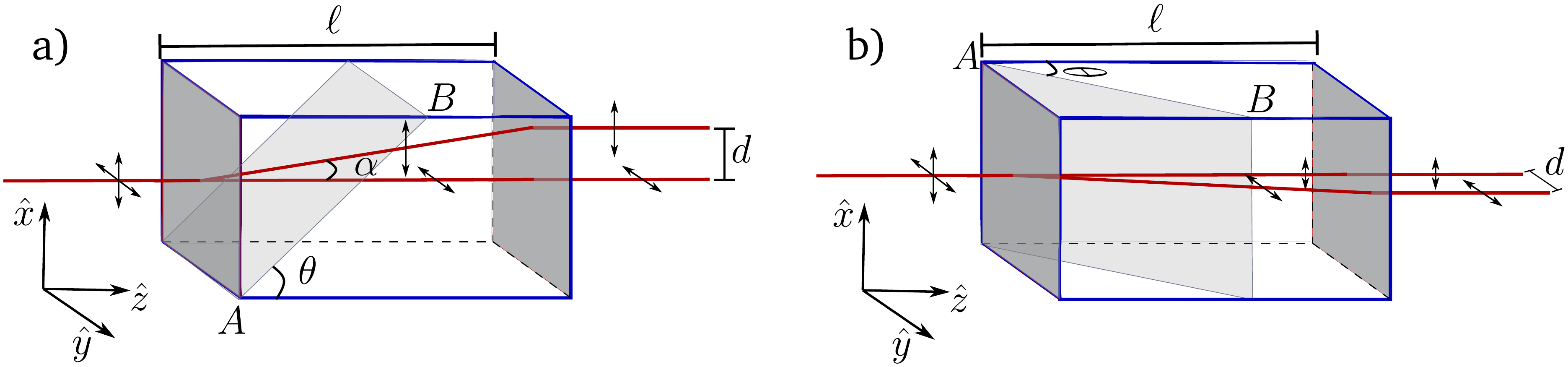}
\par\end{centering}
\caption{Beam displacer operating in the \textbf{a)} vertical and \textbf{b)}
horizontal polarizations.}
\label{fig:BDoperating}
\end{figure}

Given that the input beam is collimated and narrow compared to the
final separation, then the two diffracted beams define two independent
non-overlapping path modes we label simply by $0$ and $1$. For a
photon, these two spatial modes are two orthogonal states we represent
by $\ket{0}$ and $\ket{1}$, reducing the continuous momentum degree
of freedom to a discrete two level one. Moreover, if a photon pass
through a BD starting in a state $(\alpha\ket{H}+\beta\ket{V})\otimes\ket{0}$,
after the BD it will become $\alpha\ket{H}\ket{1}+\beta\ket{V}\ket{0}$,
considering the BD as in Fig. \ref{fig:BDoperating}-b),
creating entanglement between the polarization and path degrees of
freedom. Actually, this description is only effective. The index of
refraction for each beam is different causing an optical path difference
of $\ell(n_{\vartheta}-n_{e})$ that in our case is larger than the
coherence length of the down-converted photons we use. As one path
is delayed relatively to the other, this generates a temporal degree
of freedom and the state after the BD would be better described as
$\alpha\ket{H}\ket{1}\ket{t_{0}}+\beta\ket{V}\ket{0}\ket{t_{1}}$,
where $\ket{t_{0}}$ and $\ket{t_{1}}$ not necessarily are orthogonal
but have a small overlap. Thus, as we access only the polarization
and momentum degrees of freedom and trace out the temporal one, the
photons do not leave the BD in a pure state, but in a convex combination
of $\ket{H}\ket{1}$ and $\ket{V}\ket{0}$. The pure state is recovered
if one retrieves the possibility of interfering the two paths, i.e.,
coherently recombine them, recovering the out-of-diagonal terms of
the density matrix. This is done by using a second BD exactly equal
the first one generating the same delay to the previously non-displaced
path. To recombine the paths it is necessary to use a HWP to interchange
the beams polarization such that different beams are deviated in each
BD and both acquire the same total phase (see Fig. \ref{fig:BDBD}-a).
Now, using a second HWP we can measure the polarization in diagonal
basis for example and check for interference while changing the phase
difference $\phi$. 

\begin{figure}[h]
\begin{centering}
\includegraphics[width=0.95\textwidth]{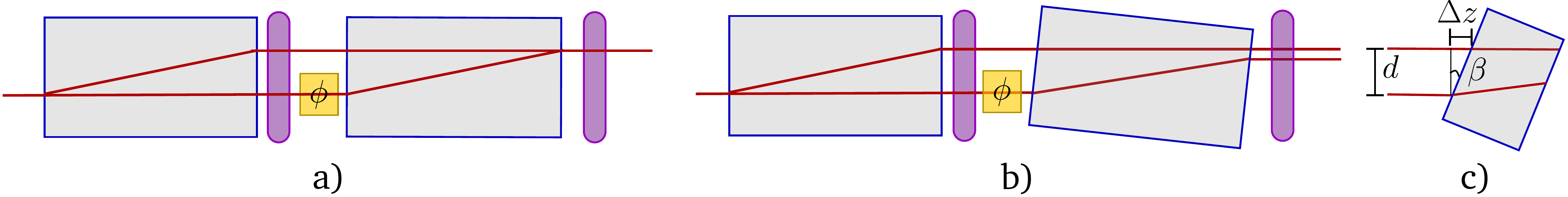}\caption{\label{fig:BDBD}Representation of an interferometer constructed with
two BDs. \textbf{a)} shows a perfectly aligned scenario, \textbf{b)}
shows an extreme case when a tilt of the second BD causes the two
beams to not overlap and \textbf{c)} shows in detail the path difference
caused by the tilt.}
\par\end{centering}
\selectlanguage{american}%
\end{figure}

If the two BDs are not well aligned as shown in Fig.\textbf{\textcolor{red}{{}
}}\ref{fig:BDBD}-b), the angle between the wave vector and the optical
axis will be different affecting the relative angle between the two
beams and possibly causing them to not overlap and consequently not
interfere at the output. This absence of overlap or a partial overlap
happens only in extreme cases when the angle of tilt is considerably
large. In a case of slight misalignment, the angle of deviation is
pretty much not affected \cite{paolino2007}, and what influences
the most in the attainment of the expected state at the output of
the second BD is the phase difference caused by the path difference
$\Delta z$ , exaggerated in Fig. \ref{fig:BDBD}-c). The phase difference,
given by $\Delta\phi=\frac{2\pi}{\lambda}\Delta z=\frac{2\pi}{\lambda}d\tan{\beta}$,
varies from zero to $2\pi$ with a small variation of approximately
$2.5\times10^{-4}$ rad in the tilt angle $\beta$ when $d=2.7$ mm
and $\lambda=650$ nm. This example shows the sensitivity of the two
BD interference with the relative tilt between them. In an experiment,
after a naked eye alignment, the second beam displacer is tilted until
a region of maximal visibility of interference is found, this means
that the beams are completely overlapping. Inside this region, the
second BD is placed in a position for which we have a maximal or minimal
power, depending on the projection we are realizing at the output. 

\section{Projective measurements\label{sec:Projective-measurements}}

Ideally one would like to be able to perform any projective measurement
in both degrees of freedom available. A projective measurement in
polarization is realized by means of a polarizing beam splitter (PBS).
A PBS (see Fig. \ref{fig:Polarizing-beam-splitter.}) is composed
by two triangular prisms made of the same transparent glass and glued
together in their hypotenuses. In the interface between the two prisms
there is a thin film layer designed such that the polarization component
parallel to the interface is completely reflected and the other is
completely transmitted \cite{macneille}. In this way this device
separates the horizontal and vertical polarization components of the
input beam in two orthogonal beams, horizontal polarization being
transmitted and vertical polarization being reflected. Then, a PBS
naturally provides a projective measurement in the basis $\left\{ \ket{H},\ket{V}\right\} $.

\begin{figure}[h]
\selectlanguage{american}%

\selectlanguage{english}%
\begin{centering}
\includegraphics[width=0.35\textwidth]{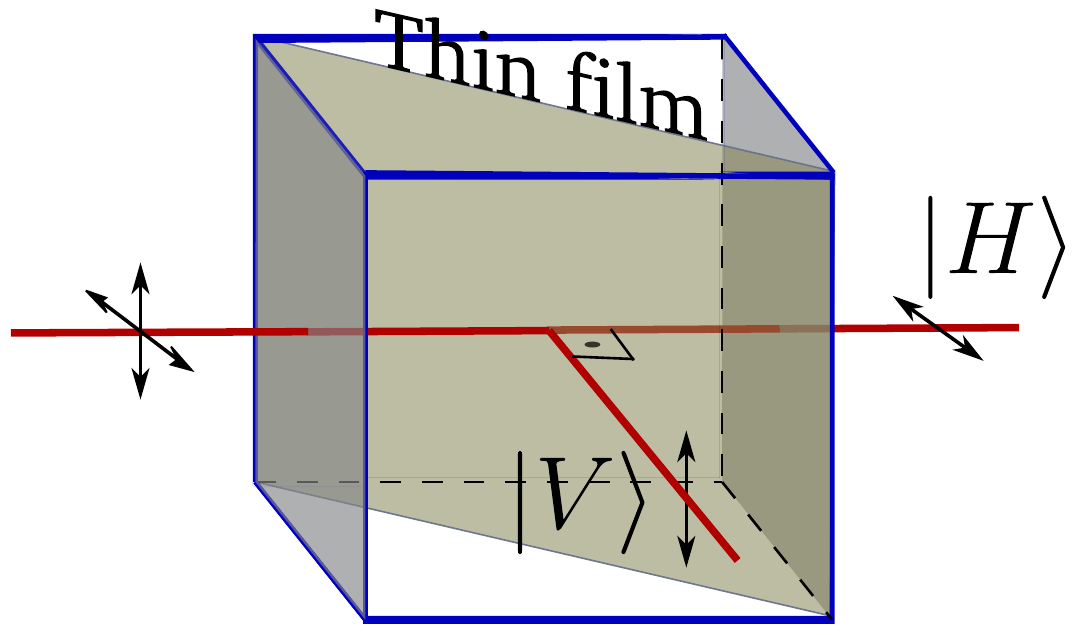}\caption{Polarizing beam splitter.\label{fig:Polarizing-beam-splitter.}}
\par\end{centering}
\selectlanguage{american}%
\end{figure}

It is possible to use a PBS also to project a polarization state into
any basis of the polarization Hilbert space if wave plates are used
to transform this basis into $\left\{ \ket{H},\ket{V}\right\} $.
Indeed, it is possible using the minimum set of a QWP followed by
a HWP. In order to understand this protocol easily, let us visualize
the effect of a HWP and a QWP in the Poincaré sphere representation.
In this representation, a general polarization state $\mathbf{E}=E_{0}\left(\cos(\theta/2),e^{i\phi}\sin(\theta/2)\right)$
becomes the normalized 3D vector 
\begin{equation}
\hat{r}_{\mathbf{E}}=\left(\cos\theta,\sin\theta\cos\phi,\sin\theta\sin\phi\right),\label{eq:poincare}
\end{equation}
whose components are the mean values of the Pauli matrices $(\left\langle \sigma_{z}\right\rangle ,\left\langle \sigma_{x}\right\rangle ,\left\langle \sigma_{y}\right\rangle )$
normalized by the total intensity \cite{damask}. It turns out that
the angle between this vector and the $x$ axis is $\theta$ and the
angle between the vector projection in the $yz$ plane and the $y$
axis is $\phi$, as illustrated in Fig. \ref{fig:a)-Polarization-state}-a).
All the linear polarization states lie in the equator of the sphere,
while the circularly polarized ones are situated in the poles. By
calculating the Poincaré vector after the action of a half wave plate
with optical axis forming angle $\gamma$ with the vertical {[}Eq.
\eqref{eq:HWP}{]}, it is easy to show that it is equivalent to apply
the operator
\begin{equation}
R_{\lambda/2}(\gamma)=\left(\begin{array}{ccc}
\cos\left(4\gamma\right) & \sin\left(4\gamma\right) & 0\\
\sin\left(4\gamma\right) & -\cos\left(4\gamma\right) & 0\\
0 & 0 & -1
\end{array}\right)
\end{equation}
to the initial Poincaré vector \eqref{eq:poincare}. This matrix is
a reflection with respect to the $x$ axis of the Poincaré space along
with a rotation of $4\gamma$ about the $z$ axis. Analogously, the
action of a QWP at angle $\alpha$ {[}Eq. \eqref{eq:QWP}{]} in Poincaré
space is obtained as the operator

\begin{equation}
R_{\lambda/4}(\alpha)=\left(\begin{array}{ccc}
\cos^{2}\left(2\alpha\right) & \sin\left(2\alpha\right)\cos\left(2\alpha\right) & \sin\left(2\alpha\right)\\
\sin\left(2\alpha\right)\cos\left(2\alpha\right) & \sin^{2}\left(2\alpha\right) & \cos\left(2\alpha\right)\\
-\sin\left(2\alpha\right) & \cos\left(2\alpha\right) & 0
\end{array}\right),
\end{equation}
which represents a $\pi/2$ rotation around the equator vector $\left(\cos(2\alpha),\sin(\alpha),0\right)$. 

\begin{figure}[h]
\begin{centering}
\includegraphics[width=0.6\textwidth]{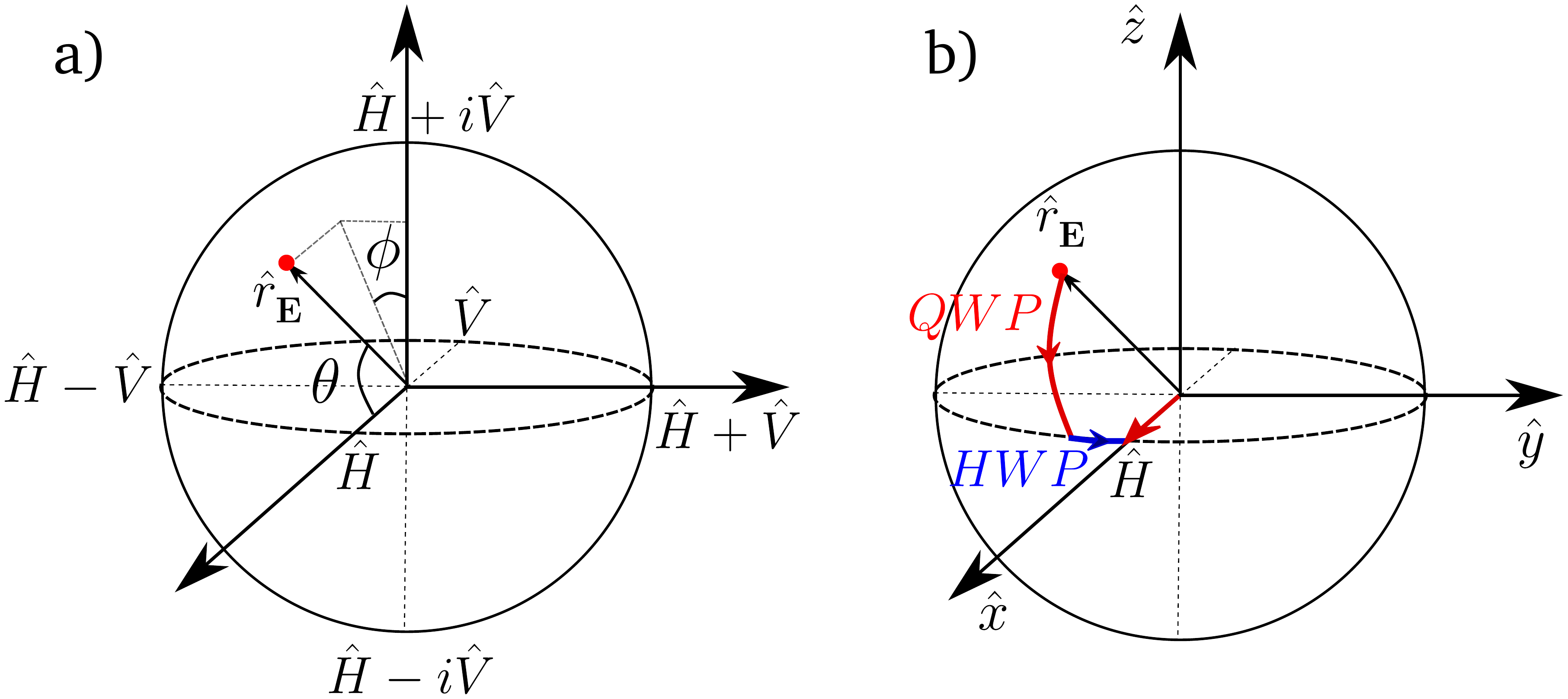}
\par\end{centering}
\caption{\textbf{a)} Polarization state representation in the Poincaré sphere.
\textbf{b)} Geometric representation of the transformation of an arbitrary
state into horizontal polarization: first a QWP rotates the state
to the equator and then a HWP rotates it to the horizontal state.\label{fig:a)-Polarization-state}}

\selectlanguage{american}%
\end{figure}

The transformation of any polarization state \eqref{eq:poincare}
into the horizontally polarized state $(1,0,0)$ can be made in two
steps represented in Fig. \ref{fig:a)-Polarization-state}-b):
\begin{enumerate}
\item Apply a QWP to the state with an angle such that the final state has
null $z$-component, lying in the equator of the Poincaré sphere.
It is achieved if the wave plate angle is chosen as to satisfy $\tan(2\alpha)=\tan\theta\cos\phi$.
The resulting state is
\[
\left(\begin{array}{c}
\cos\theta^{\prime}\\
\sin\theta^{\prime}\\
0
\end{array}\right)=\left(\begin{array}{c}
\cos^{2}\left(2\alpha\right)\cos\theta+\sin\left(2\alpha\right)\sin\theta\left[\cos\left(2\alpha\right)cos\phi+\sin\phi\right]\\
\sin^{2}\left(2\alpha\right)\sin\theta\cos\phi+\cos\left(2\alpha\right)\left[\sin\left(2\alpha\right)\cos\theta+\sin\theta\sin\phi\right]\\
0
\end{array}\right)
\]
\item Apply a HWP to the resulting state rotating it about the $z$ axis
and vanishing the second component. It is attained if the angle is
$\gamma=\theta^{\prime}/4$.
\end{enumerate}
Thus for a photon, after the combination of waveplates just described,
the original state $\ket{\theta,\phi}=\cos(\theta/2)\ket{H}+e^{i\phi}\sin(\theta/2)\ket{V}$
becomes $\ket{H}$and is transmitted by a PBS, as well as its orthogonal
state is transformed in $\ket{V}$ being reflected. If it is desirable
to really make a projection, producing state $\ket{\theta,\phi}$
after the measurement, one can reprepare the state using the inverted
sequence HWP+QWP at angles $-\gamma$ and $-\alpha$. 

The procedure was presented with this order of wave-plates to facilitate
its geometrical visualization. Nevertheless,it could be made in the
changing the order of HWP and QWP. For some states it is direct to
see this possibility. For example, to project over any state in the
equator , only a rotation around $z$ is required. The HWP can be
used for this purpose and the QWP may come after it if set to $\alpha=0$.
Another example, projecting over the right and left polarized states
in the poles require a $\pm\pi/2$ rotation about $y$ axis, which
is achieved by a QWP with $\alpha=\pm45^{\circ}$. Placing a HWP at
$0^{\circ}$ before this QWP only changes the sign of the angle $\alpha$
since it causes a minus sign in the vertical component.

\subsection{Path degree of freedom\label{subsec:Path-degree-of}}

One could directly have access to the intensity of each path, however
it is necessary to project over arbitrary path states to have access
to quantum superpositions of the momentum. It is possible by mapping
the path state into a polarization state. As discussed in Sec. \ref{sec:Generating-path-entanglement},
careful alignment of two BDs enables the coherent recombination of
the two path modes. Consider the particular case where the incoming
photons in the measurement stage in Fig. \ref{fig:Path-state-measurement:}
are in a separable state $\left(a\ket{0}+b\ket{1}\right)\left(c\ket{H}+d\ket{V}\right)$\footnote{This is the only case we consider here since it is what we have in
all the experiments.}. Then the lower path ($\ket{0}$) path passes through a HWP set at
$45^{\circ}$, which converts $\ket{H}$ into $\ket{V}$ and vice-versa.
A HWP at $0^{\circ}$ is placed in the upper path ($\ket{1}$) to
guarantee the coherence between the two paths by ensuring no optical-path-length
difference. When the photons pass through the BD, a new path is created
and the state becomes $ac\ket{0}\ket{V}+d\ket{1}(a\ket{H}+b\ket{V})+bc\ket{2}\ket{H}$
(assuming perfect coherent combination of paths). Thus, the initial
path state is transferred to the polarization state of mode $\ket{1}$
and measurements on polarization as explained before give access to
the path state.

\begin{figure}[h]
\begin{centering}
\includegraphics[width=0.4\textwidth]{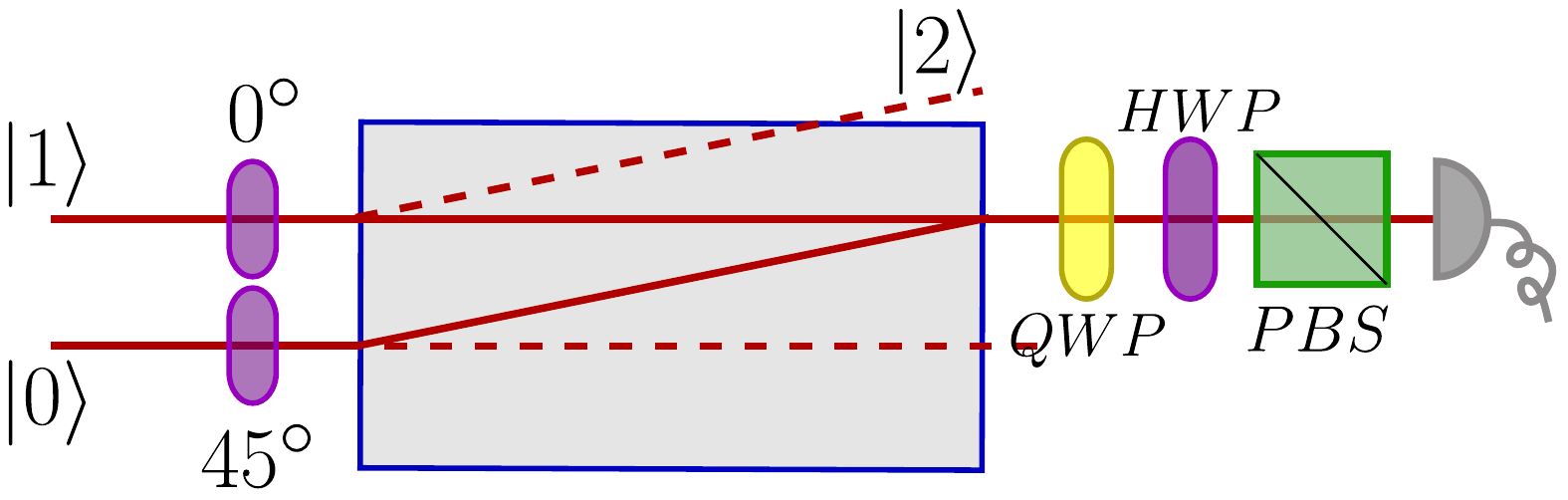}
\par\end{centering}
\caption{Path state measurement: the path state is mapped to polarization and
a projective measurement of polarization is performed. \label{fig:Path-state-measurement:}}

\selectlanguage{american}%
\end{figure}

\section{Unitary transformations \label{sec:Unitary-transformations}}

A unitary transformation of a qubit is equivalent to a rotation in
the Bloch sphere. As such, it can be specified by the rotation axis,
determined by the unit vector $\mathbf{n}=(n_{1},n_{2},n_{3})$, and
by the rotation angle $\xi$. It is expressed as
\begin{equation}
U=e^{-i\frac{\xi\mathbf{n}\cdot\boldsymbol{\sigma}}{2}}=\cos\frac{\xi}{2}\mathbb{1}-i\sin\frac{\xi}{2}(\mathbf{n}\cdot\boldsymbol{\sigma})
\end{equation}
or in matrix form
\begin{equation}
U=\left(\begin{array}{cc}
\cos\frac{\xi}{2}-in_{3}\sin\frac{\xi}{2} & -\sin\frac{\xi}{2}(in_{1}+n_{2})\\
-\sin\frac{\xi}{2}(in_{1}-n_{2}) & \cos\frac{\xi}{2}+i\sin\frac{\xi}{2}
\end{array}\right).
\end{equation}

As a unit vector, $\mathbf{n}$ is specified by two spherical-coordinates
angles. Thus, any unitary operator for a qubit is completely characterized
by three parameters. A configuration of optical elements devised to
implement any unitary transformation over a photon polarization qubit
should provide also this number of parameters to be changed as to
produce any values for $\xi$ and $\mathbf{n}$. In fact, the combination
(QWP@$\alpha_{1}$)-(HWP@$\gamma$)-(QWP@$\alpha_{2}$), represented
in Fig. \ref{fig:Unitary-operator-implementation:}, is able to realize
the unitary transformation given that the waveplate angles are chosen
to satisfy 
\begin{align}
\cos\frac{\xi}{2}=\cos\Theta\cos\Delta\quad & \quad n_{1}=\frac{\sin\Theta\cos\Delta}{\sqrt{\Omega}}\\
n_{2}=\frac{\cos\Theta\sin\Delta}{\sqrt{\Omega}}\quad & \quad n_{3}=\frac{\sin\Theta\sin\Delta}{\sqrt{\Omega}}
\end{align}
with $\Theta=\alpha_{1}-\alpha_{2}$, $\Delta=2\gamma-(\alpha_{1}+\alpha_{2})$,
and $\Omega=1-\cos^{2}\Theta\cos^{2}\Delta$. This can be verified
directly calculating the product of the waveplate operators. 

\begin{figure}[h]
\begin{centering}
\includegraphics[width=0.4\textwidth]{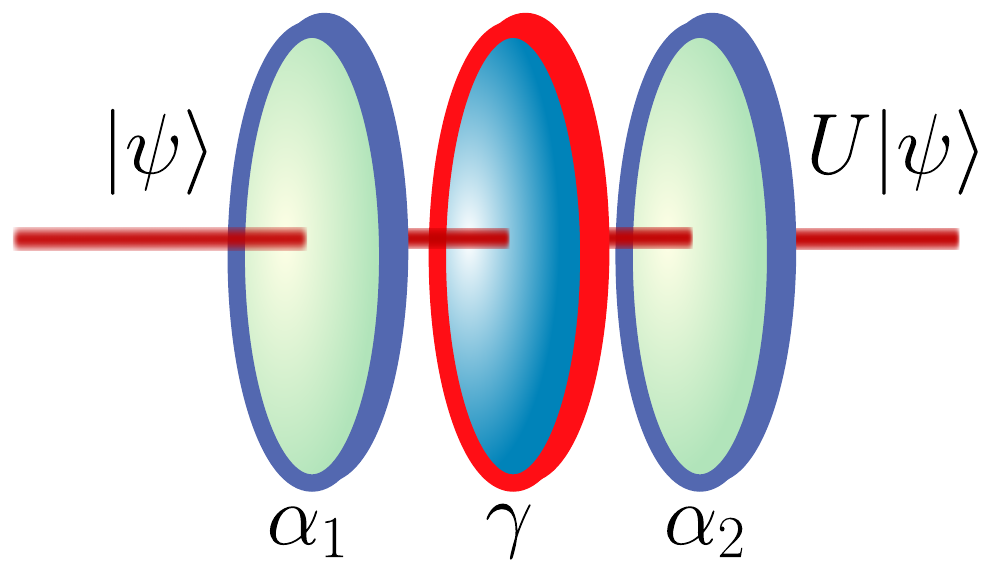}
\par\end{centering}
\caption{Unitary operator implementation on a polarization qubit: a light beam
with initial polarization state $\ket{\psi}$ passes through a sequence
of a QWP, a HWP and finally another QWP with angles chosen to implement
the unitary transformation $U$.\label{fig:Unitary-operator-implementation:}}

\selectlanguage{american}%
\end{figure}
\selectlanguage{american}

~\ihead{}

\ohead{\textbf{Chapter~\thechapter}~\leftmark}

\ifoot{}

\cfoot{}

\ofoot[
]{\thepage}

\chapter{Exposure of subtle multipartite quantum nonlocality \label{chap:Steering}}

Quantum systems can exhibit correlations that are stronger than the
ones allowed by classical physics, which can be classified as entanglement,
steering or Bell nonlocality depending on the level of characterization
of the parties involved. The definition of such quantum correlations
relies on the violation of a classical model. In this work, we show
an inconsistency on the current multipartite definition of steering
and Bell nonlocality. Namely, we show an apparent creation of such
correlations by applying a local operation on a system that is initially
believed to be uncorrelated. The inconsistency comes from the fact
that local operations are not able to increase or create nonlocal
correlations. This leads to a redefinition of these correlations,
according to which the conflicting models are allocated a subtle form
of correlation, which is exposed -- as opposed to created -- by
the local operations. Finally, we provide the first experimental demonstration
of both steering and Bell nonlocality exposure with three photonic
qubits. 

This work was done in collaboration with professors Leandro Aolita,
Gabriel Aguilar and Stephen Walborn, and with postdocs Márcio Taddei
and Ranieri Nery, all at UFRJ at the time the research was developed.
I contributed to the design of the experiment, and I was the main
experimentalist in the execution of the experiment and analysis of
data. The paper was submitted to Physical Review X and a preprint
can be found in \cite{taddei2019b}.

\section{Introduction}

Three forms of quantum correlations occur in nature --- entanglement,
Bell nonlocality and steering. The distinction between them is given
by the level of trust and control that one has on the systems involved,
as depicted in Fig. \ref{fig:The-different-levels-1}. Entanglement
(Fig. \ref{fig:The-different-levels-1}-c)), for instance, is naturally
formulated in the so-called device-dependent (DD) scenario \cite{Horodecki2009}.
There, one assumes that the system can be completely characterized
by the measurement apparatus, at least in principle. Thus, in this
scenario, the quantum state of the system is known and entanglement
is defined as the impossibility of finding a separable model for the
global state. Bell nonlocality (Fig. \ref{fig:The-different-levels-1}-d)),
in contrast, takes place in the device-independent (DI) description
\cite{Brunner2014}. There, measurement devices are treated as untrusted
black boxes whose actual measurement process is uncharacterized or
ignored, relying only on classical measurement settings (inputs) and
results (outputs). Here the description is given by the probability
distribution of the outcome results given the measurement choice.
Quantum steering (Fig. \ref{fig:The-different-levels-1}-e)), on the
other hand, is a hybrid type of correlation -- intermediate between
entanglement and Bell nonlocality -- that arises in semi-DI settings
\cite{Reid2009,Cavalcanti2017,Uola2019}. The latter involves both
DD and DI parties. In this case, the total system is described by
a hybrid mathematical object involving probability distributions and
quantum states, the so called assemblage which is presented in the
next section. 

\begin{figure}[h]
\begin{centering}
\includegraphics[width=1\textwidth]{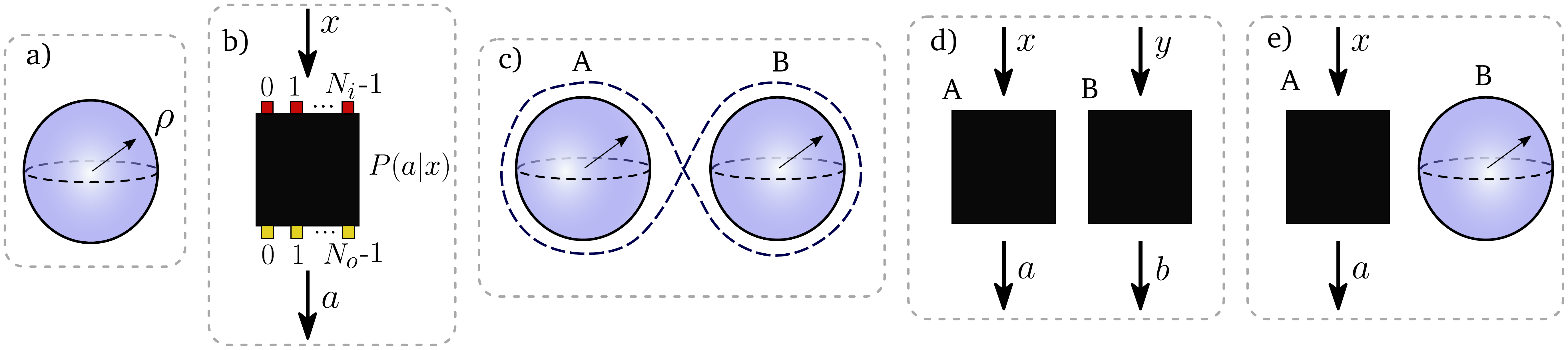}
\par\end{centering}
\caption{\textit{The different levels of characterization of a quantum system:}
\textbf{a)} the device-dependent scenario in which the system is completely
characterized and its quantum state $\rho$ is known, here represented
by a point inside a qubit Bloch sphere (obviously the system could
have any dimension); and \textbf{b)} the device-independent scenario
in which only the labels of measurement settings $x\in\{0,1,...N_{o}-1\}$
and the labels of measurement outcomes $a\in\{0,1,...N_{i}-1\}$ are
known, the underpinning mechanism being ignored, in this case the
conditional probabilities $P(a|x)$ are the only accessible information
about the system. \textit{The scenarios of the three forms of quantum
correlations for a bipartite system:} \textbf{c)} entanglement, \textbf{d)}
Bell nonlocality, \textbf{e)} quantum steering.\label{fig:The-different-levels-1}}

\selectlanguage{american}%
\end{figure}

Whereas entanglement is a resource for DD applications in quantum
information, Bell nonlocality is the key resource for DI applications
such as DI quantum key distribution \cite{Barrett2005,Acin2006,Acin2006a,Acin2007}
and DI certified randomness \cite{Colbeck2009,Colbeck2010,Pironio2010,Acin2016},
which are typically much more experimentally demanding than the corresponding
DD protocols. Steering is known to be the crucial resource for key
technological applications in the semi-DI scenario, which are generally
less technically difficult than their DI counterparts, while requiring
less assumptions than the corresponding DD protocols. These include
semi-DI entanglement certification \cite{Wiseman2007,Jones2007,Cavalcanti2017,Uola2019},
quantum key distribution \cite{Branciard2012,He2013}, certified-randomness
generation \cite{Skrzypczyk2018}, quantum secret sharing \cite{Kogias2017,Xiang2017},
as well as other useful protocols in multipartite quantum networks
\cite{Huang2019}. Moreover, there are tasks for which the presence
of steering, i. e. the capability of a entangled state to present
steering, gives the minimum amount of entanglement necessary for that
task to be successful. An example of this is the subchannel identification
task using a probe and an ancilla which are entangled, if only local
measurements and classical communication are allowed, then only steerable
states are useful \cite{Piani2015}.

These applications, as well as fundamental interest, motivated the
development of a resource theory of steering \cite{Gallego2015,Kaur2017}.
Resource theories constitute formal treatments of a physical property
as a resource, providing a complete toolbox for its quantification,
classification, and operational manipulation (see, e.g., \cite{Brandao2015,Brandao2015a,Coecke2016}).
They have been formulated for entanglement \cite{Horodecki2009} and
Bell nonlocality \cite{Gallego2012,deVicente2014,Gallego2017,Wolfe2019},
as well as for other interesting quantum properties \cite{Winter2016,Chitambar2016,Grudka2013,Amaral2018,Taddei2019,Wolfe2019}.
The cornerstone of any resource theory is the set of its \emph{free
operations}. These are unable to create the resource: they transform
every resourceless state into a resourceless state.

Interestingly, their study in fully-DI multipartite scenarios has
revealed an operational inconsistency at the very heart of the theory
\cite{Gallego2012,Bancal2013}. A fully DI description is cast in
terms of a \emph{Bell behavior}, given by a conditional probability
distribution of the outputs given the inputs. The inconsistency is
that, in a tripartite DI scenario, operations that are local in $AB$
can map tripartite Bell behaviors that are local in the $AB|C$ bipartition
into bipartite Bell behaviors that violate a Bell inequality across
$AB|C$. Bell locality implies that there exists a local-hidden-variable
(LHV) model, in which correlations are explained by a (hypothetical)
classical common cause (the hidden variable) within the common past
light-cone of the measurement events \cite{Bell1964}. Any Bell-inequality
violation implies incompatibility with LHV models, i.e. Bell nonlocality.
The observation above thus seems contradictory, as local wirings within
$AB$ are free operations of Bell nonlocality in $AB|C$ and therefore
unable to increase Bell-inequality violations. The problem, however,
lied in the definition of Bell nonlocality in multipartite scenarios
used previously \cite{Svetlichny1987}.

According to the traditional definition \cite{Svetlichny1987}, Bell
nonlocality across a system bipartition is incompatible with any LHV
model with respect to it. This includes so-called ``fine-tuned''
models \cite{Wood2015} with hidden signaling. These are LHV models
where, for each value of the hidden variable, the subsystems on each
side of the bipartition communicate, but for which the statistical
mixture over all values of the hidden variable renders the observable
correlations non-signaling. The problem is that the bilocal wiring
(taking the output of one black-box as the input of the other) can
conflict with the hidden communication in such models, giving rise
to a causal loop. For instance, to physically implement the wiring,
Bob must be in the causal future of Alice, which is inconsistent with
hidden communication from Bob to Alice. This explains why apparently
bilocal behaviors can lead to Bell violations after a bilocal wiring.
A redefinition of multipartite Bell nonlocality was then proposed
\cite{Gallego2012,Bancal2013}. This considers the correlations from
conflicting bilocal models already nonlocal across the bipartition,
so that the wiring simply exposes an already-existing subtle form
of Bell nonlocality. We refer to the latter form and effect as \emph{subtle
Bell nonlocality} and \emph{Bell-nonlocality exposure}, respectively.

The redefinition fixed the inconsistency, but also opened several
intriguing questions. First, no experimental observation of Bell-nonlocality
exposure has been reported. Second, even though steering theory is
relatively mature \cite{Cavalcanti2011,He2013,Armstrong2015,Taddei2016,Li2015},
little is known about \emph{steering exposure}. Operational consistency
relative to steering exposure was considered, in particular, in a
definition of multipartite steering \cite{He2013}, but based on models
where each party is probabilistically either trusted or untrusted.
On the other hand, a definition based on multipartite entanglement
detection in semi-DI setups with fixed trusted-versus-untrusted divisions
was proposed in Ref. \cite{Cavalcanti2015a}. There, bilocal hidden-variable
models (for multipartite assemblages) with an explicit quantum realization
are considered, which automatically rules out potentially-conflicting
fined-tuned models. Nevertheless, this has the side-effect of over-restricting
the set of unsteerable assemblages, thus potentially over-estimating
steering. Third, exposure as a resource-theoretic transformation is
yet unexplored territory. For instance, is it possible to obtain every
bipartite assemblage via exposure from some multipartite one? What
about Bell behaviors? Moreover, is there a single $N$-partite assemblage
from which all bipartite ones are obtained via exposure?

These are the questions we answer. To begin with, we show that, remarkably,
exposure of quantum nonlocality is a universal effect, in the sense
that every bipartite Bell behavior (assemblage) can be the result
of Bell-nonlocality (steering) exposure starting from some tripartite
one. This highlights the power of exposure as a resource-theoretic
transformation. However, we also delimit such power: we prove a no-go
theorem for multi-black-box universal steering bits: there exists
no single $N$-partite assemblage (with $N-1$ untrusted and 1 trusted
devices) from which all bipartite ones can be obtained through free
operations of steering. Interestingly, in the universal steering exposure
protocol, the starting behavior is not guaranteed to admit a physical
realization, i.e. it may be supra-quantum \cite{Sainz2015,Sainz2018a,Sainz2019}.
Therefore, we also derive an example that is manifestly within quantum
theory. Moreover, we show that the output assemblage of such protocol
is not only steerable but also Bell nonlocal (in the sense of producing
a nonlocal behavior upon measurements by Charlie). This is notable
as Bell nonlocality is a stronger form of quantum correlation than
steering. We refer to this effect as \emph{super-exposure of Bell
nonlocality}. In turn, we provide a redefinition of (both multipartite
and genuinely multipartite) steering to re-establish operational consistency.
Finally, we experimentally demonstrate exposure as well as super-exposure.
This is done using three degrees of freedom of two entangled photons
generated by spontaneous parametric down conversion, in a deterministic
protocol.

This chapter is organized as follows: in Section 5.2 the basic concepts
related to quantum steering are presented, including the current definition
of multipartite steering, postquantum steering, resource theory of
steering and the methods we use to detect and quantify steering in
the following sections. In the sequence, the general steering and
Bell nonlocality exposure protocols are presented together with a
quantum realizable example in Section 5.3. In Section 5.4 the experimental
implementation and experimental results are shown. Lastly, the proposed
redefinition of multipartite quantum steering is given in Section
5.5, this redefinition removes any inconsistency with the resource
theory of steering.

\section{Steering and the semi-DI setting}

The concept of quantum steering originates with the beginning of the
quantum theory. The name \textquotedbl steering\textquotedbl{} is
attributed to Schrödinger who was studying the possibility of producing
different ensembles of quantum states at a distance \cite{Schrodinger1935}
by performing local measurements. The formal treatment of steering
though was given only recently for a bipartite system \cite{Wiseman2007,Jones2007}. 

The scenario in which quantum steering is defined for a bipartite
system is as follows. A two-party system is shared between Alice and
Bob. Alice cannot characterize her measurement device such that all
the information she has is the classical input $x$ she gives to the
device and the classical output $a$. That is, Alice holds a black
box with $N_{o}$ possible choices of untrusted measurements she can
perform, and for each input $N_{i}$ different results can come out
with probability $P_{a|x}$ conditioned to the input. On the other
hand, Bob can realize tomographic measurements upon his particle to
figure out what is the quantum state he holds. If he performs quantum
state tomography conditioned to Alice's input and output, then what
he gets is a conditional state $\rho_{a|x}$, a state that has been
prepared by Alice at a distance while performing her local measurements. 

In this setting, the global system is completely characterized by
a mathematical object called an assemblage $\boldsymbol{\sigma}=\{\sigma_{a|x}\}_{a,x}$
defined as the set of sub-normalized states such that $\Tr\left[\sigma_{a|x}\right]=P_{a|x}$
and $\sigma_{a|x}/\Tr\left[\sigma_{a|x}\right]=\rho_{a|x}$ containing
all the combinations of inputs and outputs. If the joint system is
in a quantum state $\rho^{AB}$ then the assemblage elements are obtained
as $\sigma_{a|x}=\Tr_{B}\left[\left(M_{a|x}\otimes\mathbb{1}\right)\rho^{AB}\right]$
considering that Alice's action is described by the measurement operators
$M_{a|x}$. However, because of the semi-device independence, the
global state is unknown. We assume that $\boldsymbol{\sigma}$ satisfies
the no-signaling (NS) principle, by virtue of which measurement-outcome
correlations alone do not allow for communication. This physical situation
imposes the non-signaling condition to the assemblage
\begin{equation}
\sum_{a}\sigma_{a|x}=\sum_{a}\sigma_{a|x^{\prime}}=\varrho^{(B)},
\end{equation}
which means that if Bob does not know Alice's output (and he does
not without explicit communication) he cannot infer anything about
her input. Moreover, the normalization of Alice's probabilities require
\begin{equation}
\Tr\left[\sum_{a}\sigma_{a|x}\right]=1\qquad\forall x.
\end{equation}
The correlation between Alice's measurement and Bob's states is classified
as quantum steering if it cannot be explained by a classical model.
On the other hand, the assemblage $\boldsymbol{\sigma}$ is said to
be unsteerable if its elements admit a classical explanation in terms
of a classical hidden stochastic variable
\begin{equation}
\sigma_{a|x}=\sum_{\lambda}\,P_{\lambda}\:P_{a|x,\lambda}\,\varrho_{\lambda},\label{eq:LocalHS}
\end{equation}
i.e., a classical stochastic variable $\lambda$ is distributed to
Alice and Bob with probability $P_{\lambda}$, this variable is the
local common cause for Alice's probability distribution and for Bob's
state. They do not have access to this hidden variable and after unavoidably
averaging their assemblage over it, it seems that the quantum states
are nonlocaly correlated to the measurements. The description \eqref{eq:LocalHS}
is called local hidden state (LHS) model. 

\subsection{Multipartite steering}

The multipartite scenario is considerably richer than the bipartite
one. For the simplest case of three parties, the semi-device independent
setting allows for two configurations, either 1DD-2DI or 2DD-1DI ,
as shown in Fig. \ref{fig:Possible-semi-device-independent}. In this
work we focus in the former since it is enough to show the inconsistency
in the current definition of steering. Such systems are fully described
by a Bell behavior $\boldsymbol{P}^{(AB)}:=\{P_{a,b|x,y}\}_{a,b,x,y}$,
with $P_{a,b|x,y}$ the conditional probability of outputs $a,b$
given inputs $x,y$, for Alice and Bob, and an ensemble of conditional
quantum states $\varrho_{a,b|x,y}$ for Charlie. These can be encapsulated
in the assemblage $\boldsymbol{\sigma}:=\{\sigma_{a,b|x,y}\}_{a,b,x,y}$,
of sub-normalized conditional states $\sigma_{a,b|x,y}:=P_{a,b|x,y}\,\varrho_{a,b|x,y}$.
The NS-principle implies that the statistics observed by any subset
of users should be independent of the input(s) of the remaining user(s).
Mathematically, this condition reads \begin{subequations} 
\begin{alignat}{3}
 & \sum_{a}\sigma_{a,b|x,y} &  & =\sigma_{b|y}^{(BC)},\  &  & \quad\forall\ b,x,y,\label{eq:NS1}\\
 & \sum_{b}\sigma_{a,b|x,y} &  & =\sigma_{a|x}^{(AC)},\  &  & \quad\forall\ a,x,y,\label{eq:NS2}\\
 & \sum_{a}\sigma_{a|x}^{(AC)} &  & =\sum_{b}\sigma_{b|y}^{(BC)}=\varrho^{(C)},\  &  & \quad\forall\ x,y,\label{eq:NS3}
\end{alignat}
where $\boldsymbol{\sigma}^{(AC)}:=\{\sigma_{a|x}^{(AC)}\}_{a,x}$
and $\boldsymbol{\sigma}^{(BC)}:=\{\sigma_{b|y}^{(BC)}\}_{b,y}$ are
respectively the reduced assemblages on the $AC$ and $BC$ subsystems,
and $\varrho^{(C)}$ is the reduced state on $C$. \label{eq:NS}\end{subequations}

\begin{figure}[h]
\centering{}\includegraphics[width=0.6\textwidth]{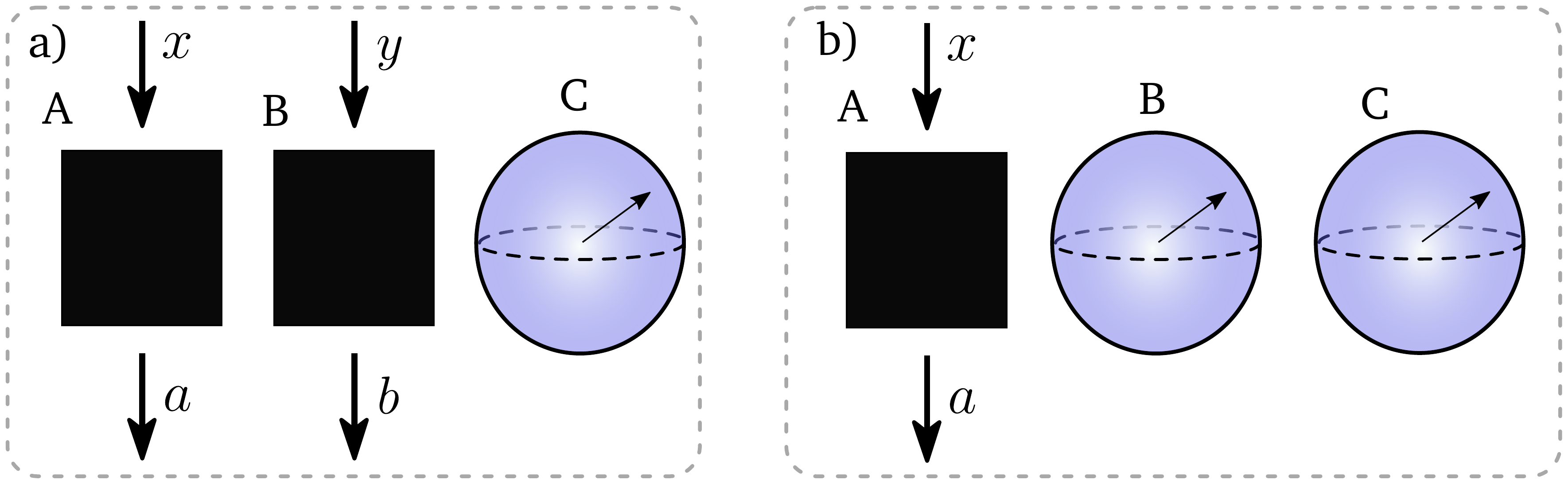}\caption{Possible semi-device independent settings: a) Two black boxes (A and
B) and a trustful device (C), i.e. 1DD-2DI, and b) one black box (A)
and two quantum systems (B and C), i.e. 2DD-1DI. \label{fig:Possible-semi-device-independent}}
\end{figure}

Unlike in Bell nonlocality or entanglement, semi-DI systems have a
natural bipartition: the one separating the trusted devices from the
untrusted ones. This is the bipartition with respect to which we define
steering throughout, unless otherwise explicitly stated. According
to the standard definition \cite{Uola2019}, $\boldsymbol{\sigma}$
is unsteerable if it admits a local hidden-state (LHS) model, namely,
if it can be decomposed as 
\begin{equation}
\sigma_{a,b|x,y}=\sum_{\lambda}\ P_{\lambda}\ \ P_{a,b|x,y,\lambda}\,\varrho_{\lambda}\ .\label{eq:LHS}
\end{equation}
Otherwise $\boldsymbol{\sigma}$ is steerable. Here, $P_{\lambda}$
is the probability of the hidden variable $\Lambda$ taking the value
$\lambda$, each $\boldsymbol{P}_{\lambda}^{(AB)}:=\{P_{a,b|x,y,\lambda}\}_{a,b,x,y}$
is a $\lambda$-dependent behavior, and $\varrho_{\lambda}$ is the
$\lambda$-th hidden state for $C$ (locally correlated with $AB$
only via $\Lambda$). Importantly, that $\boldsymbol{\sigma}$ is
non-signaling does not imply that so is each $\boldsymbol{P}_{\lambda}^{(AB)}$.
In fact, LHS models can exploit hidden communication between Alice
and Bob as long as actual communication at the observable level (i.e.
upon averaging $\Lambda$ out) is impossible. This effect is known
as fine-tuning \cite{Wood2015}; the standard definition of steering
imposes no restriction on fine-tuned LHS models. This turns out to
be critical. Indeed, we will see that unrestricted hidden signaling
is responsible for a stark conflict with the reasonable expectation
that local operations should not increase inter-party correlations.

The definition of steering as the violation of Eq. \eqref{eq:LHS}
is concerned only with the joint capability of Alice and Bob to steer
Charlie's state, without making any reference to the mechanism responsible
for it. If a violation of that model occurs, it could be the case
that only Alice is correlated to Charlie or only Bob or both. Also,
although they do not have the ability to steer Charlie, the joint
probability of Alice and Bob $P_{a,b|x,y,\lambda}$ can in principle
be non factorizable. In the particular case of a fully separable quantum
state $\rho_{sep}^{ABC}=\sum_{\lambda}\,P_{\lambda}\:\varrho_{\lambda}^{A}\otimes\varrho_{\lambda}^{B}\otimes\varrho_{\lambda}^{C}$
with the realization of the local measurements $\{M_{a|x}^{A}\}_{a,x}$
and $\{M_{b|y}^{B}\}_{b,y}$ in the parties $A$ and $B$, respectively,
the assemblage obtained is always unsteerable and moreover the Bell
behavior is separable 
\begin{equation}
\sigma_{a,b|xy}=\Tr_{AB}\left[\left(M_{a|x}^{A}\otimes M_{b|y}^{B}\otimes\mathbb{1}\right)\rho_{sep}^{ABC}\right]=\sum_{\lambda}\:P_{\lambda}\;P_{a|x,\lambda}\,P_{b|y,\lambda}\otimes\varrho_{\lambda}^{C}.\label{eq:sep}
\end{equation}
Another extreme case is that of a biseparable state $\rho_{bisep}^{ABC}=\sum_{\nu}P_{\nu}\varrho_{\nu}^{A}\otimes\varrho_{\nu}^{BC}+\sum_{\mu}P_{\mu}\varrho_{\mu}^{AB}\otimes\varrho_{\mu}^{C}+\sum_{\lambda}P_{\lambda}\varrho_{\lambda}^{B}\otimes\varrho_{\lambda}^{AC}$,
i. e. a state that is the mixture of states that are separable in
at least one bipartition. The violation of a biseparable state model
defines genuine tripartite entanglement. Analogously, a biseparable
assemblage model 
\begin{align}
\sigma_{a,b|xy} & =\Tr_{AB}\left[\left(M_{a|x}^{A}\otimes M_{b|y}^{B}\otimes\mathbb{1}\right)\rho_{bisep}^{ABC}\right]\nonumber \\
 & =\sum_{\nu}P_{\nu}P_{a|x,\nu}\sigma_{b|y,\nu}^{BC}+\sum_{\lambda}P_{\lambda}P_{b|y,\lambda}\sigma_{a|x,\lambda}^{AC}+\sum_{\mu}P_{\mu}P_{a,b|x,y,\lambda}\varrho_{\mu}^{C}\label{eq:bisep}
\end{align}
is when either Alice or Bob can steer Charlie's state (first and second
terms in the equation), but not collectively. The violation of such
a model defines genuine multipartite steering \cite{Cavalcanti2015a}.

\subsection{Post-quantum steering}

In obtaining Eqs. \eqref{eq:sep} and \eqref{eq:bisep} we used the
fact that the assemblage comes from a quantum state by performing
local quantum measurements, but by definition, steering occurs in
a semi-device independent scenario in which one does not have trustful
information about what measurements are being realized in the black-box
parties nor have access to the global quantum state. Possessing only
the Bell-behavior $P_{a,b|x,y}$ and the conditional states $\varrho_{a,b|x,y}$
satisfying the positivity condition $\sigma_{a,b|x,y}=P_{a,b|x,y}\varrho_{a,b|x,y}\geq0$,
the normalization condition $\Tr\left[\sum_{a,b}\sigma_{a,b|x,y}\right]=1$
for all $x,y$ and the no-signaling conditions \eqref{eq:NS}, one
may ask whether it is possible to find a quantum realization for such
assemblage. In other words, given a no-signaling assemblage $\sigma_{a,b|x,y}$,
is it always possible to find a tripartite state $\rho^{ABC}$ and
local measurement operators $\{M_{a|x}^{A}\}_{a,x}$ and $\{M_{b|y}^{B}\}_{b,y}$
such that $\sigma_{a,b|x,y}=\Tr_{AB}\left[\left(M_{a|x}^{A}\otimes M_{b|y}^{B}\otimes\mathbb{1}\right)\rho^{ABC}\right]$? 

In the bipartite case, the answer to this question is negative as
it is always possible to construct the bipartite state and measurement
operators yielding to any no-signaling assemblage. This is not true
for multipartite assemblages as is shown in \cite{Sainz2015}. A assemblage
for which there is no quantum realization is called postquantum. A
trivial example comes when one considers that the Bell-behavior $\{P_{a,b|x,y}\}_{a,b,x,y}$
alone has correlations stronger than the allowed by quantum theory,
as is the case of the Popescu-Rohrlich (PR) box where $P_{a,b|x,y}=1/2$
if $a\oplus b=xy$ and zero otherwise, $\oplus$ is sum modulo 2 \cite{Popescu1994}.
The authors of \cite{Sainz2015} also show cases for which there is
no postquantumness in the Bell-behavior, but still the assemblage
is postquantum, showing that it is an intrinsic feature of the assemblage
as a whole.

\subsection{Resource theory of steering }

Quantum steering is a resource for quantum information and can be
used for many tasks as mentioned before. Accordingly, a resource theory
for steering was built a few years ago \cite{Gallego2015}. For any
resource theory it is necessary to define the objects that do not
possess the resource, known as free states, and the operations that
take any free state into a free state, called free operations. In
the case of steering, the free states are those which admit a LHS
model. A useful set of free operations is the 1W-LOCCs (one way local
operations and classical communication). Consider a bipartition according
to the characterization of the parties, that is, all the DI parties
are grouped together in one partition as well as all the DD ones are
grouped in another partition. The initial assemblage $\sigma_{A|X}$
is transformed into the final assemblage $\sigma_{A_{f}|X_{f}}$.
The allowed operations that do not create steering are the following:
the quantum partition can perform stochastic generalized measurements
over her quantum system and communicate the result to the black-box
partition, which can realize the black box measurements and process
the classical information at disposal. The classical information processing
is called wiring. 

Two examples of free operations are shown in Fig. \ref{fig:Examples-of-free}.
In these examples, no quantum operation is realized in the quantum
partition and there is no classical communication from this partition
to the black boxes. In the black-box partition, classical information,
namely the classical inputs and outputs, is processed. It is intuitive
that these operations do not create the quantum correlation as they
are local in the black-box partition. In the first example in Fig.
\ref{fig:Examples-of-free}-a), Alice and Bob are no longer space-like
separated: she communicates her output to him and he uses this to
choose his input. This is an example of a bilocal \emph{wiring} (local
with respect to the bipartition $AB\vert C$). The tripartite assemblage
becomes equivalent to a bipartite one in the sense that Alice and
Bob work as only one black box with input $x$ and output $b$. In
the rest of this work we focus on this simple example and show that,
although such operations cannot create any correlations across the
bipartition, they can \emph{expose} a subtle form of multipartite
quantum nonlocality that otherwise does not violate any Bell or steering
inequality across the bipartition. In the second example in Fig. \ref{fig:Examples-of-free}-b),
a 4DI+1DD assemblage is mapped onto a 2DI+1DD one by a bilocal wiring
{[}$x_{2}=a_{3}$, $x_{3}=x_{4}$, and $a'_{1}=a_{1}\oplus a_{2}$
(sum modulo 2){]}, such that the final assemblage has only two classical
inputs ($x_{1}$ and $x_{3}$) and two outputs ($a_{1}^{\prime}$
and $a_{4}$). 

Such wirings can implement non-trivial resource-theoretic transformations.
One could ask whether there exists an $N$-partite assemblage with
$N-1$ black boxes and one quantum party from which all bipartite
ones can be produced, e.g., can be reached by means of reductions
on the number of inputs and outputs using classical information processing.
Below we this question in the negative. Although powerful, these wiring
operations are not enough to enable a multi-black-box universal \emph{steering
bit} even allowing for quantum operations and classical communication
from the DD party to the DI ones. This is formalized in the theorem
below whose demonstration is left to Appendix \ref{sec:No-go-theorem-for}.

\begin{thm}{[}No pure steering bit with higher number of parties{]}
\label{th:nobits-1} There does not exist any pure $(N-1)$-DI qubit
assemblage $\sigma_{\boldsymbol{a}|\boldsymbol{x}}^{\text{bit}}$,
where $\boldsymbol{a}=\{a_{1},...,a_{N-1}\}$, $\boldsymbol{x}=\{x_{1},...,x_{N-1}\}$
(with finite sets of input and output values), that can be transformed
via 1W-LOCCs into all qubit assemblages of minimal dimension $\sigma_{a|x}^{(\text{target})}$.
\end{thm}

\begin{figure}[h]
\begin{centering}
\includegraphics[width=0.9\textwidth]{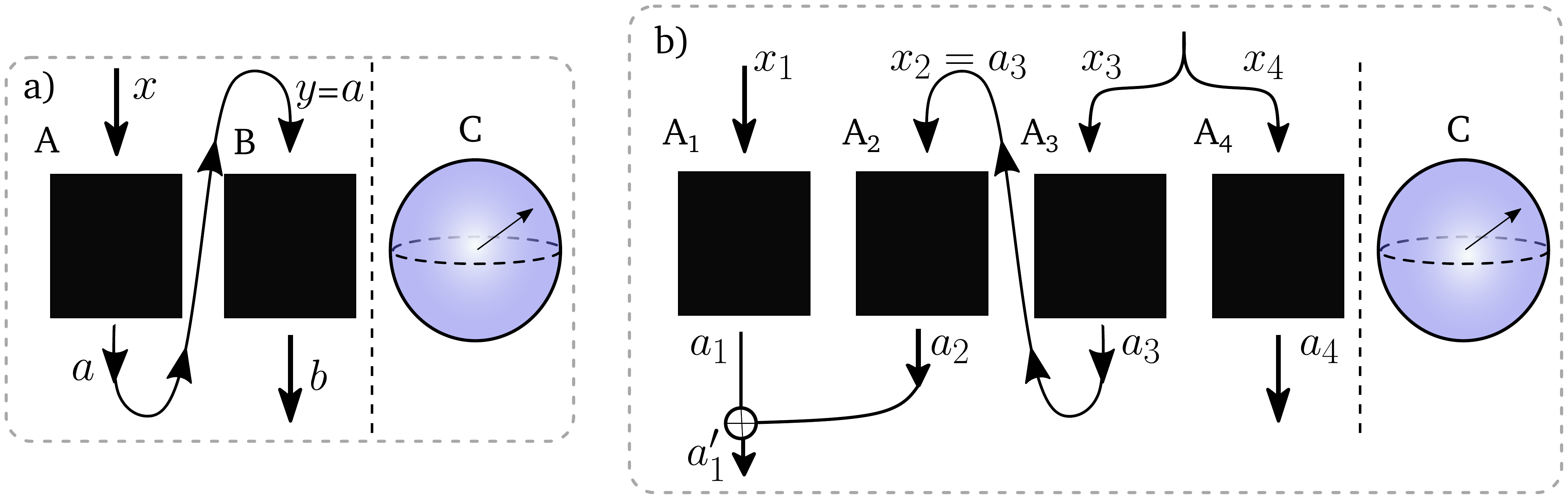}
\par\end{centering}
\caption{Examples of free operations of steering in scenarios with only one
trusted party.\label{fig:Examples-of-free}}

\selectlanguage{american}%
\end{figure}

\subsection{Steering detection, witnesses and quantifiers}

In order to detect whether a known assemblage is steerable or not,
one can directly use the definition of a non-steerable assemblage
as given by Eq. \eqref{eq:LocalHS} for bipartite assemblages or by
Eq. \eqref{eq:LHS} for a tripartite case. In this way, the problem
of steering detection amounts to search for states $\varrho_{\lambda}$
and probability distributions $P_{\lambda}$ and $P_{a|x,\lambda}$
(we consider here the bipartite case for simplicity) such that the
given assemblage can be described by a LHS model, if there do not
exist such mathematical objects, then the assemblage is steerable.
In principle, this is a hard problem since the sum in Eq. \eqref{eq:LocalHS}
has infinitely many terms. However, this problem can be stated in
a way that fits the semi-definite programming (SDP) paradigm, a class
of linear optimization problems over a convex set of positive semi-definite
operators known to be efficiently solvable with developed algorithms.

To begin with, let us consider the deterministic behaviors with one
input and one output, that is, the probability distributions such
that, given an input, it does not vanish for only one output. Given
that the black box has $N_{i}$ possible inputs and $N_{o}$ possible
outputs, the number of deterministic behaviors is $n=N_{o}^{N_{i}}$
and they are written as $D(a|x,\lambda^{\prime})=\delta_{a,\lambda^{\prime}(x)}$,
where $\lambda^{\prime}$ is defined as a string $\lambda^{\prime}=a_{x=0},...,a_{x=N_{i}-1}$
with the deterministic outputs related to each input. Any probability
distribution $P_{a|x,\lambda}$ can be written as a convex combination
of this extreme deterministic behavior with a weight that depends
on $\lambda$ 
\begin{equation}
P_{a|x,\lambda}=\sum_{\lambda^{\prime}=1}^{n}P_{\lambda^{\prime}|\lambda}D(a|x,\lambda^{\prime}).\label{eq:deterministic}
\end{equation}

Substituting Eq. \eqref{eq:deterministic} into Eq. \eqref{eq:LocalHS}
we get
\begin{equation}
\sigma_{a|x}=\sum_{\lambda^{\prime}=1}^{n}D(a|x,\lambda^{\prime})\sigma_{\lambda^{\prime}},
\end{equation}
a finite sum with only a finite number of positive semidefinite objects
$\sigma_{\lambda^{\prime}}=\sum_{\lambda}P_{\lambda^{\prime}|\lambda}P_{\lambda}\varrho_{\lambda}$
to be determined. In order to transform the problem of trying to find
the set $\{\sigma_{\lambda^{\prime}}\}_{\lambda^{\prime}}$ into an
optimization, one can define a number $\mu$ such that $\sigma_{\lambda^{\prime}}\geq\mu\mathbb{1}$.
Now the question of whether the known assemblage $\{\sigma_{a|x}\}_{a,x}$
is LHS can be written as the SDP
\begin{alignat}{1}
\text{given} & \quad\{\sigma_{a|x}\}_{a,x},\{D(a|x,\lambda^{\prime})\}_{a,x,\lambda^{\prime}}\nonumber \\
\underset{\{\sigma_{\lambda^{\prime}}\}}{\text{max}} & \quad\mu\nonumber \\
\text{s.t.} & \quad\sum_{\lambda^{\prime}=1}^{n}D(a|x,\lambda^{\prime})\sigma_{\lambda^{\prime}}=\sigma_{a|x}\quad\forall a,x\label{eq:xsdp}\\
 & \quad\sigma_{\lambda^{\prime}}\geq\mu\mathbb{1}\quad\forall\lambda^{\prime}\nonumber 
\end{alignat}
Put in this way, $\mu$ is allowed to be a negative number. The program
searches for the set $\{\sigma_{\lambda^{\prime}}\}_{\lambda^{\prime}}$
which maximizes $\mu$ keeping the LHS decomposition for the given
assemblage $\{\sigma_{a|x}\}_{a,x}$ valid. Because $\sigma_{\lambda^{\prime}}\geq0$
should be valid by definition, if the maximum value of $\mu$ is negative,
it means that there is no valid LHS decomposition for the assemblage
and therefore it is steerable. Otherwise, if the solution for $\mu$
is positive, then the assemblage is unsteerable and the program returns
its LHS decomposition. Here the problem is written for a bipartite
assemblage, but the same statement can be put forward for any number
of parties.

The SDP also has a so called dual program obtained by using Lagrange
multipliers for each of the constraints. The Lagrangian of the problem
is written as 
\begin{equation}
\mathcal{L}=\mu+\sum_{a,x}\Tr\left[w_{a|x}\left(\sigma_{a|x}-\sum_{\lambda}D(a|x,\lambda)\sigma_{\lambda}\right)\right]+\sum_{\lambda}\Tr\left[z_{\lambda}\left(\sigma_{\lambda}-\mu\mathbb{1}\right)\right],
\end{equation}
where the first term on the right hand side is the function to be
maximized, the second one is related to the first set of constraints
in \eqref{eq:xsdp} and vanishes in the optimal point, the third term
represents the second set of constraints in \eqref{eq:xsdp} and is
always bigger than zero if we impose that $z_{\lambda}\geq0$ for
all $\lambda$. Thus, the Lagrangian $\mathcal{L}$ is always bigger
than $\mu$ and its minimum value serves as a good upper bound for
$\mu$, in most problems the minimum value of $\mathcal{L}$ is actually
equal to the maximum value of $\mu$. Before writing the dual problem
as a minimization, some further simplification is possible. Let us
rewrite the Lagrangian grouping the terms related to each optimization
variable of the primal SDP
\begin{equation}
\mathcal{L}=\sum_{a,x}\Tr\left[w_{a|x}\sigma_{a|x}\right]-\sum_{\lambda}\Tr\left[\left(\sum_{a,x}w_{a|x}D(a|x,\lambda)-z_{\lambda}\right)\sigma_{\lambda}\right]+\left(1-\Tr\left[\sum_{\lambda}z_{\lambda}\right]\right)\mu.
\end{equation}
In the minimum point, the gradient of $\mathcal{L}$ vanishes, thus
the partial derivatives relative to each $\sigma_{\lambda}$ and also
relative to $\mu$ are also zero. It means that the coefficients accompanying
each of these variables must vanish due to the linearity of the Lagrangian
and the independence of the variables. Explicitly, the optimal point
satisfies 
\begin{eqnarray}
\Tr\left[\sum_{a,x}w_{a|x}D(a|x,\lambda)-z_{\lambda}\right] & = & 0\label{eq:condicao1}\\
1-\Tr\left[\sum_{\lambda}z_{\lambda}\right] & = & 0.\label{eq:condicao}
\end{eqnarray}
By substituting \eqref{eq:condicao} in \eqref{eq:condicao1} and
using the restriction that $z_{\lambda}\geq0$ for all $\lambda$,
we get the conditions for the dual minimization problem that eliminate
the variables of primal problem: $\Tr\left[\sum_{a,x,\lambda}w_{a|x}D(a|x,\lambda)\right]=1$
and $\sum_{a,x}w_{a|x}D(a|x,\lambda)\geq0$. We can finally write
down the dual problem as the minimization of the remaining term of
the Lagrangian
\begin{alignat}{1}
\text{given} & \quad\{\sigma_{a|x}\}_{a,x},\{D(a|x,\lambda)\}_{a,x,\lambda}\nonumber \\
\underset{\{w_{a|x}\}}{\text{min}} & \quad\sum_{a,x}\Tr\left[w_{a|x}\sigma_{a|x}\right]\nonumber \\
\text{s.t.} & \quad\Tr\left[\sum_{a,x,\lambda}w_{a|x}D(a|x,\lambda)\right]=1\label{eq:xsdp-1}\\
 & \quad\sum_{a,x}w_{a|x}D(a|x,\lambda)\geq0\quad\forall\lambda.\nonumber 
\end{alignat}
As the minimum of $\mathcal{L}$ coincides with the maximum of $\mu$,
only if the assemblage is steerable, the optimal value of the dual
program is negative. Moreover, in this case, the solutions $\{w_{a|x}\}_{a,x}$
define a steering witness, i.e., an inequality that, if violated,
guarantees that the assemblage is steerable. To see that, let us consider
again the decomposition of an LHS assemblage $\sigma_{a|x}^{\prime}=\sum_{\lambda}D(a|x,\lambda)\sigma_{\lambda}$
, multiply the second condition in the SDP \eqref{eq:xsdp-1} by $\sigma_{\lambda}$
and sum over $\lambda$. This results in $\sum_{a,x}w_{a|x}\sum_{\lambda}D(a|x,\lambda)\sigma_{\lambda}\geq0$
or 
\[
\Tr\left[\sum_{a,x}w_{a|x}\sigma_{a|x}^{\prime}\right]\geq0
\]
for all LHS assemblages. In particular, this inequality is violated
by the steerable assemblage that generated $\{w_{a|x}\}_{a,x}$ since
the left hand side is precisely the negative minimized Lagrangian.
The inequality can still be manipulated to change the bound or the
direction of it.

Although the above methods serve to detect steering they do not quantify
it. A good quantifier of steering must vanish for any LHS assemblage
and must not increase under 1W-LOCC operations. One such quantifier
is the steering robustness which can also be written as an SDP \cite{Piani2015}.
The steering robustness is defined as the minimum amount of an LHS
assemblage or, in other words, the minimum amount of noise that must
be mixed to the assemblage under question such that it becomes LHS,
that is 
\begin{align}
R(\sigma_{a|x}) & = & \underset{\{\xi_{\lambda}\},\{\sigma_{\lambda}\}}{\text{min}}\quad r\nonumber \\
 &  & \text{s.t.}\quad & \frac{\sigma_{a|x}+r\sum_{\lambda}D(a|x,\lambda)\xi_{\lambda}}{1+r}=\sum_{\lambda}D(a|x,\lambda)\sigma_{\lambda}\quad\forall a,x\nonumber \\
 &  &  & \sigma_{\lambda},\xi_{\lambda}\geq0\quad\forall\lambda,
\end{align}
which is clearly zero if the assemblage $\sigma_{a|x}$ already has
an LHS decomposition. Finding the robustness of an assemblage is an
optimization problem but it is not a SDP because it is not even linear.
In order to linearize it, we define $\xi_{\lambda}^{\prime}=r\xi_{\lambda}$
and $\sigma_{\lambda}^{\prime}=(1+r)\sigma_{\lambda}$ such that we
have $\sigma_{a|x}=\sum_{\lambda}D(a|x,\lambda)\left(\sigma_{\lambda}^{\prime}-\xi_{\lambda}^{\prime}\right)$
and $r=\Tr\left[\sum_{\lambda}\xi_{\lambda}^{\prime}\right]$ due
to the normalization $\sum_{\lambda}\Tr\left[\xi_{\lambda}\right]=1$.
The robustness now can be found by solving the SDP 
\begin{align}
R(\sigma_{a|x}) & = & \underset{\{\xi_{\lambda}^{\prime}\},\{\sigma_{\lambda}^{\prime}\}}{\text{min}}\quad & \Tr\left[\sum_{\lambda}\xi_{\lambda}^{\prime}\right]\nonumber \\
 &  & \text{s.t.}\quad & \sigma_{a|x}=\sum_{\lambda}D(a|x,\lambda)\left(\sigma_{\lambda}^{\prime}-\xi_{\lambda}^{\prime}\right)\quad\forall a,x\nonumber \\
 &  &  & \sigma_{\lambda}^{\prime},\xi_{\lambda}^{\prime}\geq0\quad\forall\lambda.\label{eq:rob}
\end{align}
Here we choose an LHS noise, but in fact the robustness can be defined
relatively to any subset or even to the whole convex set of assemblages
\cite{Sainz2016a}, for this reason it is better to call \eqref{eq:rob}
LHS-robustness. The motive for the name of this quantifier is obvious:
the larger the noise that must be added to extinguish the steering,
the more robust is the steering present in the assemblage.

All the methods and quantities discussed in this section can be extended
to more party assemblages. Although this quantities are well defined,
applying them to experimentally recovered assemblages to determine
if it is LHS or not may be challenging since this assemblages are
not even non-signaling in general. We discuss this issue later in
this chapter.

\subsection{Assemblage Fidelity}

To quantify the similarity between two assemblages $\boldsymbol{\sigma}_{1}=\{P_{1}(\mathbf{a}|\mathbf{x})\varrho_{1}(\mathbf{a}|\mathbf{x})\}$
and $\boldsymbol{\sigma}_{2}=\{P_{2}(\mathbf{a}|\mathbf{x})\varrho_{2}(\mathbf{a}|\mathbf{x})\}$,
we use a mean assemblage fidelity defined by 
\begin{equation}
F(\boldsymbol{\sigma}_{1},\boldsymbol{\sigma}_{2})=\frac{1}{N_{x}}\sum_{\mathbf{x},\mathbf{a}}\sqrt{P_{1}(\mathbf{a}|\mathbf{x})P_{2}(\mathbf{a}|\mathbf{x})}\mathcal{F}\left(\varrho_{1}(\mathbf{a}|\mathbf{x}),\varrho_{2}(\mathbf{a}|\mathbf{x})\right),
\end{equation}
where $\mathbf{x}$ ($\mathbf{a}$) is a list of inputs (outputs)
of all black boxes, $N_{x}$ is the number of different measurement
choices, and 
\begin{equation}
\mathcal{F}(\varrho_{1},\varrho_{2})=\Tr\sqrt{\sqrt{\varrho_{1}}\varrho_{2}\sqrt{\varrho_{1}}}\label{eq:fidelity}
\end{equation}
 is the usual fidelity between two quantum states. The above defined
fidelity can be seen as a mean of the fidelities of the quantum parts
weighted by the square root of blackbox probabilities. It has the
property of being $1$ if all elements of the two assemblages are
equal and vanishes if all quantum states are orthogonal.

\section{Steering exposure and super exposure of Bell-nonlocality}

The main result of this work is to show that the current definition
of multipartite quantum steering as the violation of model \eqref{eq:LHS}
is deficient since it presents inconsistencies with the resource theory
of steering. We begin by an exposure protocol for steering (Bell nonlocality)
that is universal in the sense of being capable of producing any bipartite
assemblage (behavior) whatsoever from an appropriate tripartite assemblage
(behavior) originally admitting an LHS (LHV) model. As in Ref. \cite{Gallego2012},
we exploit bilocal wirings as that of Fig.\,\ref{fig:Examples-of-free}-a),
which makes Bob's input $y$ equal to Alice's output $a$. This requires
that Bob's measurement is in the causal future of Alice's. Indeed,
after the wiring, systems $A$ and $B$ now behave as a single black
box with input $x$ and output $b$. In other words, exposure is a
form of conversion from tripartite correlations into bipartite ones.
Here, we restrict to the case of binary inputs and outputs (each one
can take only two values) for simplicity, where we prove the following
surprising result.

\noindent \vspace{0.2cm}
 Universal exposure of quantum nonlocality: \textit{ Any bipartite
assemblage $\boldsymbol{\sigma}^{(\text{target})}$ or Bell behavior
$\boldsymbol{P}^{(\text{target})}$ can be obtained via the wiring
$y=a$ on the tripartite assemblage $\boldsymbol{\sigma}^{(\text{initial})}$
or behavior $\boldsymbol{P}^{(\text{initial})}$, respectively, of
elements} \begin{subequations} \label{eq:universal} 
\begin{equation}
\sigma_{a,b|x,y}^{(\text{initial})}:=\frac{1}{2}\sigma_{b|x\oplus a\oplus y}^{(\text{target})}\ \label{eq:genactiv}
\end{equation}
\textit{or} 
\begin{equation}
P^{(\text{initial})}(a,b,c|x,y,z)=\frac{1}{2}P^{(\text{target})}(b,c|x\oplus a\oplus y,z)\ ,\label{eq:genblackboxactiv}
\end{equation}
\textit{where $\oplus$ stands for addition modulo 2. Moreover, $\boldsymbol{\sigma}^{(\text{initial})}$
and $\boldsymbol{P}^{(\text{initial})}$ admit respectively an LHS
and an LHV models across the $AB|C$ bipartition, for all $\boldsymbol{\sigma}^{(\text{target})}$
and $\boldsymbol{P}^{(\text{target})}$.} \end{subequations} \vspace{0.2cm}

\begin{proof}
It is straightforward to check that applying the wiring $y=a$ to
Eqs.\ \eqref{eq:genactiv} and \eqref{eq:genblackboxactiv}, the
target assemblage and behavior are obtained, i.e., $\sum_{a}\sigma_{a,b|x,y=a}^{(\text{initial})}=\sigma_{b|x}^{(\text{target})}$
and $\sum_{a}P^{(\text{initial})}(a,b,c|x,y=a,z)=P^{(\text{target})}(b,c|x,z)$.\par 
It is certainly not evident that the initial correlations are bilocal.
To prove this, we construct an explicit LHS model for the source assemblage
$\boldsymbol{\sigma}^{(\text{initial})}$. It is given by 
\begin{align}
P_{\lambda}=\frac{1}{2}\Tr\left(\sigma_{\lambda_{0}|\lambda_{1}}^{(\text{target})}\right),\ \  & \varrho_{\lambda}=\frac{\sigma_{\lambda_{0}|\lambda_{1}}^{(\text{target})}}{\Tr\left(\sigma_{\lambda_{0}|\lambda_{1}}^{(\text{target})}\right)},\label{eq:LHSorig1}\\
P_{a,b|x,y,\lambda} & =\delta_{\lambda_{0},b}\ \delta_{\lambda_{1},x\oplus a\oplus y}\ ,\label{eq:LHSorig2}
\end{align}
\label{eq:LHSorig} where $\lambda=(\lambda_{0},\lambda_{1})$ is
a two-bit hidden variable.\par For the Bell behavior, this expression
readily lends itself for a local hidden-variable decomposition of
$\boldsymbol{P}^{(\text{initial})}$ on $AB|C$, $P^{(\text{initial})}(a,b,c|x,y,z)=\sum_{\lambda}P_{\lambda}P_{a,b|x,y,\lambda}P(c|z;\lambda)$
with 
\begin{align}
P_{\lambda}=\frac{1}{2};\ \ \  & P(c|z;\lambda)=P^{(\text{target})}(\lambda_{0},c|\lambda_{1},z);%
\label{eq:LHVorig1}
\end{align}
and the same bipartite distribution from Eq.\ \eqref{eq:LHSorig2}.
\end{proof}

\noindent When the target assemblage (behavior) is steerable (Bell
nonlocal), exposure of steering (Bell nonlocality) is achieved. Furthermore,
apart from steerable, assemblages can also be Bell nonlocal in the
sense of giving rise to nonlocal behaviors under local measurements
\cite{Taddei2016}. Hence, when $\boldsymbol{\sigma}^{(\text{target})}$
is Bell nonlocal, a seemingly unsteerable system is mapped onto a
Bell nonlocal one, which is outstanding in view of the fact that unsteerable
assemblages form a strict subset of Bell-local ones.

The protocol above highlights the power of bilocal wirings as resource-theoretic
transformations. Remarkably, such wirings compose a strict subset
of well-known classes of free operations of quantum nonlocality (across
$AB|C$): local operations with classical communication (LOCCs) \cite{Horodecki2009}
for entanglement, one-way (1W) LOCCs from the trusted to the untrusted
parts \cite{Gallego2015} for steering, and local operations with
shared randomness \cite{Gallego2012,deVicente2014,Gallego2017} for
Bell nonlocality. However, there are also limitations to the power
of these wirings. In particular, in Supplementary Section VI we prove
a no-go theorem for universal steering bits in the $N$DI+1DD scenario
{[}exemplified in Fig.\ \ref{fig:Examples-of-free}-b) for $N=4${]}.
That is, we show there that there is no $N$-partite assemblage, for
all $N$, from which all bipartite ones can be obtained via arbitrary
1W-LOCCs.

Although the protocol above is universal, it is unclear whether it
can actually be physically implemented in general. This is due to
the fact that the tripartite initial correlations may be supra-quantum,
i.e.\ well-defined non-signaling correlations that can however not
be obtained from local measurements on any quantum state \cite{Sainz2015,Sainz2018a,Sainz2019,Popescu1994}.
Physical protocols for Bell-nonlocality exposure were devised in Refs.\ \cite{Gallego2012,Bancal2013},
but no such protocols have been reported for steering. Hence, we next
show an example for both steering exposure and Bell-nonlocality super-exposure
that is manifestly within quantum theory. This also exploits the bilocal
wirings of Fig.\ \ref{fig:Examples-of-free}-a), but starting from
a different initial assemblage. We describe the latter directly in
terms of its quantum realization. 

Consider a tripartite Greenberg-Horne-Zeilinger (GHZ) state $(|000\rangle+|111\rangle)/\sqrt{2}$,
with $\ket{0}$ and $\ket{1}$ the eigenvectors of the third Pauli
matrix $Z$. Bob makes von Neumann measurements on his share of the
state for both his inputs, for $y=0$ in the $Z+X$ basis and for
$y=1$ in the $Z-X$ basis, with $X$ the first Pauli matrix. Alice,
however, makes either a trivial measurement, given by the positive
operator-valued measure $\{\mathbb{1}/2,\mathbb{1}/2\}$, for $x=0$,
or a von Neumann $X$-basis measurement, for $x=1$. For the resulting
initial assemblage, $\boldsymbol{\sigma}^{(\text{GHZ})}$, the following
holds.

\vspace{0.2cm}
 Physically-realizable exposure and super-exposure:\emph{ The quantum
assemblage} $\boldsymbol{\sigma}^{(\text{GHZ})}$\emph{, of elements}
\begin{equation}
\sigma_{a,b|x,y}^{(\text{GHZ})}
=\frac{1}{8}\left\{ \mathbb{1}+\frac{(-1)^{b}}{\sqrt{2}}\left[Z+x(-1)^{a+y}X\right]\right\} \ \label{eq:quantumorig}
\end{equation}
\emph{admits an LHS model and, under the wiring $y=a$, is mapped
to the assemblage of elements 
\begin{equation}
\sigma_{b|x}
=\frac{1}{4}\left[\mathbb{1}+\frac{(-1)^{b}}{\sqrt{2}}\left(Z+xX\right)\right]\ ,\label{eq:quantumsteerable}
\end{equation}
which is both steerable and Bell-nonlocal.} \vspace{0.2cm}

Equation \eqref{eq:quantumorig} can be obtained in the same way as
in Eq. \eqref{eq:sep}, but substituting $\rho_{sep}^{ABC}$ by the
GHZ state. It is also straightforwad to show that the resulting wired
assemblage is that of Eq.\ \eqref{eq:quantumsteerable}. Now, we
proceed to prove that the physically-realizable source assemblage
$\boldsymbol{\sigma}^{(\text{GHZ})}$ in Eq.\ \eqref{eq:quantumorig}
admits an LHS model across the bipartition $AB|C$, and that the latter
is both steerable and Bell nonlocal. \begin{proof}
The LHS decomposition for Eq.\ \eqref{eq:quantumorig} is found via
solving the SDP \eqref{eq:xsdp}. The numerical results in this case
allow one to find analytic formulas for the decomposition, namely
\begin{align}
P_{\lambda}=\frac{1}{4};\ \ \ \varrho_{\lambda} & =\frac{\mathbb{1}}{2}+\frac{(-1)^{\lambda_{0}}}{2\sqrt{2}}\left[Z+(-1)^{\lambda_{1}}X\right];\label{eq:quantumLHSdecomp1}\\
P_{a,b|x,y,\lambda} & =\delta_{\lambda_{0},b}\frac{1+x(-1)^{a+y+\lambda_{1}}}{2}\ ,\label{eq:quantumLHSdecomp2}
\end{align}
 where again $\lambda=(\lambda_{0},\lambda_{1})$ is a two-bit hidden
variable.\par Let us now prove the steerability and Bell-nonlocality
of assemblage \eqref{eq:quantumsteerable}. Steerability: with the
SDP \eqref{eq:xsdp-1}, we have obtained an assemblage-like object
$W=\{w_{a|x}\}_{a,x}$ that serves as a steering witness, i.e.\ it
establishes the inequality $\sum_{a,x}\Tr\left[w_{a|x}\sigma_{a|x}\right]\leqslant1$,
which can only be violated if assemblage $\boldsymbol{\sigma}=\{\sigma_{a|x}\}_{a,x}$
is steerable. Optimized for assemblage \eqref{eq:quantumsteerable},
the witness returns a value of $1.0721$ and can be cast as 
\begin{equation}
w_{0|0}=\begin{bmatrix}p & -c\\
-c & 1-p
\end{bmatrix},\quad w_{0|1}=\begin{bmatrix}q & p/2\\
p/2 & -q
\end{bmatrix},\label{eq:witnessW}
\end{equation}
with $p=\frac{1}{2}(1+\frac{1}{\sqrt{5}}),\,c\approx0.1382,\,q\approx0.2236$,
and $w_{1|x}=Y\,w_{0|x}\,Y,\,x=0,1$. Bell-nonlocality: The necessary
and sufficient criterion from \cite{Taddei2016} yields an optimal
violation of the Clauser-Horne-Shimony-Holt (CHSH) inequality of $|-\frac{\sqrt{5}+1}{\sqrt{2}}|\approx2.29\nleqslant2$,
attained when Charlie makes von Neumann measurements in the eigenbases
of $2Z+X$ and $X$. \end{proof}

\section{Experimental implementation}

Because of experimental imperfections and even the finite statistics
inherent to the experimental data, it could be the case that, although
we have a quantum realizable example of steering exposure, we would
be unable to determine that the exposure has happened. In this section,
we present an implementation of that example and show that the exposure
of steering and Bell-nonlocality is a detectable phenomenon. 

The exposure procedure was experimentally implemented using entangled
photons produced via spontaneous parametric down conversion. The experimental
setup is shown in Fig.\,\ref{fig:setup}. A photon pair is generated
in the Bell state $|\Phi^{+}\rangle=\left(|00\rangle+|11\rangle\right)/\sqrt{2}$,
where $|0\rangle$ ($|1\rangle$) stands for horizontal (vertical)
polarization of the photons \cite{Kwiat99}. The photons in the signal
mode ($s$) pass through a calcite beam displacer (BD), which creates
two momentum modes (paths) depending on the polarization. This results
in a tripartite GHZ state, where the extra qubit is the path degree
of freedom of the photons in $s$. Alice's and Bob's qubits are the
polarization and path of the photons in mode $s$, respectively, while
Charlie's qubit is the polarization of the photon in mode $i$. The
measurements onto all the degrees of freedom required for the assemblage
production and tomography are performed as described below.

\begin{figure}[h]
\centering \includegraphics[width=0.8\columnwidth]{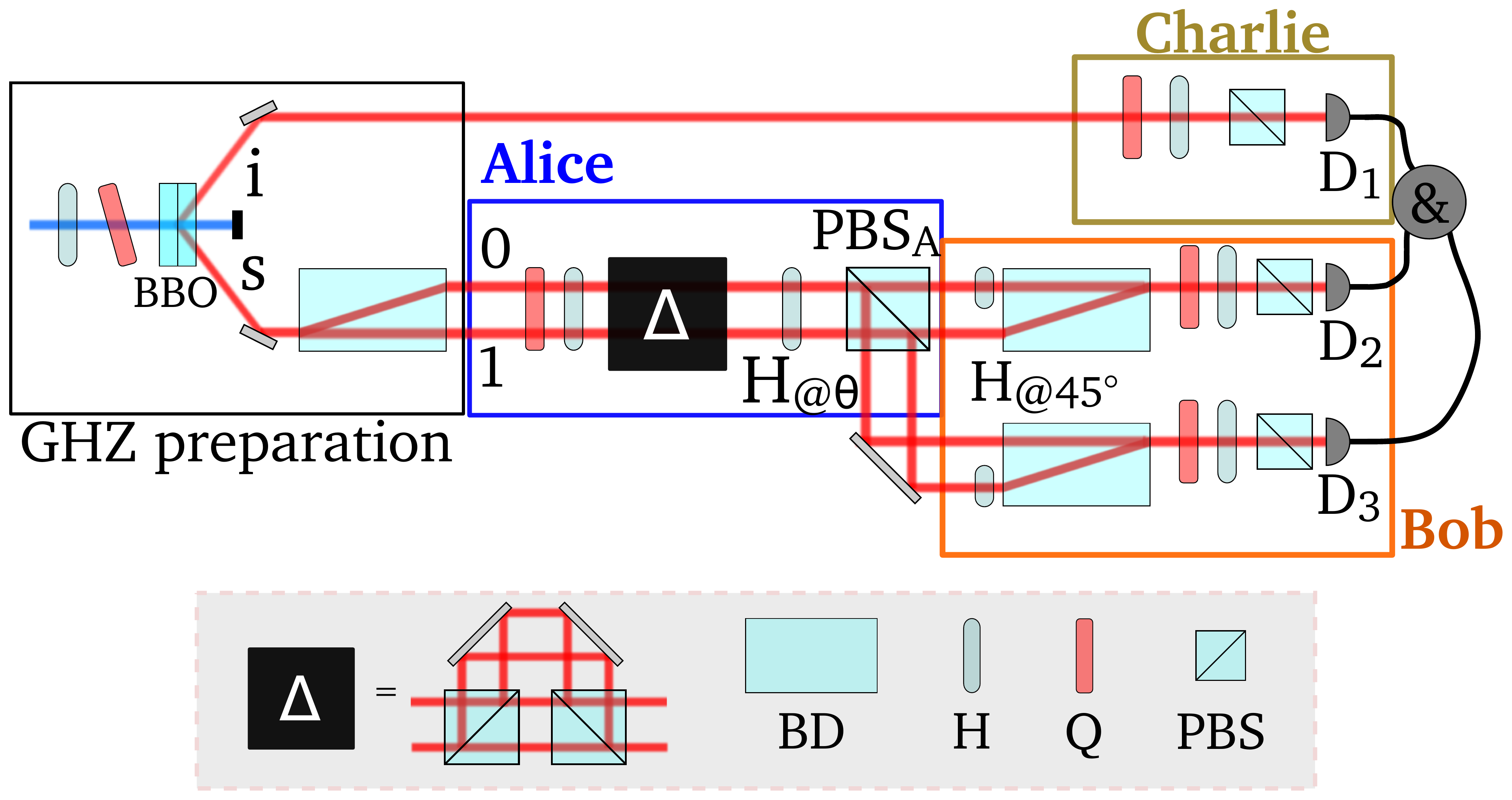}
\caption{Experimental setup. Two crossed-axis BBO crystals are pumped by a He-Cd laser centered at 325 nm, producing pairs of photons at 650 nm entangled in the polarization degree of freedom \protect\cite{Kwiat99}. The signal ($s$) photon is sent through a BD which deviates only the horizontal-polarization component, producing a tripartite GHZ state on two photons using polarization and path degrees of freedom. Idler ($i$) photons are sent directly to Charlie's polarization measurements. Signal photons are first measured in polarization by Alice, then Bob maps his path qubit onto a polarization qubit for his measurements. $H$ stands for half-wave plate, $Q$ for quarter-wave plate and $PBS$ for polarizing beam splitter.}
\label{fig:setup} 
\end{figure}

To implement the wiring from Fig. \ref{fig:Examples-of-free}-a),
Alice's polarization measurements are realized before Bob's measurements
onto the path degree of freedom. Alice's results are read from the
output of PBS$_{A}$, which determines whether D$_{2}$ ($a=0$) or
D$_{3}$ ($a=1$) clicks. For Alice's trivial measurement ($x=0$),
crucial for the original assemblage to be LHS-decomposable, both her
wave plates located before the imbalanced interferometer (represented
by $\Delta$) are kept at $0^{\circ}$, and H$_{@\theta}$ is adjusted
to $22.5^{\circ}$. The role of $\Delta$ is to remove the coherence
between horizontal and vertical polarization components, ensuring
that the photon exits PBS$_{A}$ randomly, independent of the input
polarization state. To see that it is indeed implementing the desired
measurement, consider an arbitrary pure polarization state $a\ket{H}+b\ket{V}$
entering $\Delta$. We can associate orthogonal states $\ket{s}$
and $\ket{l}$ for the photon going through the short and long paths
of the interferometer, respectively. As horizontally (vertically)
polarized photons take the short (long) path, the state of the photons
is $a\ket{H}\ket{s}+b\ket{V}\ket{l}$ inside the interferometer. After
exiting the interferometer, the paths recombine, but because of the
incoherence introduced between the two paths, the effect is equal
to tracing out the path degree of freedom obtaining the mixed state
$|a|^{2}\ket{H}\bra{H}+|b|^{2}\ket{V}\bra{V}$. Lastly, the photons
pass through the HWP after $\Delta$ which transforms horizontal (vertical)
into diagonal (anti-diagonal) polarization delivering the state $\frac{1}{2}\left[\ket{H}\bra{H}+\ket{V}\bra{V}+\left(|a|^{2}-|b|^{2}\right)\left(\ket{H}\bra{V}+\ket{V}\bra{H}\right)\right]$
which, regardless of $a$ and $b$, gives probability $1/2$ for detecting
the photon in horizontal or vertical polarization, going randomly
to D$_{2}$ or D$_{3}$. 

For $x=1$, Alice's wave plates before $\Delta$ are set to project
the polarization on the $X$ eigenstates, such that the interferometer
and H$_{@\theta}$ ($\theta=0^{\circ}$) play no role. Bob performs
his projective measurements by first mapping the path degrees of freedom
onto polarization using BDs and then projecting the polarization state
using his set of wave plates and PBSs, as was realized in Ref.\ \cite{Farias2012}
and described in Section \ref{subsec:Path-degree-of}. To reconstruct
the assemblage in Eq.\ \eqref{eq:quantumorig}, measurements for
$y=0$ and $y=1$ are made in both detectors D$_{2}$ and D$_{3}$,
by varying the angle of the wave plates in Bob's box. To collect the
data corresponding to the wired assemblage \eqref{eq:quantumsteerable}
only the $y=0$ measurement is made in D$_{2}$ ($a=0$) and only
$y=1$ is made in D$_{3}$ ($a=1$), enforcing that Bob's input equals
Alice's output ($y=a$).

Although we treat two of the qubits as black boxes, in order to ensure
that the resulting assemblage is generated by quantum measurements
performed onto a GHZ, we first performed state tomography to determine
the tripartite quantum state. This can be done without adding any
optical element to the setup. By varying the angles on Alice's quarter-wave
plate and half-wave plate before the unbalanced interferometer, we
set her apparatus to make any tomographic measurement in polarization
if we set $H_{@\theta}$ to $0^{\circ}$. The tomographic projections
for the path degree of freedom of photons in $s$ and polarization
of photons in $i$ is done using the set of wave plates just before
detectors D$_{1}$ and D$_{2}$, respectively. Using the collected
coincidence counts we reconstructed the tripartite quantum state by
maximum likelihood. The reconstructed density matrix is shown in Fig.\,\ref{fig:ghz}.
The experimental state presents fidelity (Eq. \eqref{eq:fidelity})
with GHZ state equals to $0.981\pm0.004$.

\begin{figure}[tb]
\centering \includegraphics[width=0.5\columnwidth]{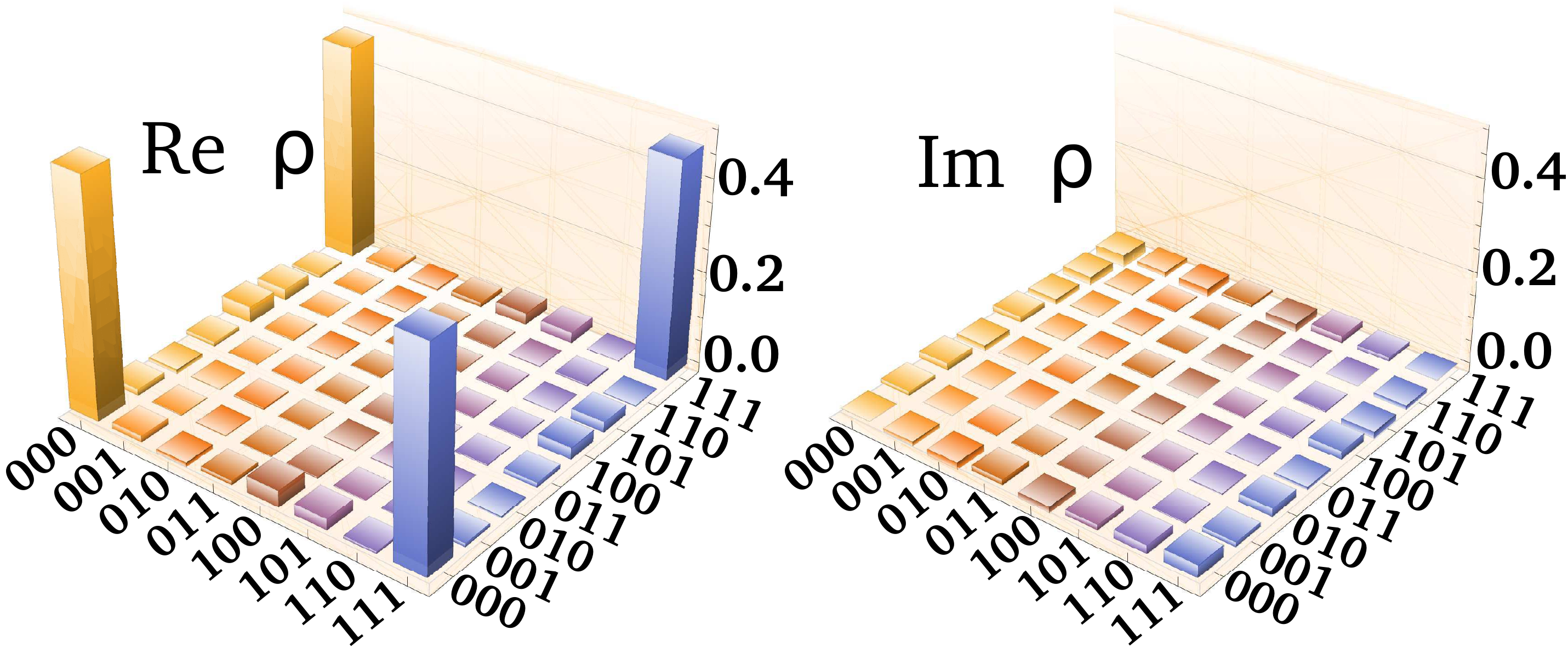}
\caption{Real and imaginary parts of the experimental reconstructed GHZ state.}
\label{fig:ghz} 
\end{figure}

The assemblage $\boldsymbol{\sigma}^{(\text{GHZ})}$ was obtained
experimentally by performing state tomography on Charlie's system
for each measurement setting and outcome of Alice and Bob. Sixteen
density matrices are obtained through maximum likelihood. Each element
of the tripartite assemblage is composed of Charlie's conditional
quantum state and the conditional probability $P_{a,b|x,y}$ for the
black boxes. All sixteen experimental density matrices of Charlie
are shown in Fig.\,\ref{fig:assembtrip} in comparison with the corresponding
theoretical ones. The associated conditional probabilities are also
shown. The assemblage presents a fidelity-like measure of $98.2\pm0.2\%$
compared to the theoretical one. 

\begin{figure}[!h]
\centering \includegraphics[width=0.8\textwidth]{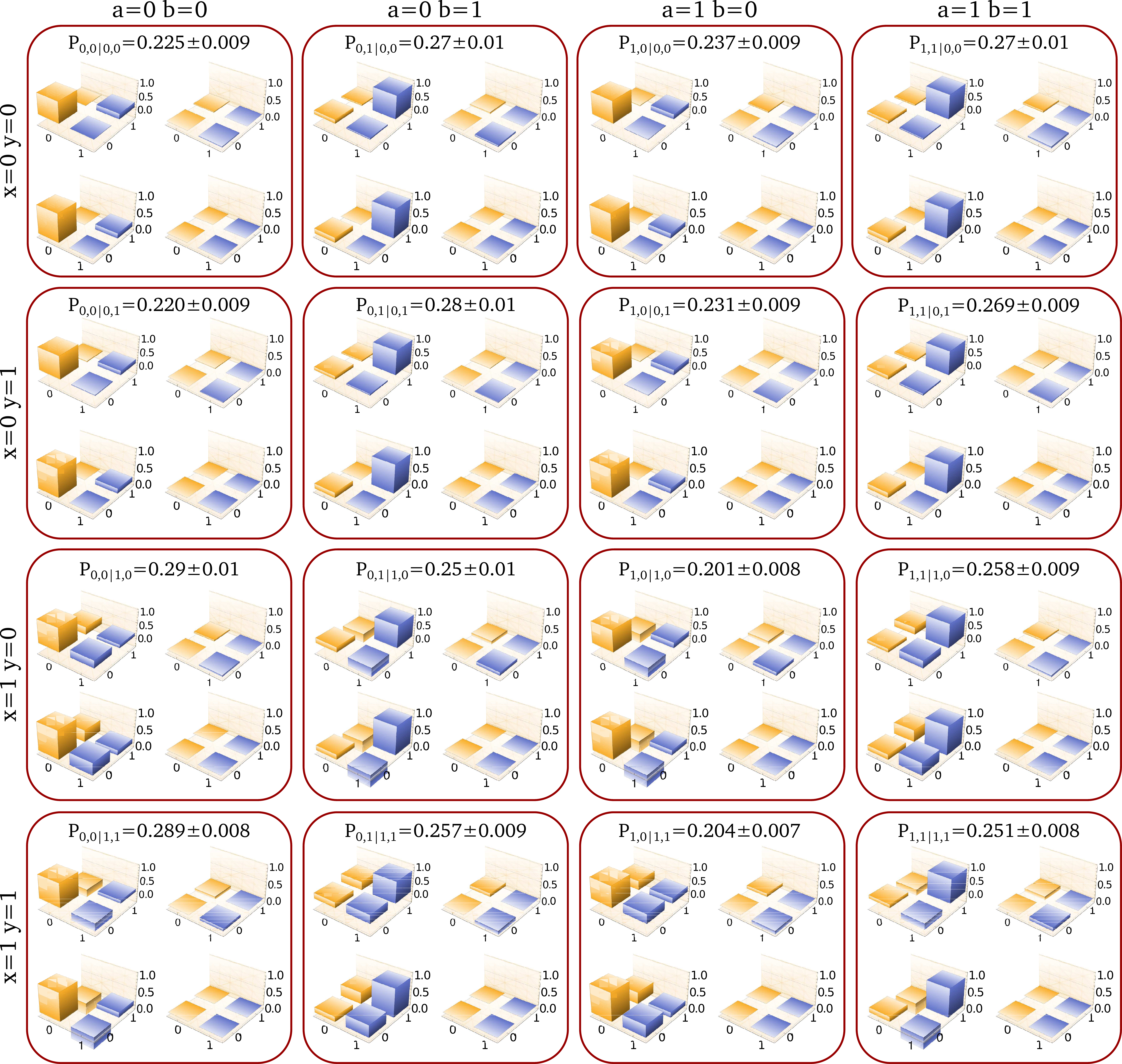}
\caption{Theoretical and experimental reconstructed assemblages for different
values of inputs $x,y$ and outputs $a,b$. Each box shows the joint
probability of measurement for the black boxes, real (top left) and
imaginary (top right) parts of the experimental density matrix of
Charlie's partition, and real (bottom left) and imaginary (bottom
right) parts of theoretical Charlie's density matrix. The theoretical
probability is $0.25$ for all measurement choices and measurement
outputs.}
\label{fig:assembtrip} 
\end{figure}

The experimental wired assemblage is shown in Fig.\ \ref{fig:wiredass}
a). For the wired assemblage, the expected conditional probability
of each outcome is $\frac{1}{2}$; the experimental values are $0.46\pm0.01$,
$0.54\pm0.01$, $0.49\pm0.01$, $0.51\pm0.01$ (following the order
in Fig.\ref{fig:wiredass}a). The imaginary components of the density
matrix average to $0.05\pm0.02$ (theoretical: zero). and returns
a fidelity of $98.1\pm0.6\%$ with respect to the theoretical wired
assemblage given in \eqref{eq:quantumsteerable}.

An exact LHS decomposition of the experimental assemblage is not feasible
due to imperfections and finite statistics --- in fact, assemblages
reproducing raw experimental data exactly are not even physical, since
they disobey the NS principle \cite{Cavalcanti2015a}. To show that
the experimental tripartite assemblage is statistically compatible
with an LHS decomposition, we proceed as follows: First, we assume
the photocounts obtained for each measured projector are averages
of Poisson distributions; with a Monte Carlo simulation, we sample
many times each of these distributions and reconstruct the corresponding
assemblages. Second, for each reconstructed assemblage, we find the
physical (NS) assemblage that best approximates it through maximum-likelihood
estimation, as well as the best LHS approximation for comparison.
As an initial indication of LHS-compatibility, the log-likelihood
error of both approximations is extremely similar, see Fig.\ref{fig:wiredass}c).
Third, for the NS approximations we calculate the LHS-robustness .
For comparison, we repeat the procedure starting with simulated finite-photocount
statistics from the theoretical LHS assemblage from Eq.\ \eqref{eq:quantumorig}.
In Fig.\ref{fig:wiredass}d) we see that the experimental robustness
has a sizable zero component and a distribution fully compatible with
that of an LHS assemblage under finite measurement statistics. To
show that the experimental wired assemblage is steerable, we tested
it on the optimal steering witness $W$ with respect to assemblage
\eqref{eq:quantumsteerable} (Eq. \eqref{eq:witnessW}). This returned
a value $1.015\pm0.009\nleqslant1$ (theoretical: $1.0721\nleqslant1$),
where the inequality violation implies steering, see Fig.\ref{fig:wiredass}b).
This allows us to conclude that the bipartite wired assemblage is
indeed steerable. The experimental error was calculated using 500
assemblages also from a Monte Carlo simulation of measurement results
with Poisson photocount statistics.

\begin{figure}[H]
\centering{}\centering \includegraphics[width=0.5\columnwidth]{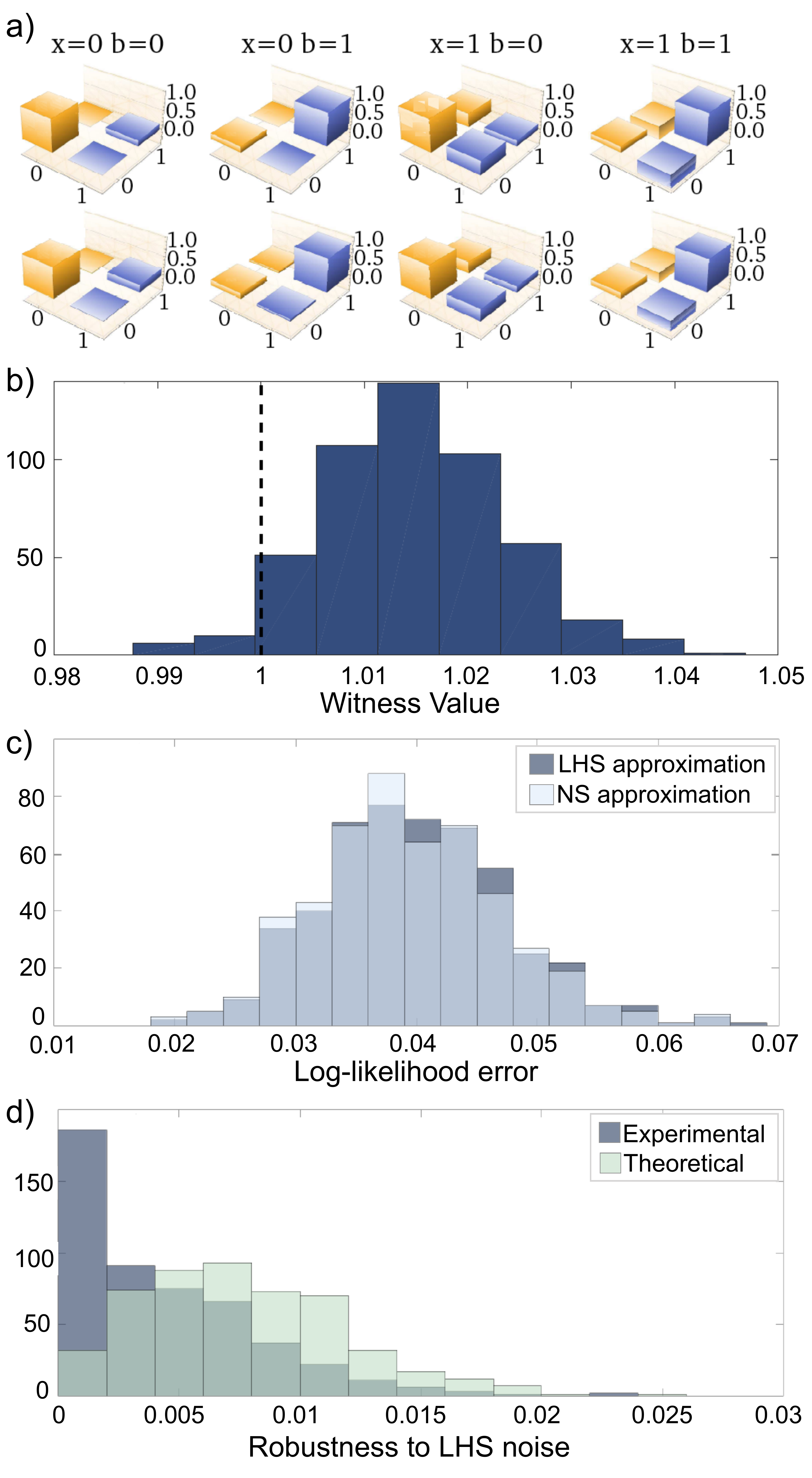}
\caption{a), b) Experimental assemblages after $y=a$ wiring. a) Real part
of Charlie's conditional density matrices, theoretical (top) and experimental
(bottom). b) Steering-witness histogram. The witness value is $1.015\pm0.009$,
meaning that the experimental assemblage is more than one standard
deviation above the steering threshold (dashed line). c), d) Compatibility
of the tripartite experimental assemblage with the naive (LHS) definition
of unsteerability {[}Eq.\ \eqref{eq:LHS}{]}. c) Histogram of the
error of approximating the tripartite assemblage by an NS and an LHS
assemblage, showing that the error of assuming the LHS decomposition
is as small as that of the physically necessary NS assumption. d)
From the best NS approximation to the experimental data, histogram
of the LHS-robustness, a measure of deviations from the set \textsf{LHS}.
Even with all experimental error, there is only a residual amount
of robustness, fully compatible with that of the theoretical LHS assemblage
solely under finite-statistics error. All histograms come from Monte
Carlo simulation assuming Poisson distributions.}
\label{fig:wiredass} 
\end{figure}

Using the same experimental setup, we can also experimentally demonstrate
super-exposure of Bell nonlocality. As argued above, the initial experimental
assemblage is compatible with an LHS model. Therefore, no matter what
measurement Charlie makes, the corresponding Bell behavior will be
compatible with an LHV model. Hence, we must only show that the experimental
wired assemblage is Bell nonlocal. In Ref.\ \cite{Taddei2016}, a
necessary and sufficient criterion for Bell nonlocality of assemblages
was derived: Given Alice and Bob's wired measurements ($y=a$) with
input bit $x$ and output bit $b$, to maximally violate a Bell inequality,
Charlie performs von Neumann measurements in the $2Z+X$ and $X$
bases, labeled by input bit $z$, obtaining binary output result $c$.
They thus obtain sixteen probabilities $P(b,c|x,z)$, which are used
to calculate the Clauser-Horne-Shimony-Holt (CHSH) inequality \cite{Clauser1969}.
We obtained an experimental violation of $2.21\pm0.04\nleqslant2$
(theoretical prediction: $2.29\nleqslant2$), showing Bell nonlocality
in a DI fashion.

Thus, we have experimentally demonstrated both exposure of steering
and super-exposure of Bell nonlocality.

\section{Redefinition of steering}

\noindent The results of the previous sections suggest the necessity
of a redefinition of steering in the multipartite scenario, since,
analogously to \cite{Gallego2012}, an assemblage can belong to LHS
and still be steerable. The existence of subtle steering implies a
stark inconsistency between the naive definition of steering from
LHS decomposability, Eq.\ \eqref{eq:LHS}, and the formulation of
its resource theory. Since the free operations that cause exposure
are classical and strictly local (fully contained in the $AB$ partition),
it is reasonable that they are unable to create not only steering
but also any form of correlations (even classical ones) across $AB|C$.
The alternative left is to redefine bipartite steering in multipartite
scenarios such that, e.g., the assemblages in Eqs.\ \eqref{eq:genactiv}
and \eqref{eq:quantumorig} are already steerable. Formally, we need
to exclude a subclass of LHS decompositions from the set of unsteerable
assemblages. In principle, no restriction must be imposed over the
probability distribution $P_{a,b|x,y,\lambda}$ in Eq.\ \eqref{eq:LHS},
once the NS conditions are satisfied for the visible assemblage. A
suitable choice is to restrict all signaling between Alice and Bob
also at the level of each $\lambda$ in Eq.\eqref{eq:LHS}; this defines
the set NS-LHS (non-signaling local hidden states). This restriction,
however, can be consistently relaxed to allow signaling between the
two as long as, for each $\lambda$, Alice and Bob's distribution
is compatible with both orders ($A$ before $B$ and $B$ before $A$);
this defines TO-LHS (time-ordered local hidden states), a strict superset
of NS-LHS; see Fig.\ref{fig:setsLHS}. This has consequences for genuine
multipartite correlations, including the possibility of certifying
genuine multipartite entanglement in a semi-DI scenario without steering.

To identify that subclass, let us apply the wiring $y=a$ to a general
$\boldsymbol{\sigma}$ fulfilling Eq.\ \eqref{eq:LHS}. This gives
$\boldsymbol{\sigma}^{(\text{wired})}$, of elements 
\begin{align}
\sigma_{b|x}^{(\text{wired})} & :=\sum_{a}\sigma_{a,b|x,a}=\sum_{\lambda}\ \ P_{\lambda}\Big(\sum_{a}P_{a,b|x,a,\lambda}\Big)\varrho_{\lambda}.\label{eq:wiredLHS}
\end{align}
This is a valid LHS decomposition as long as the term within brackets
yields a valid (normalized) conditional probability distribution (of
$B$ given $X$ and $\Lambda$). This is the case if every $\boldsymbol{P}_{\lambda}^{(AB)}$
in Eq.\ \eqref{eq:LHS} is non-signaling. In that case, by summing
over $b$ and applying the NS condition, one gets 
\begin{equation}
\sum_{a,b}P_{a,b|x,a,\lambda}=\sum_{a}P_{a|x,a,\lambda}\stackrel{\mathrm{NS}}{=}\sum_{a}P_{a|x,\lambda}=1\ ,\label{eq:norm}
\end{equation}
which renders $\boldsymbol{\sigma}^{(\text{wired})}$ indeed unsteerable.
However, this reasoning can in general not be applied if any $\boldsymbol{P}_{\lambda}^{(AB)}$
is signaling from Bob to Alice, i.e. if Alice's marginal distribution
for $a$ depends on $y$ (apart from $x$ and $\lambda$). In fact,
it can be checked that this is the case of the probability distributions
\eqref{eq:LHSorig2} and \eqref{eq:quantumLHSdecomp2} of the general
exposure protocol and of the quantum exposure example, respectively.

\begin{figure}[H]
\centering{}\includegraphics[width=0.5\columnwidth]{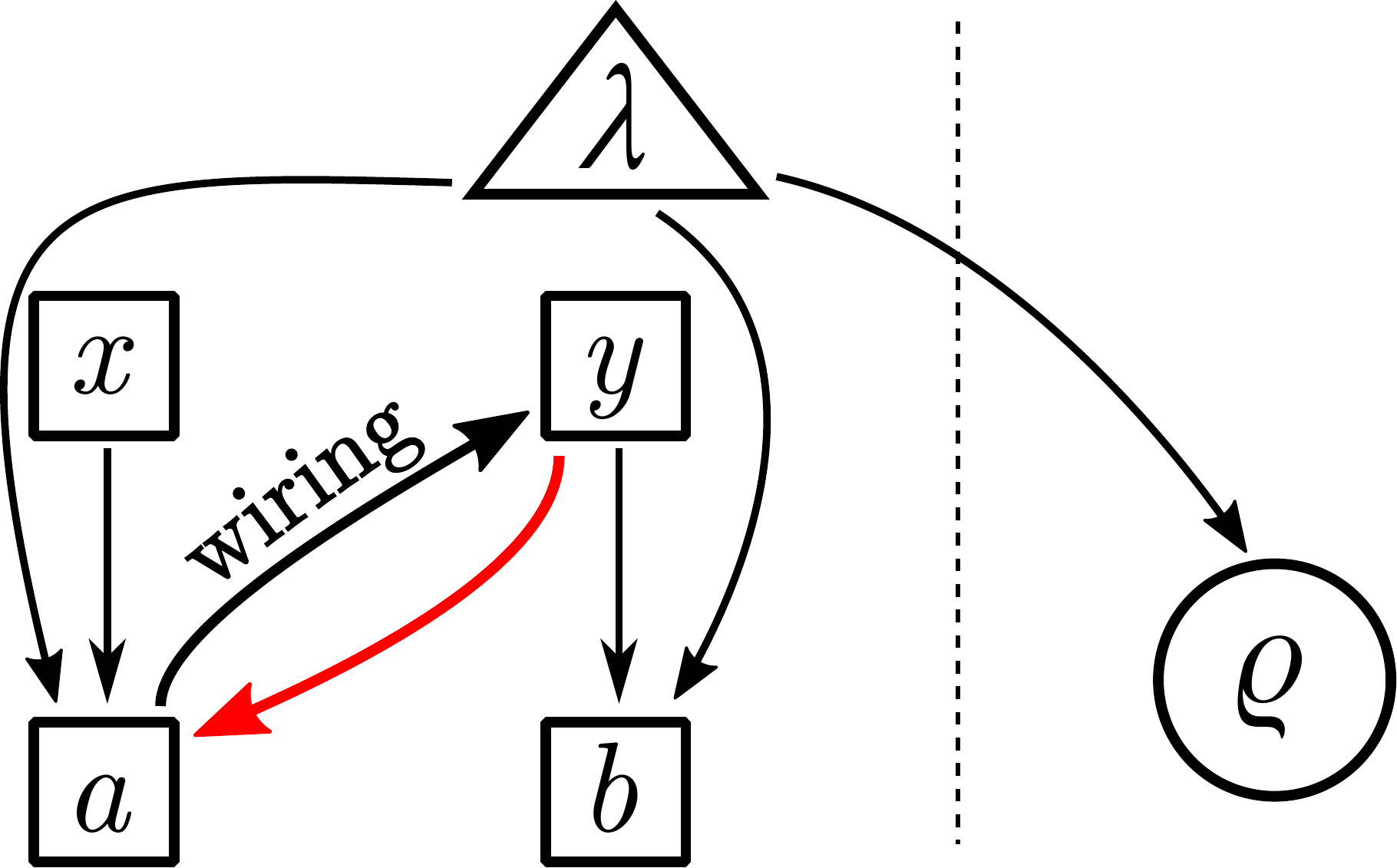}\caption{ Steering exposure as a causal loop. In the causal network underlying
LHS models, given by Eq.\ \eqref{eq:LHS}, the hidden variable $\lambda$
directly influences Charlie's quantum state $\varrho$ as well as
the Alice and Bob's outputs $a$ and $b$, which are in turn also
influenced by the inputs $x$ and $y$, respectively. Even though
the observed assemblage (after averaging $\lambda$ out) is non-signaling,
the model can still exploit hidden communication (i.e. at the level
of $\lambda$). For instance, for each $\lambda$, Alice's output
may depend (red arrow) on Bob's input in a different fine-tuned way
such that the dependence vanishes at the observable level. The wiring
of Fig.\ \ref{fig:Examples-of-free}a) forces $y=a$, closing a causal
loop that will in general conflict with the latter dependence for
some $\lambda$. As a consequence, the final assemblage resulting
from the wiring may not admit a valid LHS decomposition, exposing
steering. Hence, the exposure can in a sense be thought of as an operational
benchmark for hidden signaling in the LHS model describing the initial
assemblage. }
\label{fig:causal}
\end{figure}

Therefore, we see that the inconsistency is rooted in hidden signaling.
In fact, at the level of the underlying causal model, the phenomenon
of exposure can be understood as a causal loop between such signaling
and the applied wiring (see Fig.\ \ref{fig:causal}).

To restore consistency, hidden signaling must be restricted. An obvious
possibility would be to allow only for non-signaling $\boldsymbol{P}_{\lambda}^{(AB)}$'s
in Eq.\ \eqref{eq:LHS}. Interestingly, however, this turns out to
be over-restrictive. Following the redefinition of multipartite Bell
nonlocality \cite{Gallego2012,Bancal2013}, we propose the following
for bipartite steering in multipartite scenarios.

\vspace{0.2cm}
 Redefinition of steering: \textit{An assemblage $\boldsymbol{\sigma}$
is unsteerable if it admits }\textit{\emph{time-ordered LHS}}\textit{
(TO-LHS) decompositions both from $A$ to $B$ and from $B$ to $A$
simultaneously, i.e. if} 
\begin{align}
\sigma_{a,b|x,y}= & \sum_{\lambda}P_{\lambda}\ \ P_{a,b|x,y,\lambda}^{(A\to B)}\ \varrho_{\lambda}\label{eq:signAB}\\
= & \sum_{\lambda}P'_{\lambda}\ \ P_{a,b|x,y,\lambda}^{(B\to A)}\ \varrho'_{\lambda}\ ,\label{eq:signBA}
\end{align}
\textit{where each $\boldsymbol{P}_{\lambda}^{(A\to B)}$ is non-signaling
from Bob to Alice and each $\boldsymbol{P}_{\lambda}^{(B\to A)}$
from Alice to Bob. Otherwise $\boldsymbol{\sigma}$ is steerable.}
\label{eq:TOdef} \vspace{0.2cm}

The validity of both time orderings simultaneously prevents conflicting
causal loops. More precisely, if a wiring from Alice to Bob is applied
on $\boldsymbol{\sigma}$, one uses decomposition \eqref{eq:signAB}
to argue with the $\boldsymbol{P}_{\lambda}^{(A\to B)}$'s {[}as in
Eq.\ \eqref{eq:norm}{]} that the wired assemblage is unsteerable.
Analogously, if a wiring from Bob to Alice is performed, one argues
using the $\boldsymbol{P}_{\lambda}^{(B\to A)}$'s from decomposition
\eqref{eq:signBA}. Hence, no exposure is possible for TO-LHS assemblages,
guaranteeing consistency with bilocal wirings (as well as generic
1W-LOCCs from trusted to untrusted parts) as free operations of steering.
On the other hand, when all $\lambda$-dependent behaviors in Eqs.\ \eqref{eq:TOdef}
are fully non-signaling, then the assemblage is called \emph{non-signaling
LHS} (NS-LHS). There exists TO-LHS assemblages that are not NS-LHS,
which proves that the latter is a strict subset of the former. One
could criticize our results by arguing that it is a mere mathematical
statement and no physical realization of a TO-LHS assemblage outside
NS-LHS set is possible. To show the importance of this redefinition,
in App. \ref{app:strict_subset} we provide a quantum and a supra-quantum
example of TO-LHS assemblages that are not NS-LHS. In Fig. \ref{fig:setsLHS}
a pictorial representation of the structure of the set of NS assemblages
is shown.

In either case, the redefinition above automatically implies also
a redefinition of genuinely multipartite steering (GMS). We present
this explicitly in App. \ref{sec:def_gen_multipartite}. There, we
follow the approach of Ref. \cite{Cavalcanti2015a} in that a fixed
trusted-versus-untrusted partition is kept. However, instead of defining
GMS as incompatibility with quantum-LHS assemblages (i.e. with $\lambda$-dependent
behaviors with explicit quantum realizations) as in \cite{Cavalcanti2015a},
we use the more general TO-LHS ones. This reduces the set of genuinely
multipartite steerable assemblages safely, i.e. without introducing
room for exposure. Interestingly, this enables genuine multipartite
entanglement to be certified in the semi-DI scenario without steering
(App. \ref{sec:def_gen_multipartite}). 

\begin{figure}[H]
\centering{}\includegraphics[width=0.4\columnwidth]{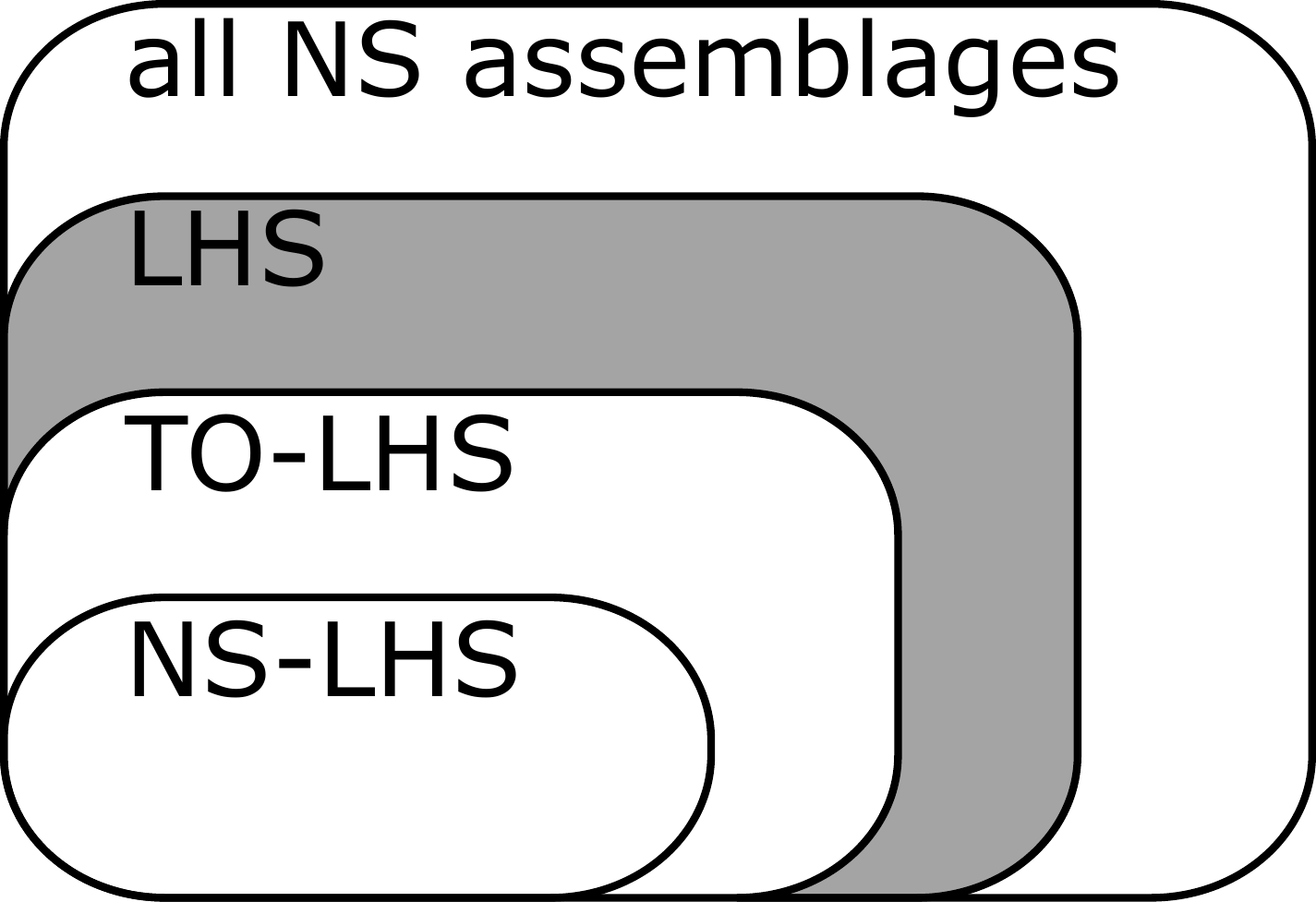}\caption{ Pictorial representation of inner structure of the set of all non-signaling
assemblages in the tripartite scenario. The subset of generic local-hidden-state
(LHS) assemblages strictly contains the subset \textsf{TO-LHS} of
time-ordered LHS ones, which in turn strictly contains the subset
\textsf{NS-LHS} of non-signaling LHS ones. The shaded region represents
the set of assemblages with subtle steering. Bilocal wirings can expose
such steering by mapping that region to the set of (bipartite) steerable
assemblages.}
\label{fig:setsLHS}
\end{figure}


\section{Concluding remarks}

We have demonstrated that the traditional definition of multipartite
steering for more than one untrusted party based on decomposability
in terms of generic bilocal hidden-state models presents inconsistencies
with a widely accepted, basic operational framework for the resource.
We have also shown how, according to such a definition, a broad set
of steerable (exposure) and even Bell-nonlocal (super-exposure) assemblages
would be created seemingly from scratch, e.g. by bilocal wirings acting
on an unsteerable assemblage. A surprising discovery that we have
made is the fact that exposure of quantum nonlocality is a universal
effect, in the sense that all steering assemblages as well as Bell
behaviors can be obtained as the result of an exposure protocol starting
from bilocal correlations in a scenario with one more untrusted party.
This highlights the power of exposure as a resource-theoretic transformation.
However, we also delimit such power: we prove a no-go theorem for
multi-black-box universal steering bits: there exists no single assemblage
with many untrusted and one trusted party from which all assemblages
with one untrusted and one trusted party can be obtained through generic
free operations of steering. To restore operational consistency, we
offer a redefinition of both bipartite steering in multipartite scenarios
and genuinely multipartite steering that does not leave room for exposure.
Finally, both steering exposure and Bell nonlocality super-exposure
have been demonstrated experimentally using an optical implementation.
This is to our knowledge the first experimental observation of exposure
of quantum nonlocality reported, not only in semi device-independent
scenarios but also in fully device-independent ones, as originally
predicted in \cite{Gallego2012,Bancal2013}.\selectlanguage{american}

~\ihead{}

\ohead{\textbf{Chapter~\thechapter}~\leftmark}

\ifoot{}

\cfoot{}

\ofoot[
]{\thepage}

\chapter{Detection of quantum non-Markovianity close to the Born-Markov approximation\label{chap:Markov}}

In this chapter we explore memory effects in quantum systems. The
presence of such memory is called non-Markovianity and appears when
a quantum system is interacting with an environment. Under some approximations,
however, this memory effects are so weak that they are not detected
by the majority of memory indicators. Here we present a study of non-Markovianity
for the decay dynamics of a two-level system in a bosonic bath. We
use an operational indicator, the so-called conditional past-future
(CPF) correlation, which relies on three measurement interventions
on the system. This indicator is able to detect memory effects even
close to the weak-interaction Born-Markov approximation. We also show
an experimental realization of the measurement of the CPF correlation
for a photonic qubit, showing the limitations of this approach in
a realistic experimental scenario.

This work was realized in collaboration with Adrián Budini from Consejo
Nacional de Investigaciones Científicas y Técnicas (CONICET - Argentina)
and the UFRJ professors Gabriel Aguilar, Marcelo Santos and Stephen
Walborn. My contribution to this work was in designing the experiment
and realizing it, as well as the data analysis. A resulting paper
is accepted for publication in Physical Review A and a preprint can
be found in Ref. \cite{thais2020}.

\section{Introduction}

Giving the time evolution of a quantum system's state as a unitary
operation is an accurate description only if the system is completely
isolated. Usually, the more general treatment of open quantum systems
(OQS) is necessary , either because it is impossible to perfectly
isolate a real system and spurious interactions with the surroundings
remain or because desirable interactions with a large uncontrolled
environment is present. In both cases one can only control and access
a small system of interest within a larger one and the reduced time
evolution of this portion is no longer unitary. The interaction with
the environment gives rise to energy dissipation, irreversible dynamics
and losses of quantum coherence and correlations. 

A first approach to the OQS problem is the so called Born-Markov approximation
(BMA) which considers that the coupling between system and environment
is much smaller then the other energy values involved and that the
environment is large enough such that its state is kept unchanged
during the OQS evolution \cite{breuerbook}. This approximation leads
to a memoryless dynamics meaning that the future state of the reduced
system only depends on its present state and not on its past story.
The BMA has been used extensively, providing excellent agreement with
many experiments in the context of quantum optics and magnetic resonance.
However the actual time evolution of any real system departures to
some extent from this idealized dynamics. For many situations it is
also the case that one has only partial information about the microscopic
details of the dynamics and it is necessary to quantify the degree
of non-Markovianity present \cite{rivas2010}. As we shall present
in the next section, most of the typical non-Markovianity measures
are based at least on the possibility of performing tomographic measurements
over the OQS state for different times during the evolution or process
tomography, or still on optimizations. From a experimental perspective,
in many cases, this renders determining the non-Markovian character
of the evolution almost impossible. The measure we employ here, the
conditional past-future (CPF) correlation \cite{budini}, relies only
on three subsequent measurements for different times, being a useful
and practical alternative to the other measurements capable of univocally
determine the non-Markovianity of the dynamics. 

To illustrate the CPF correlation capability of detecting non-Markovianity
close to BMA, in this work we study the non-Markovian features of
the spin-boson model which describes the decay of a two-level system
in a Bosonic bath. In contrast to previous memory indicators, we show
that the system propagator by itself is unable to detect quantum non-Markovianity
close the BMA. Instead, a self-convolution of the system propagator
weighted by the environment correlation becomes the proper memory
indicator. A photonic experiment that allows to measure the CPF correlation
for this system is also implemented, providing experimental support
to our main findings.

Before presenting our results, the concept of classical Markovianity
is introduced in Section \ref{sec:Markov-processes-in}. In order
to extend this concept to quantum physics, some basic elements on
quantum open systems dynamics are presented together with the Born-Markov
approximation and the quantum memory indicator used here (Section
\ref{sec:Open-quantum-system}). In sequence, it is presented the
model of a two-level system in a bosonic bath with the CPF correlation
for this model (Section \ref{sec:The-spin-boson-model}). At last,
the experiment and experimental results are presented (Section \ref{sec:Experiment}).

\section{Markov processes in classical Physics\label{sec:Markov-processes-in}}

This kind of stochastic process was named after the Russian mathematician
Andrei Andreevich Markov (1856-1922) whose contributions to number
and probability theory were fundamental to solve many subsequent problems
in science and technology \cite{basharina2004}. Roughly speaking,
Markov processes are those that do not possess memory, what means
that a result in the future, given a certain present condition, will
be the same regardless what happened in the past. 

To get to a formal definition, consider the stochastic process $X(t)$
taking discrete values $\{x_{i}\}_{i\in\mathbb{N}}$ on a finite set
$\chi$ for discrete time instants $\{t_{i}\}_{i\in\mathbb{N}}$.
It can be characterized by the hierarchy of joint probabilities $P_{n}(x_{n},t_{n};x_{n-1},t_{n-1};\cdots;x_{1},t_{1})$
that values $x_{k}$ occur at time $t_{k}$ for a given initial condition
$\{x_{0},t_{0}\}$, with $n\in\mathbb{N}$ and $t_{n}>t_{n-1}>\cdots>t_{1}>t_{0}$.
This process is said to be Markovian if the conditional probability
of the random variable at time $t_{n+1}$
\begin{equation}
P_{1|n}(x_{n+1},t_{n+1}|x_{n},t_{n},\cdots x_{1},t_{1})=\frac{P_{n+1}(x_{n+1},t_{n+1};x_{n},t_{n};\cdots;x_{1},t_{1})}{P_{n}(x_{n},t_{n};x_{n-1},t_{n-1};\cdots;x_{1},t_{1})}\label{eq:joint-prob}
\end{equation}
does not depend on the previous values of $X(t)$ but $x_{n}$, i.e.
\begin{equation}
P_{1|n}(x_{n+1},t_{n+1}|x_{n},t_{n},\cdots x_{1},t_{1})=P_{1|1}(x_{n+1},t_{n+1}|x_{n},t_{n})\quad\forall n\in\mathbb{N},\label{eq:markov}
\end{equation}
where $P_{j|k}$ is the probability distribution of $j$ events given
$k$ precedent events \cite{BreuerReview}.

Although by its definition, one would need to ensure the validity
of Eq. \eqref{eq:markov} for an infinity hierarchy of probability
distributions in order to characterize the Markovianity of a process,
there are a plethora of methods used to determine the non-Markovianity
of it which require much less information. This is the case of the
conditional past-future correlations presented in Section \eqref{subsec:Conditional-past-future}. 

An interesting property that comes directly from \eqref{eq:markov}
is the Chapman-Kolmogorov equation
\begin{equation}
P_{1|1}(x_{3},t_{3}|x_{1},t_{1})=\sum_{x_{2}\in\chi}P_{1|1}(x_{3},t_{3}|x_{2},t_{2})P_{1|1}(x_{2},t_{2}|x_{1},t_{1}),\quad\forall t_{3}>t_{2}>t_{1}\geq0\label{eq:divisibility}
\end{equation}
which can be interpreted as the possibility of writing the evolution
from the initial time $t_{1}$ to the final time $t_{3}$ as the composition
of the evolution from $t_{1}$ to an intermediate time $t_{2}$and
then from $t_{2}$ to $t_{3}$. This property is called divisibility
\cite{plenioReview}. 

\section{Open quantum system dynamics and Markovianity\label{sec:Open-quantum-system}}

The quantum mechanical description of nature is intrinsically stochastic.
Thus it is natural to translate the concept of Markovianity to processes
described by quantum mechanics. The mathematical object used to calculate
probability distributions of a quantum system is its density matrix
denoted by $\rho$. In this section we give the basic concepts of
dynamics of OQS density matrices.

The time evolution of a OQS is given by a dynamical map $\Lambda_{t}:S(\mathcal{H})\rightarrow S(\mathcal{H})$
acting on the convex set of physical states belonging to the Hilbert
space $\mathcal{H}$ of the system. If the system is initially prepared
in the state $\rho_{0}$ then its state in a posterior time $t\geq0$
is given as 
\begin{equation}
\rho_{t}=\Lambda_{t}[\rho_{0}].
\end{equation}
To be regarded as a physical map, $\Lambda_{t}$ must take physical
density matrices into physical density matrices. In other words, the
map must be trace preserving (TP) and positive, and also preserve
Hermiticity. A map is positive if $\Lambda[\rho]\geq0$ for all $\rho\geq0$.
The first property ensures the conservation of the probability and
the second one guarantees that all probabilities remain positive numbers
as they must be. A further requirement may be imposed, namely complete
positivity (CP) . If one considers that the system may be correlated
to another quantum system that is not under the action of the map,
still the map should take the composite-system initial state into
a physical state, no matter what dimension the additional Hilbert
space has. Mathematically this is expressed as $(\Lambda_{t}\otimes\mathbb{1}_{A})[\rho_{sa}]>0$
for all composite states $\rho_{sa}\in S(\mathcal{H}\otimes\mathcal{H}_{a})$
of OQS and ancilla, $\mathbb{1}_{a}$ is the identity operator in
the ancilla Hilbert space $\mathcal{H}_{a}$. For finite dimensional
Hilbert space $\mathcal{H}$ complete positivity is equivalent to
positivity in a $[\textrm{dim}(\mathcal{H})]^{2}$ space, i. e., in
the case when the ancilla Hilbert space has the same dimension as
the OQS \cite{choi}. 

A microscopic description of system and environment is usually used
to obtain the dynamical map for the OQS. The total system Hilbert
space is the tensor product of the Hilbert spaces of the OQS and the
environment $\mathcal{H}\otimes\mathcal{H}_{e}$. The reduced state
of the OQS (environment) is obtained by partial trace of the total
state over $\mathcal{H}_{e}$ ($\mathcal{H}$). The composed system
is closed and its time evolution is given by the von Neumann equation
\begin{equation}
i\frac{d}{dt}\rho_{se}(t)=\left[H_{se},\rho_{se}(t)\right]\label{eq:vonNeumman}
\end{equation}
 with Hamiltonian
\begin{equation}
H_{se}=H\otimes\mathbb{1}_{e}+\mathbb{1}_{s}\otimes H_{e}+H_{I},
\end{equation}
where $H$ ($H_{e}$) is the free Hamiltonian of the OQS (environment)
and $H_{I}$ is the interaction between system and environment. The
solution is a unitary evolution of the total initial state $\rho_{se}$.
Accordingly, the system reduced state in time $t$ is obtained as
\begin{equation}
\rho(t)=\Tr_{e}[U(t)\rho_{se}U^{\dagger}(t)],\label{eq:timeevol}
\end{equation}
being $U(t)$ the unitary operator associated with the Hamiltonian
$H_{se}$ and $\Tr_{e}[\cdot]$ denotes the partial trace over the
environment degrees of freedom.

Equation \eqref{eq:timeevol} provides the dynamical map $\Lambda_{t}$
once we can write the right hand side as the transformation of the
initial OQS state $\rho_{0}$. In the case of initially separable
state\footnote{This is a common case, since usually the system is prepared in a initial
state through a measurement, destroying any previous existing correlation
between system and environment. } $\rho_{se}=\rho_{0}\otimes\rho_{e}$ it is easy to obtain a physical
CPTP map as follows. First let us write the bath state in its spectral
decomposition $\rho_{e}=\sum_{q}p_{q}\ket{q}\bra{q}$ with $p_{q}>0$
and $\sum_{q}p_{q}=1$. Now, it is easy to identify the Kraus decomposition
of the dynamical map directly from \eqref{eq:timeevol} as 
\begin{equation}
\rho_{t}=\Lambda_{t}\left[\rho_{0}\right]=\sum_{l}E_{l}(t)\rho_{0}E_{l}^{\dagger}(t),\label{eq:Kraus}
\end{equation}
with the Kraus operators $E_{l}(t)=\sqrt{p_{q\text{\textasciiacute}}}\bra{q}U(t)\ket{q\text{\textasciiacute}}$,
$l=\{q,q\text{\textasciiacute}\}$. Because of the unitarity of the
total evolution the Kraus operators fulfill the property $\sum_{l}E_{l}(t)E_{l}^{\dagger}(t)=\mathbb{1}$.
The existence of a Kraus form for the map already ensures the complete
positivity of it \cite{kraus1983}. 

On the other hand, if the initial state of system and environment
has quantum correlations, then the dynamical map is not necessarily
CP \cite{dominy2013}. In fact, in this case it is not possible to
define such a dynamical map consistently defined for any initial system
state because the environment state is different for different system
states and the dynamics is changed. Even starting with a separable
total state, entanglement between OQS and its environment is typically
created along the evolution not allowing for the definition of dynamical
maps from intermediate times. As we are going to see in the next section,
this is behind the non-Markovianity feature of the dynamics.

\subsection{Born-Markov approximation}

Instead of directly obtaining the dynamical map, a typical approach
to OQS problems is to build a model yielding to a dynamical equation
for the OQS density matrix , the so called master equation, in which
the environment is part only by means of characteristic parameters.
It is usually a very hard problem to exactly obtain a master equation
from the unitary dynamics of the composite system and many approximations
are made necessary. Among these approximations, the most celebrated
is the Born-Markov approximation (BMA) that is presented is this section,
more details can be found in \cite{breuerbook} or any other book
on the subject of OQS. 

In the interaction picture the time evolution for the total system
is given by

\begin{equation}
i\frac{d}{dt}\rho_{se}^{I}(t)=\left[H_{I}(t),\rho_{se}^{I}(t)\right],
\end{equation}
where $\rho_{se}^{I}(t)=e^{iH_{0}t}\rho_{se}(t)e^{-iH_{0}t}$, $H_{I}(t)=e^{iH_{0}t}H_{I}e^{-iH_{0}t}$,
and $H_{0}=H\otimes\mathbb{1}_{e}+\mathbb{1}_{s}\otimes H_{e}$. The
formal solution come by integration as 
\begin{equation}
\rho_{se}^{I}(t)=\rho_{se}^{I}(0)-i\int_{0}^{t}\,d\tau\,\left[H_{I}(\tau),\rho_{se}^{I}(\tau)\right],
\end{equation}
and after the first iteration
\begin{equation}
\rho_{se}^{I}(t)=\rho_{se}^{I}(0)-i\int_{0}^{t}\,d\tau\,\left[H_{I}(\tau),\rho_{se}^{I}(0)\right]-\int_{0}^{t}\,d\tau\,\int_{0}^{\tau}\,d\tau\text{\textasciiacute}\,\left[H_{I}(\tau),\left[H_{I}(\tau\text{\textasciiacute}),\rho_{se}^{I}(\tau\text{\textasciiacute})\right]\right].\label{eq:iteration}
\end{equation}

The reduced system state is obtained from \eqref{eq:iteration} by
partial trace over the environment degrees of freedom. Moreover, a
differential equation satisfied by the reduced state is obtained by
taking the time derivative after the partial trace. Now come a series
of considerations. First of all, let us suppose that the initial state
commutes with the interaction Hamiltonian, such that the second term
in the right hand side of \eqref{eq:iteration} vanishes. After this
consideration and performing also a change of variables, one is left
with the following integro-differential equation
\begin{equation}
\frac{d}{dt}\rho(t)=-\int_{0}^{t}\,d\tau\Tr_{e}\left[H_{I}(t),\left[H_{I}(t-\tau),\rho_{se}^{I}(t-\tau)\right]\right],\label{eq:preBMA}
\end{equation}
where the superscript $I$ was suppressed only for simplicity. The
global state can be written as $\rho_{se}(t)=\rho(t)\otimes\rho_{e}(t)+\chi_{se}(t)$,
the sum of a separable part with a traceless part which contains the
correlations. Assuming that the correlations vanish in a time that
is small compared to the relaxation time of the OQS and also that
the environment is large enough such that its state is hardly affected
by the presence of the system we can approximate $\rho_{se}(t)\approx\rho(t)\otimes\rho_{e}$.
Moreover, the kernel inside the integral in Eq. \eqref{eq:preBMA},
which contains the correlation functions of the bath $\Tr_{e}\left[H_{I}(t)H_{I}(t-\tau)\rho_{e}\right]$,
vanish for $\tau$ larger than $\tau_{c}$ the correlation time of
the bath. If this correlation time is again small as compared with
the relaxation time of the OQS, for the values of $\tau$ that the
integrand does not vanish, we can approximate $\rho(t-\tau)\approx\rho(t)$
and extend the integral superior limit. The final master equation
that gives the OQS state time evolution is read
\begin{equation}
\frac{d}{dt}\rho(t)=-\int_{0}^{\infty}\,d\tau\Tr_{e}\left[H_{I}(t),\left[H_{I}(t-\tau),\rho(t)\otimes\rho_{e}\right]\right].\label{eq:BMA}
\end{equation}

The set of approximations performed above is called the \textit{Born-Markov
approximation} and it is widely used since it allows for writing a
pure differential equation for the OQS state. In the very limit of
BMA , the correlation function can be regarded as a Dirac delta function.
As can be noticed, in general this approximation can be valid only
for large enough time, such that all the correlation have been destroyed.
Increasing the time resolution to which one has access demands treating
the OQS beyond this approximation. 

Under general assumptions, the BMA equation \eqref{eq:BMA} can be
put in the form $\frac{d}{dt}\rho(t)=\mathcal{L}\rho(t)$ , whose
solutions have the semigroup expression $\rho(t)=e^{\mathcal{L}t}\rho_{0}$
, being $\mathcal{L}$ a Lindblad superoperator.

\subsection{Quantum Markovianity\label{subsec:Measures-of-non-Markovianity} }

The direct translation of the classical Markovianity to quantum theory
leads to issues related to the disturbance caused by measurement on
quantum systems \cite{vacchini2011}. Because of it, many different
definitions of Markovianity in this context have arisen in the past
few years \cite{plenioReview}.

In the quantum realm the joint probability distributions analogous
to those that define classical Markovianity are calculated from the
quantum state of the system and measurement operators $\mathcal{M}_{x}$.
Given that we chose to measure a non-degenerate observable $X=\sum_{x}x\ket{x}\bra{x}$,
the measurement operators can be the projectors over the eigenstates
of it: $\mathcal{M}_{x}=\ket{x}\bra{x}$. Not only does the joint
probability distribution obtained not satisfy basic conditions valid
for the classical counterparts \cite{BreuerReview}, but also the
measurement of the system alters the total system-environment state,
completely destroying its correlations and consequently strongly altering
the subsequent dynamics. One can easily be convinced of that from
the construction of the dynamical map shown in the previous section
and from the fact that, if the total state is $\rho_{se}(t_{i})$
in time $t_{i}$ and at this time the eigenvalue $x_{i}$ is measured,
then the total state immediately after the intervention becomes 
\begin{equation}
\rho\text{\textasciiacute}(t_{i})=\frac{\mathcal{M}_{x_{i}}\rho_{se}(t_{i})\mathcal{M}_{x_{i}}}{\Tr\left[\mathcal{M}_{x_{i}}\rho_{se}(t_{i})\mathcal{M}_{x_{i}}\right]}=\ket{x}\bra{x}\otimes\rho_{e}^{x_{i}}(t_{i}),
\end{equation}
where $\rho_{e}^{x_{i}}(t_{i})$ is the altered bath state possibly
dependent on the measurement result. 

In order to define a notion of quantum Markovianity that resembles
the classical one and rely only on the dynamics itself and not on
a particular measurement scheme, it must be related to the dynamical
map itself. One of the most used definitions of Markovianity uses
the concept of divisibility \eqref{eq:divisibility}, valid for the
classical Markovian probability distributions. A map $\Lambda_{t}$
is divisible if it is possible to define a two parameter trace preserving
map $\Lambda_{t_{2},t_{1}}$ such that $\Lambda_{t,0}=\Lambda_{t}$,
$\Lambda_{t,t}=\mathbb{1}$ and 
\begin{equation}
\Lambda_{t_{2},0}=\Lambda_{t_{2},t_{1}}\Lambda_{t_{1},0},\quad t_{2}>t_{1}>0.\label{eq:mapMarkov}
\end{equation}

A Markovian dynamics is defined as that given by a CP-divisible time
evolution map, i.e., a divisible map for which $\Lambda_{t_{2},t_{1}}$
is completely positive for all $t_{2},t_{1}\geq0$ \cite{rivas2010}.
If this is the case, then roughly speaking there is a valid CP time
evolution from any $t_{1}$ to any $t_{2}$ independent of the history
of the system or the environment, thus no memory is entailed is this
process. Clearly the solutions of the BMA master equation \eqref{eq:BMA}
are Markovian because of the semigroup property with $\Lambda_{t_{2},t_{1}}=e^{\mathcal{L}(t_{2}-t_{1})}$. 

The definition of Markovianity directly gives rise to a measure of
non-Markovianity. A divisible map is CP-divisible if and only if $\left[\Lambda_{t_{2},t_{1}}\otimes\mathbb{1}\right]\ket{\Phi}\bra{\Phi}\geq0$,
with $\ket{\Phi}=\frac{1}{\sqrt{\textrm{dim}\mathcal{H}}}\sum_{i=1}^{\textrm{dim}\mathcal{H}}\ket{n_{i}}\bra{n_{i}}$
the maximally entangled state in $\mathcal{H}\otimes\mathcal{H}$\cite{choi}.
Consider the trace norm defined as $||A||=\Tr\sqrt{A^{\dagger}A}$
that is equal to $\sum_{i}|a_{i}|$, the sum of the modulus of the
eigenvalues of $A$ if the operator is selfadjoint. Then, because
of the trace preservation property, it is true that $\left\Vert \left[\Lambda_{t_{2},t_{1}}\otimes\mathbb{1}\right]\ket{\Phi}\bra{\Phi}\right\Vert $
is equal to $1$ if the map is CP and larger than $1$ otherwise.
Then the function 
\begin{equation}
g(t)=\lim_{\epsilon\rightarrow0^{+}}\frac{\left\Vert \left[\Lambda_{t+\epsilon,t}\otimes\mathbb{1}\right]\ket{\Phi}\bra{\Phi}\right\Vert -1}{\epsilon}
\end{equation}
is positive only if the dynamics is non-Markovian and the quantity
$\Gamma=\int_{I}g(t)\:dt$ is a measure of the non-Markovianity of
the evolution in the time interval $I$ \cite{rivas2010}. This is
the commonly called RHP measure.

Another possible definition of quantum Markovianity that is in some
cases nonequivalent to the previous one can be given in terms of the
trace distance of states for different time instants \cite{haikka2011}.
The interaction between system and environment tends to diminish the
distinguishability between any two system states, which can be interpreted
as a loss of information from the system to the environment. Non-Markovian
dynamics would be mainly characterized by a reversed flow of information
from the environment to the system causing an increase in the distinguishability
of states during some interval of time \cite{breuer2009}. The distinguishability
between any two states $\rho_{1}$ and $\rho_{2}$ can be measured
by the trace distance $D(\rho_{1},\rho_{2})=\frac{1}{2}\left\Vert \rho_{1}-\rho_{2}\right\Vert $between
them. A non-Markovian dynamics is then characterized by 
\begin{equation}
\sigma(t)=\frac{d}{dt}D(\rho_{1}(t),\rho_{2}(t))>0,
\end{equation}
for any value of time and for any pair of possible initial states
$\rho_{1}$ and $\rho_{2}$, with $\rho_{i}(t)=\Lambda_{t}\rho_{i}$.
This definition also leads to a measure of non-Markovianity, the BLP
measure, defined as $\Gamma=\text{max}_{\rho_{1,2}(0)}\int_{\sigma>0}dt\,\sigma(t)$,
the maximization is over all possible pairs of initial states of the
OQS.

It is possible to show that correlations between system and environment
play a central role in non-Markovianity. As mentioned before, if the
initial total state is entangled, then the dynamical map may not be
CP. The time evolution ususally produces entanglement between system
and environment, that is why the map linking two intermediate states
in Eq. \eqref{eq:mapMarkov} is in general not CP characterizing non-Markovianity.
It is also possible to show that the trace distance always decreases
monotonically in the absence of correlations \cite{smirne2013}.

\subsection{Conditional past future correlation\label{subsec:Conditional-past-future}}

The previous mathematical definitions of classical and quantum Markovianity
are not practical in the sense that they require absolute knowledge
on infinity hierarchies of conditional probabilities (classical case)
or on the whole dynamical map (quantum case). An operational definition
or witness of non-Markovianity, on the other hand, should be obtained
from the mathematical definitions and would ideally rely on the minimum
number of measurements possible. 

Lets consider a classical or quantum system on which three sequential
measurements are performed at time instants $t_{x}<t_{y}<t_{z}$.
This is the minimum number of measurements one can think of when trying
to determine memory effects. A lack of memory is characterized by
a complete independence of the future measurement outcome $O_{z}$
on the past measurement outcome $O_{x}$ for a given present outcome
$O_{y}$, i.e. the future result is completely determined by the present
state. Mathematically, it means that the conditional-past-future (CPF)
correlation, conditioned to a fixed present outcome, defined as
\begin{equation}
C_{pf}|_{y}=\left\langle O_{x}O_{z}\right\rangle _{y}-\left\langle O_{x}\right\rangle _{y}\left\langle O_{z}\right\rangle _{y}=\sum_{x,y}\left[P(x,z|y)-P(x|y)P(z|y)\right]O_{x}O_{z},\label{eq:CPFdef}
\end{equation}
vanishes for Markov (memoryless) processes. Here $\left\langle \cdot\right\rangle _{y}$
denotes the mean value conditioned to the intermediate outcome labeled
by $y$ and the sum runs over all possible outcomes. In order to verify
that this statement is consistent with the classical definition of
a Markovian process, lets write Bayes rules for the various conditional
and joint probabilities of this triple measurement procedure. First,
one can write the triple joint probability in terms of different conditionals
as\begin{subequations} \label{eq:classprob} 
\begin{eqnarray}
P(x,y,z) & =P(z|x,y)P(y|x)P(x),\\
P(x,y,z) & =P(x,z|y)P(y).
\end{eqnarray}
 \end{subequations}Bayes rule can be applied again to invert the
present-past dependence since it gives $P(x,y)=P(y|x)P(x)=P(x|y)P(y)$.
Therefore, in general Eqs. \eqref{eq:classprob} yield $P(x,z|y)=P(z|x,y)P(x|y)$
and the CPF correlation \eqref{eq:CPFdef} is not null. For Markovian
processes, on the other hand, Eq. \eqref{eq:markov} tells that $P(z|x,y)=P(z|y)$
and hence the correlation is zero regardless the value of the present
outcome and also the time instants at which the system is measured. 

The CPF correlation can also be calculated for a quantum system. In
this case, a measurement is represented by a set of operators $\{\Omega_{x}\}$
satisfying $\sum_{x}\Omega_{x}\Omega_{x}^{\dagger}=\mathbb{1}$, $x$
running over all possible measurement outcomes. The probability of
a result $x$ of a measurement performed on a system in a state $\rho$
is calculated as $P(x)=\Tr\left[\Omega_{x}\rho\Omega_{x}^{\dagger}\right]$
and , after the measurement, the state of the system is $\rho^{(x)}=\Omega_{x}\rho\Omega_{x}^{\dagger}/\Tr\left[\Omega_{x}\rho\Omega_{x}^{\dagger}\right]$.
For the CPF correlation protocol, three sequential measurements represented
by the possibly different sets of operators $\{\Omega_{x}\}$, $\{\Omega_{y}\}$
and $\{\Omega_{z}\}$ are realized. Between the measurements the system
evolves in contact with its environment according to Eq. \eqref{eq:vonNeumman}
as illustrated in Fig. \ref{fig:CPFrep}. The initial state $\rho_{se}^{(0)}$
goes to the state $\rho_{se}^{(x)}=\Omega_{x}\rho_{se}^{(0)}\Omega_{x}^{\dagger}/\Tr\left[\Omega_{x}\rho_{se}^{(0)}\Omega_{x}^{\dagger}\right]$
after the first (past) measurement with probability $P(x|0)=\Tr\left[\Omega_{x}\rho_{se}^{(0)}\Omega_{x}^{\dagger}\right]$,
where $\Omega_{x}$ must be understood as $\Omega_{x}\otimes\mathbb{1}_{e}$.
After a period $t=t_{y}-t_{x}$ of interaction between OQS and bath,
the second (present) measurement delivers the result $y$ with probability
$P(y|x)=\Tr\left[\Omega_{y}U_{t}\rho_{se}^{(x)}U_{t}^{\dagger}\Omega_{y}^{\dagger}\right]$,
letting the composite system in the state 
\begin{equation}
\rho_{se}^{(y)}(t)=\frac{\Omega_{y}U_{t}\rho_{se}^{(x)}U_{t}^{\dagger}\Omega_{y}^{\dagger}}{\Tr\left[\Omega_{y}U_{t}\rho_{se}^{(x)}U_{t}^{\dagger}\Omega_{y}^{\dagger}\right]}.
\end{equation}
The retrodicted probability of past given present can then be calculated
as $P(x|y)=P(x,y)/P(y)$, where $P(y)=\sum_{x}P(x,y)$ and the joint
probability is obtained from the predictive probabilities as $P(x,y)=P(y|x)P(x|0)$.
The last stage is a time evolution for period $\tau=t_{z}-t_{y}$
followed by a last (future) measurement on the OQS, the probability
of an outcome $z$ being $P(z|x,y)=\Tr\left[\Omega_{z}U_{\tau}\rho_{se}^{(y)}(t)U_{\tau}^{\dagger}\Omega_{y}^{\dagger}\right]$.
Provided that the state $\rho_{se}^{(y)}(t)$ after the present intervention
does not depend on the past measurement, $P(z|x,y)=P(z|y)$, the independence
between past and future is attained, and the CPF correlation given
by Eq. \eqref{eq:CPFdef} vanishes. Usually, this is not the case.
If the intermediate measurement is projective over the eigenstates
of an observable $O=\sum_{y}O_{y}\ket{y}\bra{y}$, the resulting composite
state is
\begin{equation}
\rho_{se}^{(y)}(t)=\ket{y}\bra{y}\otimes\rho_{e}^{(y|x)}(t)=\ket{y}\bra{y}\otimes\Tr_{s}\left[\frac{\Omega_{y}U_{t}\rho_{se}^{(x)}U_{t}^{\dagger}\Omega_{y}^{\dagger}}{\Tr\left[\Omega_{y}U_{t}\rho_{se}^{(x)}U_{t}^{\dagger}\Omega_{y}^{\dagger}\right]}\right].
\end{equation}
The interaction between system and environment creates entanglement
between the two parts causing the reduced state of the environment
$\rho_{e}^{(y|x)}(t)$ to depend on the measurement result of the
system in the past intervention, although the state of the system
itself does not depend on the past.

If the system is not interacting with the environment $(H_{I}=0)$,
then $\rho_{e}^{(y|x)}(t)$ does not depend on $x$ and $C_{pf}|_{y}=0$
for any measurement choice. It means that the measurement process
itself does not violate the past-future independence, any violation
comes from the memory effects induced by the environment, the characteristic
of non-Markovian dynamics.

Another remarkable particular case is when the Born-Markov approximation
is valid. In this case, the state of the bath is approximately constant,
and $\rho_{e}^{(y|x)}(t)=\rho_{e}$ equals the initial state leading
to a null CPF correlation. Thus, the CPF correlation can also be seen
as a measure of departure from the BMA.

\begin{figure}[h]
\centering
\includegraphics[width=0.7\columnwidth]{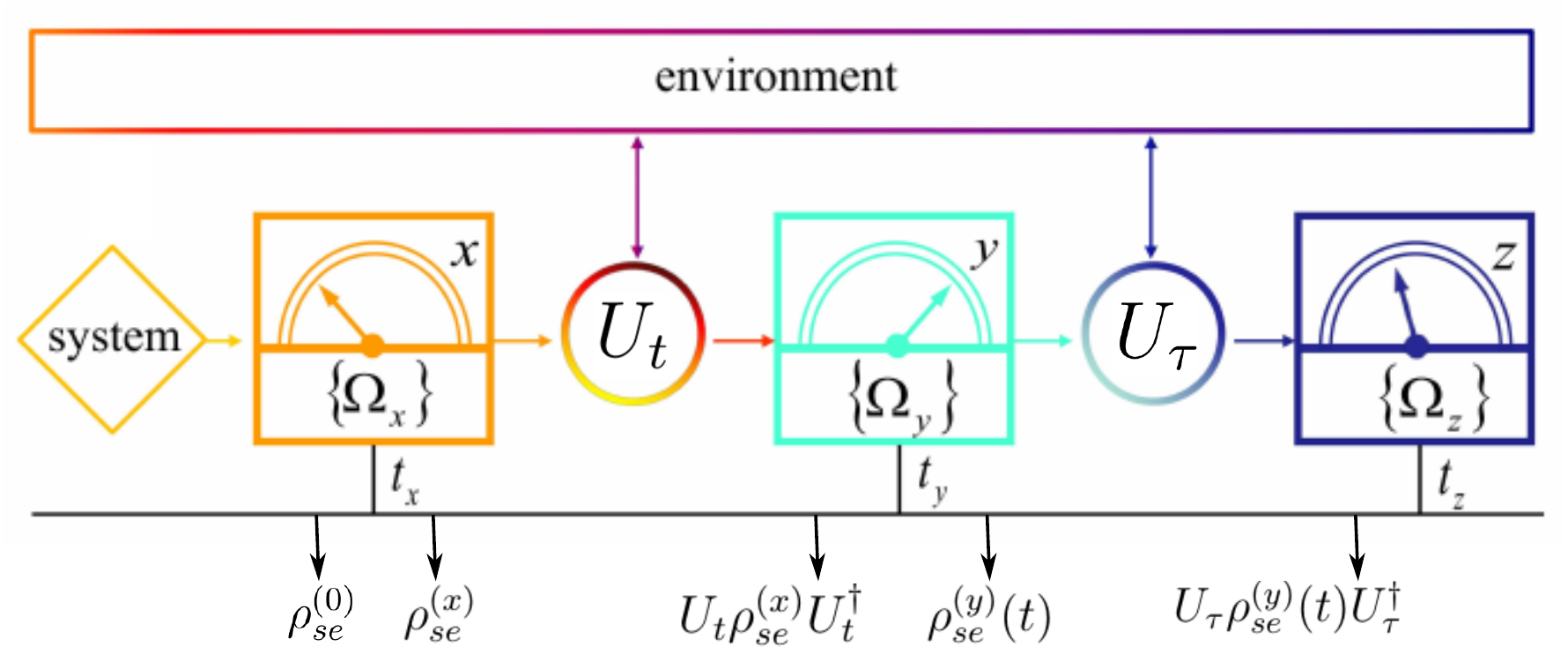}
\caption{Representation of the CPF correlation protocol. Here we consider that
system and environment do not interact at the before $t_{x}$, but
this is not mandatory for the protocol. Three measurements are performed
on the OQS without any intervention in the environment, between two
measurements system and environment are let to interact and unitarily
evolve. The composite state before and after each measurement is represented
bellow the time line. Adapted from \protect\cite{budini}. }
\label{fig:CPFrep}
\end{figure}

\section{The spin-boson model\label{sec:The-spin-boson-model}}

The spin-boson model describes a two-level system (spin) interacting
with a quantum environment composed of a continuum of bosonic modes
\cite{breuerbook}. This is a paradigmatic model in quantum optics
used to describe the dipole interaction of a two-level atom with the
electromagnetic field \cite{scullybook} , as well as in condensed
matter physics where it is used to describe defects in a solid interacting
with the phononic environment \cite{leggett1987}, only to give two
examples.

The total Hamiltonian is given by
\begin{equation}
H_{\mathrm{tot}}=\frac{\omega_{0}}{2}\sigma_{z}+\sum_{k}\omega_{k}b_{k}^{\dag}b_{k}+\sum_{k}\sigma_{x}(g_{k}b_{k}+g_{k}^{\ast}b_{k}^{\dag}),\label{Bosonic-1}
\end{equation}
where $\sigma_{z}$ is the $z$-Pauli matrix, $\omega_{0}$ is the
energy difference between the two levels of the qubit. The index $k$
labels the reservoir mode with frequency $\omega_{k}$ which couples
to the qubit with coupling constant $g_{k}$. The bosonic operators
satisfy the relations $[b_{k},b_{k}^{\dag}]=1.$ The first two terms
on the right hand side are the free energies of system and environment,
respectively, while the last term gives the interaction between the
two sub-parts.

Here the rotating wave approximation is considered, the terms of the
Hamiltonian that do not conserve energy are dropped off and it becomes
\begin{equation}
H_{\mathrm{tot}}=\frac{\omega_{0}}{2}\sigma_{z}+\sum_{k}\omega_{k}b_{k}^{\dag}b_{k}+\sum_{k}(g_{k}\sigma_{+}b_{k}+g_{k}^{\ast}\sigma_{-}b_{k}^{\dag}),\label{Bosonic-2}
\end{equation}
where $\sigma_{+}=\ket{\uparrow}\bra{\downarrow}$ and $\sigma_{-}=\ket{\downarrow}\bra{\uparrow}$
are the raising and lowering operators of the qubit in the natural
base \{$\ket{\uparrow},\ket{\downarrow}\}.$

As usual, we assume that the total initial state is the pure separable
wave vector $|\Psi_{0}\rangle=(a\ket{\uparrow}+b\ket{\downarrow})\otimes|0\rangle,$
where the environment vacuum state is $|0\rangle\equiv\prod_{k}|0\rangle_{k}$.
As the commutator of the total number of excitations $N=\sigma_{+}\sigma_{-}+\sum_{k}b_{k}^{\dagger}b_{b}$
with the total Hamiltonian vanishes ($[H_{tot},N]=0$), this quantity
is conserved and a good ansatz for the evolved state in the interaction
picture is 
\begin{equation}
|\Psi_{t}\rangle=\Big{[}a(t)\ket{\uparrow}+b(t)\ket{\downarrow}+\ket{\downarrow}\sum_{k}c_{k}(t)b_{k}^{\dag}\Big{]}|0\rangle,\label{evolved_state_t-1}
\end{equation}
as it has only terms with at most one excitation. The time evolution
in this representation is given by 
\begin{equation}
i\frac{d\ket{\Psi_{t}}}{dt}=H_{I}(t)\ket{\Psi_{t}},\label{eq:SchEq}
\end{equation}
with $H_{I}(t)=\sum_{k}(g_{k}e^{i\omega_{0}t}\sigma_{+}e^{i\omega_{k}b_{k}^{\dag}b_{k}}b_{k}e^{-i\omega_{k}b_{k}^{\dag}b_{k}}+h.c.)$.
After solving the Schrödinger equation \eqref{eq:SchEq} (the details
are left for Appendix ), the system density matrix $\rho_{t}=\mathrm{Tr}_{e}[|\Psi_{t}\rangle\langle\Psi_{t}|]$
in the interaction representation can be found as
\begin{equation}
\rho_{t}=\left(\begin{array}{cc}
|a|^{2}|G(t)|^{2} & ab^{\ast}G(t)\\
a^{\ast}bG^{\ast}(t) & 1-|a|^{2}|G(t)|^{2}
\end{array}\right),\label{Rho(t)-1}
\end{equation}
which fulfills the non-Markovian master equation $(d\rho_{t}/dt)=\frac{-i}{2}\omega(t)[\sigma_{z},\rho_{t}]+\gamma(t)([\sigma_{-}\rho_{t},\sigma_{+}]+[\sigma_{-},\rho_{t}\sigma_{+}])$
\cite{breuerbook}. The time-dependent decay rate and frequency are
defined as $\gamma(t)+i\omega(t)=-(d/dt)\ln[G(t)].$ The ``wave vector
propagator\textquotedblright \ $G(t)$ is defined by 
\begin{equation}
\frac{d}{dt}G(t)=-\int_{0}^{t}f(t-t^{\prime})G(t^{\prime})dt^{\prime},\label{G-1}
\end{equation}
where the memory kernel is defined by the bath correlation $f(t)\equiv\sum_{k}|g_{k}|^{2}\exp[+i(\omega_{0}-\omega_{k})t]$\footnote{This function is called bath correlation because it can be written
as $f(t-t^{\prime})=\textrm{Tr}_{e}\left(B(t)B^{\dagger}(t^{\prime})\rho_{e}\right)e^{i\omega(t-t^{\prime})}$,
where $B=\sum_{k}g_{k}b_{k}$ is the bath operator which participate
in the interaction, and $\rho_{e}=\ket{0}\bra{0}$ is the environment
initial state.}. For a continuous of modes with isotropic interaction the bath correlation
becomes $f(t)=\int d\omega\,J(\omega)\exp[+i(\omega_{0}-\omega)t]$,
$J(\omega)$ is called spectral function.

In the Born-Markov limit, the bath correlation remains for a very
short time and the bath correlation function can be approximated by
a Dirac delta function. In this limit the Green function that satisfies
Eq. \eqref{G-1} represents a pure exponential decay and the dynamics
is Markovian. The usual measures as the ones described in Section
\eqref{subsec:Measures-of-non-Markovianity} are able to detect non-Markovianity
only if the decay rate $\gamma(t)$ is negative in some time interval
\cite{addis2014}. For the model (\ref{Bosonic-2}), standard memory
witness, such as the ones presented in Section \ref{subsec:Measures-of-non-Markovianity},
coincide \cite{breuer}. In fact, these measures are able to detect
non-Markovianity only if the decay rate $\gamma(t)$ is negative in
some time interval \cite{addis2014}. Equivalently, this means that
if $|G(t)|^{2}$ decays monotonically, giving place to a monotonous
decay from the upper level $\ket{\uparrow}$ to the lower state $\ket{\downarrow}$,
then the dynamics is considered Markovian. Nevertheless, in this regime
it is not necessarily within the BMA. Non-Markovianity close to the
BMA can be detected with a CPF correlation \cite{budini,budiniPRA}.

\subsection{CPF correlation }

For different measurement schemes, the CPF correlation associated
to the dynamics Eq. (\ref{Bosonic-1}) can be calculated in an exact
way. Considering the initial condition $|\Psi_{0}\rangle=(a\ket{\uparrow}+b\ket{\downarrow})\otimes|0\rangle$
and performing three projective measurements in the $\sigma_{z}$
direction of the Bloch sphere ($\sigma_{z}-\sigma_{z}-\sigma_{z}$),
which implies $x=\pm1,$ $y=\pm1,$ $z=\pm1,$ the exact CPF correlation
reads (see Appendixes) $C_{pf}(t,\tau)|_{y=+1}\underset{\sigma_{z}\sigma_{z}\sigma_{z}}{=}0,$
while for the conditional $y=-1,$ it reads
\begin{equation}
C_{pf}(t,\tau)|_{y=-1}\underset{\sigma_{z}\sigma_{z}\sigma_{z}}{=}\left\{ \frac{4|a|^{2}|b|^{2}}{[(1-|G(t)|^{2})|a|^{2}+|b|^{2}]^{2}}\right\} |G(t,\tau)|^{2}.\label{zzz}
\end{equation}
Alternatively, by performing the successive measurement in the $\sigma_{x}-\sigma_{z}-\sigma_{x}$
directions, we get $C_{pf}(t,\tau)|_{y=+1}\underset{\sigma_{x}\sigma_{z}\sigma_{x}}{=}0,$
while for the conditional $y=-1,$ it reads
\begin{equation}
C_{pf}(t,\tau)|_{y=-1}\underset{\sigma_{x}\sigma_{z}\sigma_{x}}{=}-\left\{ \frac{1-[2\mathrm{Re}(ab^{\ast})]^{2}}{1-|G(t)|^{2}/2}\right\} \mathrm{Re}[G(t,\tau)].\label{xzx}
\end{equation}
In the previous two expressions, the function $G(t,\tau)$ is
\begin{equation}
G(t,\tau)\equiv\int_{0}^{t}dt^{\prime}\int_{0}^{\tau}d\tau^{\prime}f(\tau^{\prime}+t^{\prime})G(t-t^{\prime})G(\tau-\tau^{\prime}).\label{GG}
\end{equation}

The exact result $C_{pf}(t,\tau)|_{y=+1}\underset{\sigma_{z}\sigma_{z}\sigma_{z}}{=}0$
jointly with $C_{pf}(t,\tau)|_{y=+1}\underset{\sigma_{x}\sigma_{z}\sigma_{x}}{=}0,$
follow from the symmetry of the problem. In fact, the conditional
$y=+1$ implies that the system evolution during the first two measurements
(interval $t$) is exactly the same than in the interval between the
second and third measurements (interval $\tau$). Thus, the CPF correlation
vanishes \cite{budini,budiniPRA}. This accidental symmetry does not
appear for the conditional $y=-1.$

Besides normalization factors proportional to the initial system condition
and the propagator $G(t),$ Both Eq. (\ref{zzz}) and (\ref{xzx})
are determined by $G(t,\tau)$ {[}Eq. (\ref{GG}){]}. Thus, in contrast
to previous approaches, where $G(t)$ takes the main role, here $G(t,\tau)$
is the main mathematical object capturing the memory effects. It consist
in a convolution involving two system propagators mediated by the
environment correlation. It is simple to check that $G(t,\tau)\rightarrow0$
when $f(t)$ approaches a delta function. Consequently, $G(t,\tau)$
measures departures with respect to the BMA, even close to its validity.

\textit{Backflow of information}: Given that the underlying dynamics
admits an exact treatment, a simple relation between a non-operational
backflow of information \cite{breuer} and an operational one can
be established as follows: Let us consider that the system is at the
initial time in the upper state, a non-monotonous decay of the conditional
probability $P(\uparrow,t|\uparrow,0)=|G(t)|^{2}$ determines the
presence of an environment-to-system backflow of information (non-operational
way). In contrast, under the same initial condition, an operational
backflow of information can be defined by the probability $P(\uparrow,t+\tau|\downarrow,t;\uparrow,0)=|G(t,\tau)|^{2}/[1-|G(t)|^{2}]$,
which measures the capacity of the environment of reexciting the system
after it has been found in the lower state at an intermediate time.
This probability only vanishes in the Markovian limit. These two clearly
different physical scenarios determine the possibility of detecting
departure from the BMA or not, which in turn\ may be read as different
notions of environment-to-system backflow of information.

\section{Experiment\label{sec:Experiment}}


In order to demonstrate the experimental feasibility of measuring
memory effects close to the BMA, we developed a photonic platform
that simulates the non-Markovian system dynamics. 
The CPF correlation is measured through the sequence $X\rightarrow U(t)\rightarrow Y\rightarrow U(\tau)\rightarrow Z,$\ where
$X,$ $Y,$ and $Z$ are the measurement processes while $U(t)$ and
$U(\tau)$ are the unitary transformation maps associated to the total
Hamiltonian (\ref{Bosonic-2}). These maps represent the system-environment
total changes between consecutive measurement processes. Although
the real environment is composed of an infinite number of modes, the
system reduced dynamical map can be obtained if the environment is
regarded also as a two-level system \cite{ruskai2002}. The map $U(t)$
is defined by the transformations \begin{subequations} \label{eqmap1}
\begin{align}
\left\vert {\downarrow}\right\rangle \otimes\left\vert {0}\right\rangle  & \rightarrow\left\vert {\downarrow}\right\rangle \otimes\left\vert {0}\right\rangle ,\\
\left\vert {\uparrow}\right\rangle \otimes\left\vert {0}\right\rangle  & \rightarrow\cos(2\theta)\left\vert {\uparrow}\right\rangle \otimes\left\vert {0}\right\rangle +\sin(2\theta)\left\vert {\downarrow}\right\rangle \otimes\left\vert {1}\right\rangle .
\end{align}
Here, $\left\vert {0}\right\rangle $ and $\left\vert {1}\right\rangle $
represent the bath in its ground state and (first) excited state respectively.
The angle $\theta$ is such that $\cos(2\theta)=G(t).$ Given that
the intermediate (second) measurement may leave the system in its
ground state and the bath in an excited state, the map associated
to $U(\tau)$ involves one extra initial state, \end{subequations}
\begin{subequations} \label{eqmap2} 
\begin{eqnarray}
\left\vert {\downarrow}\right\rangle \otimes|0\rangle & \rightarrow & \left\vert {\downarrow}\right\rangle \otimes|0\rangle,\\
\left\vert {\uparrow}\right\rangle \otimes|0\rangle & \rightarrow & \cos(2\tilde{\theta})\left\vert {\uparrow}\right\rangle \otimes|0\rangle+\sin(2\tilde{\theta})\left\vert {\downarrow}\right\rangle \otimes|1\rangle,\\
\left\vert {\downarrow}\right\rangle \otimes|1\rangle & \rightarrow & \sin(2\tilde{\theta}^{\prime})\left\vert {\uparrow}\right\rangle \otimes|0\rangle+\cos(2\tilde{\theta}^{\prime})\left\vert {\downarrow}\right\rangle \otimes|1\rangle.\ \ \ \ 
\end{eqnarray}
The angles are given by the relations $\cos(2\tilde{\theta})=G(\tau),$
and $\sin(2\tilde{\theta}^{\prime})=G(t,\tau)/\sqrt{1-|G(t)|^{2}}.$
This last term measures the capacity of the environment of reexciting
the system. It involves a normalization proportional to the decay
probability in the interval $(0,t)$ (see Appendix \ref{chap:Anexo1}).

The previous maps can be experimentally simulated by encoding the
system states $\{|{\downarrow\rangle},\left\vert {\uparrow}\right\rangle \}$
into polarization of a photon $\{|{H\rangle},\left\vert {V}\right\rangle \},$
while the bath states are encoded into the path degree of freedom
of the same photons. Angles $\{\theta,\tilde{\theta},\tilde{\theta}^{\prime}\}$
are chosen as a function of the simulated bath properties \cite{brasil,rio}.
We approach the spectral density by a Lorentzian $J(\omega)=(1/2\pi)\gamma\tau_{c}^{-2}/[(\omega-\omega_{0})^{2}+\tau_{c}^{-2}]$,
which implies the exponential correlation $f(t)=(\gamma/2\tau_{c})\exp[-|t|/\tau_{c}]$,
$\gamma$ is the strength of the coupling between system and environment
and $\tau_{c}$ is correlation time of the bath. In this case,\ the
propagator (\ref{G-1}) reads \end{subequations} 
\begin{equation}
G(t)=e^{-t/2\tau_{c}}\Big{[}\cosh(\frac{t\chi}{2\tau_{c}})+\frac{1}{\chi}\sinh(\frac{t\chi}{2\tau_{c}})\Big{]},
\end{equation}
where $\chi\equiv\sqrt{1-2\gamma\tau_{c}}.$ Furthermore, Eq. (\ref{GG})
becomes
\begin{equation}
G(t,\tau)=\frac{2\gamma\tau_{c}}{\chi^{2}}e^{-(t+\tau)/2\tau_{c}}\sinh(\frac{t\chi}{2\tau_{c}})\sinh(\frac{\tau\chi}{2\tau_{c}}).
\end{equation}
As $\omega(t)=0$ ($G(t)$ real), in the considered case there is
no time-dependent energy shift. In the weak coupling limit $\gamma\ll1/\tau_{c},$
where the correlation time $\tau_{c}$\ of the bath is the minor
time scale of the problem, it follows that $G(t)\simeq\exp[-\gamma t/2],\ G(t,\tau)\simeq0,$
which in turn implies that, independently of the measurement scheme,
a Markovian limit is approached $C_{pf}(t,\tau)|_{y}\simeq0.$

The specific experimental setup is illustrated in Fig. \ref{fig:setup-1}.
A continuous-wave (CW) laser, centered at 325 nm, is sent to a beta-barium-borate
(BBO) crystal. Degenerated pairs of photons (wavelength centered at
650 nm), are produced in the modes signal ``s\textquotedblright \ and
idler ``i\textquotedblright \ via spontaneous-parametric-down-conversion
\cite{Kwiat99}. The photons in mode i are sent directly to detection
as they only herald the presence of photons in mode s, while the photons
in mode s pass through nested interferometers, which emulate the maps
$U(t)$ and $U(\tau)$ \cite{rio}. Projective measurements are introduced
in modules $X,$ $Y,$ $Z.$ The CPF correlation is extracted by using
coincidence counts for all the different combinations of past, present
and future outcomes.

\begin{figure}[h]
\centering{}\includegraphics[width=0.7\columnwidth]{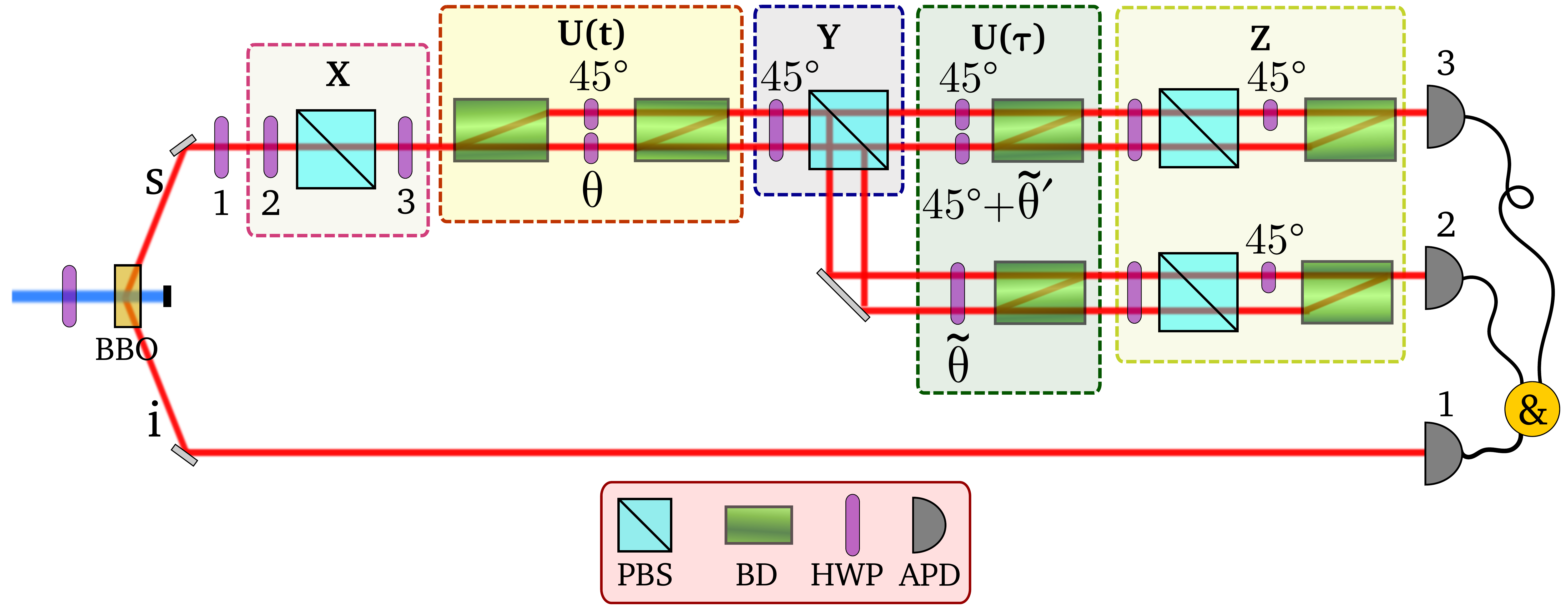}
\caption{Experimental Setup. Modules $X,$ $Y,$ and $Z$ perform the projective
measurements. Modules $U(t)$ and $U(\tau)$ implement the unitary
system-environment maps. From coincidence counting, the avalanche
photon detectors (APD) allow measuring the CPF correlation for photon
signal (``s''), while photon idler (``i'') only heralds its presence
(see text). Spatial mode 0 (1) corresponds to the upper (lower) path
after the first beam displacer (BD). PBS - polarizing beam splitter;
HWP - half wave plate. }
\label{fig:setup-1} 
\end{figure}


Given that the photons created in the BBO crystal are horizontally
polarized, we prepare any initial linear polarization state $(a\left\vert {H}\right\rangle +b\left\vert {V}\right\rangle ,\quad a,b\in\mathbb{R})$\ using
a half-wave plate (HWP$_{1}$). The past measurement $X$ is performed
using a set of two HWPs and a polarizing beam-splitter (PBS), which
transmits the horizontal polarization and reflects the vertical one.
In this measurement, the angle set in HWP$_{2}$ selects the linear
polarization state mapped to $H$ and hence transmitted by the PBS,
while HWP$_{3}$ prepares the projected state from the transmitted
horizontal polarization. After this module, the map $U(t)$ {[}Eq.
(\ref{eqmap1}){]} is implemented by coupling the polarization with
the path degrees of freedom. For this, we use an interferometer composed
of two beam-displacers (BD), each one transmitting (deviating) the
vertical (horizontal) polarization, and two HWPs, one at each path
mode. HWP$_{\theta}$ rotates the polarization such that part of the
light exits the interferometer in (spatial) mode $|0\rangle$ (upper
path) and part in mode $|1\rangle$ (lower path), depending on $\theta$.
HWP$_{45\lyxmathsym{º}}$ simply rotates the photons from $H$ to
$V,$ such that all photons of this mode are mapped to mode $|0\rangle$
at the output of the interferometer. Posteriorly, measurement $Y$
is performed using a HWP and a PBS. We restrict ourselves to perform
projections onto the $\sigma_{z}$ basis. This is done by fixing a
HWP at 45º to correct the polarization state such that the $H$-polarized
photons are transmitted and $V$-polarized ones are reflected. The
map $U(\tau)$ {[}Eq. (\ref{eqmap2}){]}, characterized by angles
$\tilde{\theta}$ and $\tilde{\theta}^{\prime},$ is implemented in
a similar way, noticing that slightly different dynamics take place
depending on the result of the Y measurement ($|\downarrow\rangle$
or $|\uparrow\rangle$, equivalent here to transmitted or reflected).
The photons on both path are coherently combined at the two BDs. The
final $Z$ measurement is also implemented by two sets of HWP and
PBS, one set for the transmitted light and the other to the reflected
light. The last two BDs, which are just before the detectors Det$_{2}$
and Det$_{3}$, are used to trace out the path degrees of freedom.

An example of angle values used in the experiment is shown in Fig.
\ref{fig:Angle-values-as}. Specially for the HWP$_{\tilde{\theta}^{\prime}}$
the angle changes must be performed very carefully and for this purpose
the half-wave plates responsible for the three dynamical angles $\{\theta,\tilde{\theta},\tilde{\theta}^{\prime}\}$
are motorized and moved with precision of $0.001^{\circ}$.

\begin{figure}[h]
\begin{centering}
\includegraphics[width=0.5\columnwidth]{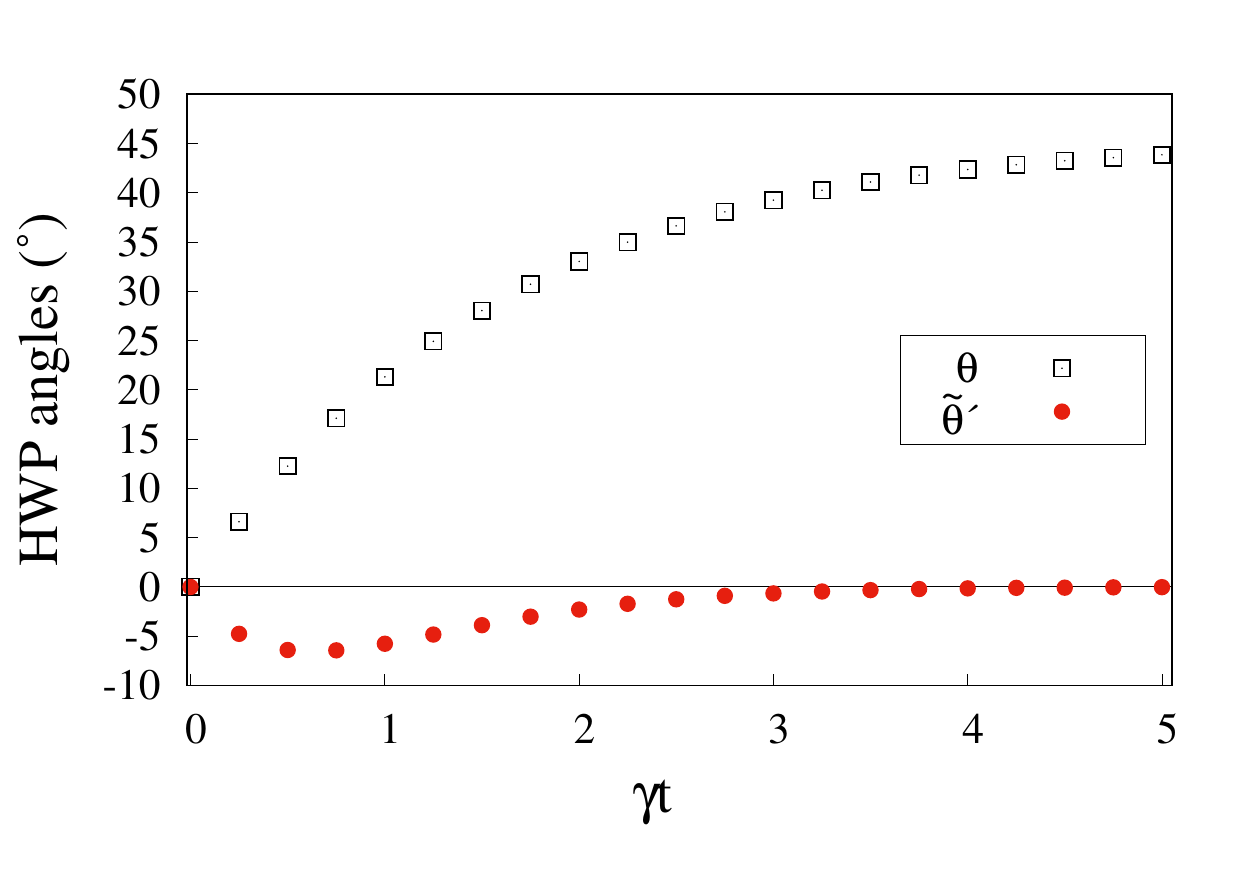}
\par\end{centering}
\caption{Calculated angle values as functions of time for the case of $\gamma\tau_{c}=0.5$.\label{fig:Angle-values-as}}

\selectlanguage{american}%
\end{figure}

From an experimental viewpoint, to condition the probabilities on
the result $y$ of the intermediate measurement means to consider
only the coincidence counts between Det$_{1}$ and Det$_{3}$ (Det$_{1}$
and Det$_{2}$) for $y=-1$ ($y=+1$). Let $N_{z,x}^{(j)}$ denotes
the number of coincidences registered between Det$_{1}$ and Det$_{j}$
when the past and future projective measurements are set to $x$ and
$z$ correspondent eigenvectors, respectively. Let also $y^{(j)}$
be the value of the intermediate outcome corresponding to Det$_{j}$.
The probabilities used to calculate the CPF correlation \eqref{eq:CPFdef}
can be obtained as $P(z,x|y^{(j)})=N_{z,x}^{(j)}/\left(\sum_{x^{\prime},z^{\prime}}N_{z^{\prime},x^{\prime}}^{(j)}\right)$,
while $P(z|y^{(j)})=\sum_{x}P(z,x|y^{(j)})$ and $P(x|y^{(j)})=\sum_{z}P(z,x|y^{(j)})$.

\subsection{Results}

In Fig. \ref{fig:CPFexperimental} we plot both the theoretical results
(full lines) as well as the experimental ones (symbols) for the CPF
correlation at equal times, $C_{pf}(t,t)|_{y=-1}.$ Both the $\hat{z}$-$\hat{z}$-$\hat{z}$
{[}Eq. (\ref{zzz}){]} and $\hat{x}$-$\hat{z}$-$\hat{x}$ {[}Eq.
(\ref{xzx}){]} measurement schemes were implemented (upper and lower
curves respectively). While for the chosen bath correlation parameters
the propagator $G(t)$ decays in a monotonous way, detection of memory
close to the BMA is confirmed for different bath correlation times
$\tau_{c}.$ An excellent agreement between theory and experiment
is observed. In particular, at time $t=0,$ null values of the CPF
correlation are experimentally observed, meaning that correlation
between the system and environment are negligible at the preparation
stage \cite{budiniChina}. While the modulus of $C_{pf}(t,t)|_{y=-1}$
depends on the initial system state, \ we note that it is smaller
in the $\hat{z}$-$\hat{z}$-$\hat{z}$ scheme when compared with
the $\hat{x}$-$\hat{z}$-$\hat{x}$ measurement scheme. In fact,
$|G(t,\tau)|^{2}\leq|\mathrm{Re}[G(t,\tau)]|$ {[}see Eqs.~(\ref{zzz})
and Eq. (\ref{xzx}){]}. This feature also reflects that in the former
case, in contrast to the last one, the dynamics between measurements
is incoherent.

\begin{figure}[h]
\centering{}\includegraphics[width=0.6\columnwidth]{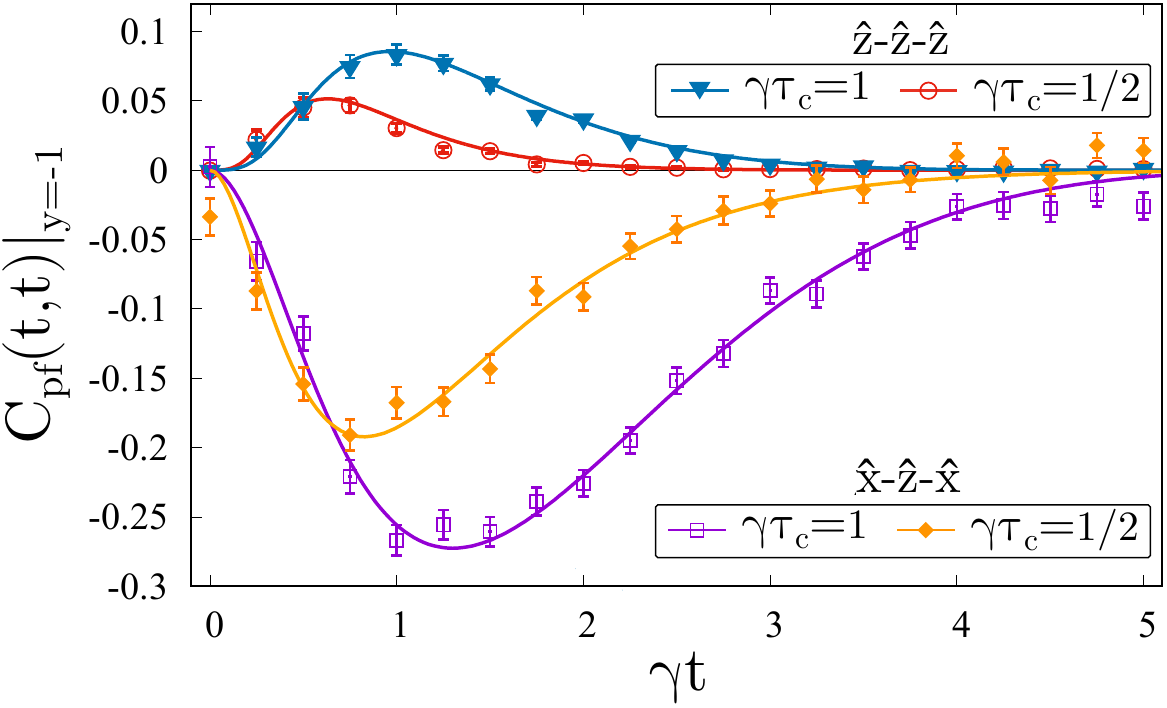}
\caption{CPF correlation for different projective measurements and bath correlation
times. Theoretical results (full lines), experimental results (symbols).
The two upper curves correspond to the $\hat{z}$-$\hat{z}$-$\hat{z}$
measurements and the lower ones to $\hat{x}$-$\hat{z}$-$\hat{x}$
measurements. The initial system state is $(\sqrt{p}\left\vert {\uparrow}\right\rangle +\sqrt{1-p}\left\vert {\downarrow}\right\rangle )$
with $p=0.8$ (upper curves) and $p=1$ (lower curves). From top to
bottom, the bath parameters are $\gamma\tau_{c}=1,$ $1/2,$ $1/2,$
$1.$}
\label{fig:CPFexperimental} 
\end{figure}

We also used the experimental setup for measuring memory effects even
closer to the BMA, that is, for smaller bath correlation times. Experimental
limitations emerge due to different aspects, as explained in the next
section . For instance, reduced visibility in the interferometers
degrades the quality of our operations, weakening agreement between
theory and experiment. The finite count statistics also become more
relevant when approaching the Markovian limit, as it becomes unclear
if a nonnull CPF comes from memory or fluctuation effects. In spite
of these limitations, our experiment demonstrates the total feasibility
of measuring quantum non-Markovian effects close and beyond the BMA.

The CPF correlation was also measured for different time intervals
for the two unitary evolution steps $t\neq\tau$ and the result is
shown in Fig. \ref{fig:CPFtTau} together with the expected theoretical
correlation. For $t=0$ the correlation vanishes meaning that the
initial state presents no correlation between system and environment.
For $t,\tau>0$ the correlation is as big as the time intervals are
close to the correlation time of the bath $\tau_{c}$. As the dynamics
goes to the Born-Markov limit ($t>>\tau_{c}$) the correlation vanishes.
This experimental results suffer even more from the aforementioned
finite statistics fluctuations as we have used approximately one fifth
of the total counts we used in the $t=\tau$ case to build each experimental
point. 

\begin{figure}[H]
\includegraphics[width=1\textwidth]{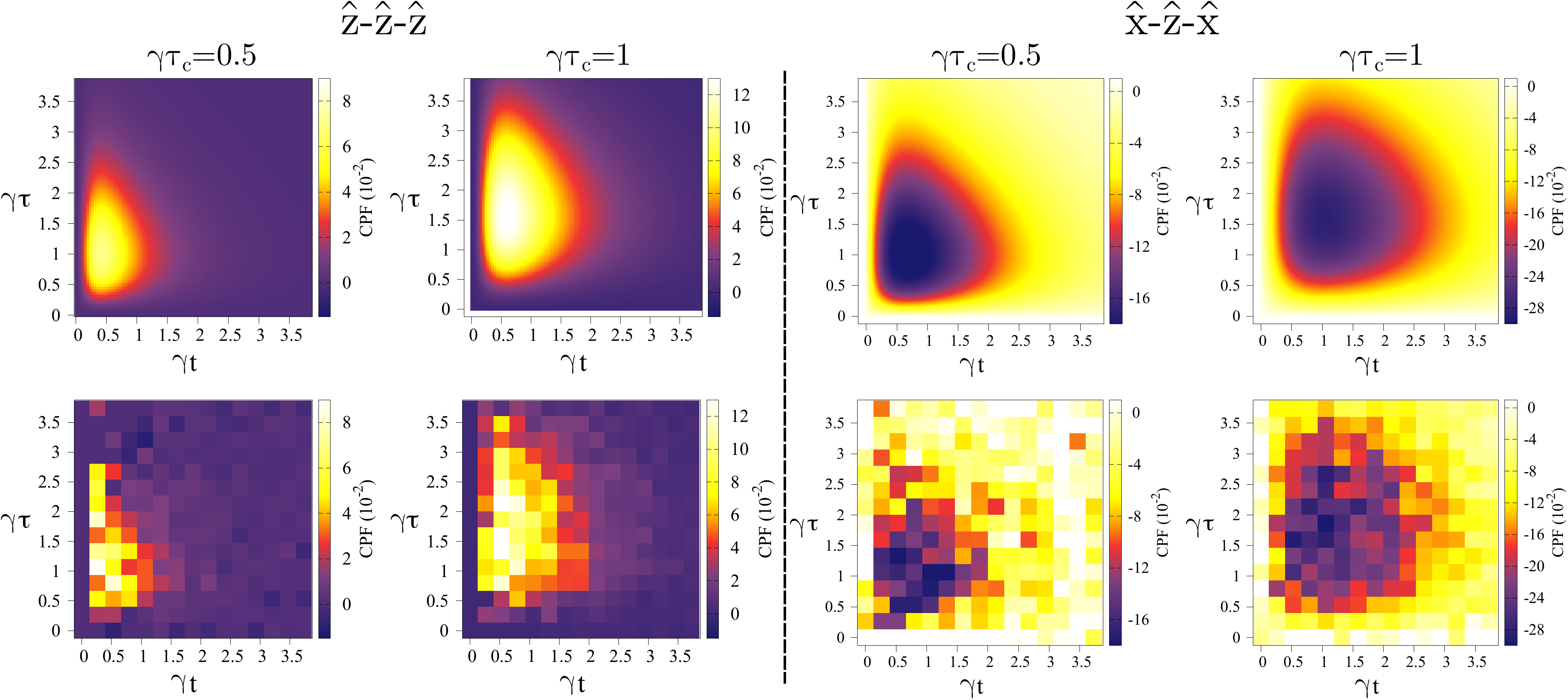}

\caption{CPF correlation as a function of the time intervals $t$ and $\tau$
of the two unitary system-environment evolutions. The upper row shows
the theoretical result and in the line below, the experimental data
is shown. The initial system state is $(\sqrt{p}\left\vert {\uparrow}\right\rangle +\sqrt{1-p}\left\vert {\downarrow}\right\rangle )$
with $p=0.8$ ($\hat{z}$-$\hat{z}$-$\hat{z}$ scheme) and $p=1$
($\hat{x}$-$\hat{z}$-$\hat{x}$ scheme).\label{fig:CPFtTau}}

\selectlanguage{american}%
\end{figure}

\subsection{Robustness of the experimental setup}

In this section we study the behavior of the CPF correlation in real
world implementations. In particular, we consider two limitations
of our experimental setup, namely the finite counts statistics and
the non-unit visibility of the interferometers. The last one is an
issue only for the $\hat{x}$-$\hat{z}$-$\hat{x}$ scheme, since
the evolution in the $\hat{z}$-$\hat{z}$-$\hat{z}$ scheme is incoherent
and no interference take place in this case. In Fig. \ref{fig:CPFsimulado}
we show results of simulations when these issues are considered. In
Fig \ref{fig:CPFsimulado}-a) we show in black hollow squares the
results for the ideal case of visibility V equals to one and infinite
counts. In red circles, we also show results for V$=1$ but considering
finite counts such as the ones we have in the experiment (around $10000$
events in total). One can see that the circles are dispersed around
the theoretical prediction, giving rise to values of the CPF correlation
up to 15$\%$ greater than what is expected theoretically. This shows
that the CPF correlation is quite sensitive to statistical fluctuations.
In Fig.\ref{fig:CPFsimulado}-b) we show results of simulations for
V$=0.9$. The results do not coincide with the theoretical prediction
even in the case of infinite counts (blue hollow squares). Moreover,
when non perfect visibility and finite counts are considered together,
experimental values could differ from theory for more than 25$\%$.
When V$=0.8$, results in Fig. \ref{fig:CPFsimulado}-c), the dispersion
of the simulated values is even larger, obtaining high discrepancy
between theory and data. As consequence, to restore the agreement
between theory and experiment it would be necessary to introduce dephasing
in the theoretical description.

\begin{figure}[h]
\centering{}\includegraphics[width=0.8\columnwidth]{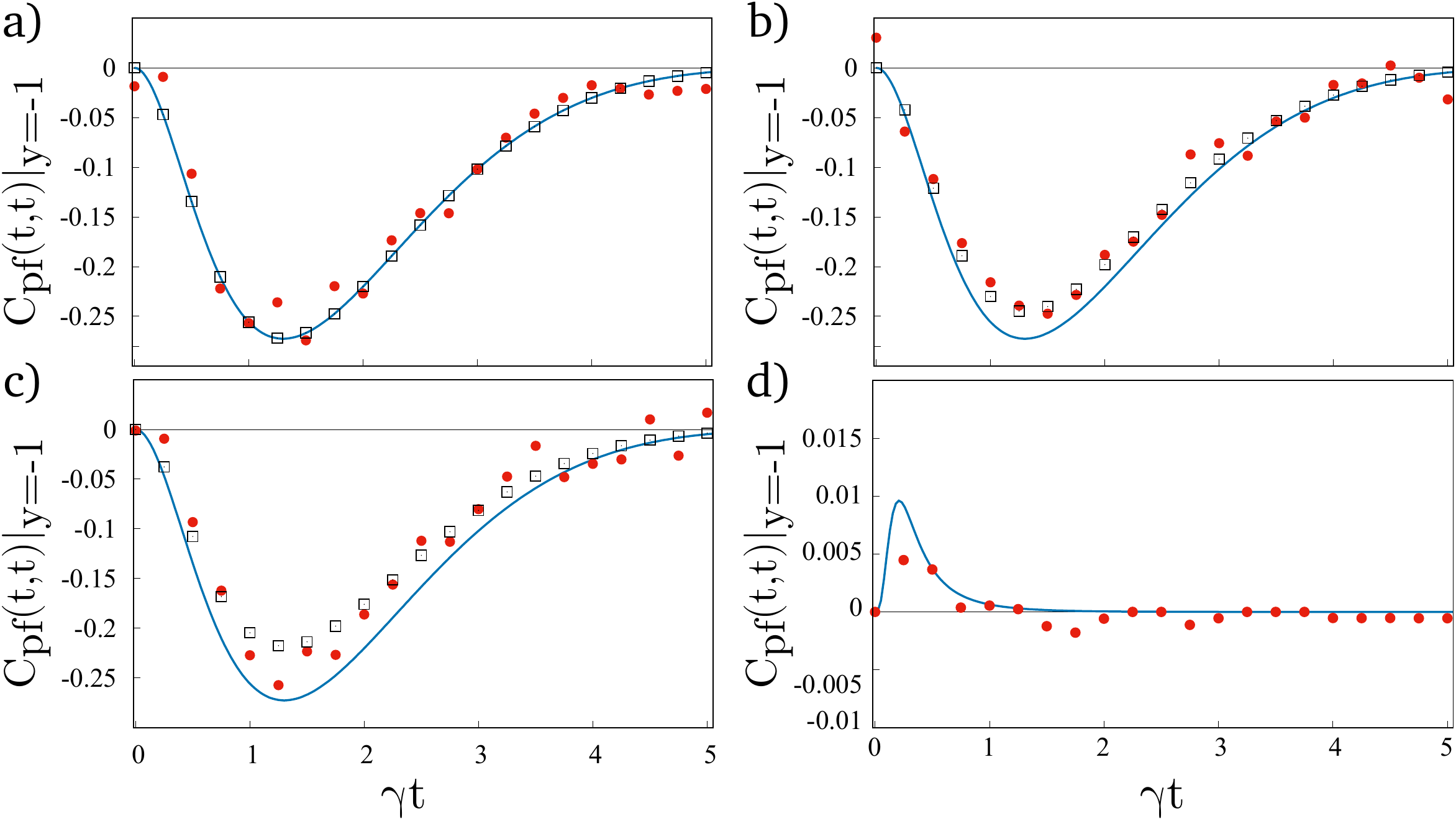}
\caption{CPF correlation as a function of time for simulated data taking into
account experimental issues. a), b) and c): $\hat{x}$-$\hat{z}$-$\hat{x}$
scheme of measurement with $\gamma\tau_{c}=1,$ and initial state
$|\uparrow\rangle$. In d) the $\hat{z}$-$\hat{z}$-$\hat{z}$ scheme
of measurement is used with initial state $\sqrt{p}|\uparrow\rangle+\sqrt{1-p}|\downarrow\rangle$,
$p=0.8$. More details in the text. }
\label{fig:CPFsimulado} 
\end{figure}


As mentioned above, we find further experimental issues closer to
BMA limit ($\tau_{c}\,\rightarrow\,0$). In Fig. \ref{fig:CPFsimulado}-d)
we show the exact value of the CPF correlation (blue solid curve)
and a theoretical simulation including finite statistic effects (red
circles) for $\tau_{c}\gamma=0.1$ in the $\hat{z}$-$\hat{z}$-$\hat{z}$
scheme of measurement. In this case, the values of CPF correlation
and its experimental variations due to fluctuations in the number
of counts are comparable. This alone prevents us to assign a non vanishing
correlation to memory effects instead of fluctuations, without considering
any other experimental issue.

\begin{figure}[t!]
\centering{}\vspace{10pt}
 \includegraphics[width=0.8\columnwidth]{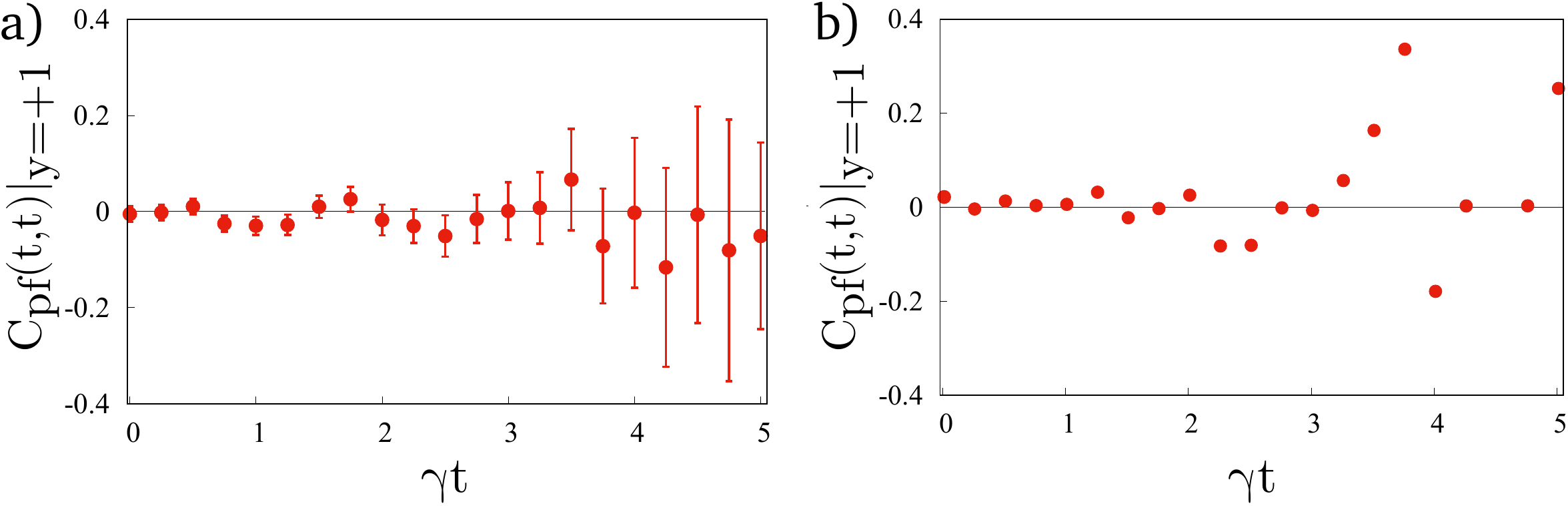}
\caption{ CPF correlation when $y=+1$ in the $\hat{x}$-$\hat{z}$-$\hat{x}$.
$\gamma\tau_{c}=1,$ and $p=1$ are used in this case. a) Experimental
data and b) simulation assuming Poissonian fluctuations. }
\label{fig:CPFzero} 
\end{figure}


In Fig. \ref{fig:CPFzero}-a), we plot the experimental values of
CPF correlation when the measurement outcome in the present is $y=+1$.
In this case, the correlation is null within the error bars, in agreement
with what is predicted theoretically. One can see that the error bars
increase substantially while time passes. This is related with the
fact that the system excitation tends to decay to the reservoir, making
the probabilities to find it in an excited state $(y=+1)$ almost
null for values of $\gamma t$ larger that 3. In our setup, this is
translated as a reduction of the number of coincidence counts, causing
the probabilities to be much more sensitive to statistical fluctuations.
The fluctuations observed experimentally are compatible with finite
count statistics as shown in Fig. \ref{fig:CPFzero}-b), where we
plot the result of a simulation assuming Poissonian fluctuations around
the ideal theoretical value of the counts.

\section{Conclusions}

Detection of quantum non-Markovianity close to the Born-Markov approximation
was characterized through an operational-based memory witness. The
CPF correlation was calculated for the decay dynamics of a two-level
system coupled to a bosonic environment. Instead of the propagator,
here the relevant object associated to memory effects consists in
the convolution of two system propagators weighted by the environment
correlation. This structure can be related to an alternative formulation
of the phenomenon of environment-to-system backflow of information,
where an intermediate condition on the system state allows to detects
memory effects even close to the validity of the BMA. A photonic experiment
corroborates the feasibility of detecting quantum memory effects close
to the BMA with excellent agreement with the theory.

These results provide a relevant contribution to the understanding
of operational-based quantum memory witnesses. In particular, our
study elucidates which structure replaces the system propagator when
studying these alternative approaches. The validity of the present
conclusions to arbitrary system-environment dynamics can be established
by using perturbation techniques \cite{bonifacio}.\selectlanguage{american}

~\ihead{}

\ohead{\textbf{Chapter~\thechapter}~\leftmark}

\ifoot{}

\cfoot{}

\ofoot[
]{\thepage}

\chapter{Experimental realization of an arbitrary qubit channel: a proposal\label{chap:GAD}}

In Section \ref{sec:Open-quantum-system} we briefly discussed the
general evolution of open quantum systems and how it is given by maps
that can be regarded as unitary evolutions for larger systems. In
that chapter however, the focus was to use this treatment to present
the concept of quantum non-Markovianity. In this chapter, we are concerned
only about the completely positive maps themselves, also called quantum
channels. Here we show how to use the experimental platform presented
in chapters 5 and 6 to implement a large class of quantum channels
for a photonic qubit using the path degrees of freedom of the photon
itself as ancillary systems, which allows us to introduce decoherence
in a controllable way in the qubit evolution (the polarization degree
of freedom of the photon). This is still only a proposal for an experiment
that we hope will be implemented soon.

Aiming to be self-contained, the chapter starts with a brief review
of quantum channels already presented in Chapter \ref{chap:Markov},
followed by the particular case of qubit transformations. Before presenting
the most general protocol, we motivate it in Section \ref{sec:Generalized-amplitude-damping}
by showing how to implement a particular channel called ``generalized
amplitude damping''. The general protocol is presented in Section
\ref{sec:General-protocol}, together with many examples of how to
set the experimental parameters to implement some particular channels.
The chapter ends with a brief description of quantum process tomography,
that should be used to verify what channel is indeed being realized
by the setup.

This work is being developed in collaboration with Gabriel Aguilar
(UFRJ), Gabriel Landi (USP), and the students of the Quantum Optics
Lab Rodrigo Piera and Thiago Guimarães. I contributed to the theoretical
part and design of the experiment. At this moment the experiment is
being set up.

\section{Quantum maps}

In Chapter \ref{chap:Markov}, we have introduced the concept of quantum
maps: the mathematical transformations leading from an initial to
a final state. Although there we presented this by tracing out the
environmental degrees of freedom from the global time-evolved state,
once defined, the map does not need to make reference to the passage
of time. 

Let us consider a system of interest whose Hilbert space is $\mathcal{H}$
and $S(\mathcal{H})$ denotes the set of positive self-adjoint operators
acting on $\mathcal{H}$ with trace $1$, i.e. $S(\mathcal{H})$ is
the set of all possible density matrices of the system. A quantum
map $\Lambda:S(\mathcal{H})\rightarrow S(\mathcal{H})$ describes
mathematically the state transformation from a initial state $\rho\in S(\mathcal{H})$
to a final state $\Lambda[\rho]\in S(\mathcal{H})$, without necessarily
mentioning the underlying physical process or the time it takes. To
be a proper quantum map, $\Lambda$ must keep the basic features which
define a density matrix:
\begin{itemize}
\item It must be trace preserving such that the total probability is conserved;
\item It must be positive such that the transformed state is a positive
semidefinite operator, which means that all the probabilities are
always positive numbers;
\item As the system can be a part in a bigger system , it should maintain
the positivity also of a global state in a joint Hilbert space $\mathcal{H}\otimes\mathcal{H}_{e}$,
regardless the dimension of the environment Hilbert space $\mathcal{H}_{e}$.
This feature is called complete positivity\footnote{In Section \ref{sec:Open-quantum-system}, it was mentioned that it
is not always possible to define a complete positive map if the initial
state has entanglement between system and environment. Here we disregard
these cases when we considered maps defined in the whole set of states
$S(\mathcal{H})$.}.
\end{itemize}
An operation obeying all this requirements is called a completely-positive
and trace-preserving (CPTP) map. If the nonunitary dynamics is the
result of a partial trace of a unitarily evolved global state, then
all these requirements are fulfilled. However, it is not always mathematically
friendly to solve or even to enunciate the global evolution problem.
Therefore it is often convenient to heuristically find the map and
to deal only with the smaller system problem, ignoring the source
of dissipation and decoherence. 

Every completely positive map has a operator-sum decomposition
\begin{equation}
\Lambda[\rho]=\sum_{k}E_{k}\rho E_{k}^{\dagger},\label{eq:OSD}
\end{equation}
which is not unique \cite{kraus1983}. The $E_{k}$s are called Kraus
operators. The trace preservation condition $\Tr{\Lambda[\rho]}=\sum_{k}\Tr{E_{k}E_{k}^{\dagger}\rho}=\Tr{\rho}$,
valid for all $\rho\in S(\mathcal{H})$, is attained if $\sum_{k}E_{k}E_{k}^{\dagger}=\mathbb{1}$.
The number of Kraus operators required to represent a map is not fixed,
but there is a minimal number of operators which is at most equal
to $d^{2}$ (the Hilbert space dimension squared) \cite{choi1975}.
If the initial state of the environment and the global Hamiltonian
evolution are known, the operator-sum representation can be obtained
as shown in Eq. \eqref{eq:Kraus}. It is possible to see that the
maximum number of Kraus operators depends on the Hilbert space dimension
of the environment and that the Kraus decomposition is not unique
since it depends on the environmental Hilbert space basis being used.

There is an interpretation for the operator-sum representation in
terms of measurements of the environment. Consider that the total
system is initially in the separable state $\rho\otimes\ket{e_{0}}\bra{e_{0}}$,
which evolves under the unitary $U$. If the environment is measured
and its state after the intervention is $\ket{e_{k}}$, then the system
state becomes proportional to
\begin{equation}
\Tr_{e}\{\ket{e_{k}}\bra{e_{k}}U(\rho\otimes\ket{e_{0}}\bra{e_{0}})\ket{e_{k}}\bra{e_{k}}\}.
\end{equation}
It is easy to see from Eq. \eqref{eq:Kraus} that the normalized state
is given by 
\[
\rho_{k}=\frac{E_{k}\rho E_{k}^{\dagger}}{\Tr\left[E_{k}\rho E_{k}^{\dagger}\right]}.
\]
Now, the probability of getting the outcome $k$ is given by 
\begin{equation}
p(k)=\Tr\left[\ket{e_{k}}\bra{e_{k}}U(\rho\otimes\ket{e_{0}}\bra{e_{0}})\ket{e_{k}}\bra{e_{k}}\right]=\Tr\left[E_{k}\rho E_{k}^{\dagger}\right].
\end{equation}
Thus, if the measurement outcome is kept unrevealed , the state of
the system is 
\begin{equation}
\rho^{\prime}=\sum_{k}p(k)\rho_{k}\label{eq:unrevealing}
\end{equation}
which is exactly the operator-sum decomposition \eqref{eq:OSD}. Thus
the effect of the quantum channel is equivalent to taking the initial
state and randomly replacing it by the states $\rho_{k}$, the inherent
randomness coming from the unknown measurement of the reservoir.

Once pursuing a Kraus decomposition, it is possible to solve the converse
problem and obtain a unitary transformation in a larger Hilbert space
that gives rise to the reduced transformation given by the map. Many
environmental dimensions and interactions may give rise to the same
dynamics for the main system. Considering that the environment is
initially in a pure state, the minimal dimension of the environment
Hilbert space is given by the minimal number of Kraus operators in
the map decomposition. In this case, the unitary evolution operator
$U$ can be obtained as to satisfy
\begin{equation}
U\ket{\psi}\ket{e_{0}}=\sum_{k}E_{k}\ket{\psi}\ket{e_{k}},
\end{equation}
where $\ket{\psi}$ is an arbitrary pure state of the main system,
the sum runs over the orthogonal basis states $\{e_{k}\}$ of the
environment, and the environment is initially in the state $\ket{e_{0}}$
of the basis \cite{chuang2000}. In this construction, $U$ is not
uniquely defined because its action is not prescribed for the basis
states but $\ket{e_{0}}$. 

In what follows we restrict ourselves to the case of a bi-dimensional
Hilbert space. It is tempting to think that in this case it is possible
to build any map from the unitary evolution with a qubit environment
initially prepared in a mixed state \cite{Horodecki1999}. Although
there are many quantum channels for which it is indeed possible, there
are indeed a few well-known and relevant counterexamples \cite{Terhal1999}.

\section{Quantum maps of qubits}

The simplest but also one of the most interesting systems for quantum
computation consists of a two level system or a qubit. As well as
the qubit is the basic unit for unitary quantum computation, it can
be thought of as the primitive for open quantum system protocols as
well \cite{wang2013}.

The identity and the Pauli matrices form a basis for $2\times2$ matrices
with complex coefficients. In particular, any density matrix, i.e.,
any positive Hermitian operator with trace equals to one can be decomposed
in this basis as $\rho=\frac{1}{2}\left[\mathbb{1}+\mathbf{r}\cdot\boldsymbol{\sigma}\right]$
with $\mathbf{r}\in\mathbb{R}^{3}$ and $|\mathbf{r}|\leq1$ to ensure
positivity, $\boldsymbol{\sigma}$ denotes a vector with the three
Pauli matrices $\sigma_{x}$, $\sigma_{y}$ and $\sigma_{z}$ as components.
It defines the so called Bloch sphere, a 3-dimensional sphere with
radius equals to unity, inside which all the qubit states are uniquely
represented through their vector $\mathbf{r}$, with the pure states
all lying on the surface. As a positive trace preserving transformation
takes states into states, it must change only the 3-dimensional vector
$\mathbf{r}$, thus it consists of rotations, reflections, contractions
and translations provided that the vector stays inside the Bloch sphere.
The map can thus be represented as
\begin{equation}
\Lambda[\rho]=\Lambda\left[\frac{1}{2}\left[\mathbb{1}+\mathbf{r}\cdot\boldsymbol{\sigma}\right]\right]=\frac{1}{2}\left[\mathbb{1}+(\mathbf{t}+T\mathbf{r})\cdot\boldsymbol{\sigma}\right],\label{eq:mapbloch}
\end{equation}
where $\mathbf{t}$ is a real 3D vector and $T$ is a real $3\times3$
matrix \cite{king2001}. It must be pointed that not all maps \eqref{eq:mapbloch}
admit a Kraus form or are completely positive. In this parameterization
it is clear that any map is characterized by at most 12 independent
parameters (the 9 elements of $T$ plus the 3 components of $\mathbf{t}$)
which is of course the same number of parameters of a Kraus decomposition
as it has at most four $2\times2$ matrices with a $2\times2$ completeness
conditions. 

Before moving to the simulation of a more general channel, let us
present a particular case, the generalized amplitude damping channel,
as a motivation for the more general protocol.

\section{Generalized amplitude damping channel\label{sec:Generalized-amplitude-damping}}

The generalized amplitude damping (GAD) channel for one qubit is defined
by its action over a qubit state $\rho$ 
\begin{equation}
\Lambda_{GAD}[\rho]=\sum_{j=1}^{4}E_{j}\rho E_{j}^{\dagger}\label{eq:gad}
\end{equation}
with the four Kraus operators \cite{chuang2000}
\begin{align}
E_{1}=\sqrt{p}\left[\begin{array}{cc}
1 & 0\\
0 & \sqrt{\eta}
\end{array}\right]\qquad & E_{2}=\sqrt{p}\left[\begin{array}{cc}
0 & \sqrt{1-\eta}\\
0 & 0
\end{array}\right]\label{eq:GADKraus}\\
E_{3}=\sqrt{1-p}\left[\begin{array}{cc}
\sqrt{\eta} & 0\\
0 & 1
\end{array}\right]\qquad & E_{4}=\sqrt{1-p}\left[\begin{array}{cc}
0 & 0\\
\sqrt{1-\eta} & 0
\end{array}\right],\nonumber 
\end{align}
where the channel parameters $\eta$ and $p$ are positive numbers
in the interval $[0,1]$. When $p=1$ it becomes the so-called amplitude
damping channel, which characterizes the interaction of the qubit
with a bath initially with no excitation, i.e. at zero temperature.
If the system is prepared in the excited state $\ket{1}$, then it
can decay and emit one excitation to the bath with probability $1-\eta$.
This channel was presented before as the map \eqref{eqmap1}. Many
works exploring the effects of the amplitude damping channel over
entanglement with a third party and system-environment entanglement
dynamics as well as quantum Markovianity have been published using
as basic tool for photonic simulations \cite{Farias2012,brasil,Knoll2016,Almeida2007,Farias2009,Haseli2014,Salles2008-2}.
Although this particular case channel has been widely studied experimentally,
the more general GAD is lacking a proper implementation as far as
we know.

When applied to a pure state $\ket{\psi}=a\ket{0}+b\ket{1}$, the
map \eqref{eq:gad} leads to the mixed state 
\begin{multline}
\Lambda_{GAD}[\ket{\psi}\bra{\psi}]=|a|^{2}\left\{ p\rho_{0}+(1-p)\left[\eta\rho_{0}+(1-\eta)\rho_{1}\right]\right\} \\
+|b|^{2}\left\{ (1-p)\rho_{1}+p\left[\eta\rho_{1}+(1-\eta)\rho_{0}\right]\right\} +\sqrt{\eta}\left(ab^{*}\ket{0}\bra{1}+a^{*}b\ket{1}\bra{0}\right),\label{result}
\end{multline}
where $\rho_{0}=\ketbra{0}{0}$ and $\rho_{1}=\ketbra{1}{1}$. Thus
a qubit originally in state $\ket{1}$ ($a=0$) has probability $1-p$
to remain in this state and probability $p$ of passing through a
simple amplitude damping with coupling $1-\eta$ between system and
environment. A similar statement is valid for initial state $\ket{0}$,
but in this case the system can absorb one excitation from the bath.
If the original state is a coherent combination of the two states
of the basis, then the channel reduces its coherence by a factor $\sqrt{\eta}$.

Let us encode the qubit in the polarization degree of freedom of photons
\[
\ket{0}\longrightarrow\ket{H}\qquad\ket{1}\longrightarrow\ket{V}.
\]

\begin{figure}[h]
\centering \includegraphics[width=0.9\columnwidth]{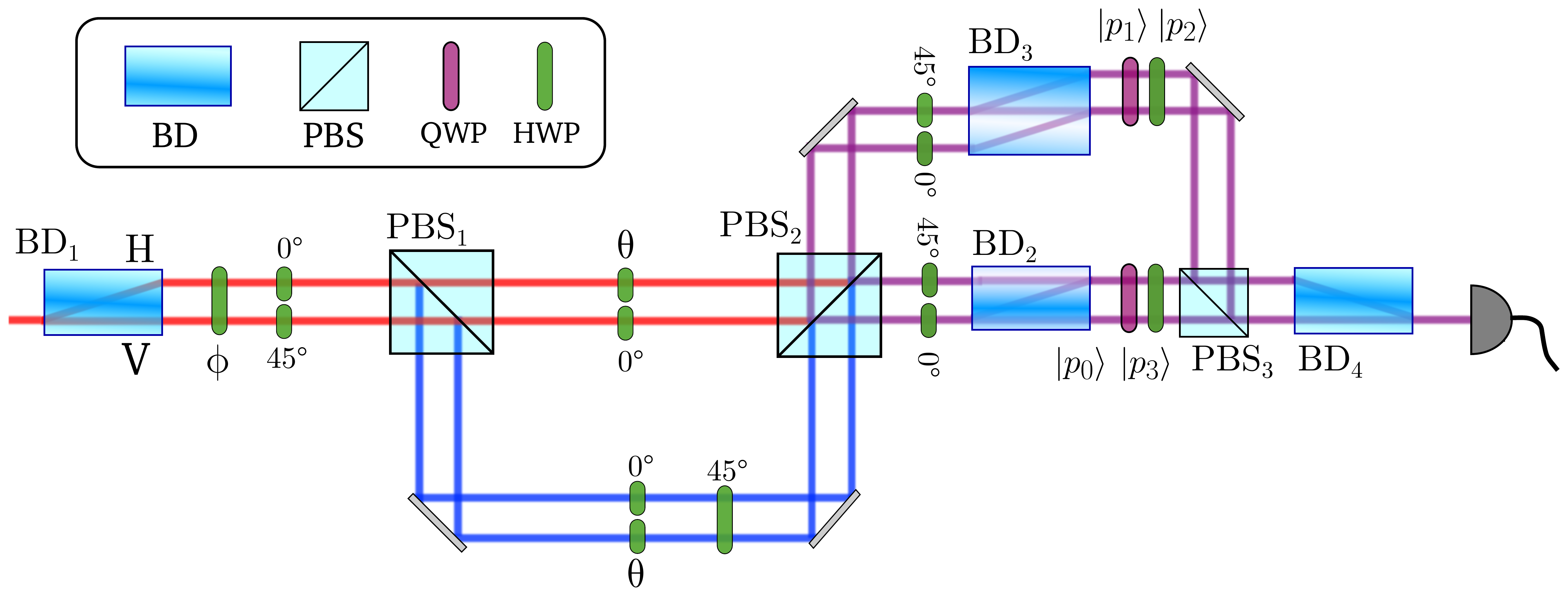}
\caption{Implementation of GAD channel for a qubit encoded in the polarization
of single photons.\label{fig:setupGad} }
\end{figure}

A setup implementing the GAD channel is shown in Fig.\ref{fig:setupGad}.
Consider a pure initial state $\ket{\psi}=\alpha\ket{H}+\beta\ket{V}$.
In the sequence that follows, the state of the photon is written after
each optical element:
\begin{itemize}
\item Beam displacer $BD_{1}$: creates a path qubit by displacing only
the horizontal polarization. Two parallel paths, up and down, come
out of this element. 
\begin{equation}
a\ket{H}+b\ket{V}\longrightarrow a\ket{H}\ket{u}+b\ket{V}\ket{d}
\end{equation}
\item Half wave plate set to an angle $\phi/2$ (the angle is chosen such
that $\sin{\phi}^{2}=p$) followed by half wave plates at $45^{\circ}$
(down path) and $0^{\circ}$ (up path, only compensates the optical
path difference caused by the other plate). 
\begin{equation}
a\ket{H}\ket{u}+b\ket{V}\ket{d}\longrightarrow a(\cos{\phi}\ket{H}-\sin{\phi}\ket{V})\ket{u}+b(-\cos{\phi}\ket{H}+\sin{\phi}\ket{V})\ket{d}
\end{equation}
\item Polarizing beam splitter ($PBS_{1}$): it creates two new paths, which
we call short ($s$) and long ($l$), by transmitting horizontally
polarized photons while reflecting the vertically polarized ones.
Because of the short coherence length of the heralded photons, the
optical path difference is enough to cause decoherence between the
two paths. This decoherence is already being considered when we attribute
orthogonal states to each path, which are traced out at the end of
the unbalanced interferometer. 
\begin{equation}
\longrightarrow a(\cos{\phi}\ket{H}\ket{s}-\sin{\phi}\ket{V}\ket{l})\ket{u}+b(-\cos{\phi}\ket{H}\ket{s}+\sin{\phi}\ket{V}\ket{l})\ket{d}
\end{equation}
\item Half wave plates at angle $\theta/2$ such that $\cos^{2}{\theta}=\eta$
\begin{align}
\longrightarrow a\left[\cos{\phi}(\cos{\theta}\ket{H}+\sin{\theta}\ket{V})\ket{s}+\sin{\phi}\ket{V}\ket{l}\right]\ket{u}\\
+b\left[-\cos{\phi}\ket{H}\ket{s}+\sin{\phi}(\sin{\theta}\ket{H}-\cos{\theta}\ket{V})\ket{l}\right]\ket{d}\nonumber 
\end{align}
\item Half wave plate at $45^{\circ}$ only in the long path 
\begin{align}
\longrightarrow a\left[\cos{\phi}(\cos{\theta}\ket{H}+\sin{\theta}\ket{V})\ket{s}+\sin{\phi}\ket{H}\ket{l}\right]\ket{u}\\
+b\left[-\cos{\phi}\ket{H}\ket{s}+\sin{\phi}(\sin{\theta}\ket{V}-\cos{\theta}\ket{H})\ket{l}\right]\ket{d}\nonumber 
\end{align}
\item Polarizing beam splitter ($PBS_{2}$): although there are only two
output ports after the PBS, we add two extra outputs to take into
account the path difference decoherence. Let us make the following
associations $\ket{H}\ket{s}\rightarrow\ket{H}\ket{p0}$, $\ket{V}\ket{s}\rightarrow\ket{V}\ket{p1}$,
$\ket{H}\ket{l}\rightarrow\ket{H}\ket{p2}$ and $\ket{V}\ket{l}\rightarrow\ket{V}\ket{p3}$.
The real paths are indicated in Fig.\ref{fig:setupGad}. 
\begin{multline}
\longrightarrow a\left[\cos{\phi}(\cos{\theta}\ket{H}\ket{p0}+\sin{\theta}\ket{V}\ket{p1})+\sin{\phi}\ket{H}\ket{p2}\right]\ket{u}\\
+b\left[-\cos{\phi}\ket{H}\ket{p0}+\sin{\phi}(\sin{\theta}\ket{V}\ket{p3}-\cos{\theta}\ket{H}\ket{p2})\right]\ket{d}
\end{multline}
\begin{multline}
=\left[a\cos{\phi}\cos{\theta}\ket{H}\ket{u}-b\cos{\phi}\ket{H}\ket{d}\right]\ket{p0}+a\cos{\phi}\sin{\theta}\ket{V}\ket{u}\ket{p1}\\
+\left[a\sin{\phi}\ket{H}\ket{u}-b\sin{\phi}\cos{\theta}\ket{H}\ket{d}\right]\ket{p2}+b\sin{\phi}\sin{\theta}\ket{V}\ket{d}\ket{p3}
\end{multline}
\item Half wave plates at $45^{\circ}$ and beam displacers: in order to
assure the right reduction on the final coherence we coherently recombine
the polarization components on paths 0 and 2 by using the wave plates
and beam displacers $BD_{2}$ and $BD_{3}$. To compensate the optical
path difference created by $BD_{1}$ only the down path is displaced.
\begin{multline}
\longrightarrow\left[a\cos{\phi}\cos{\theta}\ket{V}\ket{u}-b\cos{\phi}\ket{H}\ket{u}\right]\ket{p0}+a\cos{\phi}\sin{\theta}\ket{H}\ket{u}\ket{p1}\\
+\left[a\sin{\phi}\ket{V}\ket{d}-b\sin{\phi}\cos{\theta}\ket{H}\ket{d}\right]\ket{p2}-b\sin{\phi}\sin{\theta}\ket{V}\ket{d}\ket{p3}
\end{multline}
\item Recombining the paths all together incoherently is equivalent to tracing
out the path degrees of freedom, what leads to 
\begin{multline}
|a|^{2}\left[\sin^{2}{\phi}\ketbra{V}{V}+\cos^{2}{\phi}(\cos^{2}{\theta}\ketbra{V}{V}+\sin^{2}{\theta}\ketbra{H}{H})\right]\\
+|b|^{2}\left[\cos^{2}{\phi}\ketbra{H}{H}+\sin^{2}{\phi}(\cos^{2}{\theta}\ketbra{H}{H}+\sin^{2}{\theta}\ketbra{V}{V})\right]\\
-\cos{\theta}(ab^{*}\ketbra{V}{H}+a^{*}b\ketbra{H}{V}).
\end{multline}
The last expression is equivalent to \eqref{result} with the given
relations for the wave plate angles if we invert the polarization.
Instead of adding more HWP to correct this, we can take it into consideration
when setting the angles of the pair of HWP and QWP we use before $PBS_{3}$
to make measurements on polarization. In this scheme, the last PBS
serves not only to trace out some of the path information but also
to perform the measurement. The remaining path information is erased
by $BD_{4}$.
\end{itemize}

\subsection{Accessing the environment state}

First of all, it is necessary to identify what the environment is.
If the combination system plus environment is considered as a closed
system, then its evolution must be given by a unitary transformation.
We can consider a separable initial state $\rho_{se}$ that evolves
through the global unitary $U$: 
\begin{equation}
\rho_{se}=\rho_{s}\otimes\rho_{e}\longrightarrow U\rho_{s}\otimes\rho_{e}U^{\dagger}.
\end{equation}
The transformation over the system state is recovered by tracing out
the environment: 
\begin{equation}
\rho_{s}\longrightarrow\sum_{k}\bra{k}U\rho_{s}\otimes\rho_{e}U^{\dagger}\ket{k},\label{gen_map}
\end{equation}
where $\{\ket{k}\}$ is a orthonormal basis of the environment Hilbert
space.

If $\rho_{e}$ is a pure state, say $\rho_{e}=\ketbra{0}{0}$, then
we identify the Kraus operators with $E_{k}=\bra{e_{k}}U\ket{0}$
and the allowed number of independent Kraus operators is equal to
the environment Hilbert space dimension. For example, a GAD channel
would require a four-dimensional environment. Instead of using this,
let us recall the existence of a isomorphism between CPTP maps for
qudit states and density matrices of two qudits \cite{Horodecki1999}.
This isomorphism does not exist for all qubit channels \cite{Terhal1999},
but particularly for the GAD it does. For this isomorphism to hold
we need to allow the second qudit to be in a mixed initial state.

Thinking about the physical interpretation of the GAD channel, it
makes sense to consider the environment in a mixed initial state,
it is actually desired to be in a thermal state whose temperature
determines the channel parameter $p$. Since a thermal state includes
an infinite number of states for the reservoir, we can identify a
qubit whose state $\ket{0}$ corresponds to the ground state of the
bath and all excited states are encoded in qubit state $\ket{1}$
\begin{equation}
\rho_{e}=\frac{e^{-\beta\mathcal{H}_{e}}}{\sum_{k}e^{-\beta\mathcal{E}_{k}}}=\frac{e^{-\beta\mathcal{E}_{0}}\ketbra{0}{0}}{\sum_{k}e^{-\beta\mathcal{E}_{k}}}+\frac{\sum_{j}e^{-\beta\mathcal{E}_{j}}\ketbra{j}{j}}{\sum_{k}e^{-\beta\mathcal{E}_{k}}}\longrightarrow p\ketbra{0}{0}+(1-p)\ketbra{1}{1},
\end{equation}
 with $p=e^{-\beta\mathcal{E}_{0}}/{\sum_{k}e^{-\beta\mathcal{E}_{k}}}$.
The energy values $\mathcal{E}_{k}$ are the eigenvalues of the environment
free Hamiltonian $\mathcal{H}_{e}$ and $\beta^{-1}$ is its temperature.

Thus we can consider the reservoir as being a qubit and the initial
state of larger system (system and reservoir) as $\rho_{se}=\rho_{s}\otimes\left[p\ketbra{0}{0}+(1-p)\ketbra{1}{1}\right]$.
Plugging this state into equation (\ref{gen_map}) leads to 
\begin{equation}
\rho_{S}\longrightarrow\sum_{k=0,1}\left[\sqrt{p}\bra{k}U\ket{0}\right]\rho_{S}\left[\sqrt{p}\bra{0}U^{\dagger}\ket{k}\right]+\left[\sqrt{1-p}\bra{k}U\ket{1}\right]\rho_{S}\left[\sqrt{1-p}\bra{1}U^{\dagger}\ket{k}\right].
\end{equation}
This is a CPTP map with Kraus operators 
\begin{equation}
E_{1}=\sqrt{p}\bra{0}U\ket{0}\qquad E_{2}=\sqrt{p}\bra{1}U\ket{0}\qquad E_{3}=\sqrt{1-p}\bra{1}U\ket{1}\qquad E_{4}=\sqrt{1-p}\bra{0}U\ket{1}.\label{UKraus}
\end{equation}

Let us consider the following map (the first entry is the system and
the second one is the reservoir) 
\begin{eqnarray}
\ket{00} & \longrightarrow & \ket{00}\nonumber \\
\ket{01} & \longrightarrow & \sqrt{\eta}\ket{01}+\sqrt{1-\eta}\ket{10}\label{map}\\
\ket{10} & \longrightarrow & -\sqrt{1-\eta}\ket{01}+\sqrt{\eta}\ket{10}\nonumber \\
\ket{11} & \longrightarrow & \ket{11}\nonumber 
\end{eqnarray}
associated to the two-qubit unitary transformation 
\begin{equation}
U=\left[\begin{array}{cccc}
1 & 0 & 0 & 0\\
0 & \sqrt{\eta} & -\sqrt{1-\eta} & 0\\
0 & \sqrt{1-\eta} & \sqrt{\eta} & 0\\
0 & 0 & 0 & 1
\end{array}\right].
\end{equation}
Using \eqref{UKraus} we see that this map produces the GAD channel.

In our experiment, we would like to be able to monitor also the environment,
tracing out the system. In order to do so, first we need to identify
what are the optical states representing each state of the environment
in our description above. It cannot be the up and down paths since
they were inserted only as ancillaries allowing the amplitude damping
to be realized for the two polarization components. It remains to
consider the output ports of the PBS's. If we recombined the up and
down paths after $PBS_{1}$ we would find (remember that short and
long paths do not recombine coherently) 
\begin{equation}
(\alpha\ket{V}-\beta\ket{H})(\alpha^{*}\bra{V}-\beta^{*}\bra{H})\otimes(\cos^{2}\phi\ketbra{s}{s}+\sin^{2}\phi\ketbra{l}{l}),
\end{equation}
which is the initial state we want (after half wave plate transformations)
and path $\ket{l}$ ($\ket{s}$) is the state $\ket{0}$ ($\ket{1}$)
of the environment. Now we can analyze the output ports of $PBS_{2}$.
For that aim we consider the extreme cases $p=0$ ($\phi=0$ and reservoir
initially in excited state) and $p=1$ ($\phi=90^{\circ}$ and reservoir
initially in ground state). Each joint basis vector evolves as follows:
\begin{eqnarray}
\ket{H}\ket{l} & \longrightarrow & \ket{H}\ket{p2}\nonumber \\
\ket{H}\ket{s} & \longrightarrow & \sqrt{\eta}\ket{H}\ket{p0}+\sqrt{1-\eta}\ket{V}\ket{p1}\label{mapexp}\\
\ket{V}\ket{l} & \longrightarrow & -\sqrt{1-\eta}\ket{V}\ket{p3}+\sqrt{\eta}\ket{H}\ket{p2}\nonumber \\
\ket{V}\ket{s} & \longrightarrow & \ket{H}\ket{p0}\nonumber 
\end{eqnarray}

Notice that the final polarization generated from the initial states
that are vertically polarized is inverted, it is corrected by the
wave plates before the beam displacers. Comparing (\ref{map}) with
(\ref{mapexp}) leads to the conclusion that the paths coming out
from the up port of $PBS_{2}$ are related to reservoir state $\ket{0}$.
In the same way, the right paths represent reservoir state $\ket{1}$.
Measuring the populations on each reservoir states can be realized
by projecting the photons in each output of $PBS_{2}$ regardless
of the polarization. To obtain the coherence between environment states
$\ket{0}$ and $\ket{1}$, we first notice that the only terms which
can produce coherence are those coming from the same initial reservoir
state (because they start in an incoherent superposition) and with
the same final polarization state (because we are tracing out the
polarization in this case). This can be done by correcting the polarization
of the initially vertical state (down path) before $PBS_{2}$. This
makes the final polarization components associated with different
reservoir states coming from the same path (short or long) to recombine
in one of the beam displacers $BD_{2}$ and $BD_{3}$.

\section{General protocol\label{sec:General-protocol}}

The setup proposed for the GAD channel can be extended to implement
a more general class of maps. It is clear that if one observes that
the same angle $\theta$ is set in the HWP unitaries in up short and
down long paths, also no transformation is carried out for down-short
and up-long paths (Fig. \ref{fig:setupGad}). In fact, slightly changing
this setup it is possible to implement any qubit channel of a large
class as we argue in this section.

Recall the parameterization for qubit maps given by Eq. \eqref{eq:mapbloch}.
It would be more convenient if the matrix $T$ were diagonal. In fact,
it can be diagonalized through its singular value decomposition (SVD)
which asserts that any real $n\times n$ matrix $T$ can be written
as the product 
\begin{equation}
T=R_{1}DR_{2}^{T}
\end{equation}
of two rotation matrices $R_{1}$ and $R_{2}^{T}$ and a diagonal
matrix $D$ \footnote{Actually the SVD is more general and applies for any rectangular matrix.
Also, it states that any matrix can be written as the product of two
orthogonal matrices and a diagonal positive semidefinite matrix, but
as any orthogonal matrices are either a rotation or a product of a
rotation by a inversion, our statement is correct if we disregard
the positivity of $D$.}, $^{T}$ denotes transposition. Moreover, a rotation of the Bloch
sphere vector is the effect of a unitary operation over the density
matrix. That said, the map can be rewritten as 
\begin{equation}
\Lambda[\rho]=U_{1}\left(\Lambda_{\mathbf{t}^{\prime},D}\left[U_{2}\rho U_{2}^{\dagger}\right]\right)U_{1}^{\dagger},
\end{equation}
where $U_{1}$ and $U_{2}$ are the unitary operators associated with
the rotations $R_{1}$ and $R_{2}$, respectively, and $\Lambda_{\mathbf{t}^{\prime},D}$
is the map parameterized by 
\begin{equation}
\mathbf{t}^{\prime}=R_{1}^{T}\mathbf{t}\qquad T^{\prime}=D=\left(\begin{array}{ccc}
\lambda_{1} & 0 & 0\\
0 & \lambda_{2} & 0\\
0 & 0 & \lambda_{3}
\end{array}\right).\label{eq:diago}
\end{equation}
The unitary transformations do not alter complete positivity, consequently,
any analysis made over $\Lambda_{\mathbf{t}^{\prime},D}$ can be automatically
extended to $\Lambda$. The image of such a channel when applied to
the Bloch sphere vectors is the ellipsoid 
\begin{equation}
\left(\frac{x_{1}^{\prime}-t_{1}^{\prime}}{\lambda_{1}}\right)^{2}+\left(\frac{x_{2}^{\prime}-t_{2}^{\prime}}{\lambda_{2}}\right)^{2}+\left(\frac{x_{3}^{\prime}-t_{3}^{\prime}}{\lambda_{3}}\right)^{2}\leq1,\label{eq:ellipsoid}
\end{equation}
which must be contained inside the Bloch sphere for positivity preservation,
a necessary condition for this is $|t_{k}^{\prime}|+|\lambda_{k}|\leq1$,
$k=1,2,3$. An example of the image of a qubit map is shown in Fig.
\ref{fig:The-Bloch-sphere}, this map takes all points inside the
Bloch sphere to points inside the ellipsoid. Being contained inside
the Bloch sphere does not guarantees complete positivity.

\begin{figure}[h]
\begin{centering}
\includegraphics[width=0.45\textwidth]{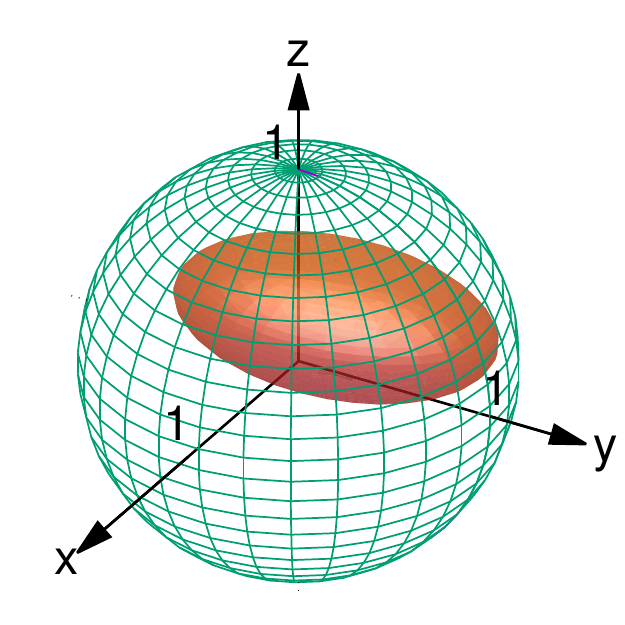}
\par\end{centering}
\caption{The Bloch sphere is represented in green. Inside of it, an example
of ellipsoid containing all the image points of a map as given by
Eq. \eqref{eq:ellipsoid}.\label{fig:The-Bloch-sphere}}

\selectlanguage{american}%
\end{figure}

To find what is the condition on the channel's parameters in order
to assure complete positivity, one may recall the result of Choi \cite{choi1975}
that a qubit map is CP if and only if the matrix
\begin{equation}
\beta(\Lambda)=\left(\mathbb{1}_{2}\otimes\Lambda\right)\left[\ket{\psi^{+}}\bra{\psi^{+}}\right]
\end{equation}
is positive semidefinite, where $\ket{\psi^{+}}=\frac{\ket{00}+\ket{11}}{\sqrt{2}}$
is one of the Bell states of two qubits. This matrix is a density
matrix in a Hilbert space of dimension 4 and as such can be put in
the block form 
\begin{equation}
\beta(\Lambda)=\left(\begin{array}{cc}
A & C\\
C^{\dagger} & B
\end{array}\right)
\end{equation}
with the $2\times2$ matrices $A$, $B$ and $C$. The positivity
condition requires that $A\ge0$, $B\ge0$ and $C=A^{\frac{1}{2}}R\,B^{\frac{1}{2}}$
for some contraction $R$ \cite{ruskai2002}. By definition a contraction
satisfies $\mathbb{1}-RR^{\dagger}\ge0$. The extreme points of the
set of CPTP maps are those for which the equality is attained, which
means that for these maps the contraction is actually a unitary matrix.
Now, using the SVD for a contraction, one gets
\begin{equation}
R=V\left(\begin{array}{cc}
\cos\theta_{1} & 0\\
0 & \cos\theta_{2}
\end{array}\right)W^{\dagger}=\frac{1}{2}V\left(\begin{array}{cc}
e^{i\theta_{1}} & 0\\
0 & e^{i\theta_{2}}
\end{array}\right)W^{\dagger}+\frac{1}{2}V\left(\begin{array}{cc}
e^{-i\theta_{1}} & 0\\
0 & e^{-i\theta_{2}}
\end{array}\right)W^{\dagger}.\label{eq:contraction}
\end{equation}
 Once $V$ and $W$ are unitary, Eq. \eqref{eq:contraction} states
that any contraction is the sum of two unitary matrices, each one
corresponding to one extreme channel. Thus, because of the linearity
of the map and the one-to-one association between map and $\beta(\Lambda)$,
it implies on the \textbf{\textit{Theorem 14}} of \cite{ruskai2002}:
\textit{Any CPTP map of a qubit can be written as the convex combination
of two extreme points of the set of all CPTP maps}. Mathematically,
this means that 
\begin{equation}
\Lambda[\rho]=p\mathcal{E}_{1}[\rho]+(1-p)\mathcal{E}_{2}[\rho],\label{eq:decomp}
\end{equation}
where $\mathcal{E}_{1}$ and $\mathcal{E}_{2}$ are extreme maps and
$0\leq p\leq1$.

The above decomposition is quite useful and has been widely used in
qubit channel simulation \cite{wang2013,Luhe2017,McCutcheon2018}.
What makes it suitable for practical purposes is that any extreme
channel has its diagonal form \eqref{eq:diago} with only one $t_{k}$
component nonnull, which can always be chosen as the third one. Moreover
these channels admit the parameterization 
\begin{equation}
\mathbf{t}=(0,0,\sin u\sin v)\quad T=\left(\begin{array}{ccc}
\cos u & 0 & 0\\
0 & \cos v & 0\\
0 & 0 & \cos u\cos v
\end{array}\right),
\end{equation}
which gives the two Kraus operators
\begin{equation}
E_{1}=\left(\begin{array}{cc}
\cos\alpha & 0\\
0 & \cos\beta
\end{array}\right)\qquad E_{2}=\left(\begin{array}{cc}
0 & \sin\beta\\
\sin\alpha & 0
\end{array}\right),\label{eq:krausextreme}
\end{equation}
with $\alpha=u-v$ , $\beta=u+v$, $u\in[0,2\pi)$ and $v\in[0,\pi)$.
Thus an extreme channel resembles a generalized amplitude damping
in which the probability of exciting the ground state ($|\sin\alpha|^{2}$)
is different from the probability of decay of the excited state $(|\sin\beta|^{2})$,
and these two processes occur in a coherent way.

All the aforementioned works using the extreme-channel decomposition
for a channel simulation use it directly to try to find the parameters
of the decomposition that fits the simulated channel. The problem
in doing this is that the decomposition \eqref{eq:decomp} has 17
free parameters: the convex parameter $p$, the four parameters of
the extreme maps in their diagonal form, plus the 6 parameters of
the unitaries used to diagonalize each of the extreme maps. Thus the
problem is over complicated, since the solution must be done before
the quantum simulation in a classical computer which requires computational
power that grows with the number of parameters. Obviously, the two
extreme channels in the decomposition are related somehow as can be
seen from Eq. \eqref{eq:contraction}. It is left for a future work
to simplify this decomposition in order to reach the number of 12
parameters, as is required for characterizing any qubit channel. Furthermore,
in the already mentioned implementations of qubit quantum channels
the convex combination is carried out classically by selecting which
extreme channel is going to be realized in each round. Contrary, our
initial goal was to realize any qubit channel in one shot. It is still
work in progress, so far we can already do this for a particular set
of quantum maps.

Let us consider a restricted class of maps for which the two extreme
points participating on its decomposition \eqref{eq:decomp} are diagonal
in the sense of \eqref{eq:diago} in the same basis or that their
SVD differs by at most one rotation. Mathematically, we are considering
maps of the form
\begin{equation}
\Lambda\left[\rho\right]=U_{3}\left(p\mathcal{E}_{1}^{\prime}\left[U_{1}\rho U_{1}^{\dagger}\right]+(1-p)U_{2}\mathcal{E}_{2}^{\prime}\left[U_{1}\rho U_{1}^{\dagger}\right]U_{2}^{\dagger}\right)U_{3}^{\dagger},\label{eq:restrict}
\end{equation}
the $\mathcal{E}_{i}$ are extreme maps whose diagonal form is denoted
by $\mathcal{E}_{i}^{\prime}$, and $U_{i}$ are unitary operators.
Although it possibly does not contain the entire set of qubit maps,
many interesting examples can be represented in this way, as is shown
in the next section.

The proposed setup for the implementation of maps of kind \eqref{eq:restrict}
is shown in Fig. \ref{Fig:setupchannel}. A realization of such a
map starts with the unitary transformation $U_{1}$ and ends also
with a unitary transformation $U_{3}$. Regarding the qubit as the
polarization of single photons, these operations are performed in
the manner described in Sec. \ref{sec:Unitary-transformations}. In
order to have all four Kraus operators implemented in each shot of
the experiment, one can proceed as for the GAD channel, transferring
the polarization state to a path degree of freedom using a BD and
reseting the polarization to state $\ket{H}$ (pink box in Fig.\ref{Fig:setupchannel}).
The path works as a qubit ancilla with possible states $\ket{d},\:\ket{u}$,
for down and up paths, respectively. The qubit state at this stage
becomes 
\begin{equation}
\left(a\ket{H}+b\ket{V}\right)\ket{d}\rightarrow\ket{H}\left(a\ket{u}+b\ket{d}\right),
\end{equation}
where $a$ and $b$ are the coefficients of the polarization state
after the unitary $U_{1}$. We are assuming the initial state to be
pure without any lost of generality. The polarization now becomes
free to be used as a control for which extreme channel is going to
be applied with the right probability. The probability is controlled
by the angle $\phi$ of a HWP set such that $\cos^{2}2\phi=p$. When
the photon passes through a PBS, two new paths are created, one for
each extreme channel (yellow box in Fig.\ref{Fig:setupchannel}).
Again this new path degree of freedom works as a qubit ancilla with
the two states $\ket{s},\:\ket{l}$, for short and long paths. The
coherence between the two extreme-maps action is removed by the optical
path difference between long and short paths which is greater than
the coherence length of the photons. 

\begin{figure}[h]
\centering{}\includegraphics[width=0.95\textwidth]{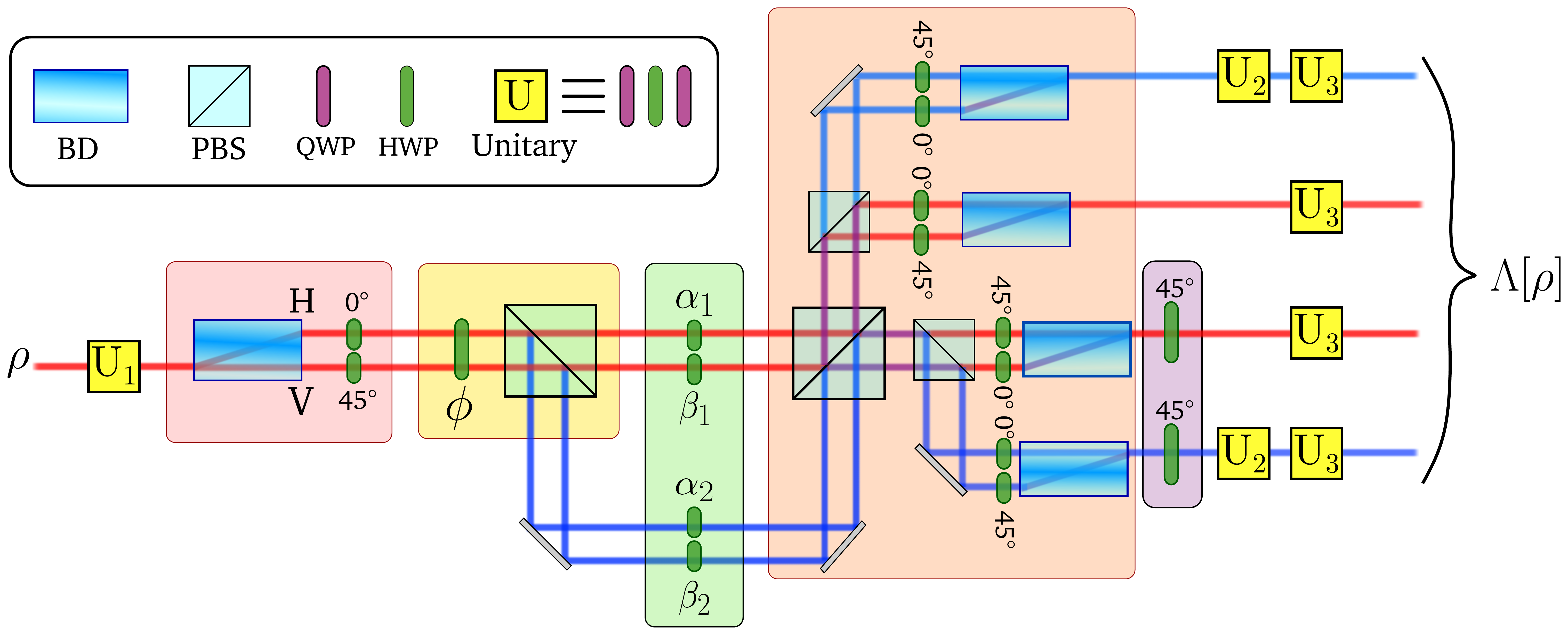}\caption{Setup for the channel simulation. A photon in a polarization state
$\rho$ ends up in a polarization state $\Lambda[\rho]$ after passing
through the setup if the final path information is erased. The beam
colors are different only to indicate which beam is coming from the
short and long paths in the unbalanced interferometer. \label{Fig:setupchannel}}
\end{figure}

Now let us consider the action of an extreme map in its diagonal basis
in Bloch sphere, when its Kraus operators are given by Eq. \eqref{eq:krausextreme}
with angles $\alpha_{1}$ and $\beta_{1}$. The Kraus operators $E_{1}$
and $E_{2}$ transform a pure state, respectively , as 
\begin{align}
a\ket{0}+b\ket{1} & \rightarrow a\cos\alpha_{1}\ket{0}+b\cos\beta_{1}\ket{1}\label{eq:kreus}\\
a\ket{0}+b\ket{1} & \rightarrow b\sin\beta_{1}\ket{0}+a\sin\alpha_{1}\ket{1},\nonumber 
\end{align}
the map itself being the convex sum of this two non-normalized states.
Each Kraus operator alone keeps the coherence between the two basis
states. To implement this map over the photon qubit, we use two HWPs
set to $\alpha_{1}/2$ and $\beta_{1}/2$, positioned in the up and
down paths, respectively, in the short arm of the interferometer.
The resulting transformation is given by
\begin{equation}
\ket{H}\left(a\ket{u}+b\ket{d}\right)\rightarrow a\left(\cos\alpha_{1}\ket{H}+\sin\alpha_{1}\ket{V}\right)\ket{u}+b\left(\cos\beta_{1}\ket{H}+\sin\beta_{1}\ket{V}\right)\ket{d}.
\end{equation}
Notice that the terms that should be coherently recombined to recover
\eqref{eq:kreus} are those in orthogonal path states but with the
same polarization. The orthogonal polarizations are separated in a
PBS and then the up and down paths are recombined in BDs. This recombination
is coherent if there is no optical-path difference between the two
path states which is ensured by the waveplates before the BD which
shifts the down path up, equaling the optical-path of the up state.
Analogously, the second extreme map of the composition is realized
in the long arm of the interferometer with HWPs set to $\alpha_{2}/2$
and $\beta_{2}/2$.

After the BDs recombination the unitary transformation $U_{2}$ is
applied to the long arm photons and the unitary transformation $U_{3}$
is applied to all path states. If the four out-coming paths are traced
out, the polarization state is exactly the one given by the action
of map \eqref{eq:restrict}. To trace them out it is enough to detect
the photons with a large aperture detector, such that all photons
are detected regardless their path states. Although, if the photons
are intended to be used for a further purpose, then one should be
able to gather all the output paths together in a single spatial mode.
A realistic experimental design for this is left for a future work,
possibly using devices that are reflective on one side and transmissive
when light is incident on the other side, as the one proposed in \cite{clikeman}.
Another option for our setup would be to use fiber couplers, which
acts as a beam splitter for two input fiber paths , taking them to
two output paths, each of them being the 50:50 combination of the
inputs.\textcolor{red}{{} }Without tracing out the paths, according
to the interpretation of Eq. \eqref{eq:unrevealing}, what we have
is one unrevealed measurement result of the environment in each output
path, since each path results from one Kraus operator application.

The circuit representation of the process just described is shown
in Fig. \ref{fig:Circuit-representation-of}. The protocol is not
the most efficient possible, as it is not intended to be. It requires
two ancilla qubits and many two and even three-qubits operations.
Still, because of the features of our system it is feasible, since
the ancillary qubits are degrees of freedom of the same system that
provides the main qubit of the computation, so controlled multi-qubit
operations can be implemented deterministically with common optical
elements. 

\begin{figure}[h]
\centering{}\includegraphics[width=0.9\textwidth]{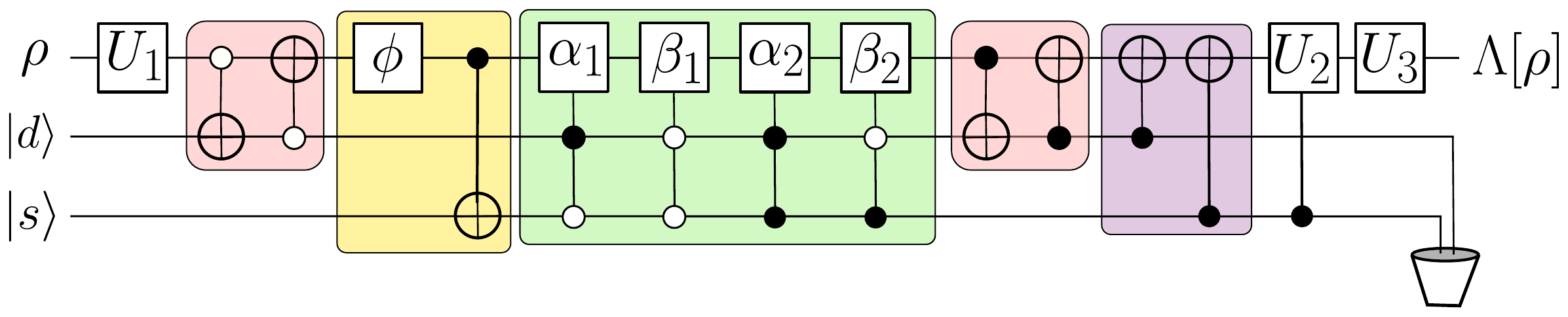}\caption{Circuit representation of the channel simulation protocol. The colors
match the ones used in the setup (Fig. \ref{Fig:setupchannel}) to
indicate equivalent steps. The two ancillary qubits are initialized
in their ground states. $\oplus$ is a ``NOT'' operation. Single
qubit operations can be controlled by the excited $(\bullet)$ or
ground $(\circ)$ state of another qubit. A sequence of two C-NOTs,
like the ones in the pink boxes, are equivalent to a SWAP operation.
The white boxes are either arbitrary unitary operators or the transformation
of a HWP. In the end the two ancillas are discarded. \label{fig:Circuit-representation-of}}
\end{figure}

\subsection{Examples}

In this section some examples of channels are shown together with
the proper parameter applicable for their simulation. This examples
are interesting since they represent much of the intuitive effects
one can think of taking place over a qubit state. Also, via change
of basis, infinitely many other channels can be obtained from these
textbook examples \cite{chuang2000}.

\subsubsection*{Bit flip channel}

The bit-flip channel is given as the action of the two Kraus operators
\begin{equation}
E_{1}=\sqrt{p}\left(\begin{array}{cc}
1 & 0\\
0 & 1
\end{array}\right)\qquad E_{2}=\sqrt{1-p}\left(\begin{array}{cc}
0 & 1\\
1 & 0
\end{array}\right),
\end{equation}
with $0\leq p\leq1$. Thus, a qubit passing through a bit-flip channel
has probability $p$ of remaining in the same state and with probability
$1-p$ it will flip from $\ket{0}$ to $\ket{1}$ and vice-versa.
There are at least two ways of implementing this channel writing it
as combination of extreme maps (Eq.\eqref{eq:restrict}), both of
them without the necessity of unitary transformations. One could either
chose $\phi=0^{\circ}$ and $\alpha_{1}=\beta_{1}=\cos^{-1}\sqrt{p}$
, or $2\phi=\cos^{-1}p$ and $\alpha_{1}=\beta_{1}=0^{\circ}$ and
$\alpha_{2}=\beta_{2}=90^{\circ}$.

\subsubsection*{Phase flip channel}

The phase-flip channel has the two Kraus operators
\begin{equation}
E_{1}=\sqrt{p}\left(\begin{array}{cc}
1 & 0\\
0 & 1
\end{array}\right)\qquad E_{2}=\sqrt{1-p}\left(\begin{array}{cc}
1 & 0\\
0 & -1
\end{array}\right),
\end{equation}
and corresponds to a $\pi$ phase shift between the two computational
basis vectors with probability $1-p$. This channel can be realized
by setting $2\phi=\cos^{-1}p$ , $\alpha_{1}=\beta_{1}=\alpha_{2}=0^{\circ}$
and $\beta_{2}=180^{\circ}$.

\subsubsection*{Depolarizing channel}

The effect of a depolarizing channel is to let the qubit state untouched
with probability $1-p$ and turn it into a complete mixture with probability
$p$. Its transformation is given as 
\begin{equation}
\Lambda[\rho]=\frac{p}{2}\mathbb{1}+(1-p)\rho,
\end{equation}
or in terms of Pauli matrices
\begin{equation}
\Lambda[\rho]=\left(1-\frac{3p}{4}\right)\rho+\frac{p}{4}\left(\sigma_{x}\rho\sigma_{x}+\sigma_{y}\rho\sigma_{y}+\sigma_{z}\rho\sigma_{z}\right).
\end{equation}
The last expression is useful for our purposes because it explicit
the operator-sum decomposition of the depolarizing channel. Before
proceeding to the simulation parameters, let us consider a slightly
more general channel
\begin{equation}
\Lambda[\rho]=p_{0}\rho+p_{1}\sigma_{x}\rho\sigma_{x}+p_{2}\sigma_{y}\rho\sigma_{y}+p_{3}\sigma_{z}\rho\sigma_{z},\label{eq:depolgeral}
\end{equation}
with $\sum_{i=0}^{3}p_{i}=1$. This channel class actually comprises
all the previous cases presented. It is possible to represent this
transformation in the form of Eq. \eqref{eq:restrict} with two extreme
channels diagonalized in the same basis. In such a convex combination,
the identity and $\sigma_{x}$ compose the first extreme map, while
$\sigma_{y}$ and $\sigma_{z}$ are the Kraus operators of the second
one. For the equality between Eqs. \eqref{eq:restrict} and \eqref{eq:depolgeral}
to hold, the parameters must obey
\begin{align}
p=p_{0}+p_{1}\qquad\alpha_{1}=\beta_{1}=\tan^{-1}\sqrt{\frac{p_{1}}{p_{0}}}\\
\alpha_{2}=\beta_{2}-180^{\circ}=\tan^{-1}\sqrt{\frac{p_{2}}{p_{3}}}\nonumber 
\end{align}

\subsubsection*{Generalized amplitude damping channel}

Finally, the largely discussed GAD channel whose Kraus operators are
given in Eq. \eqref{eq:GADKraus} can also be seen as a combination
of extreme channels with the parameters choice
\begin{equation}
\alpha_{1}=\beta_{2}=0\qquad\beta_{1}=\alpha_{2}=\cos^{-1}\sqrt{\eta}.
\end{equation}

\section{Quantum process tomography\label{sec:Quantum-process-tomography}}

We would like to certify that indeed our setup is performing the desired
operation we designed it to do. If we have access to a trustful source
of input states and we are also able to faithfully determine the output
state after the channel, then we can find out what map produced that
resulting transformation and know if it is the desired one. This method
is named quantum process tomography (QPT) \cite{Chuang1997}. 

Consider a set of pure states $\{\ket{\psi_{i}}\}_{i=1}^{N}$ such
that the associated density matrices set $\{\rho_{i}\}_{i=1}^{N}$
forms a basis for the $d\times d$ matrices, $d$ being the dimension
of the Hilbert space of the system. This implies that the set has
$N=d^{2}$ linearly independent elements. If we prepare each state
of the set, sending them through the channel and for each one a quantum
state tomography is realized, so the set of states $\{\rho_{i}^{\prime}=\Lambda[\rho_{i}]\}_{i=1}^{N}$
is obtained, then the channel is determined since for any other state
we can write
\begin{equation}
\rho=\sum_{i=1}^{N}\lambda_{i}\rho_{i}\rightarrow\Lambda[\rho]=\sum_{i=1}^{N}\lambda_{i}\rho_{i}^{\prime}.
\end{equation}

In order to determine the Kraus operators of the channel, one needs
also to set a basis for these operators, let us say $\{\tilde{E}_{n}\}$.
Any Kraus operator can be expressed as a sum $E_{k}=\sum_{n}e_{kn}\tilde{E}_{n}$.
In this operator basis the map becomes
\begin{equation}
\Lambda[\rho]=\sum_{n,m}\chi_{nm}\tilde{E}_{n}\rho\tilde{E}_{m}^{\dagger},\label{eq:QPT}
\end{equation}
being completely described by the numbers $\chi_{nm}=\sum_{k}e_{kn}e_{km}^{*}$.
Now one can look at the action of this channel upon the basis matrices
\begin{equation}
\Lambda[\rho_{j}]=\sum_{k}\lambda_{jk}\rho_{k},\label{eq:QPT2}
\end{equation}
the $\lambda_{jk}$'s are numbers experimentally determined from the
quantum state tomography. On the other hand, one can also write it
using \eqref{eq:QPT} as 
\begin{equation}
\Lambda[\rho_{j}]=\sum_{m,n,k}\chi_{nm}\beta_{jk}^{nm}\rho_{k},\label{QPT3}
\end{equation}
where we have defined $\beta_{jk}^{nm}$ from $\tilde{E}_{n}\rho_{j}\tilde{E}_{m}^{\dagger}=\sum_{m,n,k}\beta_{jk}^{nm}\rho_{k}$.
The numbers $\beta_{jk}^{nm}$ are known since they are calculated
using only the states and operators bases. By comparing \eqref{eq:QPT2}
and \eqref{QPT3} we get 
\begin{equation}
\sum_{m,n}\chi_{nm}\beta_{jk}^{nm}=\lambda_{jk},\label{QPT4}
\end{equation}
which determines operator-sum representation for the channel from
the experimental data $\lambda_{jk}$. 

For the case of a qubit, the standard choice of pure states prepared
and measured in an experiment is $\left\{ \ket{0},\ket{+}=\frac{\ket{0}+\ket{1}}{\sqrt{2}},\ket{L}=\frac{\ket{0}-i\ket{1}}{\sqrt{2}},\ket{1}\right\} $.
A possible operator basis is the set $\{\tilde{E}_{1}=\mathbb{1},\tilde{E}_{2}=\sigma_{x},\tilde{E}_{3}=-i\sigma_{y},\tilde{E}_{4}=\sigma_{z}\}$.
In this basis, Eq. \eqref{QPT4} gives 
\begin{equation}
\chi=\Omega\left(\begin{array}{cc}
\rho_{1}^{\prime} & \rho_{2}^{\prime}\\
\rho_{3}^{\prime} & \rho_{4}^{\prime}
\end{array}\right)\Omega,
\end{equation}
$\chi$ is a $4\times4$ matrix with elements $\chi_{nm}$, 
\begin{equation}
\Omega=\frac{1}{2}\left(\begin{array}{cc}
\mathbb{1} & \sigma_{x}\\
\sigma_{x} & -\mathbb{1}
\end{array}\right),
\end{equation}
and
\begin{align*}
\rho_{1}^{\prime} & =\Lambda[\rho_{1}]\\
\rho_{2}^{\prime} & =\Lambda[\rho_{2}]-i\Lambda[\rho_{3}]-(1-i)\frac{\Lambda[\rho_{1}]+\Lambda[\rho_{4}]}{2}\\
\rho_{3}^{\prime} & =\Lambda[\rho_{2}]+i\Lambda[\rho_{3}]-(1+i)\frac{\Lambda[\rho_{1}]+\Lambda[\rho_{4}]}{2}\\
\rho_{4}^{\prime} & =\Lambda[\rho_{4}],
\end{align*}
$\rho_{i}$ are the states being prepared and $\Lambda[\rho_{i}]$
is an experimentally tomographed state.

\section{Conclusion}

In this chapter a proposal for quantum channel simulation for a qubit
was presented. The same way as qubit unitary transformations are the
basis for closed-system quantum computation, qubit quantum channels
can become the basic entity for open-system quantum computation \cite{Vertraete2009}.
Thus it would be interesting to design a platform to implement an
arbitrary qubit channel. Our attempt to do so is based on the decomposition
of a qubit map in a convex combination of the extremes of the set
of maps. As a work in progress, some details are missing. For example,
the setup devised so far is able to simulate only a particular class
of channels. At the same time, what we call a particular class may
still be the whole set, since it offers even more free parameters
than is necessary to describe an arbitrary channel. However, up to
now we could not show that this is the case. An interesting feature
of our proposal in comparison to other works is that we are able to
implement the entire channel for each photon, without the necessity
to realize the convex sum in a classical probabilistic way. Moreover,
the interpretation of the operator-sum representation as unrevealed
measurements of the environment occurs exactly in our setup, since
each Kraus operator leads to a different output path and is associated
with a different states of the environment (ancillas). Our proposal
also represents an experimental problem, since it would be desirable
to gather the outputs of all Kraus operators in a single resulting
path as to have the complete transformed state.\selectlanguage{american}

~\ihead{}

\ohead{\leftmark}

\ifoot{}

\cfoot{}

\ofoot[
]{\thepage}

\addchap{Final Remarks}

During my PhD the goal was two explore as many techniques and thus
as many degrees of freedom of light as available in the Quantum Optics
Laboratory of Federal University of Rio de Janeiro. It was possible
by also exploring many different aspects of the quantum theory itself.
The result is the series of experiments presented in this thesis.
It contained two experiments using the transverse degrees of freedom
of classical light beams (which could equivalently be performed using
single photons with the same spatial profile), one of them uses also
the polarization degree of freedom. The SLM, one of the crucial devices
used, is applied in two different ways: as a phase modulator or as
a amplitude mask. These two experiments also differ in their detection
method which is an intensity profile captured by a CCD camera in one
case, and the detection of the number of photons in the attenuated
beam using a free space avalanche detector in the other case. The
thesis also contains one experiment using the entanglement in polarization
and path degrees of freedom between two photons produced in a nonlinear
crystal by spontaneous parametric down conversion. At last, two experiments
(one of them only proposed but not realized yet) using the SPDC process
as a source of single photons are also presented. In these experiments,
the polarization of single photons represent the computational qubit
while the path degrees of freedom are used as the qubit environment
or simply as ancillary systems.

As the results of the individual works were summarized at the end
of each chapter, here I would like to just point some clear open questions
and further investigations regarding some of the works realized. Relative
to the content of Chapter 2, it is still open if a slight change in
our simulation protocol would enable us to simulate an interacting
particle. Maybe more interesting and direct is to investigate if it
is possible to observe the Hegerfeldt paradox in our simulation. This
paradox is related to the superluminal propagation of relativistic
wavefunctions that were initially localized in a finite region of
space \cite{hegerfeldt1974,hegerfeldt1985}. Using our experiment,
we can check its occurrence or not in both Dirac and FW representations.
The two major open questions about the content presented in Chapter
3 are: is there any practical application of the nice construction
we showed? And as done in Ref. \cite{Ketterer16}, is it possible
to relate the PCG MUMs to angular momentum for quantum information
processing? At last, about the subject of Chapter 7 it is left to
prove or disprove that the parameterization we simulate is able to
describe a general qubit channel. Also, as was mentioned, one must
solve the experimental issue of grouping all the output paths, transforming
them in a single output which state is the result of the map application
to the initial state. Moreover, we are now seeking for applications
of the experimental platform in the investigation of correlation dynamics,
since the photon used can be entangled with the herald photon, or
quantum thermodynamics phenomena.\selectlanguage{american}

\cleardoublepage{}

\appendix

\ihead{}

\ohead{\textbf{Appendix~\thechapter}~\leftmark}

\ifoot{}

\cfoot{}

\ofoot{\thepage}

\chapter{Birefringent materials\label{chap:Anexo2}}

Most of the linear optical devices for manipulating polarization available
are made of birefringent materials. This birefringence or double refraction
property is present in many crystals and liquid crystals were a preferred
direction is naturally defined by the internal symmetry of the material.
In the experiments presented in this thesis many are the examples
of birefringence-based devices: wave plates, polarizing beam splitters,
beam displacers and spatial light modulators. In many cases, it makes
necessary to understand the physical principals behind the desired
final effect such to correct imperfections and obtain the best performance
of the devices. This Section is devoted to present the general principles
valid for propagation and incidence in any birefringent media. As
it is valid for all dispositives used in the experiments, it is considered
a non-magnetic and lossless crystal. The treatment presented is completely
classical but it extends to the modes of quantum electromagnetic field. 

In a homogeneous linear medium, an electric field $\mathbf{E}$ produces
a linearly dependent polarization $\mathbf{P}$ and thus a linear
electric displacement $\mathbf{D}=\epsilon_{0}\mathbf{E}+\mathbf{P}=\epsilon\mathbf{E}$.
$\epsilon_{0}$ is the permittivity of vacuum and $\epsilon$ is the
permittivity of the medium. If the medium is in addition isotropic,
then its permittivity is just a scalar. If it is not the case and
the medium presents anisotropy, as is the common case in crystals,
then the polarization direction and strength depends on the electric
field direction and is no longer co-linear with it. In this case,
the electric permittivity is an order 2 tensor and, because of energy
conservation considerations, it can be represented by a $3\times3$
Hermitian matrix \cite{zangwill}. Thus its eigenvectors are all orthogonal
and in this eigenbasis, called the principal directions of the material,
the permittivity tensor reads
\begin{equation}
\boldsymbol{\epsilon}=\left(\begin{array}{ccc}
\epsilon_{1} & 0 & 0\\
0 & \epsilon_{2} & 0\\
0 & 0 & \epsilon_{3}
\end{array}\right).\label{eq:epsilon-1}
\end{equation}
In terms of electromagnetic waves, if the polarization of the wave
is an eigenvector of the electric permittivity , then the polarization
is not altered during propagation inside the medium.

For the so called uniaxial crystals like calcite and quartzo, because
of the symmetry present, the electric responses in two of the principal
directions are equal $\epsilon_{1}=\epsilon_{2}=\epsilon_{\vartheta}$.
The direction with different value of permittivity $\epsilon_{3}=\epsilon_{e}$
define the optical axis of the crystal.

Let us consider a plane wave with frequency $\omega$ propagating
through the crystal, such that all the fields vary in space and time
as $e^{i\omega t-i\mathbf{k}\cdot\mathbf{r}}$, $\mathbf{k}=n\frac{\omega}{c}\mathbf{s}$
is the wave vector in the direction of the unity vector $\mathbf{s}$,
$n$ is the index of refraction and $c$ is the light velocity. The
Maxwell equations impose the following relations for the amplitudes
of the fields
\begin{equation}
\begin{array}{cc}
\mathbf{k}\times\mathbf{E}=\omega\mu\mathbf{H}\quad & \quad\mathbf{k}\times\mathbf{H}=-\omega\mathbf{D}\\
\mathbf{k}\cdot\mathbf{D}=0\quad & \quad\mathbf{k}\cdot\mathbf{H}=0,
\end{array}\label{eq:fields-1}
\end{equation}
$\mu$ is the scalar magnetic permeability. The relation between the
field vectors and wave vector is shown in Fig. \ref{fig:vectors-1}-a).
The first two equations imply that $\mathbf{k}$ is perpendicular
to the electric displacement $\mathbf{D}$ and to the magnetic fields
$\mathbf{B}$, $\mathbf{H}$. Also $\mathbf{E}$ and $\mathbf{D}$
are both perpendicular to $\mathbf{B}$, $\mathbf{H}$. Thus the vectors
$\mathbf{k}$, $\mathbf{E}$ and $\mathbf{D}$ lie in the same plane,
but as $\mathbf{E}$ and $\mathbf{D}$ are not colinear, the electric
field is not , in general, perpendicular to the direction of propagation
is this material. Moreover, as the energy flow is given by the Poynting
vector $\mathbf{E}\times\mathbf{H}$, surprisingly the wave vector
is not in the energy flow direction generally.

\begin{figure}[h]
\begin{centering}
\includegraphics[width=0.95\textwidth]{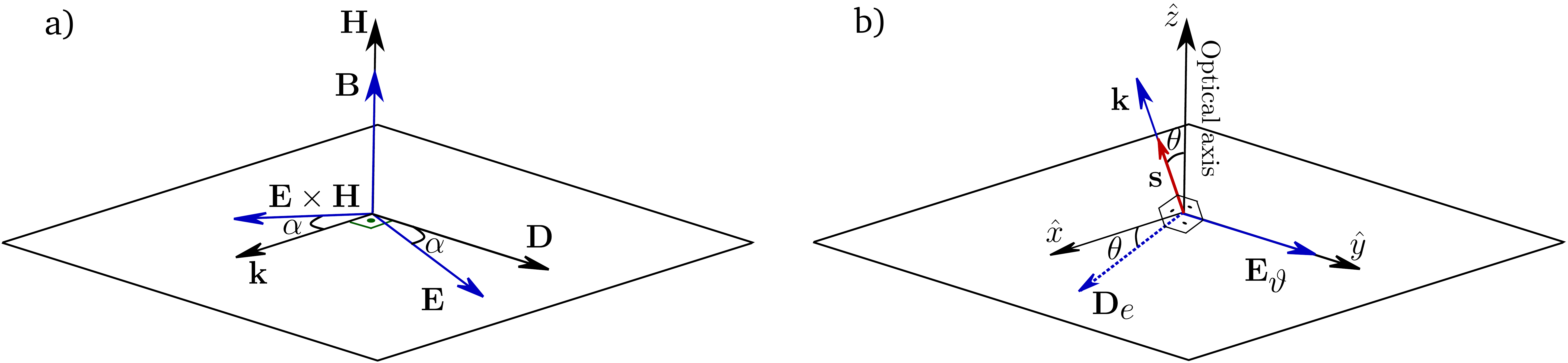}
\par\end{centering}
\caption{\foreignlanguage{english}{\label{fig:vectors-1}Representation of the field vectors for an anisotropic
material.\textbf{ a)} All the field vectors for an arbitrary material
and an arbitrary mode of propagation. Vectors with the same color
are orthogonal. \textbf{b)} The ordinary and extraordinary modes of
propagation in a uniaxial crystal.The plane perpendicular to the optical
axis contains all the directions for which $\epsilon=\epsilon_{\vartheta}.$}\selectlanguage{american}
}
\end{figure}

By eliminating $\mathbf{H}$ in Eqs. (\ref{eq:fields-1}) gives for
the electric field $\mathbf{k}\times(\mathbf{k}\times\mathbf{E})+\omega^{2}\mu\epsilon\mathbf{E}=0$
or explicitly in the principal direction basis
\begin{equation}
\left(\begin{array}{ccc}
\frac{\epsilon_{\vartheta}}{\epsilon_{0}}-n^{2}(s_{y}^{2}+s_{z}^{2}) & n^{2}s_{x}s_{y} & n^{2}s_{x}s_{z}\\
n^{2}s_{x}s_{y} & \frac{\epsilon_{\vartheta}}{\epsilon_{0}}-n^{2}(s_{x}^{2}+s_{z}^{2}) & n^{2}s_{z}s_{y}\\
n^{2}s_{x}s_{z} & n^{2}s_{z}s_{y} & \frac{\epsilon_{e}}{\epsilon_{0}}-n^{2}(s_{x}^{2}+s_{y}^{2})
\end{array}\right)\left(\begin{array}{c}
E_{x}\\
E_{y}\\
E_{z}
\end{array}\right)=N\mathbf{E}=0,
\end{equation}
where the substitution $\mathbf{k}=n\frac{\omega}{c}\mathbf{s}$ was
made. The nontrivial solutions for $\mathbf{E}$ come when the determinant
of $N$ vanishes. It implies that the index of refraction $n$ must
satisfy the equation \cite{yariv1984}
\begin{equation}
\left[\frac{n^{2}}{n_{\vartheta}^{2}}-1\right]\left[\frac{n^{2}}{n_{e}^{2}}\left(s_{x}^{2}+s_{y}^{2}\right)+\frac{n^{2}}{n_{\vartheta}^{2}}s_{z}^{2}-1\right]=0,\label{eq:index-1}
\end{equation}
 the principal indexes of refraction being defined as $n_{\vartheta}^{2}\equiv\frac{\epsilon_{\vartheta}}{\epsilon_{0}}$
and $n_{e}^{2}\equiv\frac{\epsilon_{e}}{\epsilon_{0}}$ . The above
equation has two solutions, what means that, for each propagation
direction, there are two propagating modes with different indexes
of refraction. The first one is independent of the direction of propagation
and has index of refraction $n=n_{\vartheta}$. This mode is called
the ordinary wave. It is possible to show that the electric field
of the ordinary mode is perpendicular to the wave vector, causing
all the optical phenomena to have the same behavior as for isotropic
materials. The other solution is called extraordinary wave. Lets choose
the $x$ and $y$ axis such that $\mathbf{k}$ is contained in the
$xz$ plane, as shown in Fig. \ref{fig:vectors-1}-b). This arbitrary
choice is possible because $\epsilon_{\vartheta}$ is degenerate.
According to Eq. (\ref{eq:index-1}), the index of refraction of the
extraordinary wave depends on the angle $\theta$ between the wave
vector and the optical axis as 
\begin{equation}
\frac{1}{n^{2}}=\frac{\cos^{2}\theta}{n_{\vartheta}^{2}}+\frac{\sin^{2}\theta}{n_{e}^{2}}.
\end{equation}

From the fact that the two modes are known to be orthogonal, the electric
displacement $\mathbf{D}_{e}$ lies in the $xz$ plane, being written
as $\mathbf{D}_{e}=\epsilon_{\vartheta}E_{x}\hat{x}+\epsilon_{e}E_{z}\hat{z}$.
Using the orthogonality of $\mathbf{k}$ and $\mathbf{D}_{e}$ and
the scalar product $\mathbf{D}_{e}\cdot\mathbf{E}_{e}$ give the angle
$\alpha$ between the electric field and the electric displacement
as 
\begin{equation}
\cos\alpha=\cos\theta\left(1+\frac{n_{\vartheta}^{2}}{n_{e}^{2}}\tan\theta\right)\left[1+\left(\frac{n_{\vartheta}^{2}}{n_{e}^{2}}\tan\theta\right)^{2}\right]^{-\frac{1}{2}},\label{eq:alpha}
\end{equation}
this is also the angle between the wave vector and the direction of
energy flow. As particular cases: if the two principal indexes of
refraction are equal then $\cos\alpha=1$ and $\alpha=0$ and the
two modes are perpendicular to the wave vector as expected for a isotropic
material; if $\theta=0$ then $\alpha=0$ and the two modes are ordinary
waves; if $\theta=\frac{\pi}{2}$ the above equation is actually not
valid , in this case the ordinary electric field and electric displacement
are parallel and in the direction of the optical axis, the index of
refraction being $n=n_{e}$.

When a electromagnetic wave is propagating in a isotropic medium and
reach an interface with an anisotropic material, the same boundary
conditions as if both media are isotropic are still valid. As the
index of refraction depends on the direction of propagation and on
the polarization of the wave, two different waves are refracted. In
particular, all the wave vectors are contained in the incidence plane
and the phase of the wave at the interface must be continuous leading
to a Snell's law for each refracted wave 
\begin{equation}
n_{i}\sin\theta_{i}=n_{\vartheta}\sin\theta_{\vartheta}=n(\theta_{e})\sin\theta_{e},
\end{equation}
$n_{i}$ is the index of refraction of the incident medium, $\theta_{i}$
is the angle of incidence, $\theta_{\vartheta}$ and $\theta_{e}$
are the angles of refraction of the ordinary and extraordinary waves,
respectively. Again, the ordinary wave behaves like if it was a isotropic
medium while to find the angle of refraction for the extraordinary
wave it is necessary to solve a quartic equation because of the dependence
of the index of refraction on the direction of propagation.

\ihead{}

\ohead{\textbf{Appendix~\thechapter}~\leftmark}

\ifoot{}

\cfoot{}

\ofoot{\thepage}

\chapter{Further experimental results of several PCG MUM\label{chap:Anexo0}}

In this appendix, the measurements for all combinations of preparation
and measurement phase space directions are shown. All plots show the
Shannon entropy as a function of the period of the measurement mask
($T_{j}$ should be changed by $T_{k}$). The preparation mask is
fixed in $M_{0}^{(j)}(q_{j},T_{j})$ with $T_{j}$ satisfying the
MUM condition for the $m_{jk}$ values shown in Table \ref{tab:1}.
In this way, it was expected that for the same measurement direction,
all three preparations would give the maximum entropy for the same
measurement period values. This can be verified in the plots and particularly
this happens for the periods corresponding to the right $m_{jk}$
values.

\subsubsection*{Measurement $k=0$}

\begin{figure}[H]
\noindent \begin{centering}
\subfloat[]{\noindent \begin{centering}
\includegraphics[width=0.45\columnwidth]{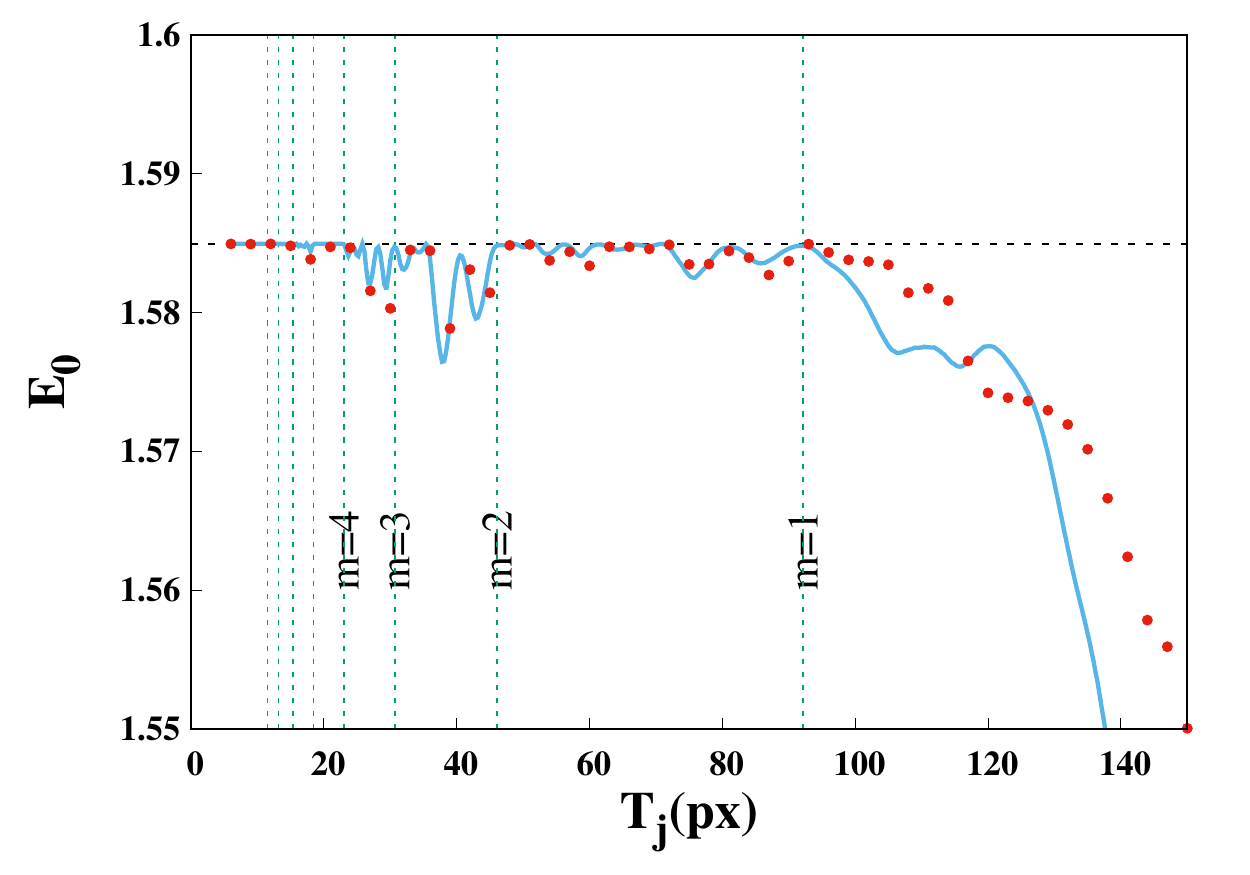}
\par\end{centering}
} \subfloat[]{\noindent \centering{}\includegraphics[width=0.45\columnwidth]{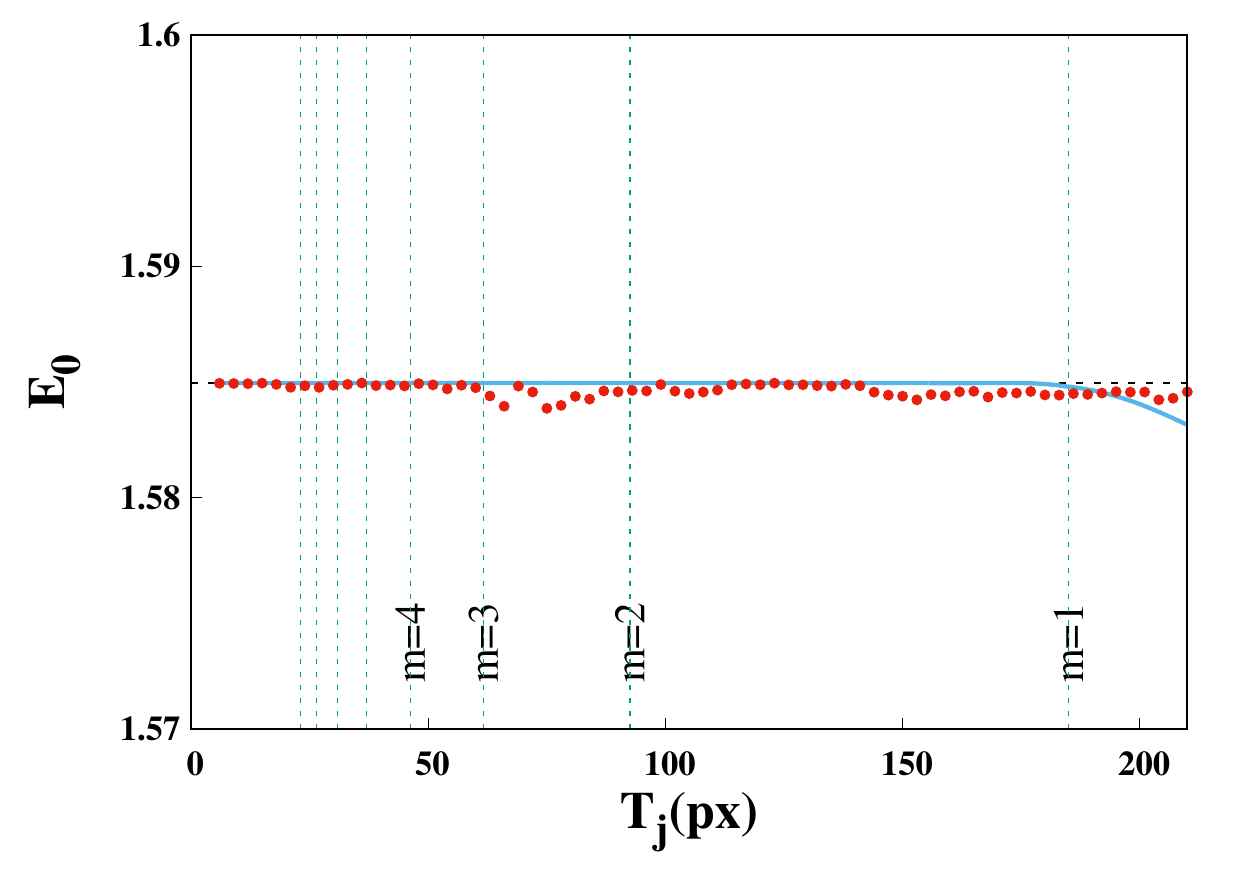}}
\par\end{centering}
\noindent \begin{centering}
\subfloat[]{\noindent \centering{}\includegraphics[width=0.45\columnwidth ]{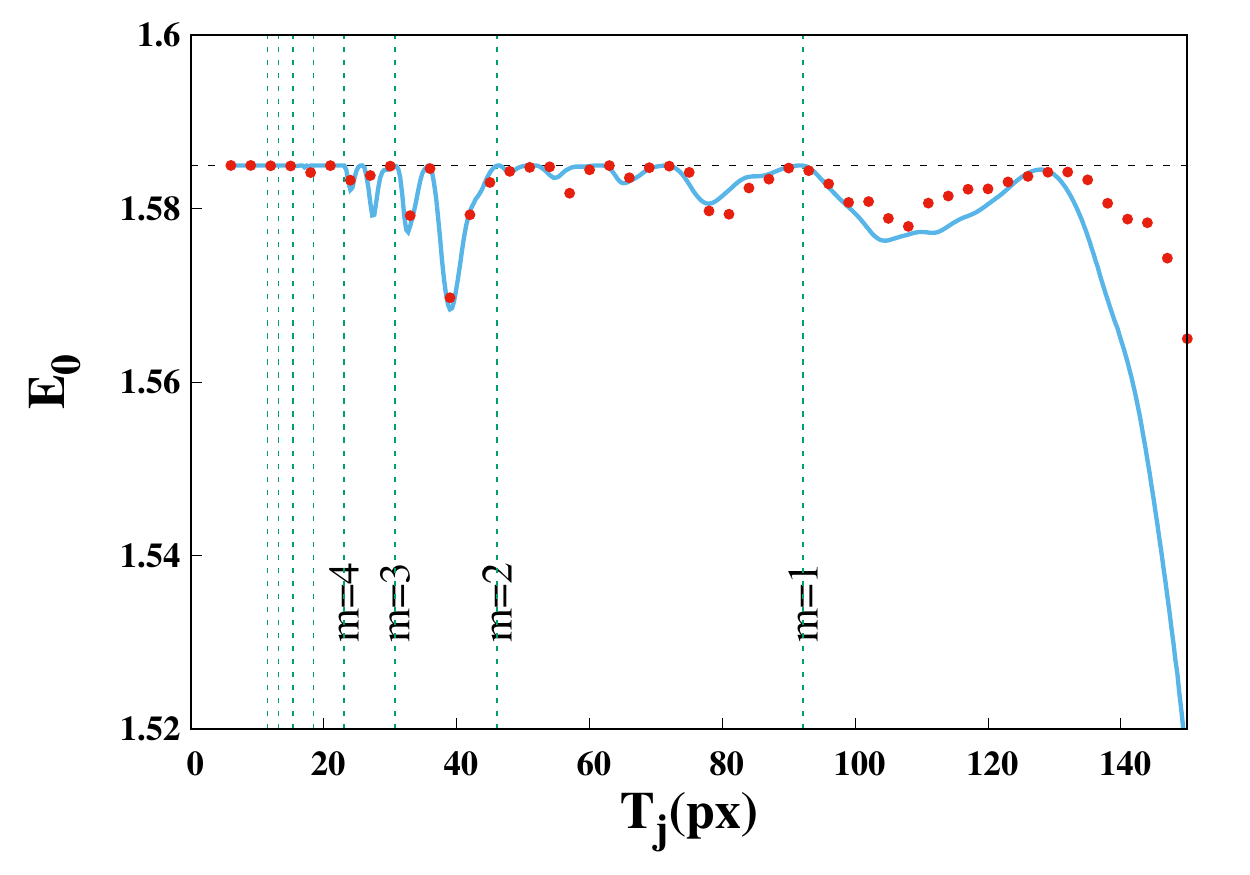}}
\par\end{centering}
\caption{Measurements on $k=0$ direction for preparation a) $j=1$, b) $j=2$
and c) $j=3$.}

\end{figure}

\subsubsection*{Measurement $k=1$}

\begin{figure}[H]
\noindent \begin{centering}
\subfloat[]{\noindent \centering{}\includegraphics[width=0.45\columnwidth ]{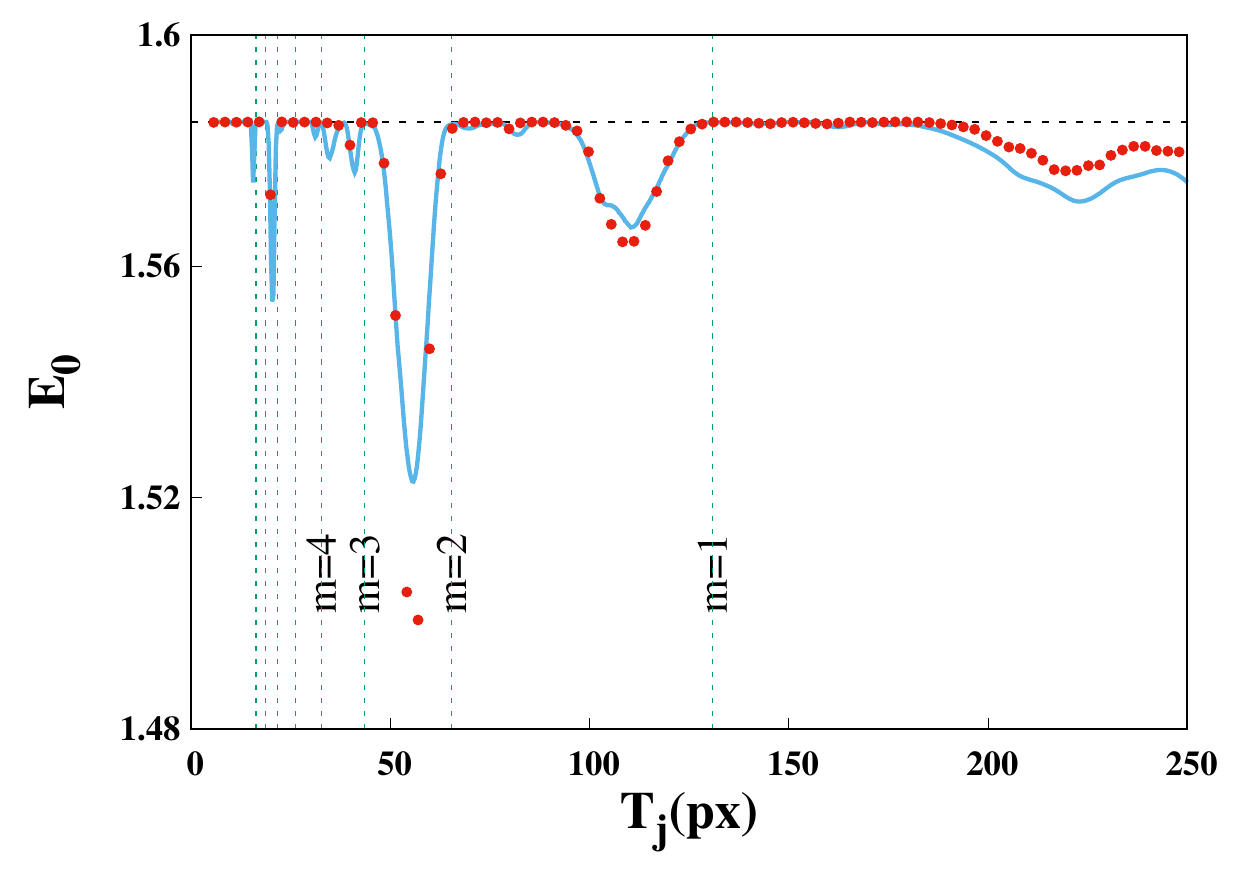}}
\subfloat[]{\noindent \centering{}\includegraphics[width=0.45\columnwidth ]{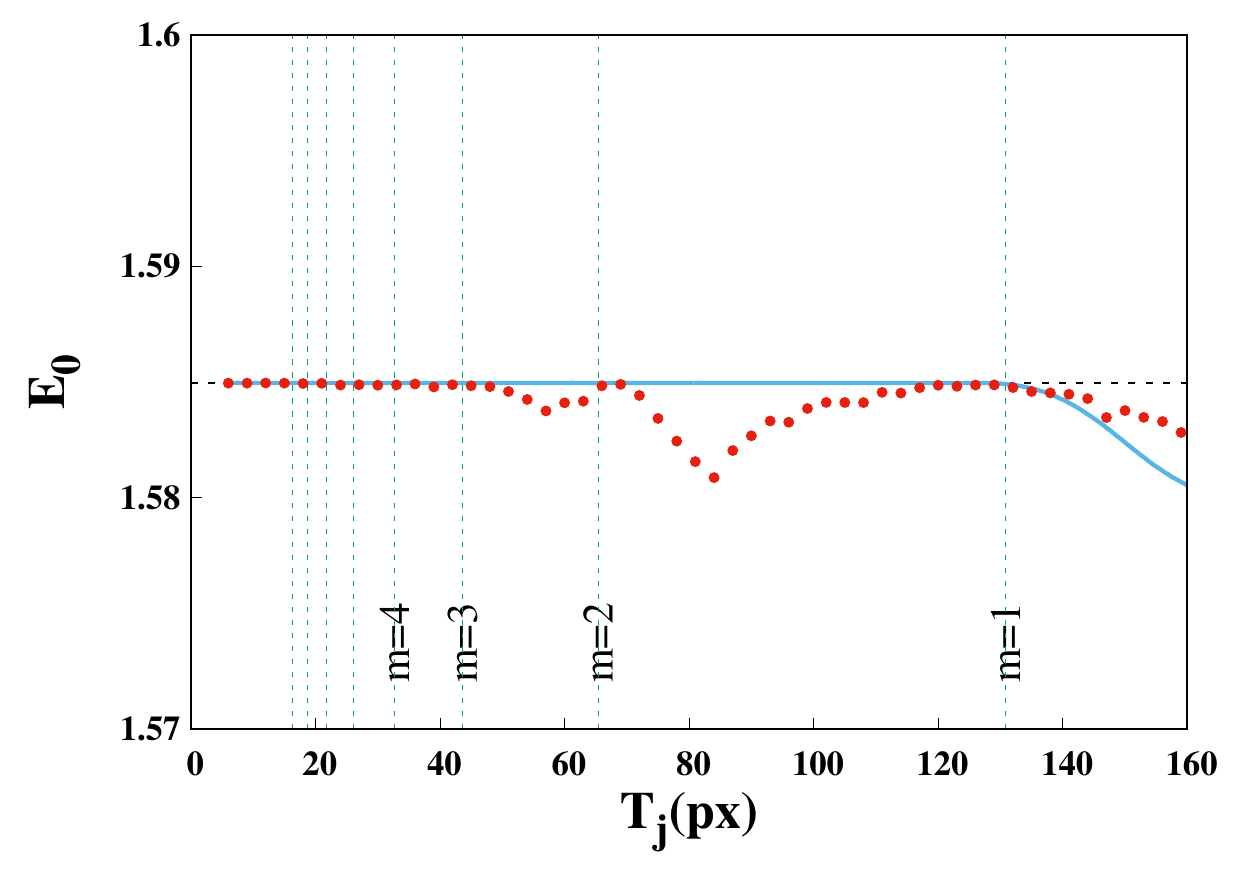}}
\par\end{centering}
\noindent \begin{centering}
\subfloat[]{\noindent \centering{}\includegraphics[width=0.45\columnwidth ]{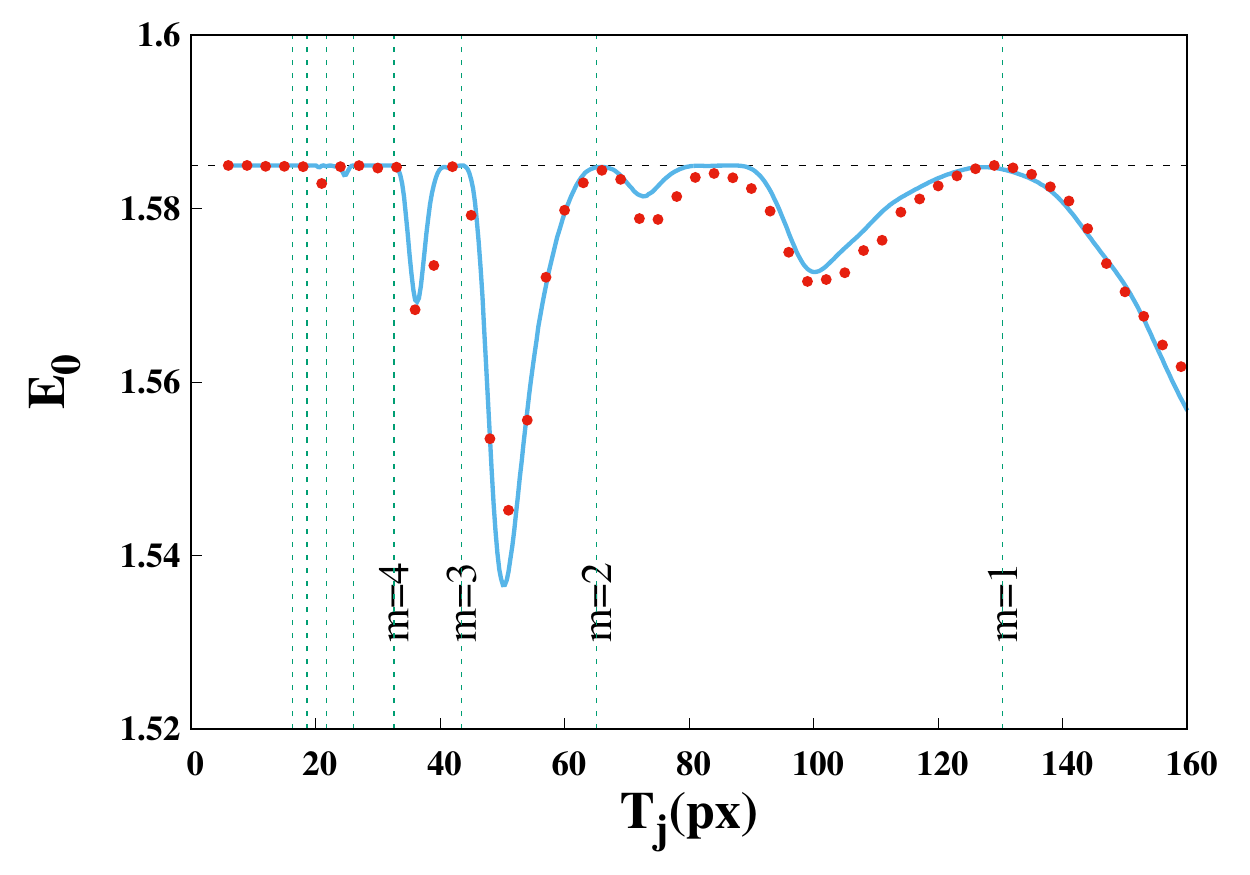}}
\par\end{centering}
\caption{Measurements on $k=1$ direction for preparation a) $j=0$, b) $j=2$
and c) $j=3$.}
\end{figure}

\subsubsection*{Measurement $k=2$}

\begin{figure}[H]
\noindent \begin{centering}
\subfloat[]{\noindent \centering{}\includegraphics[width=0.45\columnwidth ]{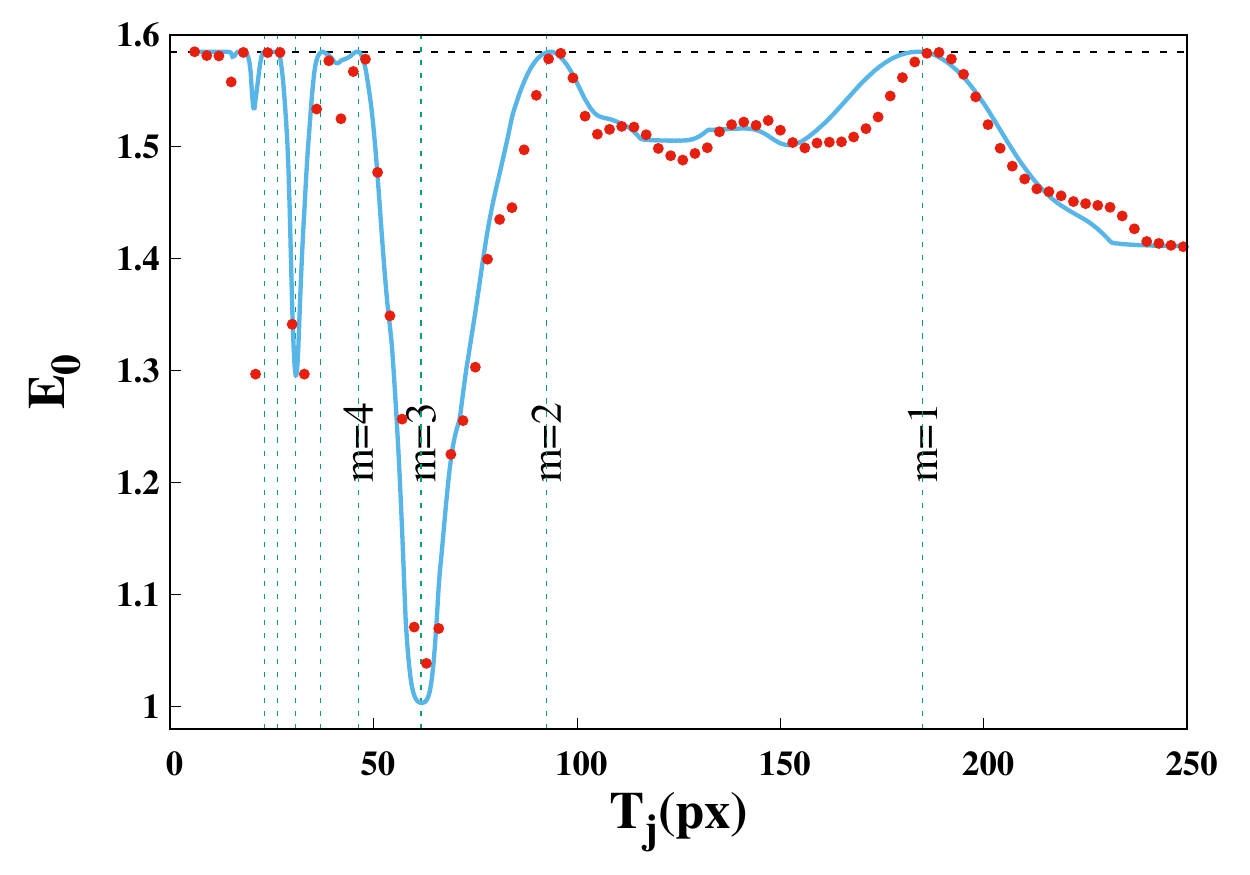}}
\subfloat[]{\noindent \centering{}\includegraphics[width=0.45\columnwidth ]{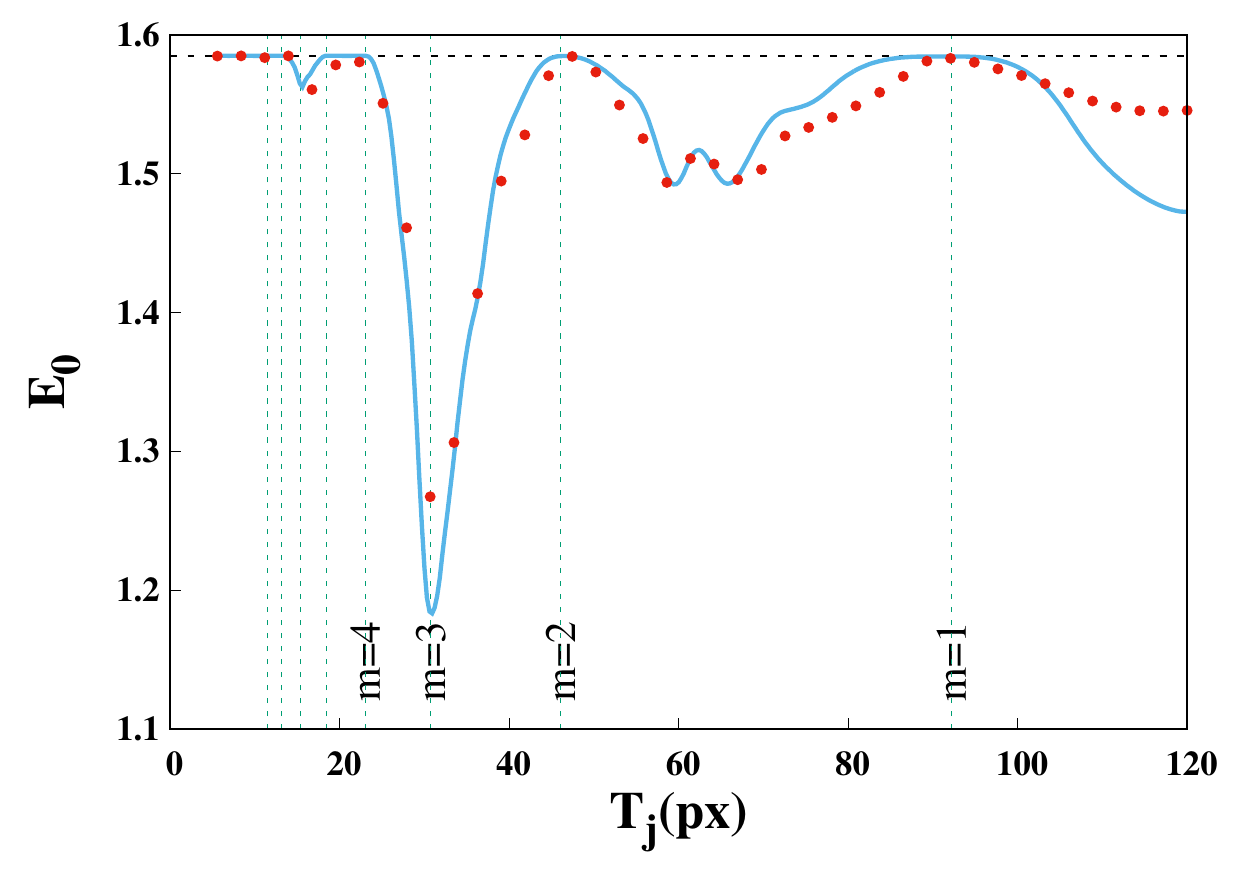}}
\par\end{centering}
\noindent \begin{centering}
\subfloat[]{\noindent \centering{}\includegraphics[width=0.45\columnwidth ]{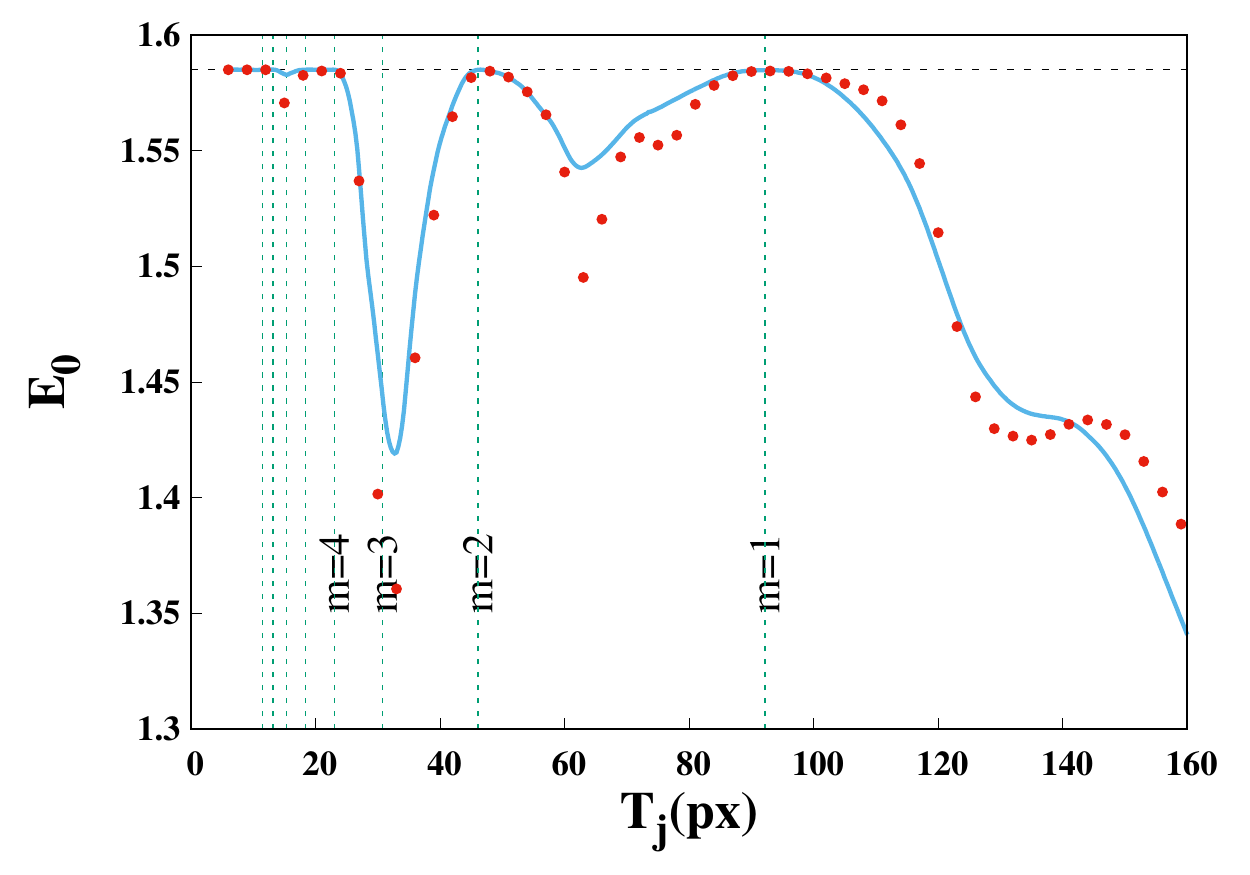}}
\par\end{centering}
\caption{Measurements on $k=2$ direction for preparation a) $j=0$, b) $j=1$
and c) $j=3$.}
\end{figure}

\subsubsection*{Measurement $k=3$}

\begin{figure}[H]
\noindent \begin{centering}
\subfloat[]{\noindent \centering{}\includegraphics[width=0.45\columnwidth ]{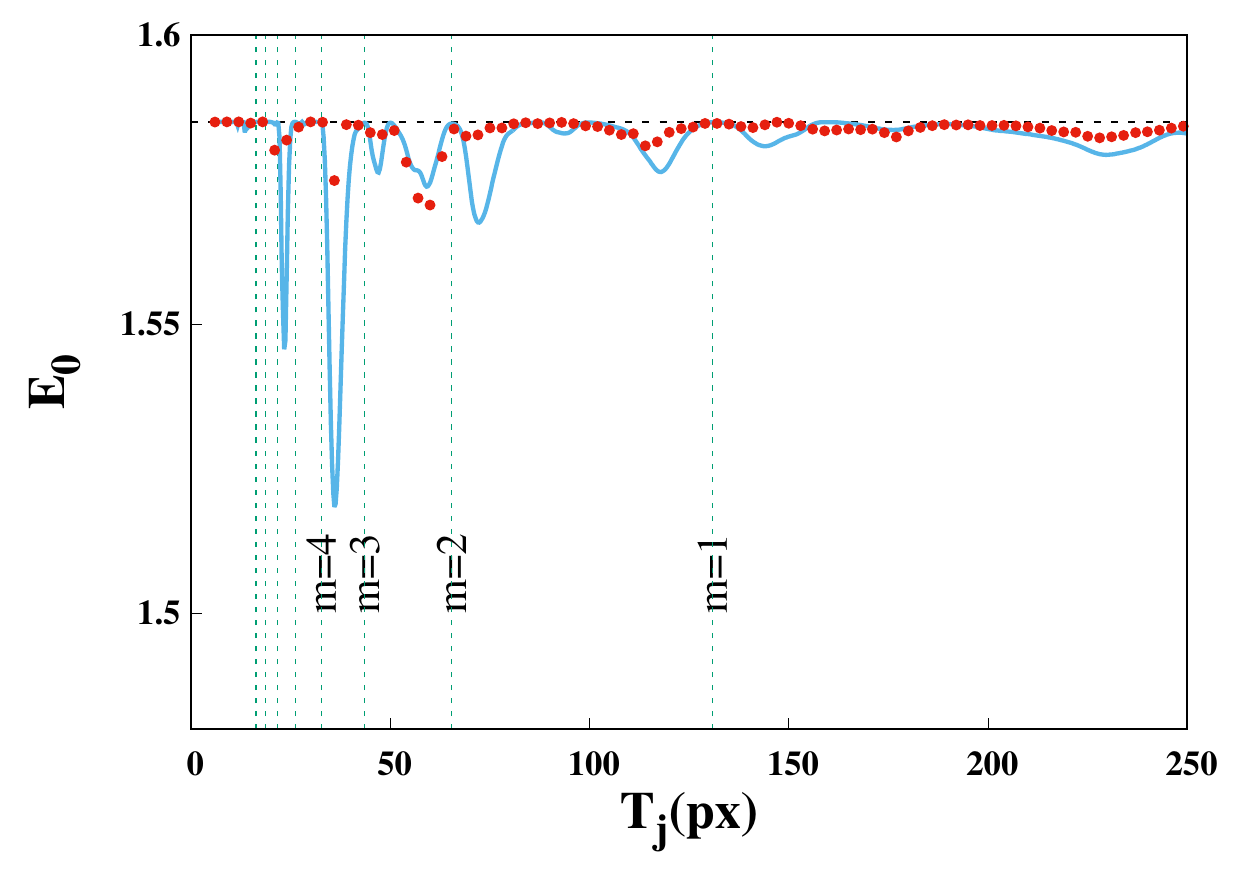}}
\subfloat[]{\noindent \centering{}\includegraphics[width=0.45\columnwidth ]{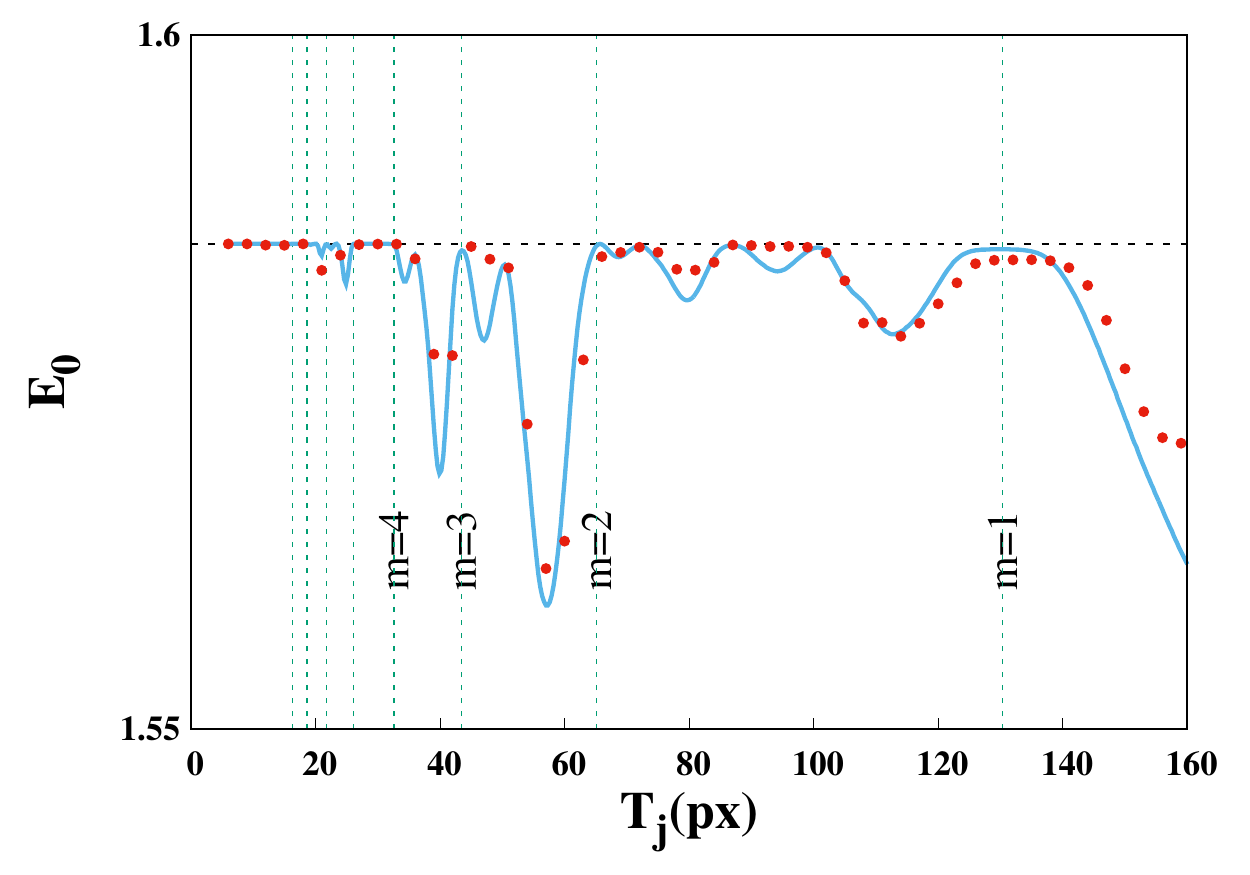}}
\par\end{centering}
\noindent \begin{centering}
\subfloat[]{\noindent \centering{}\includegraphics[width=0.45\columnwidth ]{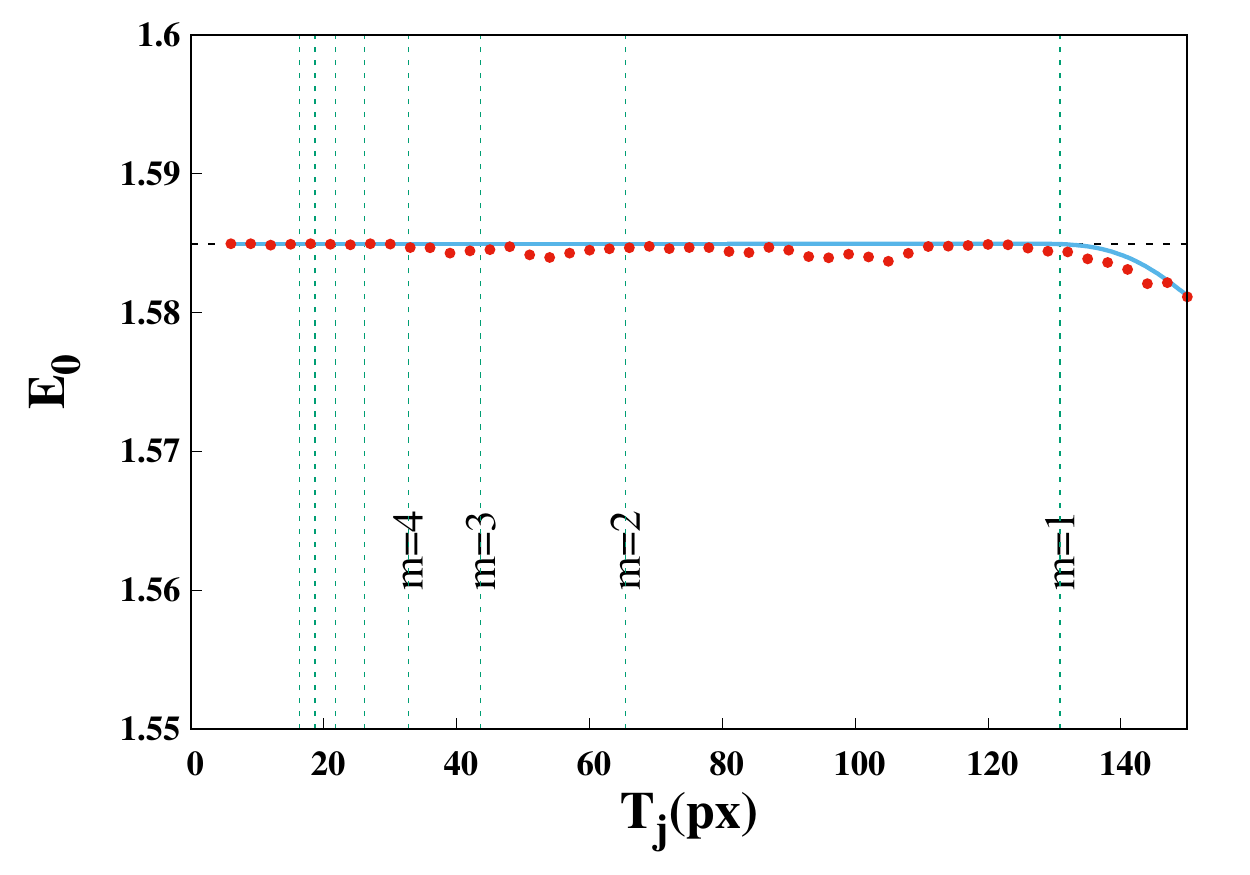}}
\par\end{centering}
\caption{Measurements on $k=3$ direction for preparation a) $j=0$, b) $j=1$
and c) $j=2$.}
\end{figure}

\ihead{}

\ohead{\textbf{Appendix~\thechapter}~\leftmark}

\ifoot{}

\cfoot{}

\ofoot{\thepage}

\chapter{Supplement to Steering Exposure\label{chap:Anexo4}}

\section{No-go theorem for multi-black-box universal steering bits\label{sec:No-go-theorem-for}}

\label{app:proof_nobits} In contrast to the protocols exploring the
capabilities of wirings within the $AB$ partition, in this section
we present a no-go theorem limiting their transformation power. Since
it is known \cite{Gallego2015} that in minimal dimension there is
no steering bit --- i.e. no ``universal'' minimal-dimension assemblage
that can be transformed into any other under 1W-LOCCs --- one can
ask whether reduction from a higher number of inputs, outputs or parties
allows such a steering bit to be established. We answer in the negative
even in minimal dimension.

\begin{thm}{[}No pure steering bit with higher number of parties{]}
\label{th:nobits} There does not exist any pure $(N-1)$-DI qubit
assemblage $\sigma_{\boldsymbol{a}|\boldsymbol{x}}^{\text{bit}}$,
where $\boldsymbol{a}=\{a_{1},...,a_{N-1}\}$, $\boldsymbol{x}=\{x_{1},...,x_{N-1}\}$
(with finite sets of input and output values), that can be transformed
via 1W-LOCCs into all qubit assemblages of minimal dimension $\sigma_{a|x}^{(\text{target})}$.
\end{thm}

\begin{proof} The proof is similar in spirit to that of Theorem 5
of \cite{Gallego2015}. We consider a pure $(N-1)$-DI qubit assemblage
as a candidate for higher-dimensional ``bit'' assemblage. With the
more detailed notation of \cite{Gallego2015}, it reads 
\begin{align}
\sigma_{\boldsymbol{a}|\boldsymbol{x}}^{\text{bit}}=P_{\boldsymbol{A}|\boldsymbol{X}}(\boldsymbol{a}|\boldsymbol{x})\ \ket{\psi({\boldsymbol{a},\boldsymbol{x}})}\bra{\psi({\boldsymbol{a},\boldsymbol{x}})}\ .%
\label{eq:N-1DI}
\end{align}
We assume the NS principle only between the DD party and all others,
the $N-1$ DI parties may signal to each other at will. We will show
that no single choice of $\sigma_{\boldsymbol{a}|\boldsymbol{x}}^{\text{bit}}$
can be freely transformed into members of a family of minimal-dimension
assemblages $\sigma_{a_{f}|x_{f}}^{\theta}=\frac{1}{2}\ket{\psi^{\theta}(a_{f},x_{f})}\bra{\psi^{\theta}(a_{f},x_{f})}$
for all $\theta\in{}]0,\pi/2[{}$, where \begin{subequations}
\begin{align}
\ket{\psi^{\theta}(0,0)} & =\ket{0}\label{eq:psitheta00}\\
\ket{\psi^{\theta}(1,0)} & =\ket{1}\label{eq:psitheta10}\\
\ket{\psi^{\theta}(0,1)} & =\ \ \cos\theta\ket{0}+\sin\theta\ket1\label{eq:psitheta01}\\
\ket{\psi^{\theta}(1,1)} & =-\sin\theta\ket{0}+\cos\theta\ket1\ .\label{eq:psitheta11}
\end{align}
\label{eq:psitheta}\end{subequations}\par The most general form
of a 1W-LOCC applied to $\sigma_{\boldsymbol{a}|\boldsymbol{x}}^{\text{bit}}$
is 
\begin{equation}
\sum_{\boldsymbol{a},\boldsymbol{x},\omega}P_{\boldsymbol{X}|X_{f},\Omega}^{\theta}(\boldsymbol{x}|x_{f},\omega)\ P_{A_{f}|\boldsymbol{A},\boldsymbol{X},\Omega,X_{f}}^{\theta}(a_{f}|\boldsymbol{a},\boldsymbol{x},\omega,x_{f})\ P_{\boldsymbol{A}|\boldsymbol{X}}(\boldsymbol{a}|\boldsymbol{x})\ K_{\omega}^{\theta}\ket{\psi({\boldsymbol{a},\boldsymbol{x}})}\bra{\psi({\boldsymbol{a},\boldsymbol{x}})}K_{\omega}^{\theta\dagger}\ ,\label{eq:1WLOCCapplied}
\end{equation}
where $\Omega$ is a variable (with values $\omega$) representing
information sent by the quantum party to the classical ones, $P_{\boldsymbol{X}|X_{f},\Omega}^{\theta}$
and $P_{A_{f}|\boldsymbol{A},\boldsymbol{X},\Omega,X_{f}}^{\theta}$
are conditional probability distributions, and $K_{\omega}^{\theta}$
is a Kraus operator \cite{Gallego2015}; the three may depend on $\theta$.
Since this transformed assemblage is intended to equal the rank-1
assemblage $\sigma_{a_{f}|x_{f}}^{\theta}$, we can conclude that
$\forall\ a_{f},x_{f}$ 
\begin{multline}
\sum_{\boldsymbol{a},\boldsymbol{x}}P_{\boldsymbol{X}|X_{f},\Omega}^{\theta}(\boldsymbol{x}|x_{f},\omega)P_{A_{f}|\boldsymbol{A},\boldsymbol{X},\Omega,X_{f}}^{\theta}(a_{f}|\boldsymbol{a},\boldsymbol{x},\omega,x_{f})P_{\boldsymbol{A}|\boldsymbol{X}}(\boldsymbol{a}|\boldsymbol{x})\ K_{\omega}^{\theta}\ket{\psi({\boldsymbol{a},\boldsymbol{x}})}\bra{\psi({\boldsymbol{a},\boldsymbol{x}})}K_{\omega}^{\theta\dagger}\\
\sim\ket{\psi^{\theta}(a_{f},x_{f})}\bra{\psi^{\theta}(a_{f},x_{f})}\ ,\label{eq:1WLOCCequality}
\end{multline}
where $\sim$ signifies ``is either null or proportional to'' and
we have used the fact that the relation, valid for the sum in $\omega$,
is also valid for each $\omega$ term.\par We will assume for now
that $\sigma_{\boldsymbol{a}|\boldsymbol{x}}^{\text{bit}}$ is not
a single-state assemblage, i.e., there is no state $\ket{\psi_{\text{single}}}$
such that $\ket{\psi(\boldsymbol{a},\boldsymbol{x})}=\ket{\psi_{\text{single}}}$
for all $\boldsymbol{a},\boldsymbol{x}$ (for our purposes throughout
this proof, states are equal if they differ only by an global phase).\par We
now notice that, due to normalization, $\forall\ x_{f},\omega$, $\exists\ \tilde{\boldsymbol{x}},\tilde{\boldsymbol{a}},\tilde{a}_{f}$
such that $P_{\boldsymbol{X}|X_{f},\Omega}^{\theta}(\tilde{\boldsymbol{x}}|x_{f},\omega)\times P_{A_{f}|\boldsymbol{A},\boldsymbol{X},\Omega,X_{f}}^{\theta}(\tilde{a}_{f}|\tilde{\boldsymbol{a}},\tilde{\boldsymbol{x}},\omega,x_{f})\times P_{\boldsymbol{A}|\boldsymbol{X}}(\tilde{\boldsymbol{a}}|\tilde{\boldsymbol{x}})\neq0$.
For these values, then, 
\begin{equation}
K_{\omega}^{\theta}\ket{\psi(\tilde{\boldsymbol{a}},\tilde{\boldsymbol{x}})}\sim\ket{\psi^{\theta}(\tilde{a}_{f},x_{f})}\ .\label{eq:Kpsi}
\end{equation}
In fact, there must be at least two different values $\tilde{\boldsymbol{a}}$
for each $\tilde{\boldsymbol{x}}$ for which Eq.\ (\ref{eq:Kpsi})
is true, with the corresponding pure states $\ket{\psi(\tilde{\boldsymbol{a}},\tilde{\boldsymbol{x}})}$
being not all equal
: if, for some $\tilde{\boldsymbol{x}}$ there is a single $\tilde{\boldsymbol{a}}$
with $P_{\boldsymbol{A}|\boldsymbol{X}}(\tilde{\boldsymbol{a}}|\tilde{\boldsymbol{x}})\neq0$,
then by purity and the NS property between the DD and DI partitions,
$\sigma_{\boldsymbol{a}|\boldsymbol{x}}^{\text{bit}}$ would be a
single-state assemblage; if for all values $\tilde{\boldsymbol{a}}$,
$\ket{\psi(\tilde{\boldsymbol{a}},\tilde{\boldsymbol{x}})}$ is the
same, it would also be a single-state assemblage due to NS and purity.\par Let
us now exclude the possibility of $K_{\omega}^{\theta}\ket{\psi(\boldsymbol{a},\boldsymbol{x})}=0$
with $K_{\omega}^{\theta}\neq0$. If that were the case, $K_{\omega}^{\theta}$
would have a rank-1 support, hence a rank-1 span: $K_{\omega}^{\theta}\ket{\psi(\boldsymbol{a},\boldsymbol{x})}\sim\ket{k_{\omega}^{\theta}}\ \forall\ \boldsymbol{a},\boldsymbol{x}$.
From (\ref{eq:1WLOCCequality}) and the independence of $x_{f}$ from
$\omega$, this would require either $\ket{\psi^{\theta}(a_{f},0)}\propto\ket{k_{\omega}^{\theta}}\propto\ket{\psi^{\theta}(\tilde{a}_{f},1)}$
{[}contradiction with Eq.\ (\ref{eq:psitheta}){]} or that, for some
value of $x_{f}$, for the corresponding $\tilde{\boldsymbol{x}}$,
$K_{\omega}^{\theta}\ket{\psi(\tilde{\boldsymbol{a}},\tilde{\boldsymbol{x}})}=0$
for all $\tilde{\boldsymbol{a}}$ with $P_{\boldsymbol{A}|\boldsymbol{X}}(\tilde{\boldsymbol{a}}|\tilde{\boldsymbol{x}})\neq0$
{[}contradiction with there existing two different states $\ket{\psi(\tilde{\boldsymbol{a}},\tilde{\boldsymbol{x}})}${]}.\par Finally,
we can conclude from the dependencies of the three probabilities $P_{\boldsymbol{X}|X_{f},\Omega}^{\theta},$
$P_{A_{f}|\boldsymbol{A},\boldsymbol{X},\Omega,X_{f}}^{\theta},$
$P_{\boldsymbol{A}|\boldsymbol{X}}$ on $x_{f},\omega,\boldsymbol{x},\boldsymbol{a},\tilde{a}_{f}$,
that 
\begin{equation}
K_{\omega}^{\theta}\ket{\psi(\tilde{\boldsymbol{a}},\tilde{\boldsymbol{x}})}\propto\ket{\psi^{\theta}(\tilde{a}_{f},x_{f})}\ .\label{eq:Kpsiprop}
\end{equation}
The validity conditions of this equation are as follows: for all $(x_{f},\omega)$,
there exists some value $\tilde{\boldsymbol{x}}$ for which (\ref{eq:Kpsiprop})
holds; for each $\tilde{\boldsymbol{x}}$, there are at least two
values $\tilde{\boldsymbol{a}}$ for which (\ref{eq:Kpsiprop}) holds;
and for each choice of $(x_{f},\omega,\tilde{\boldsymbol{x}},\tilde{\boldsymbol{a}})$
there is some value $\tilde{a}_{f}$ for which (\ref{eq:Kpsiprop})
holds. Moreover, for given $\tilde{\boldsymbol{x}}$, the corresponding
$\ket{\psi(\tilde{\boldsymbol{a}},\tilde{\boldsymbol{x}})}$ (for
varying $\tilde{\boldsymbol{a}}$) are not all equal.\par Let us
explore the possible ways of satisfying Eq.\ (\ref{eq:Kpsiprop})
by case analysis. A first possibility is that, for the two different
values $x_{f}=0,1$, the values of $\tilde{\boldsymbol{x}}$ for which
(\ref{eq:Kpsiprop}) holds intersect at some value $\tilde{\boldsymbol{x}}_{\text{int}}$.
Then $\exists\ \tilde{\boldsymbol{a}},\tilde{a}_{f0},\tilde{a}_{f1}$
such that 
\begin{equation}
\begin{split}K_{\omega}^{\theta}\ket{\psi(\tilde{\boldsymbol{a}},\tilde{\boldsymbol{x}}_{\text{int}})} & \propto\ket{\psi^{\theta}(\tilde{a}_{f0},x_{f}=0)}\ ,\\
K_{\omega}^{\theta}\ket{\psi(\tilde{\boldsymbol{a}},\tilde{\boldsymbol{x}}_{\text{int}})} & \propto\ket{\psi^{\theta}(\tilde{a}_{f1},x_{f}=1)}\ ,
\end{split}
\label{eq:x_int}
\end{equation}
which is incompatible with Eq.\ (\ref{eq:psitheta}). We are then
left with the values $\tilde{\boldsymbol{x}}$ for $x_{f}=0$ and
$x_{f}=1$ being all different. Taking the liberty to relabel our
variables, let us consider a value $\tilde{\boldsymbol{x}}=\boldsymbol{0}$
for $x_{f}=0$ and a value $\tilde{\boldsymbol{x}}=\boldsymbol{1}$
for $x_{f}=1$, ignoring the other possible values of $\tilde{\boldsymbol{x}}$
for which Eq.\ (\ref{eq:Kpsiprop}) holds. Let us call $\tilde{\boldsymbol{a}}=\boldsymbol{0}$
and $\tilde{\boldsymbol{a}}=\boldsymbol{1}$ the two values of $\tilde{\boldsymbol{a}}$
for which, given $\tilde{\boldsymbol{x}}$, Eq.\ (\ref{eq:Kpsiprop})
holds. We see that $\tilde{a}_{f}$ could take any value for each
$\tilde{\boldsymbol{a}}$. However, if $\tilde{a}_{f}$ is the same
for the same $(x_{f},\tilde{\boldsymbol{x}})$ and two different $\tilde{\boldsymbol{a}}$,
e.g., 
\begin{equation}
\begin{split} & K_{\omega}\ket{\psi(\boldsymbol{0},\boldsymbol{1})}\propto\ket{\psi^{\theta}(0,1)}\\
 & K_{\omega}\ket{\psi(\boldsymbol{1},\boldsymbol{1})}\propto\ket{\psi^{\theta}(0,1)}\ ,
\end{split}
\label{eq:cases_b_c}
\end{equation}
then Eq.\ (\ref{eq:Kpsiprop}) cannot be satisfied for all $x_{f}$.
This is because $\{\ket{\psi(\boldsymbol{0},\boldsymbol{1})},\ket{\psi(\boldsymbol{1},\boldsymbol{1})}\}$
form a basis of the qubit Hilbert space, hence $K_{\omega}$ has a
1-rank span given by $\ket{\psi^{\theta}(0,1)}$, which does not span
$\ket{\psi^{\theta}(\tilde{a}_{f},0)}$ as needed. Hence $\tilde{a}_{f}$
is different for each $\tilde{\boldsymbol{a}}$ value.\par We can
then conclude that, up to relabeling, there must be states $\ket{\psi(\tilde{\boldsymbol{a}},\tilde{\boldsymbol{x}})}$
belonging to $\boldsymbol{\sigma}^{\text{bit}}$ which obey \begin{subequations}
\begin{align}
 & K_{\omega}\ket{\psi(\boldsymbol{0},\boldsymbol{0})}\propto\ket{\psi^{\theta}(0,0)}\label{eq:case_a00}\\
 & K_{\omega}\ket{\psi(\boldsymbol{1},\boldsymbol{0})}\propto\ket{\psi^{\theta}(1,0)}\label{eq:case_a10}\\
 & K_{\omega}\ket{\psi(\boldsymbol{0},\boldsymbol{1})}\propto\ket{\psi^{\theta}(0,1)}\label{eq:case_a01}\\
 & K_{\omega}\ket{\psi(\boldsymbol{1},\boldsymbol{1})}\propto\ket{\psi^{\theta}(1,1)}\label{eq:case_a11}
\end{align}
\label{eq:case_a}\end{subequations} to obtain the family of assemblages
$\{\boldsymbol{\sigma}^{\theta}\}_{\theta\in{}]0,\pi/2[{}}$. We will
choose the following parametrization: 
\begin{equation}
\ket{\psi(\tilde{\boldsymbol{a}},\tilde{\boldsymbol{x}})}=\cos(\varphi_{\tilde{\boldsymbol{a}},\tilde{\boldsymbol{x}}})\ket0+e^{i\alpha_{\tilde{\boldsymbol{a}},\tilde{\boldsymbol{x}}}}\sin(\varphi_{\tilde{\boldsymbol{a}},\tilde{\boldsymbol{x}}})\ket1\ ,\label{eq:parametrization}
\end{equation}
where $\varphi_{\tilde{\boldsymbol{a}},\tilde{\boldsymbol{x}}}\in[0,\pi/2]$.
It should be noted that $(\varphi_{\tilde{\boldsymbol{a}},\tilde{\boldsymbol{x}}},\alpha_{\tilde{\boldsymbol{a}},\tilde{\boldsymbol{x}}})$
may depend on $\theta$ through $\tilde{\boldsymbol{a}},\tilde{\boldsymbol{x}}$:
because $P_{\boldsymbol{X}|X_{f},\Omega}^{\theta}$ may depend on
$\theta$, the values $\tilde{\boldsymbol{a}},\tilde{\boldsymbol{x}}$
for which Eq.\ (\ref{eq:Kpsiprop}) holds may vary for different
values of $\theta$. However, for finitely many values of $\boldsymbol{a}$,
$\boldsymbol{x}$, there are only finitely many states and finitely
many $(\varphi_{{\boldsymbol{a}},{\boldsymbol{x}}},\alpha_{{\boldsymbol{a}},{\boldsymbol{x}}})$
to pick from, so some choice of states as in Eq.\ (\ref{eq:parametrization})
must still be able to satisfy Eq.\ (\ref{eq:case_a}) for a continuous
set of values $\theta$.\par Substituting Eqs.\ (\ref{eq:psitheta})
and (\ref{eq:parametrization}) in (\ref{eq:case_a00},\ref{eq:case_a10}),
respectively, we see that 
\begin{equation}
\frac{K_{\omega00}^{\theta}}{K_{\omega01}^{\theta}}=-e^{i\alpha_{\boldsymbol{1}\boldsymbol{0}}}\tan\varphi_{\boldsymbol{1}\boldsymbol{0}};\ \frac{K_{\omega10}^{\theta}}{K_{\omega11}^{\theta}}=-e^{i\alpha_{\boldsymbol{0}\boldsymbol{0}}}\tan\varphi_{\boldsymbol{0}\boldsymbol{0}};\label{eq:Kraus_1}
\end{equation}
where $K_{\omega ij}^{\theta}:=\braket{i}{K_{\omega}^{\theta}|j}$.
Doing the same in (\ref{eq:case_a01},\ref{eq:case_a11}) and substituting
(\ref{eq:Kraus_1}), we find, respectively, 
\begin{align}
\frac{K_{\omega11}^{\theta}}{K_{\omega01}^{\theta}} & =\tan\theta\frac{\tan\varphi_{\boldsymbol{0}\boldsymbol{1}}e^{i\alpha_{\boldsymbol{0}\boldsymbol{1}}}-\tan\varphi_{\boldsymbol{1}\boldsymbol{0}}e^{i\alpha_{\boldsymbol{1}\boldsymbol{0}}}}{\tan\varphi_{\boldsymbol{0}\boldsymbol{1}}e^{i\alpha_{\boldsymbol{0}\boldsymbol{1}}}+\tan\varphi_{\boldsymbol{0}\boldsymbol{0}}e^{i\alpha_{\boldsymbol{0}\boldsymbol{0}}}}\label{eq:Kraus_2}\\
\frac{K_{\omega11}^{\theta}}{K_{\omega01}^{\theta}} & =\frac{-1}{\tan\theta}\frac{\tan\varphi_{\boldsymbol{1}\boldsymbol{1}}e^{i\alpha_{\boldsymbol{1}\boldsymbol{1}}}-\tan\varphi_{\boldsymbol{1}\boldsymbol{0}}e^{i\alpha_{\boldsymbol{1}\boldsymbol{0}}}}{\tan\varphi_{\boldsymbol{1}\boldsymbol{1}}e^{i\alpha_{\boldsymbol{1}\boldsymbol{1}}}-\tan\varphi_{\boldsymbol{0}\boldsymbol{0}}e^{i\alpha_{\boldsymbol{0}\boldsymbol{0}}}}.\label{eq:Kraus_3}
\end{align}
Equating the two, we have 
\begin{equation}
\begin{split}\tan^{2}\theta\left(\frac{\tan\varphi_{\boldsymbol{0}\boldsymbol{1}}e^{i\alpha_{\boldsymbol{0}\boldsymbol{1}}}-\tan\varphi_{\boldsymbol{1}\boldsymbol{0}}e^{i\alpha_{\boldsymbol{1}\boldsymbol{0}}}}{\tan\varphi_{\boldsymbol{0}\boldsymbol{1}}e^{i\alpha_{\boldsymbol{0}\boldsymbol{1}}}+\tan\varphi_{\boldsymbol{0}\boldsymbol{0}}e^{i\alpha_{\boldsymbol{0}\boldsymbol{0}}}}\right)++\left(\frac{\tan\varphi_{\boldsymbol{1}\boldsymbol{1}}e^{i\alpha_{\boldsymbol{1}\boldsymbol{1}}}-\tan\varphi_{\boldsymbol{1}\boldsymbol{0}}e^{i\alpha_{\boldsymbol{1}\boldsymbol{0}}}}{\tan\varphi_{\boldsymbol{1}\boldsymbol{1}}e^{i\alpha_{\boldsymbol{1}\boldsymbol{1}}}-\tan\varphi_{\boldsymbol{0}\boldsymbol{0}}e^{i\alpha_{\boldsymbol{0}\boldsymbol{0}}}}\right)=0\ ,\end{split}
\label{eq:tantheta}
\end{equation}
which, for fixed $\varphi_{\tilde{\boldsymbol{a}},\tilde{\boldsymbol{x}}},\alpha_{\tilde{\boldsymbol{a}},\tilde{\boldsymbol{x}}}$,
must hold for a continuous set of values $\theta$. This is only possible
if both parentheses are zero, which in turn implies $(\varphi_{\boldsymbol{0},\boldsymbol{1}},\alpha_{\boldsymbol{0},\boldsymbol{1}})=(\varphi_{\boldsymbol{1},\boldsymbol{0}},\alpha_{\boldsymbol{1},\boldsymbol{0}})=(\varphi_{\boldsymbol{1},\boldsymbol{1}},\alpha_{\boldsymbol{1},\boldsymbol{1}})$,
or $\ket{\psi(\boldsymbol{0},\boldsymbol{1})}=\ket{\psi(\boldsymbol{1},\boldsymbol{0})}=\ket{\psi(\boldsymbol{1},\boldsymbol{1})}$,
contradicting the established relation $\ket{\psi(\boldsymbol{0},\boldsymbol{1})}\neq\ket{\psi(\boldsymbol{1},\boldsymbol{1})}$.
This concludes the demonstration for non-single-state assemblages.\par Finally,
let us show that a single-state assemblage is unable to do the task.
From (\ref{eq:1WLOCCapplied}), 
\begin{equation}
\begin{split}\sum_{\boldsymbol{a},\boldsymbol{x}}P_{\boldsymbol{X}|X_{f},\Omega}^{\theta}(\boldsymbol{x}|x_{f},\omega) & P_{A_{f}|\boldsymbol{A},\boldsymbol{X},\Omega,X_{f}}^{\theta}(a_{f}|\boldsymbol{a},\boldsymbol{x},\omega,x_{f})P_{\boldsymbol{A}|\boldsymbol{X}}(\boldsymbol{a}|\boldsymbol{x})\times\\
 & \times K_{\omega}^{\theta}\ket{\psi_{\text{single}}}\bra{\psi_{\text{single}}}K_{\omega}^{\theta\dagger}\sim\ket{\psi^{\theta}(a_{f},x_{f})}\bra{\psi^{\theta}(a_{f},x_{f})}\ .
\end{split}
\label{eq:1WLOCC_singlestate}
\end{equation}
 The sum on the left-hand side is not zero for at least two pairs
$(a_{f},x_{f})$, hence $K_{\omega}^{\theta}\ket{\psi_{\text{single}}}$
must be proportional to $\ket{\psi^{\theta}(a_{f},x_{f})}$ for both
these pairs. This is incompatible with Eq.\ (\ref{eq:psitheta}),
since none of the $\ket{\psi^{\theta}(a_{f},x_{f})}$ are proportional
to one another. \end{proof}

\section{Redefinition of genuinely multipartite steering}

\label{sec:def_gen_multipartite} Although our discussion has focused
on steering along a fixed bipartition, it has a bearing on genuine
multipartite steering as well. This concept hinges on bi-separability
over all possible bipartitions, as used by D. Cavalcanti \emph{et
al} to define genuine multipartite steering in \cite{Cavalcanti2015a}.
Interestingly, however, our results can be used to generalize that
definition.

\noindent Redefinition of genuinely multipartite steering: \emph{An
assemblage $\boldsymbol{\sigma}$ is genuinely multipartite steerable
if it does not} admit a decomposition of the form \begin{subequations}
\begin{eqnarray}
\sigma_{a,b|x,y}=\sum_{\mu}p_{\mu}^{A|BC} & P_{a|x;\mu} & \sigma_{b|y}^{C}(\mu)\label{eq:biseparableassemblage1}\\
+\sum_{\nu}p_{\nu}^{B|AC} & P_{b|y;\nu} & \sigma_{a|x}^{C}(\nu)\label{eq:biseparableassemblage2}\\
+\sum_{\lambda}p_{\lambda}^{AB|C} & P_{a,b|x,y,\lambda}\  & \varrho^{C}(\lambda)\label{eq:biseparableassemblage3}
\end{eqnarray}
\label{eq:biseparableassemblage}\end{subequations} \emph{where the
last sum can be any TO-LHS assemblage.}

The difference from D. Cavalcanti \emph{et al}'s definition is that
they consider assemblages obtained from a quantum realization with
bi-separable states. Reproducing Eqs.\ (4,5,6) of \cite{Cavalcanti2015a},
a tripartite state $\varrho^{ABC}$ is bi-separable when decomposable
as \begin{subequations}
\begin{eqnarray}
\varrho^{ABC}=\sum_{\mu}p_{\mu}^{A|BC} & \ \varrho_{\mu}^{A} & \otimes\varrho_{\mu}^{BC}\label{eq:biseparablestate1}\\
+\sum_{\nu}p_{\nu}^{B|AC} & \ \varrho_{\nu}^{B} & \otimes\varrho_{\nu}^{AC}\label{eq:biseparablestate2}\\
+\sum_{\lambda}p_{\lambda}^{AB|C} & \ \varrho_{\lambda}^{AB} & \otimes\varrho_{\lambda}^{C}\ .\label{eq:biseparablestate3}
\end{eqnarray}
\label{eq:biseparablestate}\end{subequations} Under local measurements
on the $A$ and $B$ partitions, this yields a 2DI+1DD assemblage
of the form (\ref{eq:biseparableassemblage}) (akin to Eqs.\ (7,8,9)
of \cite{Cavalcanti2015a}), but with a distribution $P_{a,b|x,y,\lambda}$
in Eq.\ (\ref{eq:biseparableassemblage3}) necessarily quantum-realizable
(a subset of NS distributions). In other words, they only allow the
sum in Eq.\ (\ref{eq:biseparableassemblage3}) to be quantum-realizable
NS-LHS assemblages. Our redefinition, then, reduces the set of genuinely
multipartite steerable assemblages.

Morover, we show in Section \ref{app:strict_subset} that there are,
in fact, quantum-realizable assemblages affected by this change. These
assemblages are decomposable as in Eq.\ (\ref{eq:biseparableassemblage})
only with a TO-LHS (not NS-LHS) term in Eq.\ (\ref{eq:biseparableassemblage3}),
and hence their quantum realization requires genuinely multipartite
entangled states {[}i.e.\ not decomposable as Eq.\ (\ref{eq:biseparablestate}){]}.
Interestingly, in this case genuine multipartite entanglement is certified
in the semi-DI scenario without steering: the need for a TO-LHS term
in Eq.\ (\ref{eq:biseparableassemblage3}) implies the inexistence
of a bi-separable decomposition (\ref{eq:biseparablestate}) for the
underlying quantum state, and also implies unsteerability.

\section{On the sets of LHS assemblages, TO-LHS assemblages, and NS-LHS assemblages}

\label{app:strict_subset}

{\scriptsize{}}
\begin{table*}[t]
{\scriptsize{}\setlength{\tabcolsep}{0.5em} 
}%
\begin{tabular}{cc|cc|c}
{\scriptsize{}$a$} & {\scriptsize{}$b$} & {\scriptsize{}$x$} & {\scriptsize{}$y$} & {\scriptsize{}$\sigma_{a,b|x,y}^{W}$ }\tabularnewline
\hline 
{\scriptsize{}0} & {\scriptsize{}0} & {\scriptsize{}0} & {\scriptsize{}0} & {\scriptsize{}$\frac{1}{6}\left[2\eta^{2}\ket0\bra0+(1+\sqrt{1-\eta^{2}}-\eta^{2}/2)\ket1\bra1+\eta(1+\sqrt{1-\eta^{2}})X\right]$ }\tabularnewline
{\scriptsize{}0} & {\scriptsize{}1} & {\scriptsize{}0} & {\scriptsize{}0} & {\scriptsize{}$\frac{1}{6}\left[2(1-\eta^{2})\ket0\bra0+\eta^{2}/2\ket1\bra1-\eta\sqrt{1-\eta^{2}}X\right]$ }\tabularnewline
{\scriptsize{}1} & {\scriptsize{}0} & {\scriptsize{}0} & {\scriptsize{}0} & {\scriptsize{}$\frac{1}{6}\left[2(1-\eta^{2})\ket0\bra0+\eta^{2}/2\ket1\bra1-\eta\sqrt{1-\eta^{2}}X\right]$ }\tabularnewline
{\scriptsize{}1} & {\scriptsize{}1} & {\scriptsize{}0} & {\scriptsize{}0} & {\scriptsize{}$\frac{1}{6}\left[2\eta^{2}\ket0\bra0+(1-\sqrt{1-\eta^{2}}-\eta^{2}/2)\ket1\bra1-\eta(1-\sqrt{1-\eta^{2}})X\right]$ }\tabularnewline
\hline 
{\scriptsize{}0} & {\scriptsize{}0} & {\scriptsize{}0} & {\scriptsize{}1} & {\scriptsize{}$\frac{1}{12}\left[2(1+2\eta\sqrt{1-\eta^{2}})\ket0\bra0+(1-\eta+\sqrt{1-\eta^{2}}-\eta\sqrt{1-\eta^{2}})\ket1\bra1+(1+\eta+\sqrt{1-\eta^{2}}-2\eta^{2})X\right]$ }\tabularnewline
{\scriptsize{}0} & {\scriptsize{}1} & {\scriptsize{}0} & {\scriptsize{}1} & {\scriptsize{}$\frac{1}{12}\left[2(1-2\eta\sqrt{1-\eta^{2}})\ket0\bra0+(1+\eta+\sqrt{1-\eta^{2}}+\eta\sqrt{1-\eta^{2}})\ket1\bra1-(1-\eta+\sqrt{1-\eta^{2}}-2\eta^{2})X\right]$ }\tabularnewline
{\scriptsize{}1} & {\scriptsize{}0} & {\scriptsize{}0} & {\scriptsize{}1} & {\scriptsize{}$\frac{1}{12}\left[2(1-2\eta\sqrt{1-\eta^{2}})\ket0\bra0+(1-\eta-\sqrt{1-\eta^{2}}+\eta\sqrt{1-\eta^{2}})\ket1\bra1-(1+\eta-\sqrt{1-\eta^{2}}-2\eta^{2})X\right]$ }\tabularnewline
{\scriptsize{}1} & {\scriptsize{}1} & {\scriptsize{}0} & {\scriptsize{}1} & {\scriptsize{}$\frac{1}{12}\left[2(1+2\eta\sqrt{1-\eta^{2}})\ket0\bra0+(1+\eta-\sqrt{1-\eta^{2}}-\eta\sqrt{1-\eta^{2}})\ket1\bra1+(1-\eta-\sqrt{1-\eta^{2}}-2\eta^{2})X\right]$ }\tabularnewline
\hline 
{\scriptsize{}$a$} & {\scriptsize{}$b$} & {\scriptsize{}1} & {\scriptsize{}0} & {\scriptsize{}$\sigma_{a,b|1,0}^{W}=\sigma_{b,a|0,1}^{W}$}\tabularnewline
\hline 
{\scriptsize{}0} & {\scriptsize{}0} & {\scriptsize{}1} & {\scriptsize{}1} & {\scriptsize{}$\frac{1}{6}\left[2(1-\eta^{2})\ket0\bra0+(1-\eta-(1-\eta^{2})/2)\ket1\bra1+\sqrt{1-\eta^{2}}(1-\eta)X\right]$ }\tabularnewline
{\scriptsize{}0} & {\scriptsize{}1} & {\scriptsize{}1} & {\scriptsize{}1} & {\scriptsize{}$\frac{1}{6}\left[2\eta^{2}\ket0\bra0+(1-\eta^{2})/2\ket1\bra1+\eta\sqrt{1-\eta^{2}}X\right]$ }\tabularnewline
{\scriptsize{}1} & {\scriptsize{}0} & {\scriptsize{}1} & {\scriptsize{}1} & {\scriptsize{}$\frac{1}{6}\left[2\eta^{2}\ket0\bra0+(1-\eta^{2})/2\ket1\bra1+\eta\sqrt{1-\eta^{2}}X\right]$ }\tabularnewline
{\scriptsize{}1} & {\scriptsize{}1} & {\scriptsize{}1} & {\scriptsize{}1} & {\scriptsize{}$\frac{1}{6}\left[2(1-\eta^{2})\ket0\bra0+(1+\eta-(1-\eta^{2})/2)\ket1\bra1-\sqrt{1-\eta^{2}}(1+\eta)X\right]$ }\tabularnewline
\end{tabular}{\scriptsize{}\caption{Example quantum assemblage to demonstrate strict inclusion of \textsf{NS-LHS}
in \textsf{TO-LHS}.}
\label{tab:Wassemblage} }
\end{table*}
{\scriptsize\par}

We now state a theorem that sustains Fig.\ref{fig:setsLHS} b), concerning
the inclusion relations between the sets NS-LHS, TO-LHS, and LHS.
\begin{thm} NS-LHS $\subset$ TO-LHS $\subset$ LHS, and these relations
also hold strictly if we restrict to quantum-realizable assemblages.\label{th:sets_LHS}\end{thm}

\begin{proof} From the definitions in Eqs.\ (\ref{eq:LHS},\ref{eq:TOdef}),
it is clear that NS-LHS $\subseteq$ TO-LHS $\subseteq$ LHS. The
phenomenon of exposure implies that the assemblages in Eqs.\ (\ref{eq:genactiv},\ref{eq:quantumorig})
belong to LHS, but not to TO-LHS, so the inclusion of one in the other
is strict (notice that assemblage (\ref{main-eq:quantumorig}) is
quantum realizable). To prove that NS-LHS is a strict subset of TO-LHS,
we need an example of a TO-LHS assemblage that does not belong to
NS-LHS. 
One way to do so is to follow the reasoning of \cite{Gallego2012}:
take the time-ordered decomposition of the distribution $\boldsymbol{P}$
from \cite{Gallego2011} that violates the guess-your-neighbor's-input
(GYNI) inequality and find the $\varrho_{\lambda}$ that best mimic
the marginal $P_{a|x,\lambda}$ --- this effectively amounts to a
one-time program \cite{Roehsner2018}. The resulting TO-LHS assemblage
violates GYNI, hence is not NS-LHS, but it is also supra-quantum,
since no quantum state can violate the GYNI inequality.\par To find
a quantum-realizable assemblage that belongs to TO-LHS, but not to
NS-LHS, we take inspiration from Bancal \emph{et al} \cite{Bancal2013},
who have found Bell behaviors obtainable from noisy $W$ states with
the analogous DI-scenario property (TO-LHV, but not NS-LHV). A pure
$W$ state is given by $\ket W:=(\ket{001}+\ket{010}+\ket{100})/\sqrt{3}$,
its noisy version with visibility $v$, by 
\begin{equation}
\rho_{W}=v\ \ket W\bra W+(1-v)\ \mathbb{1}^{(ABC)}/8\ .\label{eq:noisyWstate}
\end{equation}
Alice and Bob make von Neumann measurements on the bases $\eta X+\sqrt{1-\eta^{2}}Z$
($x$ or $y=0$) and $\sqrt{1-\eta^{2}}X-\eta Z$ ($x$ or $y=1$),
with $\eta\approx0.97177$, which yields the assemblage 
\begin{equation}
\sigma_{a,b|x,y}^{\text{noisy }W}=v\ \sigma_{a,b|x,y}^{W}+(1-v)\ \mathbb{1}^{C}/8\ ,\label{eq:noisyWassemb}
\end{equation}
where $\sigma_{a,b|x,y}^{W}$ is given in Table \ref{tab:Wassemblage}.
These measurements, together with an appropriate measurement by Charlie,
yield in \cite{Bancal2013} a DI-inequality violation requiring minimal
visibility.\par We obtain the optimal NS-LHS witness $\boldsymbol{W}=\{W_{abxy}\}_{a,b,x,y}$
for $\sigma_{a,b|x,y}^{\text{noisy }W}$ for $v=0.58$, i.e. $\boldsymbol{W}$
satisfies the property 
\begin{equation}
-1\leq\sum_{a,b,x,y}\Tr[W_{abxy}\,\sigma_{a,b\vert x,y}^{\text{NS-LHS}}]\leq0\label{eq:witness_bound}
\end{equation}
for every NS-LHS assemblage $\boldsymbol{\sigma}^{\text{NS-LHS}}$.
Its components $W_{abxy}$ are given in Table \ref{tab:NS-LHSwitness}.
This witness is violated by $\sigma_{a,b|x,y}^{\text{noisy }W}$ from
$v\approx0.58$ onwards; for $v=0.64$, it returns 0.0301.\par 
\global\long\def\mspc{\hphantom{-}}%

{\footnotesize{} \begin{table*}[htb]
\setlength{\tabcolsep}{0.5em} 
\normalsize\centering 
{\footnotesize \begin{tabular}{|c|cccc|}\hline \diagbox{$x,y$}{$a,b$} &$00$&$01$&$10$&$11$\\ \hline  &&&&\\[-2ex] $00$ & $\begin{bmatrix}-0.0056 & \mspc 0.1194 \\ \mspc 0.1194 & -0.1205	\end{bmatrix}$& $\begin{bmatrix}-0.1394 & -0.0603\\ -0.0603 & \mspc 0.0662				\end{bmatrix}$& $\begin{bmatrix}-0.1394 & -0.0603\\ -0.0603 & \mspc 0.0662				\end{bmatrix}$& $\begin{bmatrix}\mspc 0.0239 & -0.0656 \\ -0.0656 & -0.1869				\end{bmatrix}$\\[2ex] $01$ & $\begin{bmatrix}\mspc 0.0233 & -0.0324 \\ -0.0324 & -0.1706				\end{bmatrix}$& $\begin{bmatrix}-0.2194 & \mspc 0.1346 \\ \mspc 0.1346 & -0.0079	\end{bmatrix}$& $\begin{bmatrix}-0.0560 &\mspc 0.1109\\ \mspc 0.1109 &\mspc 0.0114\end{bmatrix}$& $\begin{bmatrix}-0.0417 & -0.1490 \\ -0.1490 & -0.1079						\end{bmatrix}$\\[2ex] $10$ & $\begin{bmatrix}\mspc 0.0233 & -0.0324 \\-0.0324 & -0.1706				\end{bmatrix}$& $\begin{bmatrix}-0.0560 &\mspc 0.1109 \\ \mspc 0.1109&\mspc 0.0114\end{bmatrix}$& $\begin{bmatrix}-0.2194 &\mspc 0.1346 \\ \mspc 0.1346 & -0.0079		\end{bmatrix}$& $\begin{bmatrix}-0.0417 & -0.1490 \\ -0.1490 & -0.1079						\end{bmatrix}$ \\[2ex] $11$& $\begin{bmatrix}-0.0410 & -0.0560 \\ -0.0560 & \mspc 0.0863				\end{bmatrix}$& $\begin{bmatrix}\mspc 0.0665 &\mspc 0.0431\\ \mspc 0.0431 &-0.2194\end{bmatrix}$& $\begin{bmatrix}\mspc 0.0665 &\mspc 0.0431\\ \mspc 0.0431 &-0.2194\end{bmatrix}$& $\begin{bmatrix}	-0.4431 &  -0.0727 \\   -0.0727 &  \mspc 0.0239	\end{bmatrix}$ \\[2ex] \hline \end{tabular}} \caption{Elements of witness $W_{abxy}$ used to demonstrate strict inclusion of \textsf{NS-LHS} in \textsf{TO-LHS}.} \label{tab:NS-LHSwitness} 
\end{table*}

\begin{table*}[htb]
{\scriptsize \begin{tabular}{|cc||cc||cc||cc|}\hline \ \ \ $\lambda$ \ \ \ & $\sigma_\lambda$&\ \ \ $\lambda$\ \ \ &$\sigma_\lambda$&\ \ \ $\lambda$\ \ \ &$\sigma_\lambda$&\ \ \ $\lambda$\ \ \ & $\sigma_\lambda$ \rule{0pt}{2.5ex} \\[.5ex] \hline {0}&$\begin{bmatrix} 	\mspc0.0045 & \mspc0.0013\\	\mspc0.0013 & \mspc0.0009\\\end{bmatrix}$& {1}&$\begin{bmatrix}	\mspc0.0928 & \mspc0.0246\\	\mspc0.0246 & \mspc0.0070\\\end{bmatrix}$& {2}&$\begin{bmatrix}	\mspc0.0036 & \mspc0.0011\\	\mspc0.0011 & \mspc0.0009\\\end{bmatrix}$& {3}&$\begin{bmatrix}	\mspc0.0244 & \mspc0.0068\\	\mspc0.0068 & \mspc0.0024\\\end{bmatrix}$\rule{0pt}{4ex}\\[2ex] \hline {4}&$\begin{bmatrix}	\mspc0.0055 & \mspc0.0058\\	\mspc0.0058 & \mspc0.0071\\\end{bmatrix}$& {5}&$\begin{bmatrix}	\mspc0.0084 & \mspc0.0071\\	\mspc0.0071 & \mspc0.0067\\\end{bmatrix}$& {6}&$\begin{bmatrix}	\mspc0.0066 & \mspc0.0076\\	\mspc0.0076 & \mspc0.0098\\\end{bmatrix}$& {7}&$\begin{bmatrix}	\mspc0.0100 & \mspc0.0090\\	\mspc0.0090 & \mspc0.0089\\\end{bmatrix}$\rule{0pt}{4ex}\\[2ex]  \hline {8}&$\begin{bmatrix}	\mspc0.0048 & -0.0029\\	-0.0029 & \mspc0.0025\\\end{bmatrix}$& {9}&$\begin{bmatrix}	\mspc0.0118 & -0.0052\\	-0.0052 & \mspc0.0029\\\end{bmatrix}$& {10}&$\begin{bmatrix}	\mspc0.0040 & -0.0026\\	-0.0026 & \mspc0.0024\\\end{bmatrix}$& {11}&$\begin{bmatrix}	\mspc0.0079 & -0.0037\\	-0.0037 & \mspc0.0024\\\end{bmatrix}$\rule{0pt}{4ex}\\[2 ex] \hline {12}&$\begin{bmatrix}	\mspc0.0007 & -0.0004\\	-0.0004 & \mspc0.0024\\\end{bmatrix}$& {13}&$\begin{bmatrix}	\mspc0.0008 & -0.0002\\	-0.0002 & \mspc0.0014\\\end{bmatrix}$& {14}&$\begin{bmatrix}	\mspc0.0006 & -0.0004\\	-0.0004 & \mspc0.0029\\\end{bmatrix}$& {15}&$\begin{bmatrix}	\mspc0.0007 & -0.0002\\	-0.0002 & \mspc0.0015\\\end{bmatrix}$\rule{0pt}{4ex}\\[2 ex] \hline {16}&$\begin{bmatrix}	\mspc0.0219 & \mspc0.0118\\	\mspc0.0118 & \mspc0.0064\\\end{bmatrix}$& {17}&$\begin{bmatrix}	\mspc0.0001 & \mspc0.0002\\	\mspc0.0002 & \mspc0.0010\\\end{bmatrix}$& {18}&$\begin{bmatrix}	\mspc0.0028 & -0.0005\\	-0.0005 & \mspc0.0001\\\end{bmatrix}$& {19}&$\begin{bmatrix}	\mspc0.0002 & -0.0002\\	-0.0002 & \mspc0.0004\\\end{bmatrix}$\rule{0pt}{4ex}\\[2 ex] \hline {20}&$\begin{bmatrix}	\mspc0.0612 & \mspc0.0411\\	\mspc0.0411 & \mspc0.0277\\\end{bmatrix}$& {21}&$\begin{bmatrix}	\mspc0.0034 & \mspc0.0126\\	\mspc0.0126 & \mspc0.0467\\\end{bmatrix}$& {22}&$\begin{bmatrix}	\mspc0.0007 & -0.0001\\	-0.0001 & \mspc0.0001\\\end{bmatrix}$& {23}&$\begin{bmatrix}	\mspc0.0002 & -0.0002\\	-0.0002 & \mspc0.0004\\\end{bmatrix}$\rule{0pt}{4ex}\\[2 ex] \hline {24}&$\begin{bmatrix}	\mspc0.0007 & \mspc0.0003\\	\mspc0.0003 & \mspc0.0002\\\end{bmatrix}$& {25}&$\begin{bmatrix}	\mspc0.0001 & \mspc0.0001\\	\mspc0.0001 & \mspc0.0010\\\end{bmatrix}$& {26}&$\begin{bmatrix}	\mspc0.0135 & -0.0036\\	-0.0036 & \mspc0.0010\\\end{bmatrix}$& {27}&$\begin{bmatrix}	\mspc0.0074 & -0.0106\\	-0.0106 & \mspc0.0153\\\end{bmatrix}$\rule{0pt}{4ex}\\[2 ex] \hline {28}&$\begin{bmatrix}	\mspc0.0006 & \mspc0.0003\\	\mspc0.0003 & \mspc0.0003\\\end{bmatrix}$& {29}&$\begin{bmatrix}	\mspc0.0010 & \mspc0.0073\\	\mspc0.0073 & \mspc0.0545\\\end{bmatrix}$& {30}&$\begin{bmatrix}	\mspc0.0008 & -0.0002\\	-0.0002 & \mspc0.0001\\\end{bmatrix}$& {31}&$\begin{bmatrix}	\mspc0.0015 & -0.0025\\	-0.0025 & \mspc0.0045\\\end{bmatrix}$\rule{0pt}{4ex}\\[2 ex] \hline {32}&$\begin{bmatrix}	\mspc0.0020 & \mspc0.0006\\	\mspc0.0006 & \mspc0.0016\\\end{bmatrix}$& {33}&$\begin{bmatrix}	\mspc0.0049 & \mspc0.0013\\	\mspc0.0013 & \mspc0.0013\\\end{bmatrix}$& {34}&$\begin{bmatrix}	\mspc0.0017 & \mspc0.0006\\	\mspc0.0006 & \mspc0.0018\\\end{bmatrix}$& {35}&$\begin{bmatrix}	\mspc0.0038 & \mspc0.0011\\	\mspc0.0011 & \mspc0.0014\\\end{bmatrix}$\rule{0pt}{4ex}\\[2 ex] \hline {36}&$\begin{bmatrix}	\mspc0.0020 & -0.0013\\	-0.0013 & \mspc0.0022\\\end{bmatrix}$& {37}&$\begin{bmatrix}	\mspc0.0031 & -0.0012\\	-0.0012 & \mspc0.0014\\\end{bmatrix}$& {38}&$\begin{bmatrix}	\mspc0.0018 & -0.0013\\	-0.0013 & \mspc0.0024\\\end{bmatrix}$& {39}&$\begin{bmatrix}	\mspc0.0026 & -0.0011\\	-0.0011 & \mspc0.0015\\\end{bmatrix}$\rule{0pt}{4ex}\\[2 ex] \hline {40}&$\begin{bmatrix}	\mspc0.0037 & -0.0000\\	-0.0000 & \mspc0.0009\\\end{bmatrix}$& {41}&$\begin{bmatrix}	\mspc0.0261 & \mspc0.0009\\	\mspc0.0009 & \mspc0.0007\\\end{bmatrix}$& {42}&$\begin{bmatrix}	\mspc0.0029 & -0.0000\\	-0.0000 & \mspc0.0010\\\end{bmatrix}$& {43}&$\begin{bmatrix}	\mspc0.0125 & \mspc0.0005\\	\mspc0.0005 & \mspc0.0008\\\end{bmatrix}$\rule{0pt}{4ex}\\[2 ex] \hline {44}&$\begin{bmatrix}	\mspc0.0069 & -0.0040\\	-0.0040 & \mspc0.0032\\\end{bmatrix}$& {45}&$\begin{bmatrix}	\mspc0.0227 & -0.0094\\	-0.0094 & \mspc0.0045\\\end{bmatrix}$& {46}&$\begin{bmatrix}	\mspc0.0055 & -0.0034\\	-0.0034 & \mspc0.0030\\\end{bmatrix}$& {47}&$\begin{bmatrix}	\mspc0.0140 & -0.0060\\	-0.0060 & \mspc0.0033\\\end{bmatrix}$\rule{0pt}{4ex}\\[2 ex] \hline {48}&$\begin{bmatrix}	\mspc0.0062 & \mspc0.0036\\	\mspc0.0036 & \mspc0.0022\\\end{bmatrix}$& {49}&$\begin{bmatrix}	\mspc0.0011 & \mspc0.0051\\	\mspc0.0051 & \mspc0.0258\\\end{bmatrix}$& {50}&$\begin{bmatrix}	\mspc0.0031 & -0.0006\\	-0.0006 & \mspc0.0002\\\end{bmatrix}$& {51}&$\begin{bmatrix}	\mspc0.0007 & -0.0011\\	-0.0011 & \mspc0.0018\\\end{bmatrix}$\rule{0pt}{4ex}\\[2 ex] \hline {52}&$\begin{bmatrix}	\mspc0.0009 & \mspc0.0005\\	\mspc0.0005 & \mspc0.0003\\\end{bmatrix}$& {53}&$\begin{bmatrix}	\mspc0.0001 & \mspc0.0005\\	\mspc0.0005 & \mspc0.0034\\\end{bmatrix}$& {54}&$\begin{bmatrix}	\mspc0.0035 & -0.0008\\	-0.0008 & \mspc0.0003\\\end{bmatrix}$& {55}&$\begin{bmatrix}	\mspc0.0193 & -0.0303\\	-0.0303 & \mspc0.0479\\\end{bmatrix}$\rule{0pt}{4ex}\\[2 ex] \hline {56}&$\begin{bmatrix}	\mspc0.0044 & \mspc0.0023\\	\mspc0.0023 & \mspc0.0013\\\end{bmatrix}$& {57}&$\begin{bmatrix}	\mspc0.0002 & \mspc0.0004\\	\mspc0.0004 & \mspc0.0024\\\end{bmatrix}$& {58}&$\begin{bmatrix}	\mspc0.0287 & -0.0055\\	-0.0055 & \mspc0.0011\\\end{bmatrix}$& {59}&$\begin{bmatrix}	\mspc0.0008 & -0.0011\\	-0.0011 & \mspc0.0018\\\end{bmatrix}$\rule{0pt}{4ex}\\[2 ex] \hline {60}&$\begin{bmatrix}	\mspc0.0008 & \mspc0.0004\\	\mspc0.0004 & \mspc0.0003\\\end{bmatrix}$& {61}&$\begin{bmatrix}	\mspc0.0001 & \mspc0.0002\\	\mspc0.0002 & \mspc0.0015\\\end{bmatrix}$& {62}&$\begin{bmatrix}	\mspc0.0967 & -0.0246\\	-0.0246 & \mspc0.0063\\\end{bmatrix}$& {63}&$\begin{bmatrix}	\mspc0.0206 & -0.0300\\	-0.0300 & \mspc0.0440\\\end{bmatrix}$\rule{0pt}{4ex}\\[2 ex] \hline \end{tabular}} \caption{Non-normalized states $\sigma_\lambda$ needed in Eq.\ \eqref{eq:DeterministicTOLHS} for the TO-LHS decomposition of the assemblage \eqref{eq:noisyWassemb}.} \label{tab:sigmalambda} 
\end{table*}}{\footnotesize\par}

\par However, there is a TO-LHS decomposition of $\sigma_{a,b|x,y}^{\text{noisy }W}$
for $v=0.64$ (hence for $v<0.64$), which, equivalently to Eq.\ (\ref{eq:TOdef}),
can be written as \begin{subequations} 
\begin{align}
\sigma_{a,b|x,y}^{\text{noisy W}} & =\sum_{\lambda}D_{\lambda}(a|x)D_{\lambda}(b|x,y)\,\sigma_{\lambda}\label{eq:DeterministicTOLHS1}\\
 & =\sum_{\lambda}D_{\lambda}(a|x,y)D_{\lambda}(b|y)\,\sigma_{\lambda}\ ,\label{eq:DeterministicTOLHS2}
\end{align}
\label{eq:DeterministicTOLHS}\end{subequations} where the $D_{\lambda}$
are deterministic response functions and $\sigma_{\lambda}:=p_{\lambda}\rho_{\lambda}$
are non-normalized states. Each $D_{\lambda}(a|x)$ is specified by
$a_{x}$, the deterministic outcome $a$ conditioned on $x$; the
notation follows analogously for $D_{\lambda}(b|x,y)$, $D_{\lambda}(a|x,y)$,
and $D_{\lambda}(b|y)$ ($b_{xy}$, $a_{xy}$, and $b_{y}$, respectively).
These are given by {\footnotesize{}
\[
\begin{tabular}{|c|cccccc|}
\hline  \ensuremath{\lambda}  &  \ensuremath{a_{0}}  &  \ensuremath{a_{1}}  &  \ensuremath{b_{00}}  &  \ensuremath{b_{01}}  &  \ensuremath{b_{10}}  &  \ensuremath{b_{11}} \\
\hline  0  &  0  &  0  &  0  &  0  &  0  &  0 \\
 1  &  0  &  0  &  0  &  0  &  0  &  1 \\
 2  &  0  &  0  &  0  &  0  &  1  &  0 \\
 3  &  0  &  0  &  0  &  0  &  1  &  1 \\
\multicolumn{7}{c}{\ensuremath{\vdots}}\\
 62  &  1  &  1  &  1  &  1  &  1  &  0 \\
 63  &  1  &  1  &  1  &  1  &  1  &  1 \\
\hline   &   &   &   &   &   &  
\end{tabular}\ \ \ \ \begin{tabular}{|c|cccccc|}
\hline  \ensuremath{\lambda}  &  \ensuremath{a_{00}}  &  \ensuremath{a_{01}}  &  \ensuremath{a_{10}}  &  \ensuremath{a_{11}}  &  \ensuremath{b_{0}}  &  \ensuremath{b_{1}} \\
\hline  0  &  0  &  0  &  0  &  0  &  0  &  0 \\
 1  &  0  &  0  &  0  &  0  &  0  &  1 \\
 2  &  0  &  0  &  0  &  0  &  1  &  0 \\
 3  &  0  &  0  &  0  &  0  &  1  &  1 \\
\multicolumn{7}{c}{\ensuremath{\vdots}}\\
 62  &  1  &  1  &  1  &  1  &  1  &  0 \\
 63  &  1  &  1  &  1  &  1  &  1  &  1 \\
\hline   &   &   &   &   &   &  
\end{tabular}\ ,
\]
}where in each table, the six columns to the right are the binary
expression of the leftmost column ($\lambda$). The states $\sigma_{\lambda}$
are given in Table \ref{tab:sigmalambda}. \end{proof}\selectlanguage{american}

\ihead{}

\ohead{\textbf{Appendix~\thechapter}~\leftmark}

\ifoot{}

\cfoot{}

\ofoot{\thepage}

\chapter{Time evolution and CPF correlation for the decay of a two level system
in a bosonic bath \label{chap:Anexo1}}

In this Appendix, it is shown the detailed calculation to get to the
time evolved state (Eq. (\ref{Rho(t)-1}) in the main text) and to
the expression of the CPF correlation for the different schemes (Eq.
(\ref{zzz}) and Eq. (\ref{xzx}) in the main text). It is also shown
in more details the relation between the experimental HWP angles and
the theoretical time evolution.

\section{Solution of Eq. (\ref{eq:SchEq}) for initial separable states }

The dynamics of the composite system (qubit + bosonic environment)
is given by the total Hamiltonian (\ref{Bosonic-2}) that can be regarded
as the sum of a free evolution term $H_{0}=\frac{\omega_{0}}{2}\sigma_{z}+\sum_{k}\omega_{k}b_{k}^{\dag}b_{k}$
with an interaction term $H_{I}=\sum_{k}(g_{k}\sigma_{+}b_{k}+g_{k}^{\ast}\sigma_{-}b_{k}^{\dag})$.
In the interaction picture the total state satisfies the Schrödinger
equation (\ref{eq:SchEq}) with the time dependent interaction Hamiltonian
$H_{I}(t)=e^{iH_{0}t}H_{I}e^{-iH_{0}t}=\sum_{k}(g_{k}e^{i\omega_{0}t}\sigma_{+}e^{i\omega_{k}b_{k}^{\dag}b_{k}}b_{k}e^{-i\omega_{k}b_{k}^{\dag}b_{k}}+h.c.)$.
We assume that the this composite system is closed and initially in
a pure separable state\footnote{This is the case for the very initial state because the environment
starts to evolve from vacuum and the system is measured before each
step of time evolution.}. The commutator of the total number of excitations $N=\sigma_{+}\sigma_{-}+\sum_{k}b_{k}^{\dagger}b_{b}$
with the total Hamiltonian vanishes ($[H_{tot},N]=0$), therefore
this quantity is conserved. As in the initial instant we consider
that the system has at most one excitation and the environment is
in its vacuum state, then the state at time $t>0$ must be the general
state
\begin{equation}
|\Psi_{t}\rangle=\Big{[}a(t)\ket{\uparrow}+b(t)\ket{\downarrow}+\ket{\downarrow}\sum_{k}c_{k}(t)b_{k}^{\dag}\Big{]}|0\rangle,\label{evolved_state_t-2}
\end{equation}
which is the superposition of all possible composite states with at
most one excitation. 

From Schrödinger equation, the coefficients evolves as
\begin{equation}
\frac{d}{dt}b(t)=0.
\end{equation}
Therefore, $b(t)=b(0)=b.$ In addition, it follows that
\begin{eqnarray}
\frac{d}{dt}a(t) & = & -i\sum_{k}g_{k}\exp(+i\phi_{k}t)c_{k}(t),\\
\frac{d}{dt}c_{k}(t) & = & -ig_{k}^{\ast}\exp(-i\phi_{k}t)a(t),
\end{eqnarray}
where $\phi_{k}\equiv\omega_{0}-\omega_{k}.$ Integrating the last
equation as
\begin{equation}
c_{k}(t)=c_{k}(0)-ig_{k}^{\ast}\int_{0}^{t}dt^{\prime}\exp(-i\phi_{k}t^{\prime})a(t^{\prime}),\label{Ck-1}
\end{equation}
the evolution for $a(t)$ becomes
\begin{equation}
\frac{d}{dt}a(t)=-\int_{0}^{t}f(t-t^{\prime})a(t^{\prime})dt^{\prime}-ig(t).
\end{equation}
Here, $f(t)$ defines the bath correlation
\begin{equation}
f(t)\equiv\sum_{k}|g_{k}|^{2}\exp(+i\phi_{k}t),
\end{equation}
while the inhomogeneous term is
\begin{equation}
g(t)\equiv\sum_{k}g_{k}\exp(+i\phi_{k}t)c_{k}(0).\label{eq:g}
\end{equation}

Defining the Green function $G(t)$ by the evolution
\begin{equation}
\frac{d}{dt}G(t)=-\int_{0}^{t}f(t-t^{\prime})G(t^{\prime})dt^{\prime},\label{Propagator-1}
\end{equation}
with $G(0)=1,$ the coefficient $a(t)$ can be written as
\begin{equation}
a(t)=G(t)a(0)-i\int_{0}^{t}G(t-t^{\prime})g(t^{\prime})dt^{\prime}.\label{ASolution-1}
\end{equation}

\subsection*{Evolution in the time interval $(0,t)$}

The total system is prepared in the initial state $|\Psi_{0}\rangle=(a\ket{\uparrow}+b\ket{\downarrow})\otimes|0\rangle$
. Imediatelly after the preparation, the OQS is measured and the composite
system's state becomes $|\Psi_{0}^{x}\rangle=(a_{x}\ket{\uparrow}+b_{x}\ket{\downarrow})\otimes|0\rangle$,
with new coefficients dependent on the measurement outcome $x$.

Thus, the initial conditions are 
\begin{equation}
a(0)=a_{x},\ \ \ \ ,b(0)=b_{x}\ \ \;\ c_{k}(0)=0,\label{CIFirst-1}
\end{equation}
implying $g(t)=0$. The coefficients can be expressed as
\begin{equation}
a(t)=G(t)a_{x},\ \ \ \ \ \ c_{k}(t)=-ia_{x}g_{k}^{\ast}\int_{0}^{t}dt^{\prime}\exp(-i\phi_{k}t^{\prime})G(t^{\prime}).\label{TimeSolution-1}
\end{equation}

\subsection*{Evolution in the time interval $(t,t+\tau)$}

The measurement module $Y$ is applied, projecting the OQS into one
of its energy eigenstates $|\uparrow\rangle$ or $|\downarrow\rangle$.
Different initial conditions must be used for this second step of
time evolution , depending on this measurement result. The initial
state for this step is the projection of Eq. (\ref{evolved_state_t-2})
on $|\uparrow\rangle$ or $|\downarrow\rangle$ for $y=+1$ or $y=-1$,
respectively.

\textit{First Initial conditions}: when $y=+1$
\begin{equation}
\tilde{a}(0)=1,\ \ \ \;\tilde{b}(0)=0,\;\ \ \ \tilde{c}_{k}(0)=0,\label{CITilde-1}
\end{equation}
which implies $\tilde{g}(\tau)=0,$ if follows the solution 
\begin{equation}
\tilde{a}(\tau)=G(\tau),\label{eq:SolTil}
\end{equation}
while from Eq. (\ref{Ck-1}) we get
\begin{equation}
\tilde{c}_{k}(\tau)=-ig_{k}^{\ast}\int_{0}^{\tau}d\tau^{\prime}\exp[-i\phi_{k}(\tau^{\prime}+t)]G(\tau^{\prime}).
\end{equation}
These solutions are equivalent to the previous ones {[}Eq.~(\ref{TimeSolution-1}){]}
under the replacement $t\rightarrow\tau.$

\textit{Second initial conditions}: If $y=-1$, the set of initial
conditions for this second step is given by 
\begin{equation}
\tilde{a}^{\prime}(0)=0,\ \ \ \tilde{b}^{\prime}(0)=b_{x}/\sqrt{1-|G(t)|^{2}}\ \ \ \tilde{c}_{k}^{\prime}(0)=c_{k}(t)/\sqrt{1-|G(t)|^{2}}.\label{CITildePrimas-1}
\end{equation}
 From Eq.~(\ref{Ck-1}) we write $c_{k}(t)=-ig_{k}^{\ast}\int_{0}^{t}dt^{\prime}\exp(-i\phi_{k}t^{\prime})G(t^{\prime}).$
Thus, Eq.~(\ref{eq:g}) becomes
\begin{eqnarray}
\tilde{g}(\tau) & = & -i\frac{\int_{0}^{t}dt^{\prime}\sum_{k}|g_{k}|^{2}\exp[+i\phi_{k}(\tau+t-t^{\prime})]G(t^{\prime})}{\sqrt{1-|G(t)|^{2}}},\nonumber \\
 & = & \frac{-i}{\sqrt{1-|G(t)|^{2}}}\int_{0}^{t}dt^{\prime}f(\tau+t-t^{\prime})G(t^{\prime}).
\end{eqnarray}
From Eq. (\ref{ASolution-1}), $\tilde{a}^{\prime}(\tau)=-i\int_{0}^{\tau}G(\tau-\tau^{\prime})\tilde{g}(\tau^{\prime})d\tau^{\prime},$
delivering
\begin{equation}
\tilde{a}^{\prime}(\tau)=\frac{-\int_{0}^{\tau}d\tau^{\prime}\int_{0}^{t}dt^{\prime}f(\tau^{\prime}+t-t^{\prime})G(\tau-\tau^{\prime})G(t^{\prime})}{\sqrt{1-|G(t)|^{2}}},\label{Aprevio-1}
\end{equation}
which can be rewritten as \begin{subequations} \label{SolTildePrima-1}
\begin{equation}
\tilde{a}^{\prime}(\tau)=\frac{-G(t,\tau)}{\sqrt{1-|G(t)|^{2}}}.
\end{equation}
This equation defines the function $G(t,\tau).$ Moreover, from Eq.
(\ref{Ck-1}), the other coefficients read 
\begin{eqnarray}
\tilde{c}_{k}^{\prime}(\tau) & = & \frac{1}{\sqrt{1-|G(t)|^{2}}}\Big{\{}c_{k}(t)+ig_{k}^{\ast}\\
 &  & \times\int_{0}^{\tau}d\tau^{\prime}\exp[-i\phi_{k}(\tau^{\prime}+t)]G(\tau^{\prime},t)\Big{\}}.\nonumber 
\end{eqnarray}
The function $G(t,\tau),$ after a change of integration variables
in Eq. (\ref{Aprevio-1}), can be written as \end{subequations} 
\begin{equation}
G(t,\tau)=\int_{0}^{t}dt^{\prime}\int_{0}^{\tau}d\tau^{\prime}f(\tau^{\prime}+t^{\prime})G(t-t^{\prime})G(\tau-\tau^{\prime}).
\end{equation}
Eq. (\ref{SolTildePrima-1}) shows that $G(t,\tau)$\textit{\ measures
the probability of finding the system in the upper state at time }$\tau$\textit{\ given
that at the initial time }$t$\textit{\ it was in the ground state}
{[}Eq. (\ref{CITildePrimas-1}){]}.

\section{Calculation of the CPF correlation}

Here\ we explicitly calculate the CPF\ correlation defined as:
\begin{equation}
C_{pf}(t,\tau)|_{y}=\langle O_{z}O_{x}\rangle_{y}-\langle O_{z}\rangle_{y}\langle O_{x}\rangle_{y}.
\end{equation}
Equivalently, $C_{pf}(t,\tau)|_{y}=\sum_{zx}O_{z}O_{x}[P(z,x|y)-P(z|y)P(x|y)],$
for different possible measurement schemes. The conditional values
explicitly read
\begin{equation}
\langle O_{x}\rangle_{y}=\sum_{x=\pm1}xP(x|y),\ \ \ \ \ \ \ \langle O_{z}\rangle_{y}=\sum_{z=\pm1}zP(z|y),\label{Mean-1}
\end{equation}
and
\begin{equation}
\langle O_{z}O_{x}\rangle_{y}=\sum_{z,x=\pm1}zxP(z,x|y).\label{TwoMean-1}
\end{equation}
Furthermore, $P(z|y)=\sum_{x=\pm1}P(z,x|y),$ and $P(x|y)=\sum_{z=\pm1}P(z,x|y).$
Measurement outcomes are indicated by $x,$ $y,$ and $z,$ while
directions in Bloch sphere are given by the eigenvectors of the Pauli
matrices, $\sigma_{x},$ $\sigma_{y},$ and $\sigma_{z}.$

\subsection{First scheme, measurements $\sigma_{z}-\sigma_{z}-\sigma_{z}$}

The three measurements necessary to obtain the CPF correlations are
performed in in the same $\sigma_{z}-$direction, with corresponding
measurement projectors $\Pi_{+1}=\ket{\uparrow}\bra{\uparrow}$ and
$\Pi_{-1}=\ket{\downarrow}\bra{\downarrow}$. The initial condition
is taken as
\begin{equation}
|\Psi_{0}\rangle=(a\ket{\uparrow}+b\ket{\downarrow})\otimes|0\rangle.
\end{equation}

After the first $x$-measurement (measurement in the past), the total
state suffers the transformation $|\Psi_{0}\rangle\rightarrow|\Psi_{0}^{x}\rangle=\Pi_{\hat{z}=x}|\Psi_{0}\rangle/\sqrt{\langle\Psi_{0}|\Pi_{\hat{z}=x}|\Psi_{0}\rangle}$
delivering $(x=\pm1)$
\begin{equation}
|\Psi_{0}^{x}\rangle=|x\rangle\otimes|0\rangle,
\end{equation}
where we disregarded a global phase contribution. The probability
of each option $P(x)=\langle\Psi_{0}|\Pi_{\hat{z}=x}|\Psi_{0}\rangle,$
reads
\begin{equation}
P(x=+1)=|a|^{2},\ \ \ \ \ P(x=-1)=|b|^{2}.
\end{equation}

After the x-measurement, the system and environment evolve with the
Hamiltonian dynamics during a time interval $t,$ $|\Psi_{0}^{x}\rangle\rightarrow|\Psi_{t}^{x}\rangle.$
We get,
\begin{equation}
\begin{tabular}{cccc}
 \ensuremath{x}  &  \vline\   &  \ensuremath{|\Psi_{t}^{x}\rangle}  &  \vline\\
\hline  \ensuremath{+}  &  \vline\   &  \ensuremath{[a(t)\ket{\uparrow}+\ket{\downarrow}\sum_{k}c_{k}(t)b_{k}^{\dag}]|0\rangle}  &  \vline\\
 \ensuremath{-}  &  \vline\   &  \ensuremath{\ket{\downarrow}\otimes|0\rangle}  &  \vline
\end{tabular},\label{PsiTime-1}
\end{equation}
with $a(0)=1,$ $c_{k}(0)=0$ and normalization $|a(t)|^{2}+\sum_{k}|c_{k}(t)|=1.$
Thus, from Eq. (\ref{CIFirst-1}), these coefficients are explicitly
given by Eq. (\ref{TimeSolution-1}).

Posteriorly, the second $y$-measurement, correspondent to the present,
is performed. The conditional probability of outcomes $y,$ given
the previous outcomes $x,$ is given by $P(y|x)=\langle\Psi_{t}^{x}|\Pi_{\hat{z}=y}|\Psi_{t}^{x}\rangle.$
The joint probability of both outcomes is $P(y,x)=P(y|x)P(x).$ The
retrodicted probability of past outcomes given the present ones is
$P(x|y)=P(y,x)/P(y),$ where $P(y)=\sum_{x}P(y,x).$ We get
\begin{equation}
\begin{tabular}{ccccccccc}
 \ensuremath{y}  &  \ensuremath{x}  &  \vline\   &  \ensuremath{P(y|x)}  &  \vline\   &  \ensuremath{P(y,x)}  &  \vline\   &  \ensuremath{P(x|y)}  &  \vline\\
\hline  \ensuremath{+}  &  \ensuremath{+}  &  \vline\   &  \ensuremath{|a(t)|^{2}}  &  \vline\   &  \ensuremath{|G(t)|^{2}|a|^{2}}  &  \vline\   &  \ensuremath{1}  &  \vline\\
 \ensuremath{+}  &  \ensuremath{-}  &  \vline\   &  \ensuremath{0}  &  \vline\   &  \ensuremath{0}  &  \vline\   &  \ensuremath{0}  &  \vline\\
 \ensuremath{-}  &  \ensuremath{+}  &  \vline\   &  \ensuremath{1-|a(t)|^{2}}  &  \vline\   &  \ensuremath{(1-|G(t)|^{2})|a|^{2}}  &  \vline\   &  \ensuremath{\frac{(1-|G(t)|^{2})|a|^{2}}{(1-|G(t)|^{2})|a|^{2}+|b|^{2}}}  &  \vline\\
 \ensuremath{-}  &  \ensuremath{-}  &  \vline\   &  \ensuremath{1}  &  \vline\   &  \ensuremath{|b|^{2}}  &  \vline\   &  \ensuremath{\frac{|b|^{2}}{(1-|G(t)|^{2})|a|^{2}+|b|^{2}}}  &  \vline
\end{tabular}.\label{P(x|y)-1}
\end{equation}

After the second measurement, the total state suffer the transformation
$|\Psi_{t}^{x}\rangle\rightarrow|\Psi_{t}^{yx}\rangle=\Pi_{\hat{z}=y}|\Psi_{t}^{x}\rangle/\sqrt{\langle\Psi_{t}^{x}|\Pi_{\hat{z}=y}|\Psi_{t}^{x}\rangle}.$
Posteriorly, starting at time $t,$ $|\Psi_{t}^{yx}\rangle$ evolves
with the total unitary dynamics during a time interval $\tau,$ leading
to the transformation $|\Psi_{t}^{yx}\rangle\rightarrow|\Psi_{t+\tau}^{yx}\rangle.$
From Eq. (\ref{PsiTime-1}) the states conditioned to the output of
each measurement are
\begin{equation}
\begin{tabular}{ccccccc}
 \ensuremath{y}  &  \ensuremath{x}  &  \vline\   &  \ensuremath{|\Psi_{t}^{yx}\rangle}  &  \vline\   &  \ensuremath{|\Psi_{t+\tau}^{yx}\rangle}  &  \vline\\
\hline  \ensuremath{+}  &  \ensuremath{+}  &  \vline\   &  \ensuremath{\ket{\uparrow}\otimes|0\rangle}  &  \vline\   &  \ensuremath{[\tilde{a}(\tau)\ket{\uparrow}+\ket{\downarrow}\sum_{k}\tilde{c}_{k}(\tau)b_{k}^{\dag}]|0\rangle}  &  \vline\\
 \ensuremath{+}  &  \ensuremath{-}  &  \vline\   &  \ensuremath{\nexists}  &  \vline\   &  \ensuremath{\nexists}  &  \vline\\
 \ensuremath{-}  &  \ensuremath{+}  &  \vline\   &  \ensuremath{\ket{\downarrow}\frac{\sum_{k}c_{k}(t)b_{k}^{\dag}|0\rangle}{\sqrt{1-|a(t)|^{2}}}}  &  \vline\   &  \ensuremath{[\tilde{a}^{\prime}(\tau)\ket{\uparrow}+\ket{\downarrow}\sum_{k}\tilde{c}_{k}^{\prime}(\tau)b_{k}^{\dag}]|0\rangle}  &  \vline\\
 \ensuremath{-}  &  \ensuremath{-}  &  \vline\   &  \ensuremath{\ket{\downarrow}\otimes|0\rangle}  &  \vline\   &  \ensuremath{\ket{\downarrow}\otimes|0\rangle}  &  \vline
\end{tabular}.
\end{equation}
The solution form $(y,x)=(+,+)$ comes from Eq.~(\ref{CITilde-1})
{[}solutions~(\ref{SolTilde-1}){]}, while for $(y,x)=(-,+)$ follows
from Eq.~(\ref{CITildePrimas-1}) {[}solutions~(\ref{SolTildePrima-1}){]}.

Finally, the third $z$-measurement is performed (measurement in the
future). The probability $P(z|yx)$ of outcome $z$ given the previous
outcomes $y$ and $x,$ is given by $P(z|yx)=\langle\Psi_{t+\tau}^{yx}|\Pi_{\hat{z}=z}|\Psi_{t+\tau}^{yx}\rangle.$
The conditional probability of past and future event is $P(z,x|y)=P(z|y,x)P(x|y),$
where $P(x|y)$ follows from Eq.~(\ref{P(x|y)-1}). We get 
\begin{equation}
\begin{tabular}{cccccccc}
 \ensuremath{z}  &  \ensuremath{y}  &  \ensuremath{x}  &  \vline\   &  \ensuremath{P(z|y,x)}  &  \vline\   &  \ensuremath{P(z,x|y)}  &  \vline\\
\hline  \ensuremath{+}  &  \ensuremath{+}  &  \ensuremath{+}  &  \vline\   &  \ensuremath{|\tilde{a}(\tau)|^{2}}  &  \vline\   &  \ensuremath{|G(\tau)|^{2}}  &  \vline\\
 \ensuremath{+}  &  \ensuremath{+}  &  \ensuremath{-}  &  \vline\   &  \ensuremath{0}  &  \vline\   &  \ensuremath{0}  &  \vline\\
 \ensuremath{+}  &  \ensuremath{-}  &  \ensuremath{+}  &  \vline\   &  \ensuremath{|\tilde{a}^{\prime}(\tau)|^{2}}  &  \vline\   &  \ensuremath{\frac{|G(t,\tau)|^{2}|a|^{2}}{(1-|G(t)|^{2})|a|^{2}+|b|^{2}}}  &  \vline\\
 \ensuremath{+}  &  \ensuremath{-}  &  \ensuremath{-}  &  \vline\   &  \ensuremath{0}  &  \vline\   &  \ensuremath{0}  &  \vline\\
 \ensuremath{-}  &  \ensuremath{+}  &  \ensuremath{+}  &  \vline\   &  \ensuremath{1-|\tilde{a}(\tau)|^{2}}  &  \vline\   &  \ensuremath{1-|G(\tau)|^{2}}  &  \vline\\
 \ensuremath{-}  &  \ensuremath{+}  &  \ensuremath{-}  &  \vline\   &  \ensuremath{0}  &  \vline\   &  \ensuremath{0}  &  \vline\\
 \ensuremath{-}  &  \ensuremath{-}  &  \ensuremath{+}  &  \vline\   &  \ensuremath{1-|\tilde{a}^{\prime}(\tau)|^{2}}  &  \vline\   &  \ensuremath{\frac{(1-|G(t,\tau)|^{2}-|G(t)|^{2})|a|^{2}}{(1-|G(t)|^{2})|a|^{2}+|b|^{2}}}  &  \vline\\
 \ensuremath{-}  &  \ensuremath{-}  &  \ensuremath{-}  &  \vline\   &  \ensuremath{1}  &  \vline\   &  \ensuremath{\frac{|b|^{2}}{(1-|G(t)|^{2})|a|^{2}+|b|^{2}}}  &  \vline
\end{tabular}.\label{P(zx|y)-1}
\end{equation}
The conditional probability of the last measurement follows from $P(z|y)=\sum_{x}P(z,x|y),$
delivering
\begin{equation}
\begin{tabular}{ccccc}
 \ensuremath{z}  &  \ensuremath{y}  &  \vline\   &  \ensuremath{P(z|y)}  &  \vline\\
\hline  \ensuremath{+}  &  \ensuremath{+}  &  \vline\   &  \ensuremath{|G(\tau)|^{2}}  &  \vline\\
 \ensuremath{+}  &  \ensuremath{-}  &  \vline\   &  \ensuremath{\frac{|G(t,\tau)|^{2}|a|^{2}}{(1-|G(t)|^{2})|a|^{2}+|b|^{2}}}  &  \vline\\
 \ensuremath{-}  &  \ensuremath{+}  &  \vline\   &  \ensuremath{1-|G(\tau)|^{2}}  &  \vline\\
 \ensuremath{-}  &  \ensuremath{-}  &  \vline\   &  \ensuremath{\frac{(1-|G(t,\tau)|^{2}-|G(t)|^{2})|a|^{2}+|b|^{2}}{(1-|G(t)|^{2})|a|^{2}+|b|^{2}}}  &  \vline
\end{tabular}.\label{P(z|y)-1}
\end{equation}

From Eqs. (\ref{P(x|y)-1}) and (\ref{P(z|y)-1}), the expectation
values {[}Eqs.~(\ref{Mean-1}) and (\ref{TwoMean-1}){]} read
\begin{equation}
\langle O_{x}\rangle_{y=1}=1,\ \ \ \ \ \ \langle O_{z}\rangle_{y=1}=2|G(\tau)|^{2}-1,
\end{equation}
while from Eq. (\ref{P(zx|y)-1}) we get
\begin{equation}
\langle O_{z}O_{x}\rangle_{y=1}=2|G(\tau)|^{2}-1.
\end{equation}
Thus, it follows 
\begin{equation}
C_{pf}(t,\tau)|_{y=+1}=0.
\end{equation}
On the other hand, for $y=-1,$ the averages read 
\begin{equation}
\langle O_{x}\rangle_{y=-1}=\frac{(1-|G(t)|^{2})|a|^{2}-|b|^{2}}{(1-|G(t)|^{2})|a|^{2}+|b|^{2}},
\end{equation}
while
\begin{equation}
\langle O_{z}\rangle_{y=-1}=\frac{(2|G(t,\tau)|^{2}+|G(t)|^{2}-1)|a|^{2}-|b|^{2}}{(1-|G(t)|^{2})|a|^{2}+|b|^{2}},
\end{equation}
and
\begin{equation}
\langle O_{z}O_{x}\rangle_{y=-1}=\frac{(2|G(t,\tau)|^{2}+|G(t)|^{2}-1)|a|^{2}+|b|^{2}}{(1-|G(t)|^{2})|a|^{2}+|b|^{2}}.
\end{equation}
The CPF correlation then is
\begin{equation}
C_{pf}(t,\tau)|_{y=-1}\underset{\hat{z}\hat{z}\hat{z}}{=}\left\{ \frac{4|a|^{2}|b|^{2}}{[(1-|G(t)|^{2})|a|^{2}+|b|^{2}]^{2}}\right\} |G(t,\tau)|^{2}.\label{ApCPFZZZ-1}
\end{equation}

\subsection{Second scheme, \^{x}-\^{z}-\^{x}}

In this scheme, the first and last measurements are performed in $\hat{x}-$direction,
with measurement projector $\Pi_{\hat{x}+1}=|+\rangle\langle+|,$
and $\Pi_{\hat{x}-1}=|-\rangle\langle-|,$ where $|\pm\rangle=(1/\sqrt{2})(\ket{\uparrow}\pm\ket{\downarrow}).$
The intermediate one is realized in $\hat{z}-$direction, with projector
$\Pi_{\hat{z}=+1}$ and $\Pi_{\hat{z}=-1}$ defined above. The initial
system-environment state is
\begin{equation}
|\Psi_{0}\rangle=(a\ket{\uparrow}+b\ket{\downarrow})\otimes|0\rangle.
\end{equation}

After the first $x$-measurement $|\Psi_{0}\rangle\rightarrow|\Psi_{0}^{x}\rangle=\Pi_{\hat{x}=x}|\Psi_{0}\rangle/\sqrt{\langle\Psi_{0}|\Pi_{\hat{x}=x}|\Psi_{0}\rangle},$
the bipartite state is
\begin{equation}
|\Psi_{0}^{x}\rangle=\frac{\ket{\uparrow}+x\ket{\downarrow}}{\sqrt{2}}\otimes|0\rangle,
\end{equation}
where global phase contributions are disregarded. The probability
of each option $(x=\pm1)$ $P(x)=\langle\Psi_{0}|\Pi_{\hat{x}=x}|\Psi_{0}\rangle,$
reads
\begin{equation}
P(x)=\frac{1}{2}|a+xb|^{2}.
\end{equation}

After the previous step, $|\Psi_{0}^{x}\rangle$ evolves with the
unitary evolution during a time interval $t,$ $|\Psi_{0}^{x}\rangle\rightarrow|\Psi_{t}^{x}\rangle.$
Using the initial conditions (\ref{CIFirst-1}) and their associated
solution (\ref{TimeSolution-1}), we get
\begin{equation}
|\Psi_{t}^{x}\rangle=\frac{1}{\sqrt{2}}\Big{[}a(t)\ket{\uparrow}+x\ket{\downarrow}+\ket{\downarrow}\sum_{k}c_{k}(t)b_{k}^{\dag}\Big{]}|0\rangle,\label{PsiAfterXyTime-1}
\end{equation}
where $a(0)=1$ and $c_{k}(0)=0.$

Posteriorly, the second $y$-measurement is performed. The conditional
probability for the outcomes is $P(y|x)=\langle\Psi_{t}^{x}|\Pi_{\hat{z}=y}|\Psi_{t}^{x}\rangle,$
which deliver
\begin{equation}
P(+|x)=\frac{|a(t)|^{2}}{2},\ \ \ \ \ \ P(-|x)=1-\frac{|a(t)|^{2}}{2},
\end{equation}
where we used $|a(t)|^{2}+\sum_{k}|c_{k}(t)|^{2}=1.$ This result
indicates that the random variable $y$ is statistically independent
of $x,$ $P(y|x)=P(y).$ Thus, the joint probability for the first
and second outcomes is $P(y,x)=P(y|x)P(x)=P(y)P(x).$ The retrodicted
probability $P(x|y)=P(y,x)/P(y),$ where $P(y)=\sum_{x}P(y,x),$ becomes
\begin{equation}
P(x|y)=P(x).\label{Pretro-1}
\end{equation}

After the second measurement, the state suffers the transformation
$|\Psi_{t}^{x}\rangle\rightarrow|\Psi_{t}^{yx}\rangle=\Pi_{\hat{z}=y}|\Psi_{t}^{x}\rangle/\sqrt{\langle\Psi_{t}^{x}|\Pi_{\hat{z}=y}|\Psi_{t}^{x}\rangle}.$
From Eq.~(\ref{PsiAfterXyTime-1}), for $y=+1$ we get
\begin{equation}
|\Psi_{t}^{+,x}\rangle=\ket{\uparrow}\otimes|0\rangle,\label{CIPlus-1}
\end{equation}
while for $y=-1,$
\begin{equation}
|\Psi_{t}^{-,x}\rangle=\frac{1}{\sqrt{2-|a(t)|^{2}}}\ket{\downarrow}\otimes\Big{[}x+\sum_{k}c_{k}(t)b_{k}^{\dag}\Big{]}|0\rangle.\label{CIMinus-1}
\end{equation}

Starting at time $t,$ $|\Psi_{t}^{yx}\rangle$ evolves with the total
unitary dynamics during a time interval $\tau,$ leading to the transformation
$|\Psi_{t}^{yx}\rangle\rightarrow|\Psi_{t+\tau}^{yx}\rangle.$ From
Eq.~(\ref{CIPlus-1}) we get
\begin{equation}
|\Psi_{t+\tau}^{+,x}\rangle=\Big{[}\tilde{a}(\tau)\ket{\uparrow}+\ket{\downarrow}\sum_{k}\tilde{c}_{k}(\tau)b_{k}^{\dag}\Big{]}|0\rangle,
\end{equation}
with $\tilde{a}(0)=1,$ $\tilde{c}_{k}(0)=0$ {[}Eq. (\ref{CITilde-1}){]},
with $|\tilde{a}(\tau)|^{2}+\sum\nolimits _{k}|\tilde{c}_{k}(\tau)|^{2}=1.$
Thus, $\tilde{a}(\tau)$ and $\tilde{c}_{k}(\tau)$ are given by Eq.~(\ref{SolTilde-1}).
On the other hand, from Eq. (\ref{CIMinus-1}), it follows
\begin{eqnarray}
|\Psi_{t+\tau}^{-,x}\rangle & = & \frac{x\ket{\downarrow}\otimes|0\rangle}{\sqrt{2-|a(t)|^{2}}}+\sqrt{\frac{1-|a(t)|^{2}}{2-|a(t)|^{2}}}\\
 &  & \times\Big{[}\tilde{a}^{\prime}(\tau)\ket{\uparrow}+\ket{\downarrow}\sum_{k}\tilde{c}_{k}^{\prime}(\tau)b_{k}^{\dag}\Big{]}|0\rangle,\nonumber 
\end{eqnarray}
where $\tilde{a}^{\prime}(0)=0$ and $\tilde{c}_{k}^{\prime}(0)=c_{k}(t)/\sqrt{1-|a(t)|^{2}}$
{[}Eq.~(\ref{CITildePrimas-1}){]} with $|\tilde{a}^{\prime}(\tau)|^{2}+\sum\nolimits _{k}|\tilde{c}_{k}^{\prime}(\tau)|^{2}=1.$
In this case, $\tilde{a}^{\prime}(\tau)$ and $\tilde{c}_{k}^{\prime}(\tau)$
are then given by Eq.~(\ref{SolTildePrima-1}).

At the final stage, the third $z$-measurement is performed, where
the corresponding conditional probability reads $P(z|yx)=\langle\Psi_{t+\tau}^{yx}|\Pi_{\hat{x}=z}|\Psi_{t+\tau}^{yx}\rangle.$
From the previous expressions, we get
\begin{equation}
P(z|+,x)=\frac{1}{2},
\end{equation}
while
\begin{equation}
P(z|-,x)=\frac{1}{2}\Big{[}1-zx\frac{G(t,\tau)+G^{\ast}(t,\tau)}{2-|G(t)|^{2}}\Big{]}.
\end{equation}
The CPF probability $P(z,x|y)=P(z|y,x)P(x|y),$ from the previous
two expressions and Eq. (\ref{Pretro-1}), reads $(y=+1)$
\begin{equation}
P(z,x|+)=\frac{1}{2}P(x)=\frac{1}{4}|a+xb|^{2},\label{cpfPyPlus-1}
\end{equation}
while $(y=-1)$
\begin{equation}
P(z,x|-)=\frac{|a+xb|^{2}}{4}\Big{[}1-zx\frac{G(t,\tau)+G^{\ast}(t,\tau)}{2-|G(t)|^{2}}\Big{]}.\label{cpfPyMinus-1}
\end{equation}

From Eqs. (\ref{cpfPyPlus-1}) and (\ref{cpfPyMinus-1}), the conditional
expectation values {[}Eqs. (\ref{Mean-1}) and (\ref{TwoMean-1}){]}
for $y=+1$ read
\begin{equation}
\langle O_{x}\rangle_{y=+1}=2\mathrm{Re}(ab^{\ast}),\ \ \ \ \ \ \langle O_{z}\rangle_{y=+1}=0,
\end{equation}
and
\begin{equation}
\langle O_{z}O_{x}\rangle_{y=+1}=0,
\end{equation}
which implies 
\begin{equation}
C_{pf}(t,\tau)|_{y=+1}=0.
\end{equation}
On the other hand, for $y=-1,$ the averages read 
\begin{equation}
\langle O_{x}\rangle_{y=-1}=2\mathrm{Re}(ab^{\ast}),
\end{equation}
while
\begin{equation}
\langle O_{z}\rangle_{y=-1}=-2\mathrm{Re}(ab^{\ast})\frac{G(t,\tau)+G^{\ast}(t,\tau)}{2-|G(t)|^{2}}.
\end{equation}
Furthermore,
\begin{equation}
\langle O_{z}O_{x}\rangle_{y=-1}=-\frac{G(t,\tau)+G^{\ast}(t,\tau)}{2-|G(t)|^{2}}.
\end{equation}
The CPF correlation then is
\begin{equation}
C_{pf}(t,\tau)|_{y=-1}\underset{\hat{x}\hat{z}\hat{x}}{=}-\left\{ \frac{1-[2\mathrm{Re}(ab^{\ast})]^{2}}{1-|G(t)|^{2}/2}\right\} \mathrm{Re}[G(t,\tau)].
\end{equation}

For a$\ \hat{y}-\hat{z}-\hat{y}$ measurements scheme, by performing
a similar calculation, the CPF correlation reads
\begin{equation}
C_{pf}(t,\tau)|_{y=-1}\underset{\hat{y}\hat{z}\hat{y}}{=}-\left\{ \frac{1-[2\mathrm{Im}(ab^{\ast})]^{2}}{1-|G(t)|^{2}/2}\right\} \mathrm{Re}[G(t,\tau)].\label{ApCPFXZX-1}
\end{equation}

\section{Map representation \ of the total unitary dynamics}

For experimental implementation, the system is encoded in the light
polarization states, while the bath is effectively implemented through
different spatial light modes \cite{brasil,rio}.

The total unitary evolution in first interval $(0,t)$ can be written
as the map \begin{subequations} \label{TimeMap-1} 
\begin{eqnarray}
\ket{\downarrow}\otimes|0\rangle & \rightarrow & \ket{\downarrow}\otimes|0\rangle,\\
\ket{\uparrow}\otimes|0\rangle & \rightarrow & \cos(\theta)\ket{\uparrow}\otimes|0\rangle+\sin(\theta)\ket{\downarrow}\otimes|1\rangle,\ \ \ \ 
\end{eqnarray}
where here $|0\rangle$ and $|1\rangle$ represent spatial modes that
respectively take into account the absence or presence of one excitation
in the environment Bosonic modes. Thus, the angle $\theta$ is given
by the relation \end{subequations} 
\begin{equation}
\cos(\theta)=a(t)=G(t),
\end{equation}
where $a(t)$ follows from Eq. (\ref{TimeSolution-1}).

In the interval $(t,t+\tau)$ the total unitary dynamics realize the
following mapping \begin{subequations} \label{TauMap-1} 
\begin{eqnarray}
\ket{\downarrow}\otimes|0\rangle & \rightarrow & \ket{\downarrow}\otimes|0\rangle,\\
\ket{\uparrow}\otimes|0\rangle & \rightarrow & \cos(\tilde{\theta})\ket{\uparrow}\otimes|0\rangle+\sin(\tilde{\theta})\ket{\downarrow}\otimes|1\rangle,\\
\ket{\downarrow}\otimes|1\rangle & \rightarrow & \sin(\tilde{\theta}^{\prime})\ket{\uparrow}\otimes|0\rangle+\cos(\tilde{\theta}^{\prime})\ket{\downarrow}\otimes|1\rangle.\ \ \ \ 
\end{eqnarray}
The angles are given by the relations \end{subequations} 
\begin{equation}
\cos(\tilde{\theta})=\tilde{a}(\tau)=G(\tau),
\end{equation}
and
\begin{equation}
\sin(\tilde{\theta}^{\prime})=\tilde{a}^{\prime}(\tau)=-\frac{G(t,\tau)}{\sqrt{1-|G(t)|^{2}}},
\end{equation}
where $\tilde{a}(\tau)$\ and $\tilde{a}^{\prime}(\tau)$\ follows
from Eqs. (\ref{TimeSolution-1}) and (\ref{SolTildePrima-1}) respectively.

From the previous mapping, it is possible to rewrite the CPF correlation
in terms of angle variables. From Eq. (\ref{ApCPFZZZ-1}) we get
\begin{equation}
C_{pf}|_{y=-1}\underset{\hat{z}\hat{z}\hat{z}}{=}\left\{ \frac{4|a|^{2}|b|^{2}}{[\sin^{2}(\theta)|a|^{2}+|b|^{2}]^{2}}\right\} \sin^{2}(\theta)\sin^{2}(\tilde{\theta}^{\prime}),
\end{equation}
while from Eq. (\ref{ApCPFXZX-1}) it follows
\begin{equation}
C_{pf}|_{y=-1}\underset{\hat{x}\hat{z}\hat{x}}{=}\left\{ \frac{1-[2\mathrm{Re}(ab^{\ast})]^{2}}{1-\cos^{2}(\theta)/2}\right\} \sin(\theta)\sin(\tilde{\theta}^{\prime}).
\end{equation}
These two expressions do not depend on angle $\tilde{\theta}.$ In
fact, this angle is relevant when $y=+1,$ where $C_{pf}|_{y=+1}\underset{\hat{z}\hat{z}\hat{z}}{=}0$
and $C_{pf}|_{y=+1}\underset{\hat{x}\hat{z}\hat{x}}{=}0.$

The previous expressions for the CPF\ correlation in terms of angle
variables can also be derived from the measurement schemes and by
using the dynamical maps Eqs. (\ref{TimeMap-1}) and (\ref{TauMap-1}).
For example, the CPF probability $P(z,x|y)$ for the $\hat{z}-\hat{z}-\hat{z}$
scheme {[}compare with Eq.~(\ref{P(zx|y)-1}){]} reads
\begin{equation}
\begin{tabular}{cccccc}
 \ensuremath{z}  &  \ensuremath{y}  &  \ensuremath{x}  &  \vline\   &  \ensuremath{P(z,x|y)}  &  \vline\\
\hline  \ensuremath{+}  &  \ensuremath{+}  &  \ensuremath{+}  &  \vline\   &  \ensuremath{\cos^{2}(\tilde{\theta})}  &  \vline\\
 \ensuremath{+}  &  \ensuremath{+}  &  \ensuremath{-}  &  \vline\   &  \ensuremath{0}  &  \vline\\
 \ensuremath{+}  &  \ensuremath{-}  &  \ensuremath{+}  &  \vline\   &  \ensuremath{\frac{\sin^{2}(\theta)\sin^{2}(\tilde{\theta}^{\prime})|a|^{2}}{\sin^{2}(\theta)|a|^{2}+|b|^{2}}}  &  \vline\\
 \ensuremath{+}  &  \ensuremath{-}  &  \ensuremath{-}  &  \vline\   &  \ensuremath{0}  &  \vline\\
 \ensuremath{-}  &  \ensuremath{+}  &  \ensuremath{+}  &  \vline\   &  \ensuremath{\sin^{2}(\tilde{\theta})}  &  \vline\\
 \ensuremath{-}  &  \ensuremath{+}  &  \ensuremath{-}  &  \vline\   &  \ensuremath{0}  &  \vline\\
 \ensuremath{-}  &  \ensuremath{-}  &  \ensuremath{+}  &  \vline\   &  \ensuremath{\frac{\sin^{2}(\theta)\cos^{2}(\tilde{\theta}^{\prime})|a|^{2}}{\sin^{2}(\theta)|a|^{2}+|b|^{2}}}  &  \vline\\
 \ensuremath{-}  &  \ensuremath{-}  &  \ensuremath{-}  &  \vline\   &  \ensuremath{\frac{|b|^{2}}{\sin^{2}(\theta)|a|^{2}+|b|^{2}}}  &  \vline
\end{tabular}.
\end{equation}
For the $\hat{x}-\hat{z}-\hat{x}$ scheme {[}compare with Eqs.~(\ref{cpfPyPlus-1})
and (\ref{cpfPyMinus-1}){]} it can be written as $(y=+1)$
\begin{equation}
P(z,x|+)=\frac{1}{4}|a+xb|^{2},
\end{equation}
while $(y=-1)$ 
\begin{equation}
P(z,x|-)=\frac{1}{4}|a+xb|^{2}\Big{[}1+2zx\frac{\sin(\theta)\sin(\tilde{\theta}^{\prime})}{2-\cos(\theta)}\Big{]}.
\end{equation}

\ihead[]{Acknowledgments}

\ohead[Acknowledgments]{}

\cleardoublepage{}

\ihead{ }

\ohead[]{\leftmark}
\chead{}

\bibliographystyle{thais}
\phantomsection\addcontentsline{toc}{chapter}{\bibname}\bibliography{ref}

\cleardoublepage{}

\ihead[]{Nomenclature}

\ohead[Nomenclature]{}

\printnomenclature[2.5cm]
\end{document}